\documentclass[journal]{IEEEtran}
\usepackage[margin=0.65in]{geometry}
\usepackage[T1]{fontenc}
\usepackage{blindtext}
\usepackage{graphicx}
\usepackage[latin9]{inputenc}
\usepackage{verbatim}
\usepackage{float}
\usepackage{amsthm}
\usepackage{enumerate}
\usepackage{amsmath,epsfig,pifont,geometry,hyperref}
\usepackage{amssymb}

\let\emptyset\varnothing
\usepackage{cite}
\usepackage{tikz}
\usepackage{algorithm}
\usepackage{multirow}
\usepackage{algpseudocode}
\usepackage{array}
\usepackage{booktabs}
\usepackage[skip=2pt]{caption}
\usepackage[labelformat=simple]{subcaption}

\usepackage{tkz-euclide}
\usepackage[protrusion=true,expansion=true]{microtype}
\pdfoutput=1
\usepackage{pgfplots}
\usepackage{pgfplotstable}
\usepackage{array}
\usepackage{colortbl}
\usepackage{makecell}
\usepackage{color}
\usepackage{tkz-euclide}
\usepackage[version=3]{mhchem}
\usepackage{mathtools}
\usepackage{enumitem}\setlist[description]{font=\textendash\enskip\scshape\bfseries}
\usepackage{cleveref}
\usepackage{dsfont}
\usepackage{enumitem}
\usepackage{svg}
\usepackage{amsmath,bm}
\usepackage{amsfonts}
\usepackage{amssymb}
\usepackage{mathrsfs}

\usetikzlibrary{chains,fit,shapes}
\usetikzlibrary{arrows,calc,shapes,decorations.pathreplacing}
\usepgfplotslibrary{fillbetween}
\usepgfplotslibrary{statistics}

\newcommand{\vast}{\bBigg@{4}}
\newcommand{\Vast}{\bBigg@{5}}

\usepackage{hyperref}
\hypersetup{%
  colorlinks=false,
  urlbordercolor=red,
  pdfborder={0 0 0.6}
}
\usepackage{xcolor}
\pgfplotstableset{
  every head row/.style={before row=\toprule,after row=\midrule},
  every last row/.style={after row=\bottomrule},
  fixed,precision=2,
}
\pgfplotsset{compat=1.17}

\title{Decomposition~Theory~Meets~Reliability~Analysis: Processing~of~Computation-Intensive~Dependent Tasks~over~Vehicular~Clouds~with~Dynamic~Resources}

\author{Payam Abdisarabshali,~\IEEEmembership{Student,~IEEE},  Minghui Liwang,~\IEEEmembership{Member,~IEEE},\\ Amir Rajabzadeh, Mahmood Ahmadi, and Seyyedali Hosseinalipour,~\IEEEmembership{Member,~IEEE}}

\begin{document}

\maketitle

\newtheorem{lemma}{Lemma}

\theoremstyle{theorem}
\newtheorem{theorem}{Theorem}

\theoremstyle{proposition}
\newtheorem{proposition}{Proposition}

\theoremstyle{problem}
\newtheorem{problem}{Problem}

\theoremstyle{Assumption}
\newtheorem{assumption}{Assumption}

\theoremstyle{corollary}
\newtheorem{corollary}{Corollary}

\newtheorem{remark}{Remark}

\theoremstyle{theorem}
\newtheorem{definition}{Definition}

\newtheorem{example}{Example}

\providecommand{\keywords}[1]{\textbf{\textit{Index Terms---}}#1}

\definecolor{Gray}{gray}{0.9}
\setlength{\textfloatsep}{0.05cm}

\begin{abstract}
Vehicular cloud (VC) is a promising technology for processing computation-intensive applications (CI-Apps) on smart vehicles. Implementing VCs over the network edge faces two key challenges: (C1) On-board computing resources of a single vehicle are often insufficient to process a CI-App; (C2) The dynamics of available resources, caused by vehicles' mobility, hinder reliable CI-App processing. This work is among the first to jointly address (C1) and (C2), while considering two common CI-App graph representations, directed acyclic graph (DAG) and undirected graph (UG). To address (C1), we consider partitioning a CI-App with $m$ dependent (sub-)tasks into $k\le m$ groups, which are dispersed across vehicles. To address (C2), we introduce a generalized reliability metric called \textit{conditional mean time to failure} (C-MTTF). Subsequently, we increase the C-MTTF of dependent sub-tasks processing via introducing a general framework of \underline{r}edundancy-based \underline{p}rocessing of dependent sub-tasks over semi-dynamic \underline{VC}s ({\tt RP-VC}). We demonstrate that {\tt RP-VC} can be modeled as a non-trivial semi-Markov process (SMP). To analyze this SMP model and its reliability, we develop a novel mathematical framework, called \textit{event stochastic algebra} ($\langle e\rangle$-algebra). Based on $\langle e\rangle$-algebra, we propose \textit{decomposition theorem} (DT) to transform the presented SMP to a decomposed SMP (D-SMP). We subsequently calculate the {C-MTTF} of our methodology. We demonstrate that $\langle e\rangle$-algebra and DT are general mathematical tools that can be used to analyze other cloud-based networks. Simulation results reveal the exactness of our analytical results and the efficiency of our methodology in terms of acceptance and success rates of CI-App processing.
\end{abstract}

\begin{keywords}
Event stochastic algebra, decomposition theory, vehicular cloud, semi-Markov process, stochastic analysis, reliable service provisioning, directed acyclic graphs (DAG) tasks/applications, undirected graph (UG) tasks/applications.
\end{keywords}

\IEEEpeerreviewmaketitle
\section{Introduction}
\noindent\IEEEPARstart{R}{e}cent years have witnessed explosive growth in the number of smart vehicles equipped with powerful on-board processors. Researchers have subsequently promoted the utilization of on-board computing resources of vehicles for innovative Internet of Things (IoT) applications (e.g., self-driving, augmented reality, and crowd processing)\cite{42,43,45}, which are predominantly computation-intensive and energy-hungry. Nevertheless, the limited on-board computing and storage resources of a single vehicle may fail to meet the execution demands of computation-intensive applications (CI-App), e.g., augmented reality. One method to overcome this limitation is to offload application data from vehicles to centralized cloud  servers\cite{4,5} or edge cloud servers, enabled via technologies such as fog computing\cite{8}, through vehicle-to-infrastructure (V2I) connections. Although this can potentially alleviate the computing resource shortage of vehicles, application offloading through V2I and data relaying across the core network may result in  extra delay, network congestion, and
excessive network resource (e.g., computing and spectrum) utilization \cite{7}.\par

To alleviate the aforementioned issues, vehicular cloud (VC) has been introduced, which orchestrates the distributed and dynamic resources (e.g., processors, storage, and sensors) of smart vehicles in a cooperative manner through both vehicle-to-vehicle (V2V) and V2I communications\cite{12}. VC has attracted tremendous attentions in supporting various services such as traffic management and entertainment (e.g., VR gaming) \cite{7,14,16}. Through exploiting V2V connections and cooperative application processing over the vehicles, VC reduces the service latency and overhead on backhaul links\cite{16}.\par
\vspace{-2.6mm}
\subsection{Motivations and Challenges}
\vspace{-.15mm}
Implementing VC on the network edge faces two challenges.
\begin{enumerate}[label={(C\arabic*)}]
  \item \label{Ch2} \textit{Heterogeneity and limitation of computing resources across vehicles:} Computing resources are often varying across the vehicles. Particularly, a vehicle may suffer from resource limitations to process a CI-Apps, such as augmented reality \cite{45} and data analysis \cite{46}, which can degrade the users' quality of experience (QoE).
  \item \label{Ch1} \textit{Dynamics and volatility of available computing resources:} Vehicles may enter and leave the VC due to their mobility, imposing new challenges in CI-App processing ranging from handling the highly dynamic resources \cite{17,36} to ensuring reliable resource provisioning \cite{24}.
\end{enumerate}
 \textbf{Addressing \ref{Ch2}.} One promising approach to tackle \ref{Ch2} is partitioning a CI-App into smaller \textit{dependent sub-tasks} that can be dispersed across vehicles \cite{42}. Dependent sub-tasks structure can be used to represent a wide range of applications (also known as graph-structured applications) \cite{21,26,42,45,46}, and have been recognized and incorporated into modern computing architectures (e.g., micro-services architecture offered via IBM\cite{26}). Execution of a graph-structured application requires processing its constituent dependent sub-tasks, where processing a sub-task is contingent on the reception of data from other sub-tasks. Hence, the failure of processing of one sub-task can lead to a chain of failures of other dependent sub-tasks\cite{21}. Hence, although partitioning an CI-App into multiple dependent sub-tasks can improve the system performance in terms of CI-App execution time (e.g., due to possible parallel processing of a fraction of sub-tasks) and   computing resource utilization \cite{45,46}, reliability enhancement (e.g., reducing the chance of failure of CI-App processing) remains an open challenge.\par
\textbf{Addressing \ref{Ch1}.} Studying the reliability of application processing in the VC is still in its early development stage and only handful of research works \cite{24,29,30,31,32} have been devoted to address \ref{Ch1} to mitigate the effect of fluctuations in available resources. These works aim to increase the \textit{mean time to failure (MTTF)} of application processing, which is one of the well-known fault tolerance metrics defined as the expected time it will take for the system to encounter a failure. The main assumption of these works is that a single vehicle has sufficient computing resources to process applications of any sizes, which is unrealistic in case of having CI-Apps.\par
\textbf{Reliability Analysis under Application Partitioning.} None of the existing works have so far jointly addressed \ref{Ch2} and \ref{Ch1}. Motivated by this, we tackle the following research questions:
\begin{enumerate}[label={(Q\arabic*)}]
    \item \label{Q1} How to enhance the processing reliability (e.g., increasing the \textit{MTTF}) in a semi-dynamic VC (e.g., a VC over a parking lot) when a CI-App with $m$ dependent sub-tasks is partitioned into $k\le m$ groups and processed distributedly on different vehicles?
    \item \label{Q2} How application partitioning and dispersing the corresponding sub-tasks can affect the reliability of CI-App processing over a semi-dynamic VC?
\end{enumerate}\par
To answer \ref{Q1}, we propose  a general framework for \underline{r}edundancy-based \underline{p}rocessing of dependent sub-tasks over semi-dynamic \underline{VC}s ({\tt RP-VC}). Roughly speaking, {\tt RP-VC} (i) disperses the processing of an application with $m$ dependent sub-tasks, modeled as a directed acyclic graph (DAG) or undirected graph (UG), over a semi-dynamic VC through application partitioning and (ii) enhances the reliability of the CI-App processing via considering redundancy-based resource provisioning to mitigate the impact of vehicles' unexpected departures from the VC.\par
To answer \ref{Q2}, we present mathematical modeling to analyze the reliability of dependent sub-tasks processing over a semi-dynamic VC, which has been an open problem and is the key contribution of our work. In particular, we introduce a reliability metric, called \textit{conditional mean time to failure (C-MTTF)}, to quantify the reliability of processing  DAG- and UG-structured CI-App. We next propose a non-trivial semi-Markov process (SMP) to model {\tt RP-VC}. We then introduce a novel mathematical framework, which we refer to as \textit{event stochastic algebra ($\langle e\rangle$-algebra)}, which makes the reliability analysis of dependent sub-tasks processing tractable. We then exploit $\langle e\rangle$-algebra and introduce \textit{decomposition theorem (DT)} to transform the proposed non-trivial SMP into a decomposed SMP (D-SMP). We obtain the closed-from expression of \textit{C-MTTF} through in-depth analysis of D-SMP. We will show that $\langle e\rangle$-algebra and DT are general mathematical tools that can be utilized for other problems in cloud-based systems.\par
It is worth mentioning that the term "decomposition" has been used in other math domains, e.g., in optimization \cite{47,48}, where an optimization problem is decomposed into several simpler sub-problems. In this paper, we develop a new notion of "decomposition", where our theory enables us to disentangle the states of $\beta$-SMP into several simpler sub-states.

\vspace{-4mm}
\subsection{Related Work}\label{relatedWorks}
VC architectures considered in literature can be roughly divided into two categories: semi-dynamic (e.g., parked vehicles) and dynamic (e.g., moving vehicles). Service provisioning through semi-dynamic VCs \cite{12} has been widely studied in literature. For example, research works \cite{24,28,29,30,31,32} have considered VCs in parking lots formed via parked vehicles. Also, dynamic VC architectures for moving vehicles have also been investigated\cite{1,33,34,35,36,42}, where most of the works aim to study the impact of the mobility patterns of the vehicles on application processing. The above-mentioned works have tackled various problems in VCs, such as resource provisioning, application partitioning, and reliability, as discussed below.\par
\subsubsection{Resource Provisioning}
There exist several works dedicated to addressing resource provisioning problems in VCs. Works \cite{34,35} proposed innovative resource provisioning strategies based on semi-Markov decision processes (SMDP). The authors in \cite{36} have modeled the VC as a cluster of connected vehicles, where a head vehicle supervisions and allocates resources in each cluster. However, these studies focused on resource provisioning without considering the reliability aspects and application partitioning.\par
\subsubsection{Application Partitioning}
Several studies have tackled effective  partitioning of CI-Apps into smaller sub-tasks \cite{42,44}.
Authors in \cite{42} have studied the joint application partitioning and power control problem in a fog computing network to optimize the long-term system utility measured in terms of execution delay and energy consumption. In \cite{44}, the authors considered a multi-user application partitioning in industrial mobile edge computing (MEC) systems, where the workload of a vehicle is partitioned and offloaded to rented edge devices in a MEC platform. However, works in this literature, none of the conducted research has dealt with the effect of application partitioning on reliability metrics, e.g., \textit{MTTF}.\par
\subsubsection{Reliability}
The dynamics of VCs has made their reliability analysis a vital research topic\footnote{The departure of a vehicle from the VC leads to a failure of the applications/tasks offloaded to the vehicle.} \cite{24,29}. Although reliability of application processing over VCs has not been investigated profoundly, some prior works have taken initial steps toward this direction via proposing a variety of methods to improve the reliability of VCs, such as migration of applications (i.e., migrating applications of the imminent leaving vehicle to the nearby
vehicles) \cite{33,39}, checkpointing (i.e., storing snapshots of the application's state at multiple time-stamps, providing an opportunity to recover the application if failures occur) \cite{40} and redundancy-based execution (i.e., multiple images/replicas of the application are executed among multiple vehicles to provide robustness against vehicles departure) \cite{24,29,30,31,32}.\par
Among these methods, redundancy-based strategies are of particular interest\footnote{Designing efficient VM migration strategies faces significant challenges in the VCs with unpredictable vehicular sojourn times since predicting the optimal moment for conducting VM migration is non-trivial\cite{29}. Moreover, checkpointing adds considerable overhead to the VC\cite{40}.} to enhance the \textit{MTTF} of application execution\cite{24,29,30,31,32}. One of the primary works has proposed two strategies called {\tt J}$_2$ and {\tt J}$_3$ to mitigate the impact of  VC dynamics on application processing\cite{24}. These strategies increase the \textit{MTTF} by allocating two/three vehicles to an application. Authors in \cite{24} also provided mathematical models based on a semi-Markov process (SMP) to calculate the \textit{MTTF}. The proposed models in \cite{24} require complete information on the probability distributions of the vehicles' sojourn times (i.e., the time during which a vehicle is parked in a parking lot) and recruitment duration (i.e., the duration of time required to recruit a new vehicle). To relax these requirements, in a follow-up work, authors in \cite{30} provided a methodology to estimate the \textit{MTTF} of  strategies {\tt J}$_2$ and {\tt J}$_3$. In addition, a mathematical model is presented to estimate the completion time of an application under the supervision of {\tt J}$_2$ in~\cite{31}. Further, the authors in \cite{32} extended the framework of \cite{31} to obtain a more accurate estimation of application completion time. Finally, a generalization of {\tt J}$_2$ to {\tt J}$_n$ is proposed in \cite{29}, in which $n$ vehicles simultaneously process the same application, where each vehicle's sojourn and recruitment times are modeled via exponential random variables\cite{29}. Strategies {\tt J}$_2$ and {\tt J}$_n$ enhance the reliability of application processing as follows.\par
{\tt J}$_2$: In this technique, two vehicles ($C_1$, $C_2$) are assigned to an application, which process it independently and simultaneously. When one of the two vehicles leaves the VC coverage area, the other one pauses processing its application and starts recruiting a new vehicle. The application processing will be resumed on both vehicles (i.e., the recruiter and recruited vehicles) if the recruitment operation completes successfully.\par
{\tt J}$_n$: The procedure conducted to enhance the reliability in this strategy is almost equivalent to {\tt J}$_2$, except that it assigns $n$ different vehicles to an application. If the first vehicle leaves the VC coverage area, strategy {\tt J}$_n$ first considers one of the other $n-1$ vehicles as the recruiter (namely $C_r$) and defines a variable $h=0$. $C_r$ starts recruiting a new vehicle and increases $h$ by one. If another vehicle leaves the VC coverage area during recruitment, $C_r$ increases $h$ by one and recruits two vehicles. Likewise, if the third vehicle leaves the VC during recruitment, $C_r$ recruits three vehicles. This process continues until the recruitment operation completes by recruiting $h$ new vehicles or all of the $n$ vehicles leave the VC, in which case the processing of the application encounters a failure.\par
The only application model adopted in the literature of redundancy-based application execution in VCs is {\tt J}$_n$ class of strategies, the processing of which is different\footnote{{\tt J}$_n$ class of strategies offload the entire application to the vehicles without considering the application's internal structure.} than applications consisting of multiple dependent sub-tasks dispersing across different vehicles. In particular, upon execution of a CI-App with multiple dependent sub-tasks, failure of a single sub-task can lead to the failure of the entire application.
To the best of our knowledge, none of the conducted studies have investigated the reliability of processing dependent sub-tasks in a VC.
\vspace{-2mm}
\subsection{Outline and Summary of Contributions}
Our major contributions can be summarizes as follows:
\begin{description}[font=$\bullet$~\normalfont]
  \item We model an application using a general graph representation enfolding DAG and UG (Sec. \ref{taskModel}). We propose {\tt RP-VC}, which is the first unified framework for studying the \textit{MTTF} of redundancy-based processing of an application with $m$ dependent sub-tasks, modeled as a DAG or UG, partitioned into $k\le m$ groups, in semi-dynamic VCs (Sec.~\ref{porposedStrategy}).
  \item We show that computing the \textit{MTTF} of processing of an application modeled as a DAG is non-trivial. We then introduce an extension of \textit{MTTF}, called \textit{conditional MTTF (C-MTTF)},
  which can quantify the reliability of processing an application
  modeled as both  DAG and UG (Sec. \ref{probDef}).
  \item To model the dynamics of our system, we introduce a general model for a class of stochastic systems, which we call \textit{stochastic event system (SeS)}. We then present a new concept, called $\beta$-inhomogeneous, to characterize SeS (Sec.~\ref{BSMPR2n}). We demonstrate that a semi-dynamic VC under {\tt RP-VC} is an SeS, which we refer to as ${\mathsf{SeS}}_{\mathsf{VC}}$.
  \item We show that the execution of dependent sub-tasks over a ${\mathsf{SeS}}_{\mathsf{VC}}$ through {\tt RP-VC} can be modeled as an SMP (Sec.~\ref{RPVC_SMP}). We demonstrate that because of $\beta$-inhomogeneous property of ${\mathsf{SeS}}_{\mathsf{VC}}$, analyzing the presented SMP, referred to as \textit{$\beta$-inhomogeneous SMP ($\beta$-SMP)}, is non-trivial (Sec.~\ref{complexityOfRP}).
  \item We develop a unified mathematical framework, called \textit{event stochastic algebra ($\langle e\rangle$-algebra)}, enabling us to investigate the dynamics of an SeS. We demonstrate the generality of the proposed $\langle e\rangle$-algebra, making it suitable for studying a variety of similar problems in literature (Sec.~\ref{CSE}).
  \item Building on the $\langle e\rangle$-algebra, we develop the foundations of \textit{decomposition theorem (DT)}, used to decompose each state of $\beta$-SMP (Sec.~\ref{DecT}). By utilizing DT, we demonstrate that $\beta$-SMP can be transformed into a \textit{decomposed SMP (D-SMP)} (Sec.~\ref{DSMPModel}).
  \item Relying on D-SMP, we derive the general closed-form expression describing \textit{C-MTTF} of our methodology under general dynamics of vehicles in a VC (Sec.~\ref{MTTFCal}). Subsequently, we obtain a special closed-form formula to compute \textit{C-MTTF} of {\tt RP-VC} for a realistic scenario, in which the sojourn times and recruitment duration of the vehicles follow exponential distribution (Sec.~\ref{ERT}).
  \item We present extensive simulations to verify the exactness of our mathematical results and demonstrate the efficiency of our proposed methodology in terms of application acceptance rate and application success rate (Sec.~\ref{simulation}).
\end{description}
\vspace{-2.3mm}
\section{System Model and Preliminaries}\label{systemModel}
\noindent In this section, we first describe the network model of a semi-dynamic VC and state the problem regarding question \ref{Q1} (Sec.~\ref{dmms}). Afterward, we introduce the application and partition models (Sec.~\ref{taskModel}). We next introduce our methodology, i.e., {\tt RP-VC} (Sec.~\ref{porposedStrategy}). Finally, we present the reliability model and problem statement regarding question \ref{Q2} (Sec.~\ref{probDef}).
\vspace{-2.3mm}
\subsection{VC Network Model and First Problem Statement}\label{dmms}
Fig.~\ref{architecture} illustrates the network architecture of a semi-dynamic VC with a set of vehicles equipped with onboard computing and storage units. A centralized controller (e.g., an edge server) is responsible for organizing services and recruiting vehicles via some form of incentivization (e.g., gas credit and parking pass) \cite{49,50}. The CI-App processing requests are sent to the controller, where an application is partitioned into several groups and distributed across vehicles. We assume that there exist enough vehicles in the VC converge area, where the controller can allocate at least two processing vehicles to each sub-task of the incoming application \cite{24,29,30,31,32}. In this architecture, we are interested in addressing the following problem:\par
\begin{figure}
\centering
\includegraphics[width=\linewidth,trim=3 3 3 3,clip]{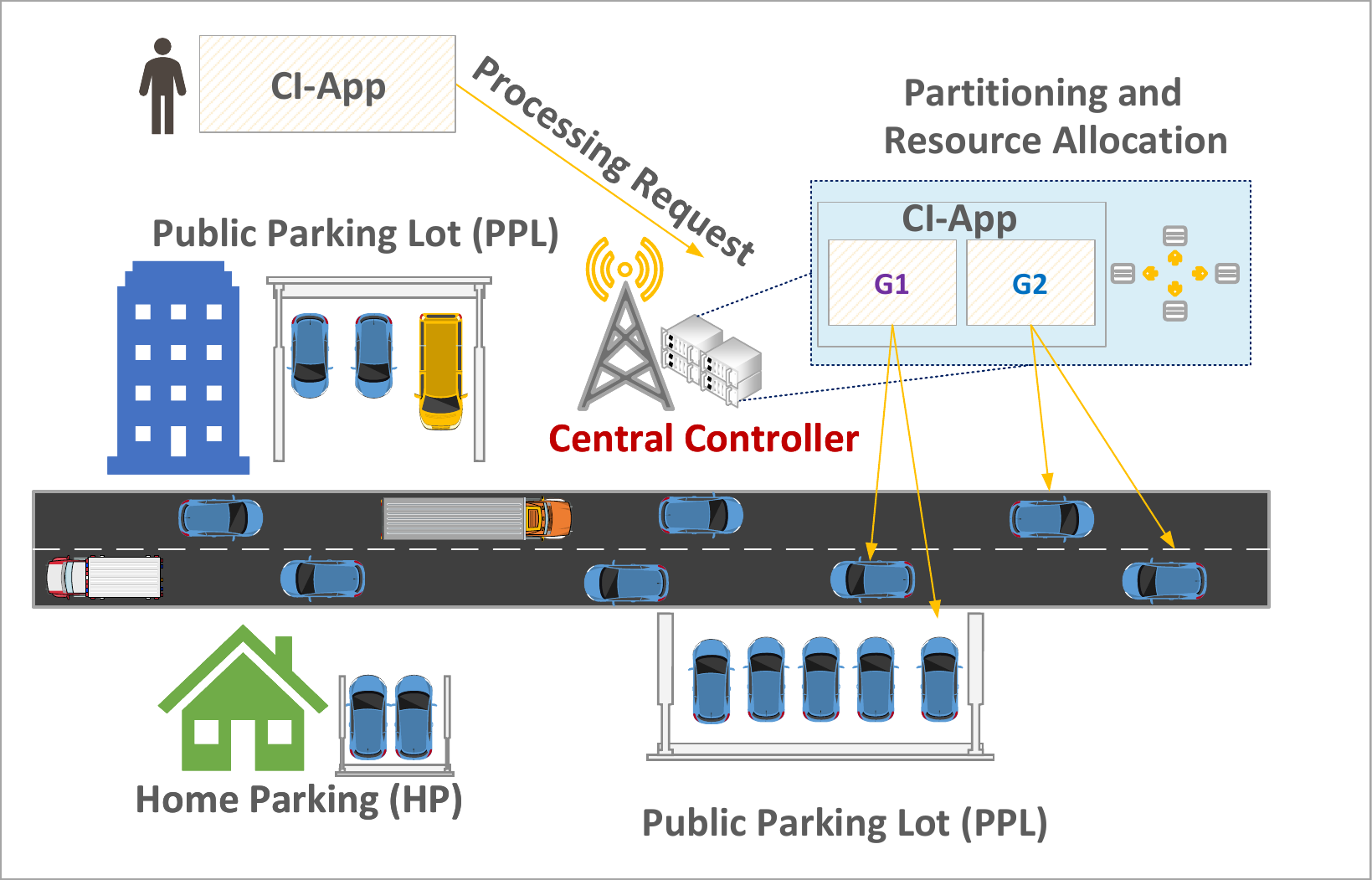}
 \caption{\small Semi-dynamic VC architecture. The vehicles enter the VC and reside for a while in one of the three types of semi-dynamic environments, i.e., (i) public parking lot (PPL), (ii) home parking (HP), and (iii) traffic jams.}\label{architecture}
\end{figure}

\vspace{-1.5mm}
\begin{problem}\label{P1}
 How to increase the \textit{MTTF} of processing a CI-App with $m$ dependent sub-tasks partitioned into $k{\le}m$ groups and processed distributedly on different vehicles in the semi-dynamic VC presented in Fig.~\ref{architecture}?
\end{problem}
In the following, we introduce {\tt RP-VC} to address Problem~\ref{P1}. We first present application and partition models.
\vspace{-2.5mm}
\subsection{CI-Apps and Partition Models}\label{taskModel}
We consider a scenario in which a CI-App is partitioned into several groups \cite{1,42,45,46,52}. We first define a  CI-App, and then we describe the partitioning procedure.
\vspace{-1mm}
\begin{definition}[Application]\label{application}
An application, denoted by $\mathcal{A}$, is a set of $m$ dependent sub-tasks, modeled by a graph:
\begin{equation}
    \mathcal{A} \triangleq \left(\mathcal{V},\mathcal{E}\right),
\end{equation}
where {\small$\displaystyle \mathcal{V}\triangleq\{T_1, T_2,\dotsc, T_m\}$} denotes the set of sub-tasks (i.e., vertices) and {\small$\mathcal{E}\triangleq\left\{(x,x')|x,x'\in\{1,2,\dotsc,m\} ,x\neq x'\right\}$} is the set of dependencies between the sub-tasks (i.e., edges).
\end{definition}
\noindent To characterize internal structures of an application, we consider the following two types of graph representation.
\begin{description}[font=$\bullet$~\normalfont]
\item \textit{Directed Acyclic Graph (DAG).} DAG describes the order of execution of sub-tasks, leading to processing the sub-tasks partially in parallel \cite{45,46}.
\item\textit{Undirected Graph (UG).} There is no order between sub-tasks and they all can be processed in parallel \cite{41,8847383}.
\end{description}
\par
Due to the resource deficiency of vehicles, a large application with $m$ dependent sub-tasks can be partitioned into $k \le m$ deployable groups and offloaded to different vehicles for processing. We define a partition of application $\mathcal{A}$ as follows.
\vspace{-1mm}
\begin{definition}[Partition]\label{partition}
A partition of an application {\small$\mathcal{A} {=} \left(\mathcal{V},\mathcal{E}\right)$} with $m$ dependent sub-tasks, referred to by $\mathcal{P}(\mathcal{A})$, is a dependency-preserving grouping of sub-tasks into $k \leq m$ mutually exclusive nonempty groups {\small$\mathcal{G}{=}\{G_1,\dotsc,G_k\}$}. That is
\begin{equation}
   \mathcal{P}(\mathcal{A})\triangleq\left(\mathcal{G}, \mathcal{E}\right),
\end{equation}
where {\small$G_{x}, G_{x'} \subset \mathcal{V}$}, {\small$\forall G_{x}, G_{x'} \in \mathcal{G}$}, and {\small$\bigcup_{x=1}^{k}G_x = \mathcal{V}$}.
\end{definition}
To clarify Definitions \ref{partition} and \ref{deployment}, Fig.~\ref{RPVCFreamwork} provides an example of a DAG of an application $\mathcal{A}{=}\big(\mathcal{V},\mathcal{E}\big)$, where  {\small$\mathcal{V}{=}\{T_1,T_2,\dotsc,T_6\}$} and {\small$\mathcal{E}{=}\{(1,2),(1,3),(2,4),(2,5),(4,6),(5,6)\}$}. {\small$\mathcal{A}$} is grouped by partition {\small$\mathcal{P}(\mathcal{A}){=}\big(\mathcal{G},\mathcal{E}\big)$} into three groups {\small$\mathcal{G}=\{G_1,G_2,G_3\}$}, where {\small$G_1{=}\{T_1,T_2\}$, $G_2{=}\{T_3,T_5\}$, and $G_3{=}\{T_4,T_6\}$}.\par
We next introduce our general reliable application execution methodology called {\tt RP-VC} to tackle Problem \ref{P1}, which is utilized for processing applications constituted of multiple dependent sub-tasks. {\tt RP-VC} extends the prior art (e.g., {\tt J}$_2$ \cite{24}, concerned with reliable resource allocation for a scenario in which an application is offloaded to the vehicles) to a scenario in
which reliable execution of application sub-tasks are of interest.
\begin{figure}
\centering
\includegraphics[width=0.9\linewidth,trim=1 1 1 1,clip]{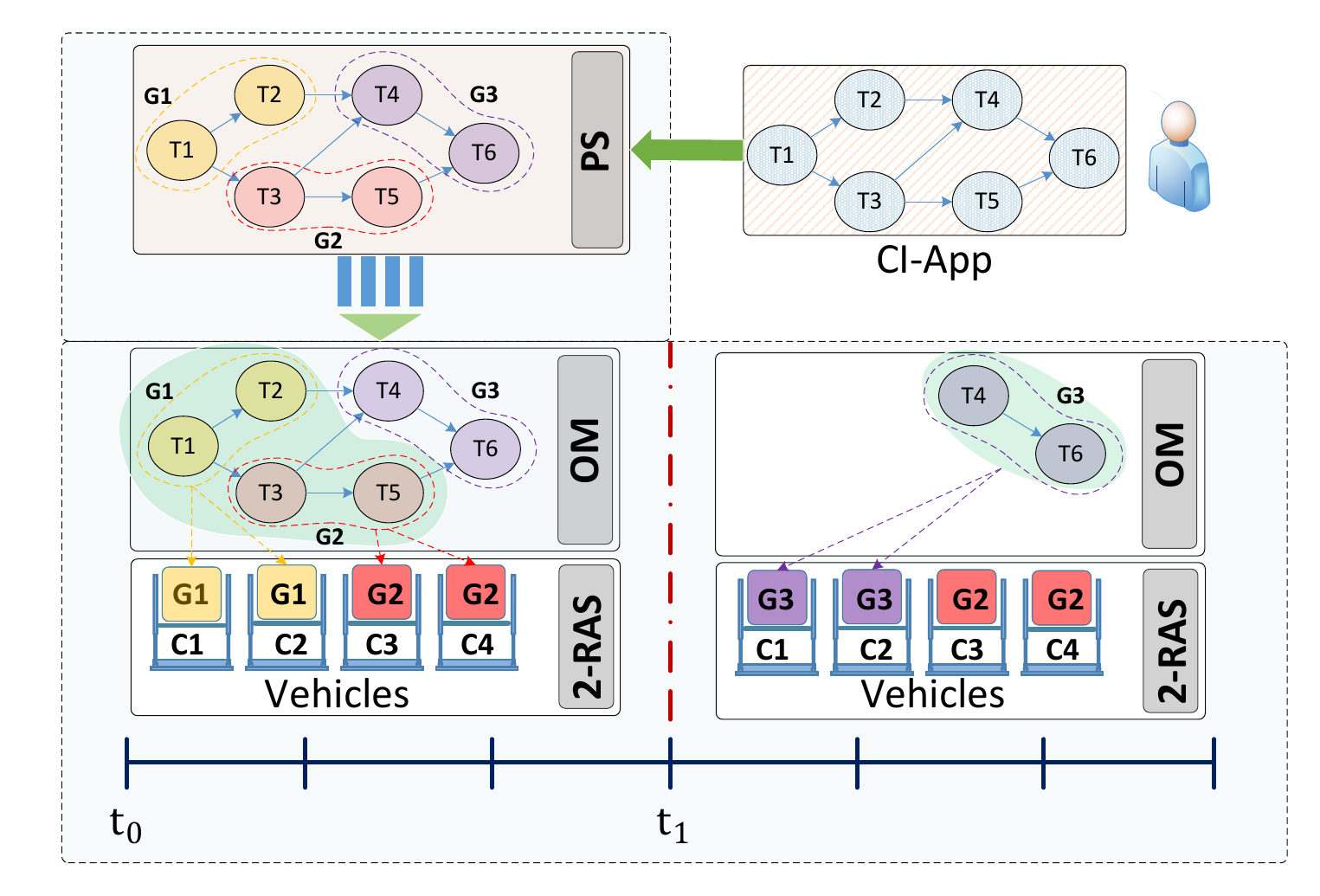}
 \caption{\small A schematic of {\tt RP-VC} methodology. A CI-App $\mathcal{A}$ is submitted to the VC. \textit{PS} partitions it into three groups $\mathcal{G}{=}\{G_1,G_2,G_3\}$. At time $t_0$, \textit{OM} specifies a subset $\overline{\mathcal{V}}_{t_0}{=}\{G_1,G_2\}$ and sends it to the next unit. \textit{2-RAS} allocates two vehicles to each group and offloads the groups to the vehicles. At time $t_1$, the processing of group $G_1$ is completed, and \textit{OM} specifies subset $\overline{\mathcal{V}}_{t_1}{=}\{G_3\}$ and sends it to the 2-RAS to be offloaded to vehicles $C_1$ and $C_2$.}
    \label{RPVCFreamwork}
\end{figure}
\vspace{-5mm}
\subsection{{\tt RP-VC}:  Redundancy-based Processing of CI-Apps}\label{porposedStrategy}
{\tt RP-VC} has three main units, (i) \textit{partition strategy (PS)}, (ii) \textit{offload manager (OM)}, and (iii) \textit{2-redundant allocation strategy (2-RAS)}, which are discussed below.
\subsubsection{Partition strategy (PS)} The \textit{PS} breaks down application $\mathcal{A}$ into smaller groups by partition {\small$\mathcal{P}(\mathcal{A})$} and sends it to the OM. A broad range of partition strategies can be utilized for this intent \cite{42,44}. In this paper, we consider application partitioning under an arbitrary strategy since we are not concerned with the efficiency of PS, which is left as future work.
\subsubsection{Offload manager (OM)}\label{OMS} Consider a partition {\small$\mathcal{P}(\mathcal{A}){=}\left(\mathcal{G}, \mathcal{E}\right)$} of an application $\mathcal{A}$. If $\mathcal{A}$ is a UG-structured application, the groups in $\mathcal{G}$ can be offloaded to different vehicles and processed in parallel. However, if $\mathcal{A}$ is a DAG-structured application, there is a definite order between the execution of the sub-tasks, and different groups should wait until the dependencies of their sub-tasks are satisfied. OM supports applications modeled by both DAG and UG structures. Consider discrete time instances $\tau=\{t_0, t_1,\dotsc\}$. At time $t_0{\in}\tau$ (i.e., the time that the processing of $\mathcal{A}$ is started), \textit{OM} specifies a subset $\overline{\mathcal{V}}_{t_0}{\subseteq}\mathcal{V}$ (the equality of the subset notation is for UG-structured applications) based on a specific strategy (e.g., \cite{52}) and sends it to \textit{2-RAS} for offloading. Likewise, at time $t_1\in\tau$, \textit{OM} specifies another subset $\overline{\mathcal{V}}_{t_1}{\subseteq}\mathcal{V}\setminus\overline{\mathcal{V}}_{t_0}$ and sends it to \textit{2-RAS}. This procedure takes place until all the groups are offload to the vehicles. Referring to Fig.~\ref{RPVCFreamwork} as an example, at time $t_0$, \textit{OM} specifies $\overline{\mathcal{V}}_{t_0}{=}\{G_1,G_2\}$ and sends it to 2-RAS for offloading. Further, at time $t_1$, \textit{OM} specifies $\overline{\mathcal{V}}_{t_1}{=}\{G_3\}$ to be offloaded to the vehicles. Same as the \textit{PS}, different strategies can be exploited for \textit{OM} \cite{52}. In this paper, we consider an arbitrary OM and dealing with different OM strategies is left for future works.
\subsubsection{2-redundant allocation strategy (2-RAS)}\label{S2RAS} \textit{2-RAS} has the responsibility of allocating $2$ different vehicles\footnote{Enhancing reliability by placing two replicas of an application/data on different servers is  common in redundancy-based strategies \cite{19,21,24}.} to each group $G_h\in\overline{\mathcal{V}}_{t}$ at each time $t\in \tau$. Note that in our mathematical analysis, we assume that each vehicle processes only one group at a specific time, which is a common assumption \cite{24,29}. In Sec. \ref{simulation}, we will investigate a scenario in which more than one group is processed by a vehicle.\par
Let $C_{\ell}$ and $C_{\ell'}$ be two vehicles assigned to group $G_h{\in}\overline{\mathcal{V}}_{t}$. As the second responsibility of \textit{2-RAS}, if $C_{\ell}$ departs the VC, \textit{2-RAS} determines $C_{\ell'}$ as the recruiter, which in turn starts recruiting a vehicle to process $G_h$. Recruitment operation for each group $G_h$ can thus be summarized as follows: the recruiter (i.e., $C_{\ell'}$) (i) finds a new vehicle from the vehicles that were in the VC at the time that processing of application $\mathcal{A}$ has been started, and (ii) transfers an image/replica of $G_h$ (i.e., the current status of processing $G_h$ and its required data) to the recruited vehicle through a migration strategy. Studying \textit{2-RAS} is our major focus in this paper.
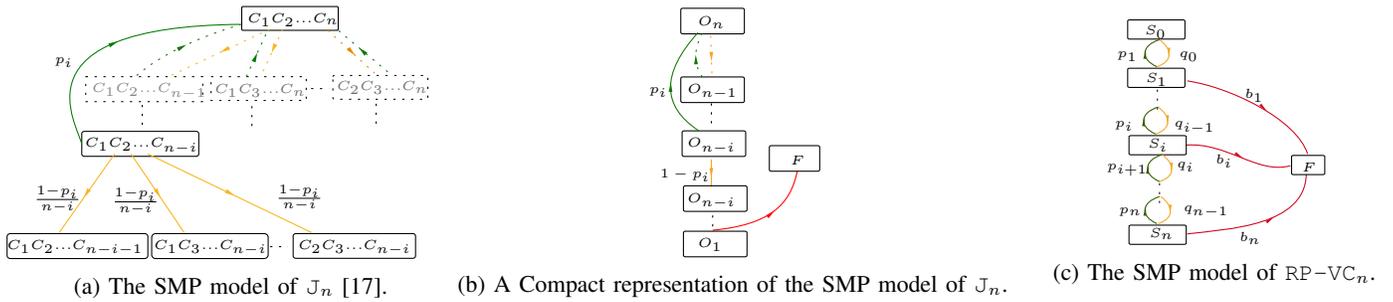
\begin{figure*}[!t]
\vspace{-5mm}
  \centering
  \tikzset{every picture/.style={line width=0.4pt}} 
  \begin{subfigure}{.33\textwidth}
\begin{tikzpicture}[x=0.5pt,y=0.45pt,yscale=-1,xscale=1]
\draw  [dash pattern={on 0.84pt off 2.51pt}]  (233.8,80) -- (247.3,80) ;
\draw   (181.3,12) .. controls (181.3,11.45) and (181.75,11) .. (182.3,11) -- (253.3,11) .. controls (253.85,11) and (254.3,11.45) .. (254.3,12) -- (254.3,30) .. controls (254.3,30.55) and (253.85,31) .. (253.3,31) -- (182.3,31) .. controls (181.75,31) and (181.3,30.55) .. (181.3,30) -- cycle ;
\draw  [dash pattern={on 0.84pt off 2.51pt}] (63.3,70) .. controls (63.3,70) and (63.3,70) .. (63.3,70) -- (151.8,70) .. controls (151.8,70) and (151.8,70) .. (151.8,70) -- (151.8,91) .. controls (151.8,91) and (151.8,91) .. (151.8,91) -- (63.3,91) .. controls (63.3,91) and (63.3,91) .. (63.3,91) -- cycle ;
\draw  [dash pattern={on 0.84pt off 2.51pt}] (158.1,70) .. controls (158.1,70) and (158.1,70) .. (158.1,70) -- (230.1,70) .. controls (230.1,70) and (230.1,70) .. (230.1,70) -- (230.1,90) .. controls (230.1,90) and (230.1,90) .. (230.1,90) -- (158.1,90) .. controls (158.1,90) and (158.1,90) .. (158.1,90) -- cycle ;
\draw  [dash pattern={on 0.84pt off 2.51pt}] (250.3,70) .. controls (250.3,69.45) and (250.75,69) .. (251.3,69) -- (321.1,69) .. controls (321.65,69) and (322.1,69.45) .. (322.1,70) -- (322.1,88) .. controls (322.1,88.55) and (321.65,89) .. (321.1,89) -- (251.3,89) .. controls (250.75,89) and (250.3,88.55) .. (250.3,88) -- cycle ;
\draw   (3.8,204.5) .. controls (3.8,203.12) and (4.92,202) .. (6.3,202) -- (107.8,202) .. controls (109.18,202) and (110.3,203.12) .. (110.3,204.5) -- (110.3,220.5) .. controls (110.3,221.88) and (109.18,223) .. (107.8,223) -- (6.3,223) .. controls (4.92,223) and (3.8,221.88) .. (3.8,220.5) -- cycle ;
\draw   (113.2,205.1) .. controls (113.2,203.39) and (114.59,202) .. (116.3,202) -- (198.2,202) .. controls (199.91,202) and (201.3,203.39) .. (201.3,205.1) -- (201.3,219.9) .. controls (201.3,221.61) and (199.91,223) .. (198.2,223) -- (116.3,223) .. controls (114.59,223) and (113.2,221.61) .. (113.2,219.9) -- cycle ;
\draw   (220.2,204.1) .. controls (220.2,202.94) and (221.14,202) .. (222.3,202) -- (310.7,202) .. controls (311.86,202) and (312.8,202.94) .. (312.8,204.1) -- (312.8,220.7) .. controls (312.8,221.86) and (311.86,222.8) .. (310.7,222.8) -- (222.3,222.8) .. controls (221.14,222.8) and (220.2,221.86) .. (220.2,220.7) -- cycle ;
\draw [color={rgb, 255:red, 0; green, 123; blue, 3 }  ,draw opacity=1 ][fill={rgb, 255:red, 0; green, 105; blue, 15 }  ,fill opacity=1 ] [dash pattern={on 0.84pt off 2.51pt}]  (291.8,66) -- (254.3,31) ;
\draw [shift={(269.98,45.63)}, rotate = 43.03] [fill={rgb, 255:red, 0; green, 123; blue, 3 }  ,fill opacity=1 ][line width=0.08]  [draw opacity=0] (7.2,-1.8) -- (0,0) -- (7.2,1.8) -- cycle    ;
\draw [color={rgb, 255:red, 229; green, 136; blue, 0 }  ,draw opacity=1 ][fill={rgb, 255:red, 0; green, 105; blue, 15 }  ,fill opacity=1 ] [dash pattern={on 0.84pt off 2.51pt}]  (246.3,31) -- (281.8,67) ;
\draw [shift={(267,51.99)}, rotate = 225.4] [fill={rgb, 255:red, 229; green, 136; blue, 0 }  ,fill opacity=1 ][line width=0.08]  [draw opacity=0] (7.2,-1.8) -- (0,0) -- (7.2,1.8) -- cycle    ;
\draw  [dash pattern={on 0.84pt off 2.51pt}]  (189,94.77) -- (189,113) ;
\draw   (60.3,118) .. controls (60.3,116.9) and (61.2,116) .. (62.3,116) -- (146.8,116) .. controls (147.9,116) and (148.8,116.9) .. (148.8,118) -- (148.8,135) .. controls (148.8,136.1) and (147.9,137) .. (146.8,137) -- (62.3,137) .. controls (61.2,137) and (60.3,136.1) .. (60.3,135) -- cycle ;
\draw [color={rgb, 255:red, 0; green, 123; blue, 3 }  ,draw opacity=1 ]   (60.3,127) .. controls (57.8,111) and (6.3,28) .. (179.8,26) ;
\draw [shift={(87.58,40.37)}, rotate = 155.2] [fill={rgb, 255:red, 0; green, 123; blue, 3 }  ,fill opacity=1 ][line width=0.08]  [draw opacity=0] (7.2,-1.8) -- (0,0) -- (7.2,1.8) -- cycle    ;
\draw  [dash pattern={on 0.84pt off 2.51pt}]  (283,92.77) -- (283,112) ;
\draw  [dash pattern={on 0.84pt off 2.51pt}]  (106,93.77) -- (106,114) ;
\draw  [dash pattern={on 0.84pt off 2.51pt}]  (202.8,213) -- (216.3,213) ;

\draw (184,14.4) node [anchor=north west][inner sep=0.75pt]  [font=\fontsize{5}{5}\selectfont]  {$C_{1} C_{2} ...C_{n}$};
\draw (66.3,73.4) node [anchor=north west][inner sep=0.75pt]  [font=\fontsize{5}{5}\selectfont,color={rgb, 255:red, 0; green, 0; blue, 0 }  ,opacity=0.48 ]  {$C_{1} C_{2} ...C_{n-1}$};
\draw (159,74.4) node [anchor=north west][inner sep=0.75pt]  [font=\fontsize{5}{5}\selectfont,color={rgb, 255:red, 0; green, 0; blue, 0 }  ,opacity=0.58 ]  {$C_{1} C_{3} ...C_{n}$};
\draw (251.3,71.4) node [anchor=north west][inner sep=0.75pt]  [font=\fontsize{5}{5}\selectfont,color={rgb, 255:red, 0; green, 0; blue, 0 }  ,opacity=0.58 ]  {$C_{2} C_{3} ...C_{n}$};
\draw (61,119.4) node [anchor=north west][inner sep=0.75pt]  [font=\fontsize{5}{5}\selectfont]  {$C_{1} C_{2} ...C_{n-i}$};
\draw (2,205.4) node [anchor=north west][inner sep=0.75pt]  [font=\fontsize{5}{5}\selectfont]  {$C_{1} C_{2} ...C_{n-i-1}$};
\draw (113.2,205.4) node [anchor=north west][inner sep=0.75pt]  [font=\fontsize{5}{5}\selectfont]  {$C_{1} C_{3} ...C_{n-i}$};
\draw (222.2,205.4) node [anchor=north west][inner sep=0.75pt]  [font=\fontsize{5}{5}\selectfont]  {$C_{2} C_{3} ...C_{n-i}$};
\draw (38,51.4) node [anchor=north west][inner sep=0.75pt]  [font=\tiny]  {$p_{i}$};
\draw (22,159.4) node [anchor=north west][inner sep=0.75pt]  [font=\tiny]  {$\frac{1-p_{i}}{n-i}$};
\draw (81,161.4) node [anchor=north west][inner sep=0.75pt]  [font=\tiny]  {$\frac{1-p_{i}}{n-i}$};
\draw (205,158.4) node [anchor=north west][inner sep=0.75pt]  [font=\tiny]  {$\frac{1-p_{i}}{n-i}$};
\draw [color={rgb, 255:red, 248; green, 181; blue, 28 }  ,draw opacity=1 ] [dash pattern={on 0.84pt off 2.51pt}]  (202.63,30) -- (127.8,69) ;
\draw [shift={(165.22,49.5)}, rotate = 332.47] [fill={rgb, 255:red, 248; green, 181; blue, 28 }  ,fill opacity=1 ][line width=0.08]  [draw opacity=0] (7.2,-1.8) -- (0,0) -- (7.2,1.8) -- cycle    ;
\draw [color={rgb, 255:red, 248; green, 181; blue, 28 }  ,draw opacity=1 ] [dash pattern={on 0.84pt off 2.51pt}]  (212.25,30) -- (195.58,70) ;
\draw [shift={(203.92,50)}, rotate = 292.62] [fill={rgb, 255:red, 248; green, 181; blue, 28 }  ,fill opacity=1 ][line width=0.08]  [draw opacity=0] (7.2,-1.8) -- (0,0) -- (7.2,1.8) -- cycle    ;
\draw [color={rgb, 255:red, 22; green, 117; blue, 5 }  ,draw opacity=1 ] [dash pattern={on 0.84pt off 2.51pt}]  (101.83,69) .. controls (119.52,52.13) and (145.91,37.76) .. (181,25.91) ;
\draw [shift={(139.32,43.22)}, rotate = 151.98] [fill={rgb, 255:red, 22; green, 117; blue, 5 }  ,fill opacity=1 ][line width=0.08]  [draw opacity=0] (7.2,-1.8) -- (0,0) -- (7.2,1.8) -- cycle    ;
\draw [color={rgb, 255:red, 0; green, 123; blue, 3 }  ,draw opacity=1 ] [dash pattern={on 0.84pt off 2.51pt}]  (184.75,70) -- (201.42,30) ;
\draw [shift={(193.08,50)}, rotate = 112.62] [fill={rgb, 255:red, 0; green, 123; blue, 3 }  ,fill opacity=1 ][line width=0.08]  [draw opacity=0] (7.2,-1.8) -- (0,0) -- (7.2,1.8) -- cycle    ;
\draw [color={rgb, 255:red, 248; green, 181; blue, 28 }  ,draw opacity=1 ]   (85.16,135) -- (43.34,201) ;
\draw [shift={(62,171.55)}, rotate = 302.36] [fill={rgb, 255:red, 248; green, 181; blue, 28 }  ,fill opacity=1 ][line width=0.08]  [draw opacity=0] (7.2,-1.8) -- (0,0) -- (7.2,1.8) -- cycle    ;
\draw [color={rgb, 255:red, 248; green, 181; blue, 28 }  ,draw opacity=1 ]   (97.57,135) -- (137.63,201) ;
\draw [shift={(119.78,171.59)}, rotate = 238.74] [fill={rgb, 255:red, 248; green, 181; blue, 28 }  ,fill opacity=1 ][line width=0.08]  [draw opacity=0] (7.2,-1.8) -- (0,0) -- (7.2,1.8) -- cycle    ;
\draw [color={rgb, 255:red, 248; green, 181; blue, 28 }  ,draw opacity=1 ]   (110.24,135) -- (233.96,201) ;
\draw [shift={(175.81,169.98)}, rotate = 208.08] [fill={rgb, 255:red, 248; green, 181; blue, 28 }  ,fill opacity=1 ][line width=0.08]  [draw opacity=0] (7.2,-1.8) -- (0,0) -- (7.2,1.8) -- cycle    ;
\end{tikzpicture}
 \caption{\small The SMP model of {\tt J}$_n$ \cite{29}.}
 \label{SMPJn}
  \end{subfigure}%
  \begin{subfigure}{.40\textwidth}
    \centering
    \begin{tikzpicture}[x=0.5pt,y=0.4pt,yscale=-1,xscale=1]
    \draw [color={rgb, 255:red, 0; green, 0; blue, 0 }  ,draw opacity=1 ] [dash pattern={on 0.84pt off 2.51pt}]  (65.83,94) -- (65.83,119) ;
    \draw [color={rgb, 255:red, 245; green, 166; blue, 35 }  ,draw opacity=1 ] [dash pattern={on 0.84pt off 2.51pt}]  (65.33,30) -- (66.33,70) ;
    \draw [shift={(65.94,54.2)}, rotate = 268.57] [fill={rgb, 255:red, 245; green, 166; blue, 35 }  ,fill opacity=1 ][line width=0.08]  [draw opacity=0] (7.2,-1.8) -- (0,0) -- (7.2,1.8) -- cycle    ;
    \draw [color={rgb, 255:red, 0; green, 0; blue, 0 }  ,draw opacity=1 ] [dash pattern={on 0.84pt off 2.51pt}]  (66.83,200) -- (66.83,214) ;
    \draw [color={rgb, 255:red, 245; green, 166; blue, 35 }  ,draw opacity=1 ]   (65.83,147) -- (65.83,173) ;
    \draw [shift={(65.83,164.2)}, rotate = 270] [fill={rgb, 255:red, 245; green, 166; blue, 35 }  ,fill opacity=1 ][line width=0.08]  [draw opacity=0] (7.2,-1.8) -- (0,0) -- (7.2,1.8) -- cycle    ;
    \draw [color={rgb, 255:red, 255; green, 0; blue, 0 }  ,draw opacity=1 ]   (66.83,214) .. controls (109.8,209) and (125.8,190) .. (130.8,158) ;
    \draw [shift={(113.65,195.56)}, rotate = 139.46] [fill={rgb, 255:red, 255; green, 0; blue, 0 }  ,fill opacity=1 ][line width=0.08]  [draw opacity=0] (7.2,-1.8) -- (0,0) -- (7.2,1.8) -- cycle    ;
    \draw   (42.7,5.5) .. controls (42.7,4.67) and (43.37,4) .. (44.2,4) -- (89.3,4) .. controls (90.13,4) and (90.8,4.67) .. (90.8,5.5) -- (90.8,26.7) .. controls (90.8,27.53) and (90.13,28.2) .. (89.3,28.2) -- (44.2,28.2) .. controls (43.37,28.2) and (42.7,27.53) .. (42.7,26.7) -- cycle ;
    \draw   (42.7,70.5) .. controls (42.7,69.67) and (43.37,69) .. (44.2,69) -- (89.3,69) .. controls (90.13,69) and (90.8,69.67) .. (90.8,70.5) -- (90.8,91.7) .. controls (90.8,92.53) and (90.13,93.2) .. (89.3,93.2) -- (44.2,93.2) .. controls (43.37,93.2) and (42.7,92.53) .. (42.7,91.7) -- cycle ;
    \draw   (43.7,121.5) .. controls (43.7,120.67) and (44.37,120) .. (45.2,120) -- (90.3,120) .. controls (91.13,120) and (91.8,120.67) .. (91.8,121.5) -- (91.8,142.7) .. controls (91.8,143.53) and (91.13,144.2) .. (90.3,144.2) -- (45.2,144.2) .. controls (44.37,144.2) and (43.7,143.53) .. (43.7,142.7) -- cycle ;
    \draw   (44.7,173.5) .. controls (44.7,172.67) and (45.37,172) .. (46.2,172) -- (91.3,172) .. controls (92.13,172) and (92.8,172.67) .. (92.8,173.5) -- (92.8,194.7) .. controls (92.8,195.53) and (92.13,196.2) .. (91.3,196.2) -- (46.2,196.2) .. controls (45.37,196.2) and (44.7,195.53) .. (44.7,194.7) -- cycle ;
    \draw   (44.7,216.5) .. controls (44.7,215.67) and (45.37,215) .. (46.2,215) -- (91.3,215) .. controls (92.13,215) and (92.8,215.67) .. (92.8,216.5) -- (92.8,237.7) .. controls (92.8,238.53) and (92.13,239.2) .. (91.3,239.2) -- (46.2,239.2) .. controls (45.37,239.2) and (44.7,238.53) .. (44.7,237.7) -- cycle ;
    \draw   (109.7,135.5) .. controls (109.7,134.67) and (110.37,134) .. (111.2,134) -- (146.3,134) .. controls (147.13,134) and (147.8,134.67) .. (147.8,135.5) -- (147.8,156.7) .. controls (147.8,157.53) and (147.13,158.2) .. (146.3,158.2) -- (111.2,158.2) .. controls (110.37,158.2) and (109.7,157.53) .. (109.7,156.7) -- cycle ;
    \draw (54,9.4) node [anchor=north west][inner sep=0.75pt]  [font=\fontsize{5}{5}\selectfont]  {$O_{n}$};
    \draw (46,73.4) node [anchor=north west][inner sep=0.75pt]  [font=\fontsize{5}{5}\selectfont]  {$O_{n-1}$};
    \draw (46,124.4) node [anchor=north west][inner sep=0.75pt]  [font=\fontsize{5}{5}\selectfont]  {$O_{n-i}$};
    \draw (54,219.4) node [anchor=north west][inner sep=0.75pt]  [font=\fontsize{5}{5}\selectfont]  {$O_{1}$};
    \draw (124,141.4) node [anchor=north west][inner sep=0.75pt]  [font=\fontsize{5}{5}\selectfont]  {$F$};
    \draw (17,76.4) node [anchor=north west][inner sep=0.75pt]  [font=\tiny]  {$p_{i}$};
    \draw (46,176.4) node [anchor=north west][inner sep=0.75pt]  [font=\fontsize{5}{5}\selectfont]  {$O_{n-i}$};
    \draw (25.2,153.4) node [anchor=north west][inner sep=0.75pt]  [font=\tiny]  {$1-p_{i}$};
    \draw [color={rgb, 255:red, 22; green, 117; blue, 5 }  ,draw opacity=1 ]   (56.4,120) .. controls (28.01,96.71) and (27.78,65.71) .. (55.73,27) ;
    \draw [shift={(35.16,74.1)}, rotate = 91.9] [fill={rgb, 255:red, 22; green, 117; blue, 5 }  ,fill opacity=1 ][line width=0.08]  [draw opacity=0] (7.2,-1.8) -- (0,0) -- (7.2,1.8) -- cycle    ;
    \draw [color={rgb, 255:red, 22; green, 117; blue, 5 }  ,draw opacity=1 ] [dash pattern={on 0.84pt off 2.51pt}]  (59.03,69) .. controls (52.21,57) and (52.27,43) .. (59.21,27) ;
    \draw [shift={(53.99,48.02)}, rotate = 88.11] [fill={rgb, 255:red, 22; green, 117; blue, 5 }  ,fill opacity=1 ][line width=0.08]  [draw opacity=0] (7.2,-1.8) -- (0,0) -- (7.2,1.8) -- cycle    ;
    \end{tikzpicture}
    \caption{\small A Compact representation of the SMP model of {\tt J}$_n$.}
\label{comSMPJn}
\end{subfigure}%
  \begin{subfigure}{.30\textwidth}
  \centering
\begin{tikzpicture}[x=0.5pt,y=0.30pt,yscale=-1,xscale=1]
\draw [line width=0.2]  [dash pattern={on 0.84pt off 2.51pt}]  (51.8,215) -- (51.8,241) ;
\draw [color={rgb, 255:red, 208; green, 2; blue, 27 }  ,draw opacity=1 ]   (72.8,86) .. controls (104.19,93.37) and (155.23,133.87) .. (163.8,178) ;
\draw [shift={(133.11,124.37)}, rotate = 224.92] [fill={rgb, 255:red, 208; green, 2; blue, 27 }  ,fill opacity=1 ][line width=0.08]  [draw opacity=0] (7.2,-1.8) -- (0,0) -- (7.2,1.8) -- cycle    ;
\draw [color={rgb, 255:red, 208; green, 2; blue, 27 }  ,draw opacity=1 ]   (71.8,167) .. controls (102.8,163) and (125.94,204.44) .. (150.8,193) ;
\draw [shift={(115.13,183.8)}, rotate = 213.42] [fill={rgb, 255:red, 208; green, 2; blue, 27 }  ,fill opacity=1 ][line width=0.08]  [draw opacity=0] (7.2,-1.8) -- (0,0) -- (7.2,1.8) -- cycle    ;
\draw [color={rgb, 255:red, 208; green, 2; blue, 27 }  ,draw opacity=1 ]   (72.8,279) .. controls (103.8,273) and (161.8,269) .. (162.8,204) ;
\draw [shift={(138.61,257.53)}, rotate = 145.2] [fill={rgb, 255:red, 208; green, 2; blue, 27 }  ,fill opacity=1 ][line width=0.08]  [draw opacity=0] (7.2,-1.8) -- (0,0) -- (7.2,1.8) -- cycle    ;
\draw [line width=0.2]  [dash pattern={on 0.84pt off 2.51pt}]  (50.2,97) -- (50.2,123) ;
\draw   (27.7,10.5) .. controls (27.7,9.67) and (28.37,9) .. (29.2,9) -- (69.5,9) .. controls (70.33,9) and (71,9.67) .. (71,10.5) -- (71,31.7) .. controls (71,32.53) and (70.33,33.2) .. (69.5,33.2) -- (29.2,33.2) .. controls (28.37,33.2) and (27.7,32.53) .. (27.7,31.7) -- cycle ;
\draw   (27.7,71.5) .. controls (27.7,70.67) and (28.37,70) .. (29.2,70) -- (69.5,70) .. controls (70.33,70) and (71,70.67) .. (71,71.5) -- (71,92.7) .. controls (71,93.53) and (70.33,94.2) .. (69.5,94.2) -- (29.2,94.2) .. controls (28.37,94.2) and (27.7,93.53) .. (27.7,92.7) -- cycle ;
\draw   (28.7,156.5) .. controls (28.7,155.67) and (29.37,155) .. (30.2,155) -- (70.5,155) .. controls (71.33,155) and (72,155.67) .. (72,156.5) -- (72,177.7) .. controls (72,178.53) and (71.33,179.2) .. (70.5,179.2) -- (30.2,179.2) .. controls (29.37,179.2) and (28.7,178.53) .. (28.7,177.7) -- cycle ;
\draw   (28.7,269.5) .. controls (28.7,268.67) and (29.37,268) .. (30.2,268) -- (70.5,268) .. controls (71.33,268) and (72,268.67) .. (72,269.5) -- (72,290.7) .. controls (72,291.53) and (71.33,292.2) .. (70.5,292.2) -- (30.2,292.2) .. controls (29.37,292.2) and (28.7,291.53) .. (28.7,290.7) -- cycle ;
\draw   (151.69,181.5) .. controls (151.69,180.67) and (152.36,180) .. (153.19,180) -- (175.3,180) .. controls (176.13,180) and (176.8,180.67) .. (176.8,181.5) -- (176.8,202.7) .. controls (176.8,203.53) and (176.13,204.2) .. (175.3,204.2) -- (153.19,204.2) .. controls (152.36,204.2) and (151.69,203.53) .. (151.69,202.7) -- cycle ;
\draw [color={rgb, 255:red, 245; green, 166; blue, 35 }  ,draw opacity=1 ]   (51.8,34) .. controls (60.8,45) and (62.8,61) .. (49.8,68) ;
\draw [shift={(58.71,56.18)}, rotate = 272.13] [fill={rgb, 255:red, 245; green, 166; blue, 35 }  ,fill opacity=1 ][line width=0.08]  [draw opacity=0] (7.2,-1.8) -- (0,0) -- (7.2,1.8) -- cycle    ;
\draw [color={rgb, 255:red, 49; green, 92; blue, 0 }  ,draw opacity=1 ]   (49.8,68) .. controls (37.07,58.91) and (37.8,41) .. (51.8,34) ;
\draw [shift={(41.43,46.19)}, rotate = 96.58] [fill={rgb, 255:red, 49; green, 92; blue, 0 }  ,fill opacity=1 ][line width=0.08]  [draw opacity=0] (7.2,-1.8) -- (0,0) -- (7.2,1.8) -- cycle    ;
\draw [color={rgb, 255:red, 245; green, 166; blue, 35 }  ,draw opacity=1 ]   (50.8,119) .. controls (59.8,130) and (61.8,146) .. (48.8,153) ;
\draw [shift={(57.71,141.18)}, rotate = 272.13] [fill={rgb, 255:red, 245; green, 166; blue, 35 }  ,fill opacity=1 ][line width=0.08]  [draw opacity=0] (7.2,-1.8) -- (0,0) -- (7.2,1.8) -- cycle    ;
\draw [color={rgb, 255:red, 49; green, 92; blue, 0 }  ,draw opacity=1 ]   (48.8,153) .. controls (36.07,143.91) and (36.8,126) .. (50.8,119) ;
\draw [shift={(40.43,131.19)}, rotate = 96.58] [fill={rgb, 255:red, 49; green, 92; blue, 0 }  ,fill opacity=1 ][line width=0.08]  [draw opacity=0] (7.2,-1.8) -- (0,0) -- (7.2,1.8) -- cycle    ;
\draw [color={rgb, 255:red, 245; green, 166; blue, 35 }  ,draw opacity=1 ]   (53.8,179) .. controls (62.8,190) and (64.8,206) .. (51.8,213) ;
\draw [shift={(60.71,201.18)}, rotate = 272.13] [fill={rgb, 255:red, 245; green, 166; blue, 35 }  ,fill opacity=1 ][line width=0.08]  [draw opacity=0] (7.2,-1.8) -- (0,0) -- (7.2,1.8) -- cycle    ;
\draw [color={rgb, 255:red, 49; green, 92; blue, 0 }  ,draw opacity=1 ]   (51.8,213) .. controls (39.07,203.91) and (39.8,186) .. (53.8,179) ;
\draw [shift={(43.43,191.19)}, rotate = 96.58] [fill={rgb, 255:red, 49; green, 92; blue, 0 }  ,fill opacity=1 ][line width=0.08]  [draw opacity=0] (7.2,-1.8) -- (0,0) -- (7.2,1.8) -- cycle    ;
\draw [color={rgb, 255:red, 245; green, 166; blue, 35 }  ,draw opacity=1 ]   (52.8,232) .. controls (61.8,243) and (63.8,259) .. (50.8,266) ;
\draw [shift={(59.71,254.18)}, rotate = 272.13] [fill={rgb, 255:red, 245; green, 166; blue, 35 }  ,fill opacity=1 ][line width=0.08]  [draw opacity=0] (7.2,-1.8) -- (0,0) -- (7.2,1.8) -- cycle    ;
\draw [color={rgb, 255:red, 49; green, 92; blue, 0 }  ,draw opacity=1 ]   (49.8,266) .. controls (37.07,256.91) and (37.8,239) .. (51.8,232) ;
\draw [shift={(41.43,244.19)}, rotate = 96.58] [fill={rgb, 255:red, 49; green, 92; blue, 0 }  ,fill opacity=1 ][line width=0.08]  [draw opacity=0] (7.2,-1.8) -- (0,0) -- (7.2,1.8) -- cycle    ;

\draw (38.77,12.45) node [anchor=north west][inner sep=0.75pt]  [font=\fontsize{5}{5}\selectfont]  {$S_{0}$};
\draw (18.04,45.55) node [anchor=north west][inner sep=0.75pt]  [font=\fontsize{5}{5}\selectfont]  {$p_{1}$};
\draw (39.69,72.29) node [anchor=north west][inner sep=0.75pt]  [font=\fontsize{5}{5}\selectfont]  {$S_{1}$};
\draw (42.26,267.99) node [anchor=north west][inner sep=0.75pt]  [font=\fontsize{5}{5}\selectfont]  {$S_{n}$};
\draw (41.26,155.67) node [anchor=north west][inner sep=0.75pt]  [font=\fontsize{5}{5}\selectfont]  {$S_{i}$};
\draw (157.87,185.56) node [anchor=north west][inner sep=0.75pt]  [font=\fontsize{5}{5}\selectfont]  {$F$};
\draw (114.55,92.37) node [anchor=north west][inner sep=0.75pt]  [font=\fontsize{5}{5}\selectfont]  {$b_{1}$};
\draw (64.83,45.27) node [anchor=north west][inner sep=0.75pt]  [font=\fontsize{5}{5}\selectfont]  {$q_{0}$};
\draw (19.31,242.38) node [anchor=north west][inner sep=0.75pt]  [font=\fontsize{5}{5}\selectfont]  {$p_{n}$};
\draw (92.55,175.37) node [anchor=north west][inner sep=0.75pt]  [font=\fontsize{5}{5}\selectfont]  {$b_{i}$};
\draw (110.55,273.37) node [anchor=north west][inner sep=0.75pt]  [font=\fontsize{5}{5}\selectfont]  {$b_{n}$};
\draw (63.26,185.09) node [anchor=north west][inner sep=0.75pt]  [font=\fontsize{5}{5}\selectfont]  {$q_{i}$};
\draw (10.31,187.38) node [anchor=north west][inner sep=0.75pt]  [font=\fontsize{5}{5}\selectfont]  {$p_{i+1}$};
\draw (16.31,133.38) node [anchor=north west][inner sep=0.75pt]  [font=\fontsize{5}{5}\selectfont]  {$p_{i}$};
\draw (61.26,134.09) node [anchor=north west][inner sep=0.75pt]  [font=\fontsize{5}{5}\selectfont]  {$q_{i-1}$};
\draw (68.26,240.09) node [anchor=north west][inner sep=0.75pt]  [font=\fontsize{5}{5}\selectfont]  {$q_{n-1}$};
\end{tikzpicture}
  \caption{\small The SMP model of {\tt RP-VC$_n$}.}
   \label{j2nMarkov1}
\end{subfigure}
  \caption{\small SMP model comparison between {\tt J}$_n$ and {\tt RP-VC$_n$}.}
  \vspace{-5mm}
\end{figure*}\par
We next present the reliability model and the key problem addressed in this paper.
\vspace{-3mm}
\subsection{Reliability Model and Second Problem Statement}\label{probDef}
Reliability is a major notion of quality of service (QoS) in cloud-based computing systems. One of the most well-known metrics to quantify reliability is \textit{mean time to failure (MTTF)}. Calculating \textit{MTTF} of processing a DAG-structured application in a semi-dynamic VC under {\tt RP-VC} faces significant challenges since, as mentioned in Sec. \ref{OMS}, different OMs may offload the groups of $\mathcal{G}$ to the vehicles progressively (see Fig.~\ref{RPVCFreamwork}). For instance, authors in \cite{52} present a strategy in which offloading a group of sub-tasks is parallelized with processing other deployed sub-tasks. Consequently, the number of deployed groups of partition $\mathcal{P}(\mathcal{A})$ on the vehicles (i.e., the groups offloaded to the vehicles for processing) varies at different time instances during processing application $\mathcal{A}$. Accordingly, we extend the \textit{MTTF} to introduce a new reliability metric, called \textit{conditional mean time to failure (C-MTTF)}, enabling us to calculate \textit{MTTF} at different time instances.

Let {\small$\mathcal{C}{=}\{C_1,\dotsc,C_r\}$} denote the set of all $r$ vehicles in the VC. Before formalizing \textit{C-MTTF}, we first introduce the following two definitions to characterize deployed groups on the vehicles at an arbitrary time instant $t$.\par
\begin{definition}[Deployment Map]\label{deployment}
At time instant t, consider partition {\small$\mathcal{P}(\mathcal{A}){=}\left(\mathcal{G}, \mathcal{E}\right)$} with $k$ groups. Deployment map of partition {\small$\mathcal{P}(\mathcal{A})$} at time $t$, referred to by {\small$\mathcal{M}\big(\mathcal{P}(\mathcal{A}),t\big){=}\mathbf{M}{=}[M_{h,\ell}]_{1\leq h\leq k, 1\leq \ell \leq r}$}, is a function outputting a $k\times r$ binary matrix {\small$\mathbf{M}$} at time $t$, where {\small$M_{h,\ell}{=}1$} implies group {\small$G_h{\in}\mathcal{G}$} is deployed on vehicle {\small$C_{\ell}\in\mathcal{C}$}.
\end{definition}
\begin{definition}[Deployment]
Consider deployment map {\small$\mathcal{M}\big(\mathcal{P}(\mathcal{A}),t\big)$}. Deployment at time $t$, shown by {\small$\mathcal{D}\Big(\mathcal{M}\big(\mathcal{P}(\mathcal{A}),t\big)\Big)$}, is a function outputting the set of all groups deployed on the vehicles at time $t$. That is
\begin{equation}
    \mathcal{D}\Big(\mathcal{M}\big(\mathcal{P}(\mathcal{A}),t\big)\Big)=\left\{G_h\Big| G_h\in \mathcal{G},1\leq \sum_{\ell=1}^{r}\mathbf{M}_{h,\ell}\right\}.
\end{equation}
\end{definition}
We next define {\small$\Theta\Big(\mathcal{M}\big(\mathcal{P}(\mathcal{A}),t\big)\Big)$} as a function outputting the set of all groups that one of their vehicles is departed from the VC. Mathematically,
\begin{equation}\label{theta}
    \Theta\Big(\mathcal{M}\big(\mathcal{P}(\mathcal{A}),t\big)\Big){=}\left\{G_h\Big| G_h{\in} \mathcal{G},\sum_{\ell=1}^{r}\mathbf{M}_{h,\ell}=1\right\}.
\end{equation}
To define \textit{C-MTTF}, we make an assumption on the applications, which holds in practical systems\cite{1,21,44,45,46}.
\begin{assumption}\label{assumption1}
Consider {\small$\mathcal{M}\big(\mathcal{P}(\mathcal{A}),t\big)$} at time $t$. If vehicle $C_{\ell}$, allocated to {\small$G_h\in \Theta\Big(\mathcal{M}\big(\mathcal{P}(\mathcal{A}),t\big)\Big)$} (i.e., the only remaining vehicle processing $G_h$), leaves the VC before completing its recruitment, the execution of $\mathcal{A}$ will encounter a failure.
\end{assumption}
To clarify Assumption \ref{assumption1}, consider the DAG presented in Fig.~\ref{RPVCFreamwork}. Assume that vehicles $C_2$ and $C_4$ are departed the VC. In this situation, the sub-tasks belonging to $G_1$ and $G_2$ are only processed by vehicles $C_1$ and $C_3$, respectively. Therefore, the departure of $C_1$ from the VC, while it processes $T_1$, results in the failure of $T_1$, leading to a chain of failures of {\small$\{T_2,T_3,\dotsc,T_6\}$}. Otherwise, if $C_1$ processes $T_1$ successfully, $T_2$ and $T_3$ will be processed by $C_1$ and $C_3$ in parallel. At this moment, if $C_1$ departs the VC, $T_2$ will be failed, leading to a chain of failures of {\small$\{T_4,T_5,T_6\}$}. Due to Assumption \ref{assumption1}, although the dependency between $T_1$ and $T_2$ is from $T_1$ to $T_2$ (i.e., vehicle $C_3$, which processes $T_2$, should wait until $C_1$ processes $T_1$), the departure of $C_3$ leads to the failure of $\mathcal{A}$, because it is costly (e.g., in terms of power consumption) to re-offload $G_2$. Hence, the failure of {\small$G_h\in\mathcal{D}\Big(\mathcal{M}\big(\mathcal{P}(\mathcal{A}),t\big)\Big)$} leads to the failure of $\mathcal{A}$.\par
Considering Assumption \ref{assumption1}, we next define \textit{C-MTTF} below.
\begin{definition}[\textit{C-MTTF}]\label{MTTFDef}
Consider {\small$\mathcal{M}\big(\mathcal{P}(\mathcal{A}),t\big)$} at time $t$. Also, let {\small$\mathcal{R}(t){=}\Big|\Theta\Big(\mathcal{M}\big(\mathcal{P}(\mathcal{A}),t\big)\Big)\Big|$}. \textit{C-MTTF} of application $\mathcal{A}$ given {\small$\mathcal{M}\big(\mathcal{P}(\mathcal{A}),t\big)$}, $\mathcal{R}(t)$, and $t$, referred to by \textit{C-MTTF}{\small$\Big(\mathcal{A}\big|\mathcal{M}\big(\mathcal{P}(\mathcal{A}),t\big),\mathcal{R}(t),t\Big)$}, is the expected time until one of the groups {\small$G_h {\in} \mathcal{D}\Big(\mathcal{M}\big(\mathcal{P}(\mathcal{A}),t\big)\Big)$} encounters a failure.
\end{definition}
It can be construed that \textit{C-MTTF} extends the definition of MTTF~\cite{24} to a scenario in which the application consists of multiple sub-tasks.
We next define the following non-trivial problem that we carefully investigate in subsequent sections.
\begin{problem}\label{p2}
Consider an application $\mathcal{A}$ partitioned by {\small$\mathcal{P}(\mathcal{A})$}. What is the impact of distributed processing of the sub-tasks of $\mathcal{A}$ according to deployment map {\small$\mathcal{M}\big(\mathcal{P}(\mathcal{A}),t\big)$} on the \textit{C-MTTF} of CI-App processing in a semi-dynamic VC? Precisely, what is \textit{C-MTTF}{\small$\Big(\mathcal{A}\big|\mathcal{M}\big(\mathcal{P}(\mathcal{A}),t\big),\mathcal{R}(t),t\Big)$} at time $t$?
\end{problem}
The main contribution of this paper is to present a mathematical framework to study Problem \ref{p2}. To this end, we conduct theoretical analysis to study the C-MTTF of {\tt RP-VC}. Note that we focus on the reliable processing of \textit{one} application given the dynamics of the VC \cite{24,29,30}. Hereafter, we use notation {\tt RP-VC$_n$} to specify the number of groups in {\small$\mathcal{D}\Big(\mathcal{M}\big(\mathcal{P}(\mathcal{A}),t\big)\Big)$} (i.e., {\small$n=\Big| \mathcal{D}\Big(\mathcal{M}\big(\mathcal{P}(\mathcal{A}),t\big)\Big)\Big|$}).
\section{Modeling CI-App Processing as an SMP}\label{mathematicalModel}
\noindent In this section, we model the dynamics of {\tt RP-VC$_n$}. To this end, we first discuss semi-Markov process (Sec.~\ref{smpDefinition}), and then overview the SMP model of {\tt J}$_n$ (Sec.~\ref{smpModel}). Finally, we introduce the SMP model of {\tt RP-VC$_n$} (Sec.~\ref{BSMPR2n}).
\vspace{-2mm}
\subsection{Semi-Markov Process: Definitions and Preliminaries}\label{smpDefinition}
A Markov process is used to describe the dynamics of stochastic systems with limited information about their past, defined as follows.
\begin{definition}[Markov Process]\label{mp}
A random process {\small$\{X_{t}:t\ge 0\}$} is a Markov process (MP) if, for any {\small$s,t\ge 0$}, the conditional probability of transition to the future state {\small$X_{s+t}$} given the present state {\small$X_{s}$} and history {\small$X_{u}$}, for $0\le u < s$, depends only on the present state. Mathematically, for non-negative integer values $i$, $j$, and $x(u)$, where $x(u)$ is an arbitrary function, we have
\begin{equation}\label{markovDef}
\begin{aligned}
    P(X_{s+t}=i &| X_{s}=j, X_{u}=x(u), 0\le u < s)\\
    &= P(X_{s+t}=i|X_{s}=j).
\end{aligned}
\end{equation}
\end{definition}
Definition~\ref{mp} is called Markov property, implying that sojourn time (i.e., the duration that SMP stays in a state before transitioning to other states) in each state has \textit{memoryless} property and thus follows an exponential distribution~\cite{53}.
A generalization of Definition~\ref{mp} where the sojourn time in a state can follow an arbitrary distribution is called semi-Markov process (SMP)\cite{53}. In short, in semi-Markov process, the transition to a future state depends not only on the present state but also on the amount of time spent in it; however, in Markov process, the transition to a future state depends only on the present state. Semi-Markov process captures the scenario of our interest in this paper, where the expected sojourn time in each state (e.g., the time that it takes to successfully recruit a new vehicle) is a function of time the system has spent in a state itself (e.g., the time spent on recruiting from when the system has entered the state).

\vspace{-2mm}
\subsection{SMP Model of {\tt J}$_n$ Overview}\label{smpModel}
Fig.~\ref{SMPJn} illustrates the SMP model of the state-of-the-art redundancy-based application execution strategy {\tt J}$_n$ for a single application $\mathcal{A}$\cite{29}. In this figure, state {\small$(C_1C_2...C_n)$} is the initial state and refers to a situation in which $n$ vehicles process application $\mathcal{A}$ simultaneously and independently. Generally, in all states {\small$\big\{(C_1C_2...C_{n-i}),(C_1C_3...C_{n-i+1}),\dotsc,(C_2C_3...C_{n-i+1})\big\}$}, {\small$n-i$} vehicles process application $\mathcal{A}$ and $i$ vehicles are departed from the VC. In this SMP model, the departure of different vehicles from the VC results in transitions to different states, captured via the links depicted in Fig.~\ref{SMPJn}. This SMP model is large and hard to handle specially since it has similar states. For instance, although in both states {\small$(C_1C_3...C_n)$} and {\small$(C_2C_3...C_n)$} one vehicle is departed from the VC, they are considered distinguished. The SMP model presented in Fig.~\ref{SMPJn} can be represented more simply and concisely by considering the symmetry of states as depicted in Fig.~\ref{comSMPJn}. For example, all states {\small$\big\{(C_1C_2...C_{n-i}),(C_1C_3...C_{n-i+1}),\dotsc,(C_2C_3...C_{n-i+1})\big\}$} are equivalent and can be wrapped into a more compact state {\small$O_{i}$}, depicted in Fig.~\ref{comSMPJn}. To elaborate on the compact model, state {\small$O_n$} on top of the figure represents having $n$ vehicles processing an application $\mathcal{A}$. Further, state {\small$O_{n-1}$} represents a state in which one vehicle departs the VC and one of the other $n-1$ vehicles recruits a new vehicle. Generally, there are $n-i$ vehicles in state {\small$O_{n-i}$}, and $i$ vehicles are departed from the VC. In state {\small$O_{n-i}$}, there is only one recruiter responsible for recruiting new vehicles to replace the departed vehicles, since all the $n-i$ vehicles execute the same application. Furthermore, completing the recruitment operation causes a transition to the initial state (i.e., {\small$O_n$}) with a probability of $p_i$. Conversely, departing a vehicle in state {\small$O_{n-i}$} results in a transition to state {\small$O_{n-i-1}$} with probability of $1-p_i$\cite{29}.
\vspace{-2mm}
\subsection{SMP Model of {\tt RP-VC$_n$}}\label{BSMPR2n}
We present a theoretical model to characterize the behavior of {\tt RP-VC$_n$}, exploited later to calculate \textit{C-MTTF}{\small$\Big(\mathcal{A}\big|\mathcal{M}\big(\mathcal{P}(\mathcal{A}),t\big),\mathcal{R}(t),t\Big)$}. We first present a general class of stochastic systems that we refer to as \textit{stochastic event systems (SeS)}, utilized to introduce the SMP model of {\tt RP-VC$_n$}.\par

\begin{definition}[Stochastic Event System (SeS)]\label{sesDef}
A SeS is a stochastic system modeled as a tuple {\small$\Delta\triangleq(\zeta,\mathcal{Q})$}, where {\small$\zeta{=}\{e_1,\dotsc,e_r\}$} is an event set and {\small$\mathcal{Q}{=}\{Q_1,\dotsc,Q_r\}$} denotes a set of random variables. Each {\small$Q_a{\in} \mathcal{Q}$} describes the time until occurring event {\small$e_a\in \zeta$}. In SeS, the following conditions hold:
\begin{enumerate}
    \item At time $t$, event {\small$e_{a}\in \zeta$} is scheduled to occur after $Q_{a}$ random time units.
    \item The occurrence of each event $e_{a}$ can lead to scheduling a set of $s$ events {\small$\{e_{a_1},\dotsc,e_{a_s}\}\subset\zeta$}.
\end{enumerate}
\end{definition}
In this paper, we consider $s=1$ for analytical simplicity (i.e., the occurrence of each event leads to scheduling of one event). One of the main characteristics of SeS is that the occurrence of the events are \textit{dependent} (i.e., an event is scheduled only at time $t_0$ or by the occurrence of another event). To characterize the events of an SeS, we present following property.
\begin{definition}[$\beta$-inhomogeneous]
Consider an SeS {\small$\Delta=(\zeta,\mathcal{Q})$}. An event {\small$e_a\in \zeta$} is $\beta$-inhomogeneous if after occurring event {\small$e_{a'\neq a}$}, the residual occurring time of event $e_a$ can follow $\beta$ different possible distributions.
\end{definition}\noindent
We next model the VC system depicted in Fig.~\ref{architecture} as an SeS.
\subsubsection{Semi-Dynamic VC as a Stochastic Event System}\label{SDSeS}
The semi-dynamic VC presented in Fig.~\ref{architecture} under the supervision of {\tt RP-VC$_n$} is an SeS, referred to by ${\mathsf{SeS}}_{\mathsf{VC}}$, defined as follows.
\begin{definition}[${\mathsf{SeS}}_{\mathsf{VC}}$]\label{SeSVC}
Let {\small$\mathcal{C}$} denote the set of all the vehicles utilized to process $\mathcal{A}$. Also, let the set of i.i.d random variables $\mathcal{Z}{=}\{Z_1,\dotsc,Z_{|\mathcal{C}|}\}$ and the set of i.i.d random variables $\mathcal{U}{=}\{U_1,\dotsc,U_{|\mathcal{C}|}\}$ denote the sojourn times and the recruitment duration of the vehicles in $\mathcal{C}$. The semi-dynamic VC presented in Fig.~\ref{architecture} under the supervision of {\tt RP-VC$_n$} is an SeS, referred to by ${\mathsf{SeS}}_{\mathsf{VC}}=(\zeta,\mathcal{Q})$, where {\small$\zeta=\{e^{\mathsf{D}}(C_\ell),e^{\mathsf{R}}(C_\ell)\}_ {C_\ell\in \mathcal{C}}$} for all vehicles and $\mathcal{Q}=\mathcal{Z}\cup \mathcal{U}$. Here, $e^{\mathsf{D}}(C_\ell)$ and $e^{\mathsf{R}}(C_\ell)$ are departure and recruitment events, defined below:
\begin{description}[font=$\bullet$~\normalfont]
  \item $e^{\mathsf{D}}(C_\ell)$: Vehicle $C_\ell$ departs the VC.
  \item $e^{\mathsf{R}}(C_\ell)$: Vehicle $C_\ell$ completes its recruitment operation.
\end{description}
\end{definition}
To clarify ${\mathsf{SeS}}_{\mathsf{VC}}$, consider an application $\mathcal{A}\hspace{-1mm}=\hspace{-1mm}\big(\mathcal{V},\mathcal{E}\big)$, where  {\small$\mathcal{V}\hspace{-1mm}=\hspace{-1mm}\{T_1,T_2,\dotsc,T_4\}$} and {\small$\mathcal{E}\hspace{-1mm}=\hspace{-1mm}\{(1,2),(1,3),(2,4)\}$}. Also, let {\small$\mathcal{A}$} be grouped by partition {\small$\mathcal{P}(\mathcal{A})\hspace{-1mm}=\hspace{-1mm}\big(\mathcal{G},\mathcal{E}\big)$} into two groups {\small$\mathcal{G}\hspace{-1mm}=\hspace{-1mm}\{G_1,G_2\}$}, where {\small$G_1\hspace{-1mm}=\hspace{-1mm}\{T_1,T_2\}$} and {\small$G_2\hspace{-1mm}=\hspace{-1mm}\{T_3,T_4\}$}. Further, assume that $G_1$ is deployed on vehicles $C_1$ and $C_2$, and $G_2$ is deployed on vehicles $C_3$ and $C_4$. Let $\mathcal{Z}{=}\{Z_1,Z_2,Z_3, Z_4\}$ and $\mathcal{U}{=}\{U_1,U_2,U_3,U_4\}$ denote the sojourn times and the recruitment duration of vehicles in $\mathcal{C}{=}\{C_1,C_2,C_3,C_4\}$. We have ${\mathsf{SeS}}_{\mathsf{VC}}{=}(\zeta,\mathcal{Q})$, where {\small$\zeta=\{e^{\mathsf{D}}(C_{1}),\dotsc,e^{\mathsf{D}}(C_{4}), e^{\mathsf{R}}(C_1),\dotsc, e^{\mathsf{R}}(C_4)\}$}. Fig.~\ref{SES} depicts ${\mathsf{SeS}}_{\mathsf{VC}}=(\zeta,\mathcal{Q})$, in which the departure events {\small$\{e^{\mathsf{D}}(C_{1}),e^{\mathsf{D}}(C_{2}),e^{\mathsf{D}}(C_{3}),e^{\mathsf{D}}(C_{4})\}\subset \zeta$} are scheduled at time $0$ to occur at times $Z_1$, $Z_2$, $Z_3$, and $Z_4$, respectively. The occurrence of events {\small$e^{\mathsf{D}}(C_{1})$} and {\small$e^{\mathsf{D}}(C_{3})$} results in scheduling recruitment events {\small$e^{\mathsf{R}}(C_{2})$} and {\small$e^{\mathsf{R}}(C_{4})$} at times $t_1$ and $t_2$, occurred at times $t_1+U_2$ and $t_2+U_4$. As can be seen from Fig.~\ref{SES}, the residual times of random variables can change after occurring each event. For instance, recruitment event {\small$e^{\mathsf{R}}(C_{2})$} is $\beta$-inhomogeneous with $\beta=3$ since it is scheduled at time $t_1$ with initial recruitment duration $U_1$. After occurring event {\small$e^{\mathsf{D}}(C_{3})$} at time $t_2$, the residual recruitment duration of {\small$e^{\mathsf{R}}(C_{2})$} is {\small$U_2-(Z_3-Z_1)$}. Finally, after occurring {\small$e^{\mathsf{R}}(C_{4})$} at time $t_3$, the residual recruitment duration of {\small$e^{\mathsf{R}}(C_{2})$} is {\small$(U_2-(Z_3-Z_1))-U_4$}.\par
\begin{figure}
    \centering
\tikzset{every picture/.style={line width=0.2pt}} 

\begin{tikzpicture}[x=0.64pt,y=0.45pt,yscale=-1,xscale=1]

\draw  [color={rgb, 255:red, 0; green, 0; blue, 0 }  ,draw opacity=0.47 ][fill={rgb, 255:red, 167; green, 195; blue, 255 }  ,fill opacity=0.16 ][dash pattern={on 0.84pt off 2.51pt}] (4.3,172.27) -- (401.95,172.27) -- (401.95,226.73) -- (4.3,226.73) -- cycle ;
\draw  [color={rgb, 255:red, 0; green, 0; blue, 0 }  ,draw opacity=0.47 ][fill={rgb, 255:red, 178; green, 178; blue, 178 }  ,fill opacity=0.16 ][dash pattern={on 0.84pt off 2.51pt}] (4.3,137.93) -- (401.22,137.93) -- (401.22,169.6) -- (4.3,169.6) -- cycle ;
\draw  [color={rgb, 255:red, 0; green, 0; blue, 0 }  ,draw opacity=0 ][fill={rgb, 255:red, 255; green, 32; blue, 32 }  ,fill opacity=0.35 ] (39.31,148.67) -- (174.99,148.67) -- (174.99,161.33) -- (39.31,161.33) -- cycle ;
\draw  [color={rgb, 255:red, 0; green, 0; blue, 0 }  ,draw opacity=0 ][fill={rgb, 255:red, 255; green, 32; blue, 32 }  ,fill opacity=0.35 ] (177.1,183.73) -- (285.93,183.73) -- (285.93,196.67) -- (177.1,196.67) -- cycle ;
\draw  [color={rgb, 255:red, 0; green, 0; blue, 0 }  ,draw opacity=0.47 ][fill={rgb, 255:red, 167; green, 195; blue, 255 }  ,fill opacity=0.16 ][dash pattern={on 0.84pt off 2.51pt}] (4.3,74.27) -- (401.22,74.27) -- (401.22,134.93) -- (4.3,134.93) -- cycle ;
\draw  [color={rgb, 255:red, 0; green, 0; blue, 0 }  ,draw opacity=0.47 ][fill={rgb, 255:red, 178; green, 178; blue, 178 }  ,fill opacity=0.16 ][dash pattern={on 0.84pt off 2.51pt}] (4.3,46.27) -- (400.68,46.27) -- (400.68,71.6) -- (4.3,71.6) -- cycle ;
\draw  [color={rgb, 255:red, 0; green, 0; blue, 0 }  ,draw opacity=0 ][fill={rgb, 255:red, 255; green, 250; blue, 32 }  ,fill opacity=0.35 ] (82.33,91.25) -- (384.37,91.25) -- (384.37,103.67) -- (82.33,103.67) -- cycle ;
\draw  [color={rgb, 255:red, 0; green, 0; blue, 0 }  ,draw opacity=0 ][fill={rgb, 255:red, 32; green, 255; blue, 115 }  ,fill opacity=0.35 ] (38.6,53.67) -- (84.9,53.67) -- (84.9,66.33) -- (38.6,66.33) -- cycle ;
\draw    (39,61) -- (83.8,61) ;
\draw [shift={(83.8,61)}, rotate = 180] [color={rgb, 255:red, 0; green, 0; blue, 0 }  ][line width=0.75]    (0,3.35) -- (0,-3.35)(-3.01,3.35) -- (-3.01,-3.35)   ;
\draw    (39,155) -- (174.3,155) ;
\draw [shift={(174.3,155)}, rotate = 180] [color={rgb, 255:red, 0; green, 0; blue, 0 }  ][line width=0.75]    (0,3.35) -- (0,-3.35)(-3.01,3.35) -- (-3.01,-3.35)   ;
\draw    (83.87,97) -- (378.88,97) ;
\draw [shift={(378.88,97)}, rotate = 180] [color={rgb, 255:red, 0; green, 0; blue, 0 }  ][line width=0.75]    (0,3.35) -- (0,-3.35)(-3.01,3.35) -- (-3.01,-3.35)   ;
\draw [color={rgb, 255:red, 255; green, 1; blue, 1 }  ,draw opacity=1 ] [dash pattern={on 0.84pt off 2.51pt}]  (86.8,32) -- (86.8,212.93) ;
\draw  [fill={rgb, 255:red, 255; green, 0; blue, 205 }  ,fill opacity=1 ] (82,154.9) .. controls (82,152.19) and (84.19,150) .. (86.9,150) .. controls (89.61,150) and (91.8,152.19) .. (91.8,154.9) .. controls (91.8,157.61) and (89.61,159.8) .. (86.9,159.8) .. controls (84.19,159.8) and (82,157.61) .. (82,154.9) -- cycle ;
\draw  [fill={rgb, 255:red, 255; green, 0; blue, 205 }  ,fill opacity=1 ] (82,96.9) .. controls (82,94.19) and (84.19,92) .. (86.9,92) .. controls (89.61,92) and (91.8,94.19) .. (91.8,96.9) .. controls (91.8,99.61) and (89.61,101.8) .. (86.9,101.8) .. controls (84.19,101.8) and (82,99.61) .. (82,96.9) -- cycle ;
\draw    (182,189.9) -- (278.52,189.9) ;
\draw [shift={(278.52,189.9)}, rotate = 180] [color={rgb, 255:red, 0; green, 0; blue, 0 }  ][line width=0.75]    (0,3.35) -- (0,-3.35)(-3.01,3.35) -- (-3.01,-3.35)   ;
\draw  [fill={rgb, 255:red, 255; green, 0; blue, 205 }  ,fill opacity=1 ] (172.2,189.9) .. controls (172.2,187.19) and (174.39,185) .. (177.1,185) .. controls (179.81,185) and (182,187.19) .. (182,189.9) .. controls (182,192.61) and (179.81,194.8) .. (177.1,194.8) .. controls (174.39,194.8) and (172.2,192.61) .. (172.2,189.9) -- cycle ;
\draw [color={rgb, 255:red, 255; green, 1; blue, 1 }  ,draw opacity=1 ] [dash pattern={on 0.84pt off 2.51pt}]  (177.1,32) -- (177.1,212.93) ;
\draw [color={rgb, 255:red, 255; green, 1; blue, 1 }  ,draw opacity=1 ] [dash pattern={on 0.84pt off 2.51pt}]  (279.8,32.87) -- (279.8,213.8) ;
\draw  [fill={rgb, 255:red, 255; green, 0; blue, 205 }  ,fill opacity=1 ] (274.9,97.9) .. controls (274.9,95.19) and (277.09,93) .. (279.8,93) .. controls (282.51,93) and (284.7,95.19) .. (284.7,97.9) .. controls (284.7,100.61) and (282.51,102.8) .. (279.8,102.8) .. controls (277.09,102.8) and (274.9,100.61) .. (274.9,97.9) -- cycle ;
\draw  [fill={rgb, 255:red, 255; green, 0; blue, 205 }  ,fill opacity=1 ] (172.2,96.9) .. controls (172.2,94.19) and (174.39,92) .. (177.1,92) .. controls (179.81,92) and (182,94.19) .. (182,96.9) .. controls (182,99.61) and (179.81,101.8) .. (177.1,101.8) .. controls (174.39,101.8) and (172.2,99.61) .. (172.2,96.9) -- cycle ;
\draw [color={rgb, 255:red, 0; green, 87; blue, 255 }  ,draw opacity=1 ]   (32.16,42) -- (394.15,42) (78.16,38) -- (78.16,46)(124.16,38) -- (124.16,46)(170.16,38) -- (170.16,46)(216.16,38) -- (216.16,46)(262.16,38) -- (262.16,46)(308.16,38) -- (308.16,46)(354.16,38) -- (354.16,46) ;
\draw [shift={(394.15,42)}, rotate = 180] [color={rgb, 255:red, 0; green, 87; blue, 255 }  ,draw opacity=1 ][line width=0.75]    (0,3.35) -- (0,-3.35)   ;
\draw [shift={(32.16,42)}, rotate = 180] [color={rgb, 255:red, 0; green, 87; blue, 255 }  ,draw opacity=1 ][line width=0.75]    (0,3.35) -- (0,-3.35)   ;
\draw [color={rgb, 255:red, 0; green, 255; blue, 66 }  ,draw opacity=1 ]   (67,60.33) .. controls (89.36,31.41) and (41.75,36.39) .. (67.17,14.1) ;
\draw [shift={(69.27,12.33)}, rotate = 140.95] [fill={rgb, 255:red, 0; green, 255; blue, 66 }  ,fill opacity=1 ][line width=0.08]  [draw opacity=0] (3.57,-1.72) -- (0,0) -- (3.57,1.72) -- cycle    ;
\draw [color={rgb, 255:red, 255; green, 0; blue, 0 }  ,draw opacity=1 ]   (230,192.17) .. controls (238.11,212.97) and (211.34,226.49) .. (189.72,236.82) ;
\draw [shift={(187.05,238.1)}, rotate = 334.49] [fill={rgb, 255:red, 255; green, 0; blue, 0 }  ,fill opacity=1 ][line width=0.08]  [draw opacity=0] (3.57,-1.72) -- (0,0) -- (3.57,1.72) -- cycle    ;
\draw [color={rgb, 255:red, 255; green, 0; blue, 0 }  ,draw opacity=1 ]   (74.3,160.8) .. controls (64.4,215.25) and (55.48,233.44) .. (149.42,237.68) ;
\draw [shift={(152.3,237.8)}, rotate = 182.36] [fill={rgb, 255:red, 255; green, 0; blue, 0 }  ,fill opacity=1 ][line width=0.08]  [draw opacity=0] (3.57,-1.72) -- (0,0) -- (3.57,1.72) -- cycle    ;
\draw [color={rgb, 255:red, 236; green, 246; blue, 21 }  ,draw opacity=1 ]   (164.45,96.17) .. controls (200.09,36.57) and (110.5,13.83) .. (219.28,9.66) ;
\draw [shift={(220.93,9.6)}, rotate = 177.94] [fill={rgb, 255:red, 236; green, 246; blue, 21 }  ,fill opacity=1 ][line width=0.08]  [draw opacity=0] (3.57,-1.72) -- (0,0) -- (3.57,1.72) -- cycle    ;
\draw    (39.93,121.33) -- (388.95,121.33) ;
\draw [shift={(388.95,121.33)}, rotate = 180] [color={rgb, 255:red, 0; green, 0; blue, 0 }  ][line width=0.75]    (0,3.35) -- (0,-3.35)(-3.01,3.35) -- (-3.01,-3.35)   ;
\draw  [fill={rgb, 255:red, 255; green, 0; blue, 205 }  ,fill opacity=1 ] (82,120.9) .. controls (82,118.19) and (84.19,116) .. (86.9,116) .. controls (89.61,116) and (91.8,118.19) .. (91.8,120.9) .. controls (91.8,123.61) and (89.61,125.8) .. (86.9,125.8) .. controls (84.19,125.8) and (82,123.61) .. (82,120.9) -- cycle ;
\draw  [fill={rgb, 255:red, 255; green, 0; blue, 205 }  ,fill opacity=1 ] (274.9,121.9) .. controls (274.9,119.19) and (277.09,117) .. (279.8,117) .. controls (282.51,117) and (284.7,119.19) .. (284.7,121.9) .. controls (284.7,124.61) and (282.51,126.8) .. (279.8,126.8) .. controls (277.09,126.8) and (274.9,124.61) .. (274.9,121.9) -- cycle ;
\draw  [fill={rgb, 255:red, 255; green, 0; blue, 205 }  ,fill opacity=1 ] (172.2,120.9) .. controls (172.2,118.19) and (174.39,116) .. (177.1,116) .. controls (179.81,116) and (182,118.19) .. (182,120.9) .. controls (182,123.61) and (179.81,125.8) .. (177.1,125.8) .. controls (174.39,125.8) and (172.2,123.61) .. (172.2,120.9) -- cycle ;
\draw    (39.93,213.33) -- (388.95,213.33) ;
\draw [shift={(388.95,213.33)}, rotate = 180] [color={rgb, 255:red, 0; green, 0; blue, 0 }  ][line width=0.75]    (0,3.35) -- (0,-3.35)(-3.01,3.35) -- (-3.01,-3.35)   ;
\draw  [fill={rgb, 255:red, 255; green, 0; blue, 205 }  ,fill opacity=1 ] (82,212.9) .. controls (82,210.19) and (84.19,208) .. (86.9,208) .. controls (89.61,208) and (91.8,210.19) .. (91.8,212.9) .. controls (91.8,215.61) and (89.61,217.8) .. (86.9,217.8) .. controls (84.19,217.8) and (82,215.61) .. (82,212.9) -- cycle ;
\draw  [fill={rgb, 255:red, 255; green, 0; blue, 205 }  ,fill opacity=1 ] (274.9,213.9) .. controls (274.9,211.19) and (277.09,209) .. (279.8,209) .. controls (282.51,209) and (284.7,211.19) .. (284.7,213.9) .. controls (284.7,216.61) and (282.51,218.8) .. (279.8,218.8) .. controls (277.09,218.8) and (274.9,216.61) .. (274.9,213.9) -- cycle ;
\draw  [fill={rgb, 255:red, 255; green, 0; blue, 205 }  ,fill opacity=1 ] (172.2,212.9) .. controls (172.2,210.19) and (174.39,208) .. (177.1,208) .. controls (179.81,208) and (182,210.19) .. (182,212.9) .. controls (182,215.61) and (179.81,217.8) .. (177.1,217.8) .. controls (174.39,217.8) and (172.2,215.61) .. (172.2,212.9) -- cycle ;

\draw (49,48.4) node [anchor=north west][inner sep=0.75pt]  [font=\tiny]  {$Z_{1}$};
\draw (50,141.4) node [anchor=north west][inner sep=0.75pt]  [font=\tiny]  {$Z_{3}$};
\draw (102,81.4) node [anchor=north west][inner sep=0.75pt]  [font=\tiny]  {$U_{2}$};
\draw (100,141.4) node [anchor=north west][inner sep=0.75pt]  [font=\tiny]  {$Z_{3}{-}Z_{1}$};
\draw (182.2,175.4) node [anchor=north west][inner sep=0.75pt]  [font=\tiny]  {$U_{4}$};
\draw (180.1,81.4) node [anchor=north west][inner sep=0.75pt]  [font=\tiny]  {$U_{2}{-}( Z_{3}{-}Z_{1})$};
\draw (280.8,79.4) node [anchor=north west][inner sep=0.75pt]  [font=\tiny]  {$U'_2{=}\big( U_{2}{-}( Z_{3}{-}Z_{1})\big){-}U_{4}$};
\draw (3,50.4) node [anchor=north west][inner sep=0.75pt]  [font=\tiny]  {$e^{\mathsf{D}}(C_1)$};
\draw (3,146.4) node [anchor=north west][inner sep=0.75pt]  [font=\tiny]  {$e^{\mathsf{D}}(C_3)$};
\draw (46,87.4) node [anchor=north west][inner sep=0.75pt]  [font=\tiny]  {$e^{\mathsf{R}}(C_2)$};
\draw (137,180.4) node [anchor=north west][inner sep=0.75pt]  [font=\tiny]  {$e^{\mathsf{R}}(C_4)$};
\draw (35,22.4) node [anchor=north west][inner sep=0.75pt]  [font=\tiny]  {$0$};
\draw (83,22.4) node [anchor=north west][inner sep=0.75pt]  [font=\tiny]  {$t_{1}$};
\draw (172,22.4) node [anchor=north west][inner sep=0.75pt]  [font=\tiny]  {$t_{2}$};
\draw (274,22.4) node [anchor=north west][inner sep=0.75pt]  [font=\tiny]  {$t_{3}$};
\draw (72.33,4.07) node [anchor=north west][inner sep=0.75pt]  [font=\tiny]  {$Pivot\ of\ U'_{2}$};
\draw (144.3,240.2) node [anchor=north west][inner sep=0.75pt]  [font=\tiny]  {$Reducer\ of\ U'_{2}$};
\draw (224.5,3.73) node [anchor=north west][inner sep=0.75pt]  [font=\tiny]  {$Duration\ of\ U'_{2}$};
\draw (3,112.73) node [anchor=north west][inner sep=0.75pt]  [font=\tiny]  {$e^{\mathsf{D}}(C_2)$};
\draw (51,109.07) node [anchor=north west][inner sep=0.75pt]  [font=\tiny]  {$Z_{2}$};
\draw (100,108.4) node [anchor=north west][inner sep=0.75pt]  [font=\tiny]  {$Z_{2}{-}Z_{1}$};
\draw (182,109.4) node [anchor=north west][inner sep=0.75pt]  [font=\tiny]  {$Z_{2}{-}Z_{3}$};
\draw (287,109.4) node [anchor=north west][inner sep=0.75pt]  [font=\tiny]  {$Z_{2}{-}Z_{3}{-}U_{4}$};
\draw (3,202.73) node [anchor=north west][inner sep=0.75pt]  [font=\tiny]  {$e^{\mathsf{D}}(C_4)$};
\draw (51,201.07) node [anchor=north west][inner sep=0.75pt]  [font=\tiny]  {$Z_{4}$};
\draw (100,200.4) node [anchor=north west][inner sep=0.75pt]  [font=\tiny]  {$Z_{4}{-}Z_{1}$};
\draw (182,201.4) node [anchor=north west][inner sep=0.75pt]  [font=\tiny]  {$Z_{4}{-}Z_{3}$};
\draw (287,201.4) node [anchor=north west][inner sep=0.75pt]  [font=\tiny]  {$Z_{4}{-}Z_{3}{-}U_{4}$};

\end{tikzpicture}

    \caption{\small A stochastic event system consisting of six events.}\label{SES}
\end{figure}
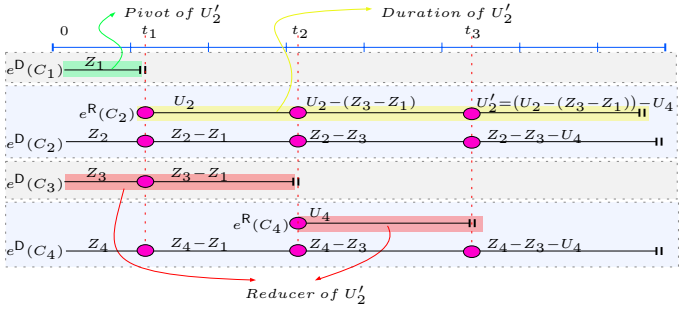
We next present an SMP model to characterize the dynamics of {\tt RP-VC$_n$} in a {\small${\mathsf{SeS}}_{\mathsf{VC}}$}.
\subsubsection{SMP Modeling of {\tt RP-VC$_n$} in an {\small${\mathsf{SeS}}_{\mathsf{VC}}$}}\label{RPVC_SMP}
Consider deployment {\small$\mathcal{D}\Big(\mathcal{M}\big(\mathcal{P}(\mathcal{A}),t\big)\Big)$} at time $t$, where {\small$\Big|\mathcal{D}\Big(\mathcal{M}\big(\mathcal{P}(\mathcal{A}),t\big)\Big)\Big|=n$}. Further, consider state set {\small$\mathcal{S}=\{S_0,S_1,\dotsc,S_n, F\}$}. We formalize the system behavior of {\tt RP-VC$_n$} in a ${\mathsf{SeS}}_{\mathsf{VC}}$ as an SMP model {\small$\{X_t\in\mathcal{S}\}$}, depicted in Fig.~\ref{j2nMarkov1}. In this model, each state {\small$S_i\in \mathcal{S}$} is defined as a set of computing vehicles utilized for processing each group {\small$G_h{\in} \mathcal{D}\Big(\mathcal{M}\big(\mathcal{P}(\mathcal{A}),t\big)\Big)$}. Let {\small$\dot{\mathcal{C}}_i$} denote the set of $i$ vehicles processing $i$ different groups, where one of the vehicles of each group has departed the VC (i.e., each group is being processed by only one vehicle).
Therefore, each vehicle {\small$C_\ell\in \dot{\mathcal{C}}_i$} is a recruiter vehicle. Also, let $\ddot{\mathcal{C}}_{2n-2i}$ denote the set of $2n-2i$ vehicles processing $n-i$ different groups, where each group is being processed by two vehicles. State {\small$S_0{=}\dot{\mathcal{C}}_0 \cup \ddot{\mathcal{C}}_{2n}$} is the initial state where each group {\small$G_h{\in} \mathcal{D}\Big(\mathcal{M}\big(\mathcal{P}(\mathcal{A}),t\big)\Big)$} has two vehicles. In state {\small$S_1{=}\dot{\mathcal{C}}_1 \cup \ddot{\mathcal{C}}_{2n-2}$}, there is a group {\small$G_{h_1}{\in} \mathcal{D}\Big(\mathcal{M}\big(\mathcal{P}(\mathcal{A}),t\big)\Big)$} with only one vehicle, and each group {\small$G_h {\in} \mathcal{D}\Big(\mathcal{M}\big(\mathcal{P}(\mathcal{A}),t\big)\Big)\setminus G_{h_1}$} has two vehicles. Generally, state {\small$S_i=\dot{\mathcal{C}}_i \cup \ddot{\mathcal{C}}_{2n-2i}$} represents a situation that $i$ groups {\small$\{G_{h_1},G_{h_2},\dotsc,G_{h_{i}}\}$} have one vehicles, and $n{-}i$ groups, i.e., each group {\small$G_h {\in} \mathcal{D}\Big(\mathcal{M}\big(\mathcal{P}(\mathcal{A}),t\big)\Big){\setminus}\{G_{h_1},G_{h_2},\dotsc,G_{h_{i}}\}$}, have one vehicle. Further, state {\small$F{\in} \mathcal{S}$} refers to the failure of {\small$\mathcal{A}$}.\par
\subsubsection{Dynamics of the SMP Model of {\tt RP-VC$_n$} in an {\small${\mathsf{SeS}}_{\mathsf{VC}}$}}\label{dynamicOfRP}
Initially, {\tt RP-VC$_n$} assigns two vehicles to each group of {\small $\mathcal{D}\Big(\mathcal{M}\big(\mathcal{P}(\mathcal{A}),t\big)\Big)$}. Let $C_{\ell_1}$ and {\small$C_{\ell'_1}$} be two vehicles assigned to a group {\small $G_{h_1}{\in} \mathcal{D}\Big(\mathcal{M}\big(\mathcal{P}(\mathcal{A}),t\big)\Big)$}. In state {\small$S_0{=}\dot{\mathcal{C}}_0 \cup \ddot{\mathcal{C}}_{2n}$}, event {\small$e^{\mathsf{D}}(C_{\ell_1})$} occurs with probability $q_0$ resulting in starting recruitment operation by {\small$C_{\ell'_1}$}. In this situation, the SMP transits to state {\small$S_1{=}\dot{\mathcal{C}}_1 {\cup} \ddot{\mathcal{C}}_{2n-2}$} since $C_{\ell_1}$ has departed the VC and should be removed from {\small$\ddot{\mathcal{C}}_{2n}$}. Besides, {\small$C_{\ell'_1}$} is now a recruiter which should be removed from {\small$\ddot{\mathcal{C}}_{2n}$} and added to $\dot{\mathcal{C}}_0$. In state $S_1$, event {\small$e^{\mathsf{D}}(C_{\ell_2})$} occurs with probability $b_1$ before completing recruitment operation of {\small$C_{\ell'_1}$} and the departure of any other vehicles, leading to the failure of the processing of {\small$\mathcal{A}$}. Conversely, in state $S_1$, event {\small$e^{\mathsf{R}}(C_{\ell'_1})$} occurs with probability $p_1$, and {\small$C_{\ell'_1}$} successfully recruits a new vehicle to allocate to {\small$G_{h_1}$}, leading to a transition of the SMP to state $S_0$ since there is no recruiter and all the groups have two vehicles. Further, let {\small$C_{\ell_2}$} and {\small$C_{\ell'_2}$} be two vehicles assigned to group {\small$G_{h_2} {\in} \mathcal{D}\Big(\mathcal{M}\big(\mathcal{P}(\mathcal{A}),t\big)\Big){\setminus} G_{h_1}$}. In state $S_1$, event {\small$e^{\mathsf{D}}(C_{\ell_2})$} occurs with probability $q_1$ before any other events. In this situation, the SMP transits to state $S_2$, and {\small$C_{\ell'_2}$} starts recruiting a new vehicle. In state $S_2$, {\small$C_{\ell'_1}$} and {\small$C_{\ell'_2}$} recruit new vehicles for groups {\small$G_{h_1}$} and {\small$G_{h_2}$}, independently. The above events can happen in the other states as well.\par
\subsubsection{Complexity of the SMP Model of {\tt RP-VC$_n$}}\label{complexityOfRP}
In addition to considering the execution of dependent sub-tasks in {\tt RP-VC$_n$}, which is fundamentally different from atomic application modeling in {\tt J}$_n$, the main differences between the SMP of {\tt RP-VC$_n$} and {\tt J}$_n$ can be outlined as follows:
\begin{description}[font=$\bullet$~\normalfont]
    \item In {\tt RP-VC$_n$}, each group $G_h$ has its own unique image/replica and is processed by two vehicles, e.g., $C_{\ell}$ and $C_{\ell'}$, independently. If $C_{\ell}$ departs the VC, only $C_{\ell'}$ can recruit a new vehicle to process sub-tasks of $G_h$. As a result, $i$ simultaneous active recruiters coexist in state $S_i$. Moreover, in state {\small$S_i$}, if a recruiter finishes its recruitment operation, the system will transit to state {\small$S_{i-1}$}. However, in {\tt J}$_n$, whenever the only existing recruiter completes its recruitment, the SMP will transit to the initial state {\small$O_n$}.
    \item In {\tt RP-VC$_n$}, processing application $\mathcal{A}$ fails with probability $b_i$ from each state $S_i$, where {\small$0 {<} i {\leq} n$} (due to the dependent sub-tasks structure of $\mathcal{A}$). While, in {\tt J}$_n$, no failure occurs until reaching the last state {\small$O_1$}, since $n$ replications of the application are processed on $n$ different vehicles.
\end{description}
Relying on the $\beta$-inhomogeneous property of the events in an {\small${\mathsf{SeS}}_{\mathsf{VC}}{=}(\zeta,\mathcal{Q})$}, we refer to the SMP presented in Sec.~\ref{BSMPR2n} by \textit{$\beta$-inhomogeneous SMP ($\beta$-SMP)}. As shown in Fig.~\ref{SES}, at each time instant {\small$t$}, residual occurring time of all {\small$e_a{\in} \zeta$} can follow different distributions.\par
To analyze {\tt RP-VC$_n$} in an {\small${\mathsf{SeS}}_{\mathsf{VC}}=(\zeta,\mathcal{Q})$} and consequently obtaining \textit{C-MTTF}{\small$\Big(\mathcal{A}\big|\mathcal{M}\big(\mathcal{P}(\mathcal{A}),t\big),\mathcal{R}(t),t\Big)$}, it is necessary to answer the following questions regarding a general SeS:
\begin{enumerate}[label={(Q\arabic*)}]
    \setcounter{enumi}{2}
    \item\label{QQ1} What are the residual occurring times of events after occurring each event {\small$e_a\in \zeta$}?
    \item\label{QQ2} What is the value of $\beta$ for an $\beta$-inhomogeneous event?
\end{enumerate}\par
To deal with \ref{QQ1}, we introduce $\langle e\rangle$-algebra as a general mathematical environment to ease conducting algebraic operations on random variables in an SeS\footnote{Note that $\langle e\rangle$-algebra is not used to calculate the subtraction of random variables, studied through convolution, and probability characteristic function\cite{54}.}. To answer question \ref{QQ2} regarding ${\mathsf{SeS}}_{\mathsf{VC}}$, we propose DT to disentangle $\beta$-SMP to a decomposed SMP (D-SMP).\par
\vspace{-2mm}
\section{Reliability Analysis via Event Stochastic Algebra and Decomposition Theorem}\label{RACSE}
\noindent We develop two new mathematical frameworks, which we refer to as \textit{event stochastic algebra ($\langle e\rangle$-algebra)} (Sec.~\ref{CSE}) and \textit{decomposition theorem (DT)} (Sec.~\ref{DecT}). These frameworks are highly effective in analyzing a variety of dynamic systems, including the semi-dynamic VC of our interest under {\tt RP-VC$_n$}.
\begin{remark}
Note that all of the subsequent mathematical modeling given in this section is among the first in the literature.
\end{remark}
\vspace{-3mm}
\subsection{Event Stochastic Algebra ($\langle e\rangle$-algebra)}\label{CSE}
The pillars of $\langle e\rangle$-algebra are built on top of two mathematical concepts that we introduce, \textit{event dynamic variable (EDV)} and \textit{event dynamic list (EDL)}, defined as follows.
\subsubsection{Event Dynamic Variable (EDV)}
An EDV represents the difference of two random variables. EDV is utilized to characterize stochastic behavior of a $\beta$-inhomogeneous event in an SeS.  The name EDV stems from the fact that in our SeS the dynamics of the occurrence of the events are represented via the difference of random variables as discussed later. In this paper, we use EDV to define EDL exploited to formulate different states of D-SMP in Sec.~\ref{DecT}.
\begin{definition}[Event Dynamic Variable]\label{event_dynamic_variable}
Consider random variables {\small$M{>}0$} and {\small$N{>}0$}. If {\small$M{>}N$}, random variable {\small$V{=}M{-}N$} is named EDV, and its distribution $\alpha$ is called e-distribution.
\end{definition}
In the above definition, {\small$M$} and {\small$N$} can themselves be EDVs. Note that in the definition of EDV, the internal components of $V$, i.e., $M$ and $N$, make $V$ an EDV; otherwise, we refer to $V$ by \textit{simple random variable (SRV)}.\par
\textbf{Characterizing an EDV.}
Consider an SeS $\Delta=(\zeta,\mathcal{Q})$. Also, let EDV {\small$V=M-N$} be the residual time until occurring event $e_{a}\in\zeta$. Further, let {\small$\alpha$ and $\alpha'$} be two e-distributions. We characterize $V$ through three notions of \textit{duration}, \textit{reducer}, and \textit{pivot}, defined as follows.
\begin{definition}[Duration]\label{duration}
The duration of {\small$V{=}M{-}N$}, shown by function {\small$\hat{\Gamma}(V)$}, is defined as follows:
\begin{equation}
\hat{\Gamma}(V)=
\begin{cases}
  \hat{\Gamma}(M), & \mbox{if } V {\sim}\alpha,\\
  V, & \mbox{if } V\mbox{ is SRV},\\
  0, & \mbox{otherwise}.
\end{cases}
\end{equation}
\end{definition}
Intuitively, {\small$\hat{\Gamma}(V)$} refers to the total time until the occurrence of event $e_{a}$. Referring to Fig.~\ref{SES}, {\small$\hat{\Gamma}(U'_2)=U_2$}.
\begin{definition}[Reducer]\label{reducer}
The reducer of {\small$V{=}M{-}N$}, referred to by function {\small$\hat{\mathbf{\delta}}(V)$}, is defined as follows:
\begin{equation}
\hat{\mathbf{\delta}}(V)=
\begin{cases}
  \hat{\mathbf{\delta}}(M)+\left(N+\hat{\mathbf{\delta}}(N)\right), & \mbox{if } V{\sim}\alpha, \\
  0 , & \mbox{if } V \mbox{ is SRV}.
\end{cases}
\end{equation}
\end{definition}
To clarify the above definition, assume that event $e_{a'}\in\zeta$ is the latest event that has occurred before $e_{a}$. The reducer of $V$ is the time passed from time $0$ until occurrence of $e_{a'}$. As can be seen from Fig.~\ref{SES}, after occurring event $e^{\mathsf{R}}(C_4)$, we have {\small$\hat{\mathbf{\delta}}(U'_2)=Z_3+U_4$}.
\begin{definition}[Pivot]\label{pivot}
The pivot of {\small$V{=}M{-}N$}, referred to by function {\small$\hat{\mathbf{\xi}}(V)$}, is defined as follows:
\begin{equation}
\hat{\mathbf{\xi}}(V)= V+\hat{\mathbf{\delta}}(V)-\hat{\Gamma}(V).
\end{equation}
\end{definition}
Intuitively, the pivot of $V$ refers to the scheduling time of event $e_{a}$. As can be seen from Fig.~\ref{SES}, we have {\small$\hat{\mathbf{\xi}}(U'_2){=}Z_1$}. Considering above definitions, if {\small$\hat{\mathbf{\delta}}(V){=}\hat{\mathbf{\xi}}(V){=}0$}, we refer to $V$ as SRV and its distribution is called \textit{simple distribution}. Next, we introduce an operator to compare different EDVs.
\begin{definition}[Pivotally Greater ({\small$\dot{\preceq}$})]\label{CG}
Consider EDVs {\small$V_1{\sim}\alpha$} and {\small$V_2{\sim}\alpha'$}, where {\small$\hat{\Gamma}(V_1){\neq} \hat{\Gamma}(V_2)$}. The relations {\small$V_1 \dot{\preceq} V_2$} or {\small$V_2 \dot{\succeq} V_1$}, read {\small$V_2$} is pivotally greater than or equal to {\small$V_1$}, are defined as
\begin{equation}
 V_2 \dot{\succeq} V_1 \equiv V_1 \dot{\preceq} V_2 \Leftrightarrow
  \left\{\hat{\mathbf{\delta}}(V_2)=\hat{\mathbf{\delta}}(V_1) \,\&\,\hat{\mathbf{\xi}}(V_1)\le\hat{\mathbf{\xi}}(V_2)\right\}.
\end{equation}
\end{definition}
\noindent By convention {\small$V_1\dot{\preceq}V_2$}, if {\small$V_1$} is an EDV and {\small$V_2$} is a SRV. \par
\textbf{Properties of EDVs.}
We define the following two properties, i.e., \textit{conditional closure} and \textit{subtraction inequality}, which are used later to do algebraic operations on EDVs.
\begin{lemma}[Conditional Closure]\label{conditional_closure}
   Consider EDVs {\small$V_1{\sim}\alpha$} and {\small$V_2{\sim}\alpha'$}, where {\small$V_1 \dot{\preceq} V_2$}. {\small$S_1=V_1-V_2$} is an EDV, calculated as follows, if {\small$S_2>0$}:
   \begin{equation}\label{eq10}
    S_1=V_1-V_2=\left(\hat{\Gamma}(V_1)-\left(\hat{\xi}(V_2)-\hat{\xi}(V_1)\right)\right)-\hat{\Gamma}(V_2),
   \end{equation}
   and {\small$S_2=V_2-V_1$} is an EDV, calculated as follows, if {\small$S_2>0$}:
   \begin{equation}\label{eq11}
    S_2=V_2-V_1=\hat{\Gamma}(V_2)-\left(\hat{\Gamma}(V_1)-\left(\hat{\xi}(V_2)-\hat{\xi}(V_1)\right)\right),
   \end{equation}
\end{lemma}
\noindent where $\hat{\Gamma}(V_1)>\left(\hat{\xi}(V_2)-\hat{\xi}(V_1)\right)>0$.
\begin{proof}
See Appendix \ref{conditional_closure_proof}.
\end{proof}
\begin{lemma}[Subtraction Inequality]\label{subtraction_inequality}
Consider EDVs {\small$V_1{\sim}\alpha$} and {\small$V_2{\sim}\alpha'$}, where {\small$V_1 \dot{\preceq} V_2$}. Also, consider {\small$V'$}, where {\small$V'$} can be EDV or SRV, and {\small$V'<V_1,V_2$}. The following inequality holds.
\begin{equation}
  V_1-V'\dot{\preceq} V_2-V',
\end{equation}
provided that
\begin{equation}
  \left\{V_1 \dot{\preceq} V' or V' \dot{\preceq} V_1\} \,\&\, \{V_2 \dot{\preceq} V' or V' \dot{\preceq} V_2\right\}.
\end{equation}
\end{lemma}
\begin{proof}
See Appendix \ref{subtraction_inequality_proof}.
\end{proof}
So far, we have provided definitions and technical results related to EDVs. As the second pillar of the $\langle e\rangle$-algebra, we next aim to introduce the concept of \textit{EDL}.
\subsubsection{Event Dynamic List (EDL)}\label{CVLsub}
EDL is a list of EDVs, utilized for modeling the states of D-SMP in Sec.~\ref{DecT}.
\begin{definition}[Event Dynamic List]\label{event_dynamic_list}
A list of $n$ EDVs, each following {\small$\alpha$} distribution, referred to by {\small$\mathcal{L}^{(\alpha)}\hspace{-1mm}\left\langle n \right\rangle{=}[V_1,V_2,\dotsc,V_n]$}, is called an EDL, if  $V_i \dot{\preceq} V_{j}$, $\forall V_i,~V_j{\in}\mathcal{L}^{(\alpha)}\hspace{-1mm}\left\langle n \right\rangle$, where $1 {\le} i{<}j {\le} n$.
\end{definition}
\noindent We let {\small$\mathcal{L}^{(\alpha)}\hspace{-1mm}\left\langle n \right\rangle\hspace{-0.7mm}(i)$} denote the $i^{th}$ member of the list.\par
\textbf{Operators To Manipulate EDLs.}
We introduce three essential operators, namely, \textit{pivotally greater}, \textit{concatenation}, and \textit{subtraction}, to manipulate EDLs. First, we extend the definition of operator {\small$\dot{\preceq}$} to compare different EDLs as follows:
\begin{definition}
Consider two EDLs {\small$\mathcal{L}_1^{(\alpha)}\hspace{-1.4mm}\left\langle n \right\rangle$} and {\small$\mathcal{L}_2^{(\alpha')}\hspace{-1.4mm}\left\langle m \right\rangle$}. {\small$\mathcal{L}_1^{(\alpha)}\hspace{-1.4mm}\left\langle n \right\rangle \dot{\preceq} \mathcal{L}_2^{(\alpha')}\hspace{-1.4mm}\left\langle m \right\rangle$} if
\begin{equation}
 V_i \dot{\preceq} V_{j}, \forall V_i\in \mathcal{L}_1^{(\alpha)}\hspace{-1.4mm}\left\langle n \right\rangle, \forall V_j\in \mathcal{L}_2^{(\alpha')}\hspace{-1.4mm}\left\langle m \right\rangle.
 \end{equation}
\end{definition}
We then define the concatenation operator as follows, enabling us to concatenate EDLs.
\begin{definition}[Concatenation ($\odot$)]\label{concatentationDef}
Let {\small$\mathcal{L}_1^{(\alpha)}\hspace{-1.4mm}\left\langle n \right\rangle{=}[V_1,V_2,\dotsc,V_n]$} and {\small$\mathcal{L}_2^{(\alpha')}\hspace{-1.4mm}\left\langle m \right\rangle{=}[V'_1,V'_2,\dotsc,V'_m]$} be two EDLs. The concatenation of {\small$\mathcal{L}_1^{(\alpha)}\hspace{-1.4mm}\left\langle n \right\rangle$ and $\mathcal{L}_2^{(\alpha')}\hspace{-1.4mm}\left\langle m \right\rangle$} is defined as follows:
\begin{equation}\small
\begin{aligned}
&\mathcal{L}_1^{(\alpha)}\hspace{-1.4mm}\left\langle n \right\rangle \odot \mathcal{L}_2^{(\alpha')}\hspace{-1.4mm}\left\langle m \right\rangle =\\
&\begin{cases}
  [V_1,V_2,\dotsc,V_n,V'_1,V'_2,\dotsc,V'_m], & \mathcal{L}_1^{(\alpha)}\hspace{-1.4mm}\left\langle n \right\rangle \dot{\preceq} \mathcal{L}_2^{(\alpha')}\hspace{-1.4mm}\left\langle m \right\rangle,\\
  [V'_1,V'_2,\dotsc,V'_m,V_1,V_2,\dotsc,V_n], & \mathcal{L}_1^{(\alpha)}\hspace{-1.4mm}\left\langle n \right\rangle \dot{\succeq} \mathcal{L}_2^{(\alpha')}\hspace{-1.4mm}\left\langle m \right\rangle.
\end{cases}
\end{aligned}
\end{equation}
\end{definition}
Finally, we introduce the following subtraction operator, utilized to subtract EDLs from an EDV.
\begin{definition}[Subtraction ($\circleddash$)]
Consider an EDL {\small$\mathcal{L}^{(\alpha)}\hspace{-1.4mm}\left\langle n \right\rangle=\{V_1,V_2,\dotsc,V_n\}$} and an EDV {\small$S{\sim}\alpha'$}. The subtraction of {\small$\mathcal{L}^{(\alpha)}\hspace{-1.4mm}\left\langle n \right\rangle$} and $S$, denoted by {\small$\mathcal{L}^{(\alpha)}\hspace{-1.4mm}\left\langle n \right\rangle \circleddash S$}, is defined as follows:
\begin{equation}
\hspace{-.5mm}
\mathcal{L}^{(\alpha)}\hspace{-1.4mm}\left\langle n \right\rangle {\circleddash} S {=} [V_1{-}S,...,V_i{-}S,...,V_n{-}S].
 \end{equation}
\end{definition}\noindent
\noindent By convention, if $S=V_i$, $V_i$ will be removed from $\mathcal{L}^{(\alpha)}\hspace{-1.4mm}\left\langle n \right\rangle$.\par
\textbf{Special Class of EDL.}
By utilizing concatenation operator, we next aim to introduce a class of EDLs, which their EDVs follow two different distributions.
\begin{definition}[Heterogeneous Event Dynamic List]\label{def19}
Let {\small$\mathcal{L}_1^{(\alpha)}\hspace{-1.4mm}\left\langle n \right\rangle{=}[V_1,V_2,\dotsc,V_n]$} and {\small$\mathcal{L}_2^{(\alpha')}\hspace{-1.4mm}\left\langle m \right\rangle{=}[V'_1,V'_2,\dotsc,V'_m]$} be two EDLs, where {\small$\alpha$} and {\small$\alpha'$} are different. Heterogeneous EDL (H-EDL), shown by {\small$\mathcal{H}^{(\alpha;\alpha')}\hspace{-1.4mm}\left\langle n \hspace{-0.3mm};\hspace{-0.3mm} m \right\rangle$}, is defined as follows:
\begin{equation}
\mathcal{H}^{(\alpha;\alpha')}\hspace{-1.4mm}\left\langle n \hspace{-0.3mm};\hspace{-0.3mm} m \right\rangle{=}\mathcal{L}_1^{(\alpha)}\hspace{-1.4mm}\left\langle n \right\rangle \odot \mathcal{L}_2^{(\alpha')}\hspace{-1.4mm}\left\langle m \right\rangle,~ \mathcal{L}_1^{(\alpha)}\hspace{-1.4mm}\left\langle n \right\rangle {\dot{\preceq}} \mathcal{L}_2^{(\alpha')}\hspace{-1.4mm}\left\langle m \right\rangle,
\end{equation}
and
\begin{equation}
\mathcal{H}^{(\alpha;\alpha')}\hspace{-1.4mm}\left\langle m \hspace{-0.3mm};\hspace{-0.3mm} n \right\rangle{=}\mathcal{L}_1^{(\alpha)}\hspace{-1.4mm}\left\langle n \right\rangle \odot \mathcal{L}_2^{(\alpha')}\hspace{-1.4mm}\left\langle m \right\rangle,~ \mathcal{L}_1^{(\alpha)}\hspace{-1.4mm}\left\langle n \right\rangle {\dot{\succeq}} \mathcal{L}_2^{(\alpha')}\hspace{-1.4mm}\left\langle m \right\rangle.
\end{equation}
\end{definition}
\noindent It is worth mentioning that if {\small$\alpha$} and $\alpha'$ are identical, we get
\begin{equation} \label{e14}
\mathcal{L}_3^{(\alpha)}\hspace{-1.4mm}\left\langle {n{+}m} \right\rangle =\mathcal{H}^{(\alpha;\alpha')}\hspace{-1.4mm}\left\langle n \hspace{-0.3mm};\hspace{-0.3mm} m \right\rangle=\mathcal{L}_1^{(\alpha)}\hspace{-1.4mm}\left\langle n \right\rangle \odot \mathcal{L}_2^{(\alpha')}\hspace{-1.4mm}\left\langle m \right\rangle,
\end{equation}\par
\noindent where {\small$\mathcal{L}_3^{(\alpha)}\hspace{-1.4mm}\left\langle {n{+}m} \right\rangle$} is an EDL with {\small$m{+}n$} EDVs.\par
We next aim to utilize the notation of $\langle e\rangle$-algebra to model {\tt RP-VC$_n$}. Then we present decomposition theorem to transform $\beta$-SMP, presented in Sec.~\ref{BSMPR2n}, to D-SMP.
\vspace{-2mm}
\subsection{Analyzing the Dynamics of {\tt RP-VC$_n$} through $\langle e\rangle$-algebra}
Let {\small$\mathcal{C}$} denote the set of all the vehicles utilized to process an application $\mathcal{A}$. Also, consider ${\mathsf{SeS}}_{\mathsf{VC}}{=}(\zeta,\mathcal{Q})$, presented in Definition~\ref{SeSVC}, where {\small$\zeta{=}\{e^{\mathsf{D}}(C_\ell),e^{\mathsf{R}}(C_\ell)\}_{C_\ell{\in} \mathcal{C}}$} and $\mathcal{Q}{=}\mathcal{Z}\cup \mathcal{U}$. Let i.i.d random variable $Z_x{\in}\mathcal{Z}$ follow general distribution $\mathfrak{D}_Z(z)$. Also, let i.i.d random variable $U_y{\in} \mathcal{U}$ follow a general distribution {\small$\mathfrak{R}_U(u)$}. We next aim to answer question \ref{QQ1} in Sec.~\ref{complexityOfRP} through utilizing the framework of $\langle e\rangle$-algebra. We first introduce a set of special EDVs, utilized to characterize the residual sojourn time (i.e., the residual time until the departure of the respective vehicle from the VC) and residual recruitment duration (i.e., the residual time until completing recruitment operation by the respective vehicle) of each vehicle after the occurrence of event $e\in\zeta$.\par
\textbf{Special EDVs and Their Relationships.}
Let {\small$Z_1$}, {\small$Z_2$}, and {\small$U$} be three independent random variables, where {\small$Z_1$} and {\small$Z_2$} follow general distribution $\mathfrak{D}_Z(z)$, and {\small$U$} has a general distribution {\small$\mathfrak{R}_U(u)$}. We define the following special EDVs, utilized to obtain our main upshots in Sec.~\ref{DecT}.
\begin{definition}[$\mathfrak{X}$ Distribution]\label{expdist}
Let {\small$Z'_1=Z_1-Z_2$} and {\small$Z'_2{=}Z_1{-}(Z_2{+}U)$} be two EDVs. The distribution of {\small$Z'_1$} and {\small$Z'_2$} are called first and second orders $\mathfrak{X}$ shown by {\small$Z'_1{\sim}\dot{\mathfrak{X}}$} and {\small$Z'_2{\sim}\ddot{\mathfrak{X}}$}.
\end{definition}
\begin{definition}[$\varphi$ Distribution]\label{phidist}
Let {\small$W{=}U{-}Z'$} be an EDV, where {\small$Z'{\sim}\dot{\mathfrak{X}}$}. The distribution of {\small$W$} is called {\small$\varphi$}.
\end{definition}
\begin{definition}[$\psi$ Distribution]\label{psidist}
Let {\small$Y{=}U{-}W$} be an EDV, where {\small$W{\sim}\varphi$}. The distribution of {\small$Y$} is called {\small$\psi$}.
\end{definition}
\begin{definition}[$\gamma$ Distribution]\label{gammadist}
Let {\small$X{=}W{-}U$} be an EDV, where {\small$W{\sim}\varphi$}. The distribution of {\small$X$} is called {\small$\gamma$}.
\end{definition}
We characterize the relations between above EDVs through \textit{decomposability} and \textit{absorbency} properties, utilized to introduce \textit{decomposition theorem} later in Sec.~\ref{DecT}.
\begin{figure}
    \centering
\tikzset{every picture/.style={line width=0.2pt}} 
    \begin{subfigure}{.25\textwidth}
    \begin{tikzpicture}[x=0.45pt,y=0.45pt,yscale=-1,xscale=1]

\draw  [color={rgb, 255:red, 24; green, 29; blue, 219 }  ,draw opacity=1 ][fill={rgb, 255:red, 255; green, 193; blue, 194 }  ,fill opacity=0.37 ] (121.8,32.5) .. controls (121.8,17.31) and (134.11,5) .. (149.3,5) .. controls (164.49,5) and (176.8,17.31) .. (176.8,32.5) .. controls (176.8,47.69) and (164.49,60) .. (149.3,60) .. controls (134.11,60) and (121.8,47.69) .. (121.8,32.5) -- cycle ;
\draw    (148.8,61) -- (59.84,101) ;
\draw [shift={(101.95,82.07)}, rotate = 335.79] [fill={rgb, 255:red, 0; green, 0; blue, 0 }  ][line width=0.08]  [draw opacity=0] (5.36,-2.57) -- (0,0) -- (5.36,2.57) -- cycle    ;
\draw    (148.8,61) -- (148.8,102) ;
\draw [shift={(148.8,84.1)}, rotate = 270] [fill={rgb, 255:red, 0; green, 0; blue, 0 }  ][line width=0.08]  [draw opacity=0] (5.36,-2.57) -- (0,0) -- (5.36,2.57) -- cycle    ;
\draw    (148.8,61) -- (241.9,101.5) ;
\draw [shift={(197.73,82.29)}, rotate = 203.51] [fill={rgb, 255:red, 0; green, 0; blue, 0 }  ][line width=0.08]  [draw opacity=0] (5.36,-2.57) -- (0,0) -- (5.36,2.57) -- cycle    ;
\draw  [color={rgb, 255:red, 152; green, 0; blue, 0 }  ,draw opacity=1 ][fill={rgb, 255:red, 250; green, 255; blue, 198 }  ,fill opacity=0.49 ] (10,105) .. controls (10,103.9) and (10.9,103) .. (12,103) -- (97.8,103) .. controls (98.9,103) and (99.8,103.9) .. (99.8,105) -- (99.8,129) .. controls (99.8,130.1) and (98.9,131) .. (97.8,131) -- (12,131) .. controls (10.9,131) and (10,130.1) .. (10,129) -- cycle ;
\draw  [color={rgb, 255:red, 152; green, 0; blue, 0 }  ,draw opacity=1 ][fill={rgb, 255:red, 250; green, 255; blue, 198 }  ,fill opacity=0.49 ] (104,105) .. controls (104,103.9) and (104.9,103) .. (106,103) -- (191.8,103) .. controls (192.9,103) and (193.8,103.9) .. (193.8,105) -- (193.8,129) .. controls (193.8,130.1) and (192.9,131) .. (191.8,131) -- (106,131) .. controls (104.9,131) and (104,130.1) .. (104,129) -- cycle ;
\draw  [color={rgb, 255:red, 152; green, 0; blue, 0 }  ,draw opacity=1 ][fill={rgb, 255:red, 250; green, 255; blue, 198 }  ,fill opacity=0.49 ] (197.8,104) .. controls (197.8,102.9) and (198.7,102) .. (199.8,102) -- (285.6,102) .. controls (286.7,102) and (287.6,102.9) .. (287.6,104) -- (287.6,128) .. controls (287.6,129.1) and (286.7,130) .. (285.6,130) -- (199.8,130) .. controls (198.7,130) and (197.8,129.1) .. (197.8,128) -- cycle ;

\draw (116,10.4) node [anchor=north west][inner sep=0.75pt]  [font=\tiny]  {$ \begin{array}{l}
V_{1} {\sim} \varphi \\
V_{2} {\sim} \varphi \\
V_{3} {\sim} \varphi
\end{array}$};
\draw (-1,102.4) node [anchor=north west][inner sep=0.75pt]  [font=\tiny]  {$ \begin{array}{l}
( V_{2}{-}V_{1}) {\sim} \psi \\
( V_{3}{-}V_{1}) {\sim} \psi
\end{array}$};
\draw (95,102.4) node [anchor=north west][inner sep=0.75pt]  [font=\tiny]  {$ \begin{array}{l}
( V_{1}{-}V_{2}) {\sim} \gamma \\
( V_{3}{-}V_{2}) {\sim} \psi
\end{array}$};
\draw (187.8,101.4) node [anchor=north west][inner sep=0.75pt]  [font=\tiny]  {$ \begin{array}{l}
( V_{1}{-}V_{3}) {\sim} \gamma \\
( V_{2}{-}V_{3}) {\sim} \gamma
\end{array}$};
\draw (55,79.4) node [anchor=north west][inner sep=0.75pt]  [font=\tiny]  {$V_{1}$};
\draw (125,79.4) node [anchor=north west][inner sep=0.75pt]  [font=\tiny]  {$V_{2}$};
\draw (175,79.4) node [anchor=north west][inner sep=0.75pt]  [font=\tiny]  {$V_{3}$};
\end{tikzpicture}
      \caption{\small Decomposability}
      \label{decomposability}
    \end{subfigure}%
    \begin{subfigure}{.25\textwidth}
    \begin{tikzpicture}[x=0.45pt,y=0.45pt,yscale=-1,xscale=1]

\draw  [color={rgb, 255:red, 24; green, 29; blue, 219 }  ,draw opacity=1 ][fill={rgb, 255:red, 255; green, 193; blue, 194 }  ,fill opacity=0.37 ] (123.13,123.17) .. controls (123.13,107.98) and (135.45,95.67) .. (150.63,95.67) .. controls (165.82,95.67) and (178.13,107.98) .. (178.13,123.17) .. controls (178.13,138.35) and (165.82,150.67) .. (150.63,150.67) .. controls (135.45,150.67) and (123.13,138.35) .. (123.13,123.17) -- cycle ;
\draw    (51.2,49.67) -- (150.63,95.67) ;
\draw [shift={(103.28,73.76)}, rotate = 204.83] [fill={rgb, 255:red, 0; green, 0; blue, 0 }  ][line width=0.08]  [draw opacity=0] (5.36,-2.57) -- (0,0) -- (5.36,2.57) -- cycle    ;
\draw    (148.8,48.5) -- (148.8,95.67) ;
\draw [shift={(148.8,74.68)}, rotate = 270] [fill={rgb, 255:red, 0; green, 0; blue, 0 }  ][line width=0.08]  [draw opacity=0] (5.36,-2.57) -- (0,0) -- (5.36,2.57) -- cycle    ;
\draw    (245.2,49.67) -- (150.63,95.67) ;
\draw [shift={(195.58,73.8)}, rotate = 334.06] [fill={rgb, 255:red, 0; green, 0; blue, 0 }  ][line width=0.08]  [draw opacity=0] (5.36,-2.57) -- (0,0) -- (5.36,2.57) -- cycle    ;
\draw  [color={rgb, 255:red, 152; green, 0; blue, 0 }  ,draw opacity=1 ][fill={rgb, 255:red, 250; green, 255; blue, 198 }  ,fill opacity=0.49 ] (10,23) .. controls (10,21.9) and (10.9,21) .. (12,21) -- (97.8,21) .. controls (98.9,21) and (99.8,21.9) .. (99.8,23) -- (99.8,47) .. controls (99.8,48.1) and (98.9,49) .. (97.8,49) -- (12,49) .. controls (10.9,49) and (10,48.1) .. (10,47) -- cycle ;
\draw  [color={rgb, 255:red, 152; green, 0; blue, 0 }  ,draw opacity=1 ][fill={rgb, 255:red, 250; green, 255; blue, 198 }  ,fill opacity=0.49 ] (104,23) .. controls (104,21.9) and (104.9,21) .. (106,21) -- (191.8,21) .. controls (192.9,21) and (193.8,21.9) .. (193.8,23) -- (193.8,47) .. controls (193.8,48.1) and (192.9,49) .. (191.8,49) -- (106,49) .. controls (104.9,49) and (104,48.1) .. (104,47) -- cycle ;
\draw  [color={rgb, 255:red, 152; green, 0; blue, 0 }  ,draw opacity=1 ][fill={rgb, 255:red, 250; green, 255; blue, 198 }  ,fill opacity=0.49 ] (197.8,23) .. controls (197.8,21.9) and (198.7,21) .. (199.8,21) -- (285.6,21) .. controls (286.7,21) and (287.6,21.9) .. (287.6,23) -- (287.6,47) .. controls (287.6,48.1) and (286.7,49) .. (285.6,49) -- (199.8,49) .. controls (198.7,49) and (197.8,48.1) .. (197.8,47) -- cycle ;

\draw (125.33,115.07) node [anchor=north west][inner sep=0.75pt]  [font=\tiny]  {$S\sim \varphi $};
\draw (-1,17.4) node [anchor=north west][inner sep=0.75pt]  [font=\tiny]  {$ \begin{array}{l}
Z'_{1} \sim \dot{exp}\\
W\sim \varphi
\end{array}$};
\draw (95,17.4) node [anchor=north west][inner sep=0.75pt]  [font=\tiny]  {$ \begin{array}{l}
Z'_{2} \sim \ddot{exp}\\
V_{1} \sim \psi
\end{array}$};
\draw (187.8,17.4) node [anchor=north west][inner sep=0.75pt]  [font=\tiny]  {$ \begin{array}{l}
Z'_{2} \sim \ddot{exp}\\
V_{2} \sim \gamma
\end{array}$};
\draw (18.33,72.4) node [anchor=north west][inner sep=0.75pt]  [font=\tiny]  {$S=W-Z'_{1}$};
\draw (94.67,50.73) node [anchor=north west][inner sep=0.75pt]  [font=\tiny]  {$S=V_{1} -Z'_{2}$};
\draw (190,72.4) node [anchor=north west][inner sep=0.75pt]  [font=\tiny]  {$S=V_{2} -Z'_{2}$};
\end{tikzpicture}
      \caption{\small Absorbency}
      \label{absorbency}
    \end{subfigure}
    \caption{\small Demonstration of decomposability and absorbency.}
    \label{fig:my_label}
\end{figure}
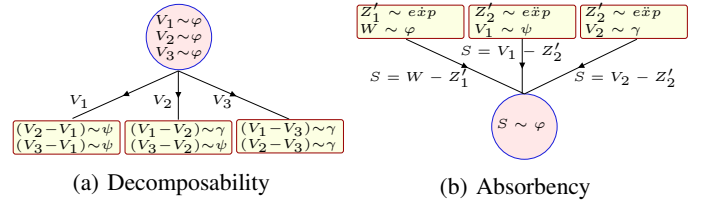
\begin{proposition}[Decomposability]\label{decomposability_props}
   Consider two EDVs {\small$V_1{\sim}\alpha$} and {\small$V_2{\sim}\alpha'$}, where {\small$\alpha$} and $\alpha'$ are
   (i) $\alpha,\alpha'=\varphi$
   or (ii) $\alpha,\alpha'=\psi$ or (iii) $\alpha,\alpha'=\gamma$ or (iv) $\alpha= \psi,\alpha'=\gamma$ or (v) $\alpha= \gamma,\alpha'=\psi$,
    and {\small$V_1 \dot{\preceq} V_2$}. EDV {\small$S_1{=}V_1{-}V_2$} follows {\small$\gamma$} distribution if {\small$S_1>0$}, and EDV {\small$S_2{=}V_2{-}V_1$} follows {\small$\psi$} distribution if {\small$S_2>0$}.
\end{proposition}
\begin{proof}
See Appendix \ref{decomposability_props_proof}.
\end{proof}
To clarify the main implication of decomposability, Fig.~\ref{decomposability} illustrates an example of three {\small$\varphi$} EDVs {\small$V_1$, $V_2$}, and {\small$V_3$}, where {\small$V_1 \dot{\succeq} V_2 \dot{\succeq} V_3$}. Decomposability states that the distributions of different subtractions of these three EDVs, i.e., {\small$\{V_1{-}V_3,V_2{-}V_3\}$}, {\small$\{V_1{-}V_2,V_3{-}V_2\}$}, and {\small$\{V_2{-}V_1,V_3{-}V_1\}$}, follow three different permutations of {\small$\gamma$} and {\small$\psi$}, i.e., {\small$\{\gamma,\gamma\}$}, {\small$\{\gamma,\psi\}$}, and {\small$\{\psi,\psi\}$}.
\begin{proposition}[Absorbency]\label{absorbency_props}
    Consider four EDVs {\small$Z'_1{\sim}\dot{\mathfrak{X}}$}, {\small$Z'_2{\sim}\ddot{\mathfrak{X}}$}, {\small$W{\sim}\varphi$} and {\small$V{\sim}\alpha$}, where {\small$\alpha$} can be $\psi$ or $\gamma$, {\small$Z'_1\dot{\preceq} W$}, and {\small$Z'_2\dot{\preceq} V$}. EDVs {\small$S_1=W{-}Z'_1$} and {\small$S_2=V{-}Z'_2$} follow {\small$\varphi$} distribution if {\small$S_1>0$} and {\small$S_2>0$}. Conversely, {\small$S_3=Z'_1{-}W$} and {\small$S_4=Z'_2{-}V$} follow {\small$\ddot{\mathfrak{X}}$} distribution if {\small$S_3>0$} and {\small$S_4>0$}.
\end{proposition}
\begin{proof}
See Appendix \ref{absorbency_props_proof}.
\end{proof}
Referring to Fig.~\ref{absorbency} as an example, consider EDVs {\small$Z'_1{\sim}\dot{\mathfrak{X}}$}, {\small$Z'_2{\sim}\ddot{\mathfrak{X}}$}, {\small$W{\sim}\varphi$}, {\small$V_1{\sim}\psi$}, and {\small$V_2{\sim}\gamma$}. Absorbency implies that {\small$W{-}Z'_1$}, {\small$V_1{-}Z'_2$}, and {\small$V_2{-}Z'_2$} result in the same EDV following $\varphi$ distribution. Decomposability and absorbency are utilized to prove Proposition~\ref{proposition3} and Proposition~\ref{proposition4}, introduced below, characterizing the residual occurring times of the departure and recruitment events in ${\mathsf{SeS}}_{\mathsf{VC}}$ after the occurrence of an event.\par
\textbf{Special EDLs and their Relationships.} In the following, we present an example to motivate Proposition~\ref{proposition3} and Proposition~\ref{proposition4}. Let $\mathcal{C}$ denote the set of vehicles utilized for processing application $\mathcal{A}$. Also, let $\mathcal{Z}$ and $\mathcal{U}$ denote two sets of random variables, where $Z_{\ell}\in \mathcal{Z}$ follows a general distribution $\mathfrak{D}_Z(z)$ referring to the sojourn time of vehicle $C_{\ell}\in \mathcal{C}$ and $U_{\ell}\in \mathcal{U}$ follows general distribution $\mathfrak{R}_U(u)$ referring to the recruitment duration of vehicle $C_{\ell}$. Consider deployment {\small$\mathcal{D}\Big(\mathcal{M}\big(\mathcal{P}(\mathcal{A}),t\big)\Big)$}, where {\small$\Big|\mathcal{D}\Big(\mathcal{M}\big(\mathcal{P}(\mathcal{A}),t\big)\Big)\Big|=n$}. Further, consider time instances $\tau=\{t_0,t_1,t_2\}$, where $t_0$ is the time that the processing of application $\mathcal{A}$ is started. Assume that {\small$\mathcal{D}\Big(\mathcal{M}\big(\mathcal{P}(\mathcal{A}),t\big)\Big)$} stays unchanged for $t\in \tau$ (i.e., the number of groups is fixed). At time $t_0$, {\small$2n$} vehicles are allocated to the groups of {\small$\mathcal{D}\Big(\mathcal{M}\big(\mathcal{P}(\mathcal{A}),t\big)\Big)$}. Assume that group $G_{h_1}$ is deployed on vehicles $C_{\ell_1}\in \mathcal{C}$ and $C_{\ell'_1}\in \mathcal{C}$, and group $G_{h_2}$ is deployed on vehicles $C_{\ell_2}\in \mathcal{C}$ and $C_{\ell'_2}\in \mathcal{C}$. Further, assume that vehicle {\small$C_{\ell_1}$} departs the VC at time $t_1$. {\small$Z_{\ell_1}\in \mathcal{Z}$} should be subtracted from the sojourn times of other vehicles because {\small$Z_{\ell_1}$} time units have passed since the start of processing $\mathcal{A}$ (i.e., $t_0$). Accordingly, we model the residual sojourn times of the vehicles (i.e., $Z_{\ell_i}-Z_{\ell_1}$, where $Z_{\ell_i}\in\mathcal{Z}$) as an EDL {\small$\mathcal{L}^{(\dot{\mathfrak{X}})}\hspace{-1.4mm}\left\langle 2n{-}1\right\rangle$}. Further, at time $t_1$, {\small$C_{\ell'_1}$} starts recruiting a new vehicle. Likewise, assume that vehicle {\small$C_{\ell_2}$} departs the VC at time $t_2$. Hence, {\small$Z_{\ell_2}{-}Z_{\ell_1}$} should be subtracted from the residual sojourn times and residual recruitment duration of the other vehicles because {\small$Z_{\ell_2}{-}Z_{\ell_1}$} time units have passed since $t_1$. Consequently, we have {\small$\mathcal{L}_2^{(\dot{\mathfrak{X}})}\hspace{-1.4mm}\left\langle 2n{-}2\right\rangle=\mathcal{L}_1^{(\dot{\mathfrak{X}})}\hspace{-1.4mm}\left\langle 2n{-}1\right\rangle \circleddash (Z_{\ell_2}{-}Z_{\ell_1})$} and {\small$U_{\ell'_1}{-}(Z_{\ell_2}{-}Z_{\ell_1})$}. Also, at time $t_2$, vehicle {\small$C_{\ell'_2}$} starts recruiting a new vehicle. It is straightforward to verify that the residual recruitment times of recruiters $C_{\ell'_1}$ and $C_{\ell'_2}$ (i.e., {\small$U_{\ell'_1}{-}(Z_{\ell_2}{-}Z_{\ell_1})$} and {\small$U_{\ell'_2}$}) can be modeled as an H-EDL {\small$\mathcal{H}^{(\varphi;\mathfrak{R})}\hspace{-1.4mm}\left\langle 1 \hspace{-0.3mm};\hspace{-0.3mm} 1 \right\rangle$}.\par
In the following, we generalize the above-mentioned example. We first define the following two special H-EDLs.
\begin{enumerate}
    \item {\small$\mathcal{H}_1^{(\gamma;\psi)}\hspace{-1.4mm}\left\langle n\hspace{-0.3mm};\hspace{-0.3mm} m \right\rangle$} : An H-EDL with $n+m$ elements, in which the first $n$ elements follow {\small$\gamma$} distribution and the last $m$ elements follow {\small$\psi$} distribution.
    \item {\small$\mathcal{H}_2^{(\varphi;\mathfrak{R})}\hspace{-1.4mm}\left\langle n{-}1; 1 \right\rangle$}: An H-EDL with $n$ elements, in which the first $n-1$ elements follow {\small$\varphi$} distribution and the last one follows a general distribution {\small$R$}.
\end{enumerate}
Accordingly, we obtain the following important technical results, utilized to prove decomposition theorem later. It is worth mentioning that decomposition theorem will be utilized to define the states of D-SMP model of {\tt RP-VC$_n$}.
\begin{proposition}\label{proposition3}
Consider two EDLs $\mathcal{H}_1^{(\varphi;\mathfrak{R})}\hspace{-1.4mm}\left\langle n{-}1 \hspace{-0.3mm};\hspace{-0.3mm} 1 \right\rangle$ and $\mathcal{L}_1^{(\dot{\mathfrak{X}})}\hspace{-1.4mm}\left\langle m\right\rangle$, where $\mathcal{L}_1^{(\dot{\mathfrak{X}})}\hspace{-1.4mm}\left\langle m\right\rangle \dot{\preceq} \mathcal{H}_1^{(\varphi;\mathfrak{R})}\hspace{-1.4mm}\left\langle n{-}1 \hspace{-0.3mm};\hspace{-0.3mm} 1 \right\rangle$. The following statements hold:
\begin{enumerate}[leftmargin=5mm]
  \item If $\mathcal{H}_1^{(\varphi;\mathfrak{R})}\hspace{-1.4mm}\left\langle n{-}1 \hspace{-0.3mm};\hspace{-0.3mm} 1 \right\rangle\hspace{-0.7mm}(k) >\vspace{1mm} \mathcal{H}_1^{(\varphi;\mathfrak{R})}\hspace{-1.4mm}\left\langle n{-}1 \hspace{-0.3mm};\hspace{-0.3mm} 1 \right\rangle\hspace{-0.7mm}(i)$ for $i\neq k$ and $\mathcal{L}_1^{(\dot{\mathfrak{X}})}\hspace{-1.4mm}\left\langle m\right\rangle\hspace{-0.7mm}(r) >\mathcal{H}_1^{(\varphi;\mathfrak{R})}\hspace{-1.4mm}\left\langle n{-}1 \hspace{-0.3mm};\hspace{-0.3mm} 1 \right\rangle\hspace{-0.7mm}(i)$, then \footnote{$\mathcal{H}_1^{(\varphi;\mathfrak{R})}\hspace{-1.4mm}\left\langle n{-}1 \hspace{-0.3mm};\hspace{-0.3mm} 1 \right\rangle\hspace{-0.7mm}(k)$ refers to the $k^{th}$ element of EDL $\mathcal{H}_1^{(\varphi;\mathfrak{R})}\hspace{-1.4mm}\left\langle n{-}1 \hspace{-0.3mm};\hspace{-0.3mm} 1 \right\rangle\hspace{-0.1mm}$.}
\begin{equation}\label{eqq1}
\begin{aligned}
\hspace{-2mm}
    \mathcal{H}_1^{(\varphi;\mathfrak{R})}\hspace{-1.4mm}\left\langle n{-}1 \hspace{-0.3mm};\hspace{-0.3mm} 1 \right\rangle &\hspace{-0.5mm}\circleddash\hspace{-0.5mm} \mathcal{H}_1^{(\varphi;\mathfrak{R})}\hspace{-1.4mm}\left\langle n{-}1 \hspace{-0.3mm};\hspace{-0.3mm} 1 \right\rangle\hspace{-0.7mm}(i){=}\hspace{-0.5mm}\mathcal{H}_2^{(\gamma;\psi)}\hspace{-1.4mm}\left\langle i{-}1 \hspace{-0.3mm};\hspace{-0.3mm} n{-}i \right\rangle\hspace{-0.5mm},
\hspace{-4mm}
\end{aligned}
\end{equation}
and
\begin{equation}\label{eq21}
\hspace{-6mm}
    \mathcal{L}_1^{(\dot{\mathfrak{X}})}\hspace{-1.4mm}\left\langle m\right\rangle \hspace{-0.6mm}\circleddash\hspace{-0.6mm} \mathcal{H}_1^{(\varphi;\mathfrak{R})}\hspace{-1.4mm}\left\langle n{-}1 \hspace{-0.3mm};\hspace{-0.3mm} 1 \right\rangle\hspace{-0.7mm}(i) \hspace{-0.6mm}{=}\hspace{-0.6mm} \mathcal{L}_2^{(\ddot{\mathfrak{X}})}\hspace{-1.4mm}\left\langle m\right\rangle\hspace{-0.5mm},
\hspace{-6mm}
\end{equation}
where $\mathcal{H}_2^{(\gamma;\psi)}\hspace{-1.4mm}\left\langle i{-}1 \hspace{-0.3mm};\hspace{-0.3mm} n{-}i \right\rangle \dot{\succeq} \mathcal{L}_2^{(\ddot{\mathfrak{X}})}\hspace{-1.4mm}\left\langle m\right\rangle$.
\vspace{1mm}
\item If $\mathcal{H}_1^{(\varphi;\mathfrak{R})}\hspace{-1.4mm}\left\langle n{-}1 \hspace{-0.3mm};\hspace{-0.3mm} 1 \right\rangle\hspace{-0.7mm}(k)>\vspace{1mm} \mathcal{L}_1^{(\dot{\mathfrak{X}})}\hspace{-1.4mm}\left\langle m\right\rangle\hspace{-0.7mm}(r)$, then
\begin{equation}\label{eq22}
\hspace{-9mm}
    \mathcal{H}_1^{(\varphi;\mathfrak{R})}\hspace{-1.4mm}\left\langle n{-}1 \hspace{-0.3mm};\hspace{-0.3mm} 1 \right\rangle \hspace{-0.6mm}\circleddash\hspace{-0.6mm} \mathcal{L}_1^{(\dot{\mathfrak{X}})}\hspace{-1.4mm}\left\langle m\right\rangle\hspace{-0.7mm}(r)\hspace{-0.6mm}=\hspace{-0.6mm} \mathcal{L}_3^{(\varphi)}\hspace{-1.4mm}\left\langle n\right\rangle\hspace{-0.6mm},
\hspace{-6mm}
\end{equation}
and
\begin{equation}\label{eq23}
\hspace{-9mm}
    \mathcal{L}_1^{(\dot{\mathfrak{X}})}\hspace{-1.4mm}\left\langle m\right\rangle \hspace{-0.6mm}\circleddash\hspace{-0.6mm} \mathcal{L}_1^{(\dot{\mathfrak{X}})}\hspace{-1.4mm}\left\langle m\right\rangle\hspace{-0.7mm}(r)\hspace{-0.6mm}= \hspace{-0.6mm}\mathcal{L}_4^{(\dot{\mathfrak{X}})}\hspace{-1.4mm}\left\langle m{-}1\right\rangle\hspace{-0.6mm},
\hspace{-6mm}
\end{equation}
where $\mathcal{L}_3^{(\varphi)}\hspace{-1.4mm}\left\langle n\right\rangle {\dot{\succeq}} \mathcal{L}_4^{(\dot{\mathfrak{X}})}\hspace{-1.4mm}\left\langle m{-}1\right\rangle$.
\end{enumerate}
\end{proposition}
\begin{proof}
See Appendix \ref{proposition3_proof}.
\end{proof}
\begin{proposition}\label{proposition4}
Consider two EDLs $\mathcal{H}_1^{(\gamma;\psi)}\hspace{-1.4mm}\left\langle n \hspace{-0.3mm};\hspace{-0.3mm} m \right\rangle$ and $\mathcal{L}_1^{(\ddot{\mathfrak{X}})}\hspace{-1.4mm}\left\langle h\right\rangle$, where $\mathcal{L}_1^{(\ddot{\mathfrak{X}})}\hspace{-1.4mm}\left\langle h\right\rangle \dot{\preceq} \mathcal{H}_1^{(\gamma;\psi)}\hspace{-1.4mm}\left\langle n \hspace{-0.3mm};\hspace{-0.3mm} m \right\rangle$. The following statements hold:
\begin{enumerate}[leftmargin=5mm]
  \item If $\mathcal{H}_1^{(\gamma;\psi)}\hspace{-1.4mm}\left\langle n \hspace{-0.3mm};\hspace{-0.3mm} m \right\rangle\hspace{-0.7mm}(k)>\vspace{1mm} \mathcal{H}_1^{(\gamma;\psi)}\hspace{-1.4mm}\left\langle n \hspace{-0.3mm};\hspace{-0.3mm} m \right\rangle\hspace{-0.7mm}(i)$ for $i \neq k$ and $\mathcal{L}_1^{(\ddot{\mathfrak{X}})}\hspace{-1.4mm}\left\langle h\right\rangle\hspace{-0.7mm}(r)> \mathcal{H}_1^{(\gamma;\psi)}\hspace{-1.4mm}\left\langle n \hspace{-0.3mm};\hspace{-0.3mm} m \right\rangle\hspace{-0.7mm}(i)$, we have
\begin{equation}\label{pro4_1}
\hspace{-6mm}
    \mathcal{H}_1^{(\gamma;\psi)}\hspace{-1.4mm}\left\langle n \hspace{-0.3mm};\hspace{-0.3mm} m \right\rangle \hspace{-0.5mm}{\circleddash} \mathcal{H}_1^{(\gamma;\psi)}\hspace{-1.4mm}\left\langle n \hspace{-0.3mm};\hspace{-0.3mm} m \right\rangle\hspace{-0.7mm}(i){=} \mathcal{H}_2^{(\gamma;\psi)} \hspace{-1.4mm}\left\langle i{-}1 \hspace{-0.3mm};\hspace{-0.3mm} m{+}n{-}i \right\rangle\hspace{-1mm},
    \hspace{-6mm}
\end{equation}
and
\begin{equation}\label{pro4_2}
\hspace{-6mm}
    \mathcal{L}_1^{(\ddot{\mathfrak{X}})}\hspace{-1.4mm}\left\langle h\right\rangle \circleddash \mathcal{H}_1^{(\gamma;\psi)}\hspace{-1.4mm}\left\langle n \hspace{-0.3mm};\hspace{-0.3mm} m \right\rangle\hspace{-0.7mm}(i)= \mathcal{L}_2^{(\ddot{\mathfrak{X}})}\hspace{-1.4mm}\left\langle h\right\rangle \hspace{-.6mm},
\hspace{-6mm}
\end{equation}
where $\mathcal{H}_2^{(\gamma;\psi)} \hspace{-1.4mm}\left\langle i{-}1 \hspace{-0.3mm};\hspace{-0.3mm} m{+}n{-}i \right\rangle \dot{\succeq} \mathcal{L}_2^{(\ddot{\mathfrak{X}})}\hspace{-1.4mm}\left\langle h\right\rangle$.
\vspace{1mm}
\item If $\mathcal{H}_1^{(\gamma;\psi)}\hspace{-1.4mm}\left\langle n \hspace{-0.3mm};\hspace{-0.3mm} m \right\rangle\hspace{-0.7mm}(k)>\vspace{1mm} \mathcal{L}_1^{(\ddot{\mathfrak{X}})}\hspace{-1.4mm}\left\langle h\right\rangle\hspace{-0.7mm}(r)$, we have
\begin{equation}\label{pro4_3}
\hspace{-6mm}
    \mathcal{H}_1^{(\gamma;\psi)}\hspace{-1.4mm}\left\langle n \hspace{-0.3mm};\hspace{-0.3mm} m \right\rangle \circleddash \mathcal{L}_1^{(\ddot{\mathfrak{X}})}\hspace{-1.4mm}\left\langle h\right\rangle\hspace{-0.7mm}(r)\hspace{-.5mm}=\hspace{-.5mm} \mathcal{L}_3^{(\varphi)}\hspace{-1.4mm}\left\langle m+n\right\rangle\hspace{-.5mm},
\hspace{-6mm}
\end{equation}
and
\begin{equation}\label{pro4_4}
\hspace{-6mm}
    \mathcal{L}_1^{(\ddot{\mathfrak{X}})}\hspace{-1.4mm}\left\langle h\right\rangle \circleddash \mathcal{L}_1^{(\ddot{\mathfrak{X}})}\hspace{-1.4mm}\left\langle h\right\rangle\hspace{-0.7mm}(r) =\mathcal{L}_4^{(\dot{\mathfrak{X}})}\hspace{-1.4mm}\left\langle h-1\right\rangle,
\hspace{-6mm}
\end{equation}
where $\mathcal{L}_3^{(\varphi)}\hspace{-1.4mm}\left\langle m+n\right\rangle \dot{\succeq} \mathcal{L}_4^{(\dot{\mathfrak{X}})}\hspace{-1.4mm}\left\langle h-1\right\rangle$.
\end{enumerate}
\end{proposition}
\begin{proof}
See Appendix \ref{proposition4_proof}.
\end{proof}
To understand the implications of the results of the above propositions, we consider an example of the result given by (\ref{eqq1}) in Proposition~\ref{proposition3} (the same argument holds for the other parts of Proposition~\ref{proposition3} and Proposition~\ref{proposition4}). Assume that there are $m$ computing vehicles in a semi-dynamic VC, from which $n{<}m$ vehicles are busy recruiting new vehicles. Further, assume that {\small$\mathcal{H}_1^{(\varphi;\mathfrak{R})}\hspace{-1.4mm}\left\langle n-1 \hspace{-0.3mm};\hspace{-0.3mm} 1 \right\rangle$} is the list of residual recruitment times of recruiters. The result of (\ref{eqq1}) states that if the $i^{th}$ recruiter completes its recruitment operation successfully, the residual recruitment times of other vehicles is {\small$\mathcal{H}_2^{(\gamma;\psi)}\hspace{-1.4mm}\left\langle i-1 \hspace{-0.3mm};\hspace{-0.3mm} n-i \right\rangle$}, which represents a direct result of decomposability (Proposition~\ref{decomposability_props}).
\vspace{-2mm}
\subsection{Decomposition Theorem}\label{DecT}
In this section we propose the \textit{decomposition theorem (DT)} based on the notations of EDV and EDL developed in Sec.~\ref{RACSE}. Specifically, DT enables us to disentangle $\beta$-SMP to a decomposed SMP model (D-SMP). The key advantage of D-SMP compared to $\beta$-SMP is that, in D-SMP, calculating transition probabilities and expected sojourn time of the process in each state are much more straightforward. We first present the following definition.
\begin{definition}[Order of Recruiter]
    For $0{\le} i{\le} n$, let {\small$\dot{\mathcal{C}}_{i}$} denote the list of all $i$ recruiter vehicles in state $S_{i}\in\mathcal{S}$, defined in Sec.~\ref{BSMPR2n}. Further, let EDV $V_{j}$ be residual recruitment time of vehicle $C_{j}\in \dot{\mathcal{C}}_{i}$. The order of recruiter $C_{\ell}\in \dot{\mathcal{C}}_{i}$, shown by $O(C_{\ell})$, is defined as follows:
    \begin{equation}
        O(C_{\ell})=\sum_{j=1, j\neq \ell}^{i} \mathds{1}_{V_{j} \preceq V_{\ell}}.
    \end{equation}
\end{definition}
\noindent We next present \textit{decomposition theorem (DT)}.
\begin{theorem}[Decomposition]\label{decompositionTheorem}
Consider ${\mathsf{SeS}}_{\mathsf{VC}}{=}(\zeta,\mathcal{Q})$ presented in Sec.~\ref{SDSeS}, and deployment {\small$\mathcal{D}\Big(\mathcal{M}\big(\mathcal{P}(\mathcal{A}),t\big)\Big)$} at time $t$, where {\small$\Big|\mathcal{D}\Big(\mathcal{M}\big(\mathcal{P}(\mathcal{A}),t\big)\Big)\Big|{=}n$}. Further, consider $\beta$-SMP presented in Fig. \ref{j2nMarkov1} with state space {\small$\mathcal{S}=\{S_0,S_1,\dotsc,S_n, F\}$}. Let $X\in\mathcal{S}$ and $X'\in\mathcal{S}$ denote the current state and next state of $\beta$-SMP, respectively. Further, let $e_X\in\zeta$ denote the event that occurred at current state $X$. Each state $S_i$, except for the $S_n$, can be decomposed into $i{+}2$ states, referred to as H-state. Mathematically,
\begin{equation}
 S_0\equiv
    \begin{cases}
     S_{0,0} & \text{initial state},
    \end{cases}
\end{equation}
\begin{equation}
 X'{=}S_0\equiv
    \begin{cases}
     S_{0,1} & \text{if  } X{=}S_{1},e_X{=}e^{\mathsf{R}}(C_{\ell}), C_{\ell} {\in} \dot{\mathcal{C}}_{1},
    \end{cases}
\end{equation}
\begin{equation}\label{decompositionTheoremEq1}
   X'{=}S_i\equiv
    \begin{cases}
        S_{i,0} & \text{if  } X=S_{i-1},e_X{=}e^{\mathsf{D}}(C_{\ell}),~ \\
            &~\,\,\,\,\,\,\,\,\,\,\,\forall C_{\ell} {\in} \ddot{\mathcal{C}}_{2n-2(i-1)}, \\
        S_{i,1} & \text{if  } X=S_{i+1},e_X{=}e^{\mathsf{R}}(C_{\ell}),\\
        &~\,\,\,\,\,\,\,\,\,\,\, O(C_{\ell})=1, \forall C_{\ell} {\in} \dot{\mathcal{C}}_{i{+}1},\\
        \dots &\dots\\
        S_{i,j} &\text{if  } X=S_{i+1},e_X{=}e^{\mathsf{R}}(C_{\ell}),\\
        &~\,\,\,\,\,\,\,\,\,\,\, O(C_{\ell})=j, \forall C_{\ell} {\in} \dot{\mathcal{C}}_{i{+}1},\\
        \dots & \dots\\
        S_{i,i+1} &\text{if  } X=S_{i+1},e_X{=}e^{\mathsf{R}}(C_{\ell}),\\
        &~\,\,\,\,\,\,\,\,\,\,\, O(C_{\ell})=i+1, \forall C_{\ell} {\in} \dot{\mathcal{C}}_{i{+}1},
    \end{cases}
  \end{equation}
  and
  \begin{equation}
 X'{=}S_n\equiv
    \begin{cases}
    S_{n,0} & \text{if  } X{=}S_{n-1},e_X{=}e^{\mathsf{D}}(C_{\ell}),~ \forall C_{\ell} {\in} \ddot{\mathcal{C}}_{2}, \\
    \end{cases}
\end{equation}
where $S_{0,0}$ is
\begin{equation}\label{S00}
      S_{0,0}=\left\{[Z_1,Z_2,\dots,Z_{2n}], [\,] \right\},
\end{equation}
$S_{0,1}$ is
\begin{equation}\label{S000}
      S_{0,1}=\left\{\mathcal{L}_2^{(\ddot{\mathfrak{X}})}\hspace{-1.4mm}\left\langle 2n\right\rangle, [\,] \right\},
\end{equation}
$S_{i,0}$, for {\small$1 \le i \le n$}, is
  \begin{equation}\label{eqqq1}
      S_{i,0}=\left\{\mathcal{L}_1^{(\dot{\mathfrak{X}})}\hspace{-1.4mm}\left\langle {2n{-}i}\right\rangle,\mathcal{H}_1^{(\varphi;\mathfrak{R})}\hspace{-1.4mm}\left\langle {i{-}1} \hspace{-0.3mm};\hspace{-0.3mm} 1 \right\rangle\right\},
  \end{equation}
and $S_{i,j}$, for {\small$1\le j\le i+1$} and {\small$1 \le i \le n-1$}, is
    \begin{equation}\label{eqqq2}
      S_{i,j}=\left\{\mathcal{L}_2^{(\ddot{\mathfrak{X}})}\hspace{-1.4mm}\left\langle 2n{-}i\right\rangle,\mathcal{H}_2^{(\gamma;\psi)}\hspace{-1.4mm}\left\langle j{-}1 \hspace{-0.3mm};\hspace{-0.3mm} i{-}(j{-}1) \right\rangle\right\}.
  \end{equation}
  Moreover, we have $\mathcal{L}_1^{(\dot{\mathfrak{X}})}\hspace{-1.4mm}\left\langle {2n{-}i}\right\rangle  {\dot{\preceq}}  \mathcal{H}_1^{(\varphi;\mathfrak{R})}\hspace{-1.4mm}\left\langle {i{-}1} \hspace{-0.3mm};\hspace{-0.3mm} 1 \right\rangle$ and $\mathcal{L}_2^{(\ddot{\mathfrak{X}})}\hspace{-1.4mm}\left\langle 2n{-}i\right\rangle   {\dot{\preceq}}  \mathcal{H}_2^{(\gamma;\psi)}\hspace{-1.4mm}\left\langle j{-}1 \hspace{-0.3mm};\hspace{-0.3mm} i{-}(j{-}1) \right\rangle$.

In \eqref{S00}, $Z_1, Z_2, \dots, Z_{2n}$ are i.i.d random variables following general distribution $\mathfrak{D}_Z(z)$ referring to the residency time of  vehicles, and $[\,]$ refers to the empty list. In~\eqref{eqqq1} and~\eqref{eqqq2},  $\mathcal{L}_1^{(\dot{\mathfrak{X}})}\hspace{-1.4mm}\left\langle {2n{-}i}\right\rangle$ and $\mathcal{L}_2^{(\ddot{\mathfrak{X}})}\hspace{-1.4mm}\left\langle 2n{-}i\right\rangle$ are two lists of EDVs referring to the residual sojourn times of $2n-i$ vehicles processing the groups of $\mathcal{P}(\mathcal{A})$ in states $S_{i,0}$ and $S_{i,j}$, respectively. Further, $\mathcal{H}_1^{(\varphi;\mathfrak{R})}\hspace{-1.4mm}\left\langle {i{-}1} \hspace{-0.3mm};\hspace{-0.3mm} 1 \right\rangle$ and $\mathcal{H}_2^{(\gamma;\psi)}\hspace{-1.4mm}\left\langle j{-}1 \hspace{-0.3mm};\hspace{-0.3mm} i{-}(j{-}1) \right\rangle$ are lists of EDVs referring to the residual recruitment duration of $i$ recruiters in states $S_{i,0}$ and $S_{i,j}$, respectively.
\end{theorem}
\begin{proof}
See Appendix \ref{decompositionTheorem_proof}.
\end{proof}
Decomposition theorem implies that, for all H-states $S_{i,j}$, where {\small$0 \le i \le n$} and {\small$0\le j\le i+1$}, the departure of a vehicle leads to the transition to H-state $S_{i+1,0}$, which is a direct result of absorbency property presented in Proposition~\ref{absorbency_props}. Intuitively, in this situation, we say that $S_{i+1,0}$ absorbs the process from H-states $S_{i,j}$. Further, due to decomposability in Proposition~\ref{decomposability_props}, if recruiter $C_k$, for $1\le k \le i$, completes its recruitment, D-SMP will transit to H-state $S_{i-1,k}$. In other words, from each H-state $S_{i,j}$, considering which recruiter completes its recruitment before any other event, D-SMP will transit to $k$ different H-states. Further, in all of the H-states, the departure of a recruiter leads to a transition to the failure state.\par
In the reminder of this paper, we refer to the first element of $S_{i,0}$ and $S_{i,j}$ by $S_{i,0}^{\mathsf{Soj}}{=}\mathcal{L}_1^{(\dot{\mathfrak{X}})}\hspace{-1.4mm}\left\langle {2n{-}i}\right\rangle$ and $S_{i,j}^{\mathsf{Soj}}{=}\mathcal{L}_2^{(\ddot{\mathfrak{X}})}\hspace{-1.4mm}\left\langle 2n{-}i\right\rangle$, respectively. Likewise, we refer to the second element of $S_{i,0}$ and $S_{i,j}$ by $S_{i,0}^{\mathsf{Rec}}{=}\mathcal{H}_1^{(\varphi;\mathfrak{R})}\hspace{-1.4mm}\left\langle {i{-}1} \hspace{-0.3mm};\hspace{-0.3mm} 1 \right\rangle$ and $S_{i,j}^{\mathsf{Rec}}{=}\mathcal{H}_2^{(\gamma;\psi)}\hspace{-1.4mm}\left\langle j{-}1 \hspace{-0.3mm};\hspace{-0.3mm} i{-}(j{-}1) \right\rangle$, respectively. Using DT, we next aim to introduce decomposed SMP model (D-SMP) of {\tt RP-VC$_n$}.
\vspace{-2mm}
\subsection{Decomposed Semi-Markov Process (D-SMP) of {\tt RP-VC$_n$}}\label{DSMPModel}
According to DT, each state $S_i$ of $\beta$-SMP presented in Fig.~\ref{j2nMarkov1} can be decomposed into $i+2$ H-states resulting in the  D-SMP model of {\tt RP-VC$_n$} depicted in Fig.~\ref{j2nMarkov}. In D-SMP model, there are {\small$\frac{(n+1)(n+2)}{2}+1$} different H-states, each of which is denoted by {\small$S_{i,j}$}, as defined in Theorem~\ref{decompositionTheorem}. We group all the H-states indexed with $i$ into a state set denoted by {\small$\hat{\mathbf{S}}_i$} (i.e., {\small$\hat{\mathbf{S}}_i=\left\{S_{i,0},S_{i,1},\dots,S_{i,j}, \dots,S_{i,i+1}\right\}$} for {\small$0\le i < n$, $0\le j \le i+1$})\footnote{$\hat{\mathbf{S}}_i$ is equivalent to $S_i$ of $\beta$-SMP presented in Fig.~\ref{j2nMarkov1}.}. Also, state set {\small$\hat{\mathbf{S}}_n=\left\{S_{n,0}\right\}$} has only one H-state. We denote the overall set of all the states by {\small$\Psi=\left\{\bigcup_{i=0}^n \hat{\mathbf{S}}_i \right\}\cup \{F\}$}, where {\small$F$} indicates the failure of processing application {\small$\mathcal{A}$}. We next aim to model the dynamics of D-SMP below.\par
\textbf{Transition Probabilities of D-SMP.} To characterize the dynamics of D-SMP, we obtain transition probabilities below.
Let $p_{i,j}^{k}$, $q_{i,j}$, and $b_{i,j}$ be the transition probabilities from state {\small$S_{i,j}$} to {\small$S_{i-1,k}$}, from state $S_{i,j}$ to $S_{i+1,0}$, and from state $S_{i,j}$ to $F$, respectively. We first specify $p_{i,j}^{k}$, which is equal to the probability that the recruiter $k$ completes its recruitment operation before any other events. Considering $S_{i,j}\in\hat{\mathbf{S}}_i$, presented in (\ref{decompositionTheoremEq1}), the general expression for $p_{i,j}^{k}$ is as follows:
\begin{equation}\label{pijk_main}
\begin{aligned}
p_{i,j}^{k} {=}\textrm{Pr}\Big[S_{i,j}^{\mathsf{Rec}}(k) {<} \min\Big\{&\min_{1\le r \le 2n-i}\left\{S_{i,j}^{\mathsf{Soj}}(r)\right\}\\
&,\min_{\substack{1\le h \le i,\\h\neq k}}\left\{S_{i,j}^{\mathsf{Rec}}(h)\right\}\Big\}\Big],
\end{aligned}
\end{equation}
where {\small$0\le i \le n$, $0\le j \le i+1$, and $1\le k \le i$}. Further, $S_{i,j}^{\mathsf{Rec}}(k)$, $S_{i,j}^{\mathsf{Rec}}(h)$, and $S_{i,j}^{\mathsf{Soj}}(r)$ refer to the residual recruitment time of vehicle $k$, residual recruitment time of vehicle $h$, and residual residency time of vehicle $r$, respectively. Further, since there are {\small$2n-i$} vehicles in state {\small$S_{i,j}$}, of which $i$ vehicles are recruiter, and also {\small$q_{i,j}{+}b_{i,j}{=}1{-}\sum_{k=1}^{i}p_{i,j}^{k}$}, we have
\begin{equation}\label{bij}
  b_{i,j}=\frac{i}{2n-i}\times \left(1-\sum_{k=1}^{i}p_{i,j}^{k}\right),
\end{equation}
and
\begin{equation}\label{qij}
  q_{i,j}=\frac{2(n-i)}{2n-i}\times \left(1-\sum_{k=1}^{i}p_{i,j}^{k}\right),
\end{equation}
where $p_{0,0}^k=p_{0,1}^k=0$ because all of the groups have two vehicles in states $S_{0,0}$ and $S_{0,1}$.
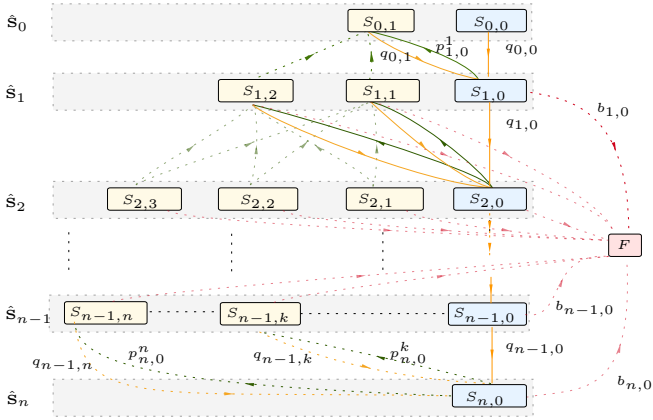
\begin{figure}[!t]
  \centering
\tikzset{every picture/.style={line width=0.2pt}} 
\begin{tikzpicture}[x=0.5pt,y=0.35pt,yscale=-1,xscale=1]

\draw  [color={rgb, 255:red, 161; green, 161; blue, 161 }  ,draw opacity=1 ][fill={rgb, 255:red, 237; green, 237; blue, 237 }  ,fill opacity=0.59 ][dash pattern={on 0.84pt off 2.51pt}] (35.77,4) -- (398.82,4) -- (398.82,44) -- (35.77,44) -- cycle ;
\draw  [color={rgb, 255:red, 161; green, 161; blue, 161 }  ,draw opacity=1 ][fill={rgb, 255:red, 237; green, 237; blue, 237 }  ,fill opacity=0.59 ][dash pattern={on 0.84pt off 2.51pt}] (36.75,79) -- (399.8,79) -- (399.8,119) -- (36.75,119) -- cycle ;
\draw  [color={rgb, 255:red, 161; green, 161; blue, 161 }  ,draw opacity=1 ][fill={rgb, 255:red, 237; green, 237; blue, 237 }  ,fill opacity=0.59 ][dash pattern={on 0.84pt off 2.51pt}] (36.75,196) -- (399.8,196) -- (399.8,236) -- (36.75,236) -- cycle ;
\draw  [color={rgb, 255:red, 161; green, 161; blue, 161 }  ,draw opacity=1 ][fill={rgb, 255:red, 237; green, 237; blue, 237 }  ,fill opacity=0.59 ][dash pattern={on 0.84pt off 2.51pt}] (33.8,319) -- (396.85,319) -- (396.85,359) -- (33.8,359) -- cycle ;
\draw  [color={rgb, 255:red, 161; green, 161; blue, 161 }  ,draw opacity=1 ][fill={rgb, 255:red, 237; green, 237; blue, 237 }  ,fill opacity=0.59 ][dash pattern={on 0.84pt off 2.51pt}] (36.75,410) -- (399.8,410) -- (399.8,450) -- (36.75,450) -- cycle ;
\draw [color={rgb, 255:red, 49; green, 92; blue, 0 }  ,draw opacity=1 ]   (369.04,203.04) .. controls (361.71,183.77) and (307.06,119.75) .. (276.88,109.94) ;
\draw [shift={(328.75,150.01)}, rotate = 46.76] [fill={rgb, 255:red, 49; green, 92; blue, 0 }  ,fill opacity=1 ][line width=0.08]  [draw opacity=0] (7.2,-1.8) -- (0,0) -- (7.2,1.8) -- cycle    ;
\draw [color={rgb, 255:red, 49; green, 92; blue, 0 }  ,draw opacity=1 ]   (369.04,203.04) .. controls (361.71,183.77) and (217.89,123.02) .. (187.72,113.21) ;
\draw [shift={(280.83,151.77)}, rotate = 24.71] [fill={rgb, 255:red, 49; green, 92; blue, 0 }  ,fill opacity=1 ][line width=0.08]  [draw opacity=0] (7.2,-1.8) -- (0,0) -- (7.2,1.8) -- cycle    ;
\draw [color={rgb, 255:red, 65; green, 117; blue, 5 }  ,draw opacity=1 ]   (358.33,85.43) .. controls (350.99,66.16) and (304.49,42.93) .. (274.31,33.13) ;
\draw [shift={(316.35,51.03)}, rotate = 28.44] [fill={rgb, 255:red, 65; green, 117; blue, 5 }  ,fill opacity=1 ][line width=0.08]  [draw opacity=0] (7.2,-1.8) -- (0,0) -- (7.2,1.8) -- cycle    ;
\draw [color={rgb, 255:red, 65; green, 117; blue, 5 }  ,draw opacity=1 ] [dash pattern={on 0.84pt off 2.51pt}]  (276.02,86.25) .. controls (279.45,71.54) and (277.56,50.14) .. (274.31,33.13) ;
\draw [shift={(277.28,54.66)}, rotate = 86.22] [fill={rgb, 255:red, 65; green, 117; blue, 5 }  ,fill opacity=1 ][line width=0.08]  [draw opacity=0] (8.4,-2.1) -- (0,0) -- (8.4,2.1) -- cycle    ;
\draw [color={rgb, 255:red, 65; green, 117; blue, 5 }  ,draw opacity=1 ] [dash pattern={on 0.84pt off 2.51pt}]  (186,87.06) .. controls (204.01,75.62) and (232.98,59.13) .. (274.31,33.13) ;
\draw [shift={(234.9,57.42)}, rotate = 148.86] [fill={rgb, 255:red, 65; green, 117; blue, 5 }  ,fill opacity=1 ][line width=0.08]  [draw opacity=0] (8.4,-2.1) -- (0,0) -- (8.4,2.1) -- cycle    ;
\draw [color={rgb, 255:red, 245; green, 166; blue, 35 }  ,draw opacity=1 ]   (274.31,33.13) .. controls (287.85,50.96) and (332.82,78.93) .. (358.33,85.43) ;
\draw [shift={(317.2,66.78)}, rotate = 211.24] [fill={rgb, 255:red, 245; green, 166; blue, 35 }  ,fill opacity=1 ][line width=0.08]  [draw opacity=0] (7.2,-1.8) -- (0,0) -- (7.2,1.8) -- cycle    ;
\draw [color={rgb, 255:red, 49; green, 92; blue, 0 }  ,draw opacity=0.52 ] [dash pattern={on 0.84pt off 2.51pt}]  (278.64,201.44) .. controls (282.93,188.36) and (282.71,125.32) .. (276.88,109.94) ;
\draw [shift={(281.54,155.4)}, rotate = 89.49] [fill={rgb, 255:red, 49; green, 92; blue, 0 }  ,fill opacity=0.52 ][line width=0.08]  [draw opacity=0] (7.2,-1.8) -- (0,0) -- (7.2,1.8) -- cycle    ;
\draw [color={rgb, 255:red, 49; green, 92; blue, 0 }  ,draw opacity=0.52 ] [dash pattern={on 0.84pt off 2.51pt}]  (278.64,201.44) .. controls (282.93,188.36) and (193.55,128.59) .. (187.72,113.21) ;
\draw [shift={(234.26,155.97)}, rotate = 39.37] [fill={rgb, 255:red, 49; green, 92; blue, 0 }  ,fill opacity=0.52 ][line width=0.08]  [draw opacity=0] (7.2,-1.8) -- (0,0) -- (7.2,1.8) -- cycle    ;
\draw [color={rgb, 255:red, 49; green, 92; blue, 0 }  ,draw opacity=0.52 ] [dash pattern={on 0.84pt off 2.51pt}]  (185.84,199.84) .. controls (196.99,187.58) and (250.99,125.32) .. (276.88,109.94) ;
\draw [shift={(229.66,152.07)}, rotate = 133.77] [fill={rgb, 255:red, 49; green, 92; blue, 0 }  ,fill opacity=0.52 ][line width=0.08]  [draw opacity=0] (7.2,-1.8) -- (0,0) -- (7.2,1.8) -- cycle    ;
\draw [color={rgb, 255:red, 49; green, 92; blue, 0 }  ,draw opacity=0.52 ] [dash pattern={on 0.84pt off 2.51pt}]  (99.44,199.84) .. controls (110.59,187.58) and (250.99,125.32) .. (276.88,109.94) ;
\draw [shift={(188.6,153.21)}, rotate = 154.45] [fill={rgb, 255:red, 49; green, 92; blue, 0 }  ,fill opacity=0.52 ][line width=0.08]  [draw opacity=0] (7.2,-1.8) -- (0,0) -- (7.2,1.8) -- cycle    ;
\draw [color={rgb, 255:red, 49; green, 92; blue, 0 }  ,draw opacity=0.52 ] [dash pattern={on 0.84pt off 2.51pt}]  (185.84,199.84) .. controls (190.13,186.76) and (193.55,128.59) .. (187.72,113.21) ;
\draw [shift={(190.61,156.3)}, rotate = 92.37] [fill={rgb, 255:red, 49; green, 92; blue, 0 }  ,fill opacity=0.52 ][line width=0.08]  [draw opacity=0] (7.2,-1.8) -- (0,0) -- (7.2,1.8) -- cycle    ;
\draw [color={rgb, 255:red, 49; green, 92; blue, 0 }  ,draw opacity=0.52 ] [dash pattern={on 0.84pt off 2.51pt}]  (99.44,199.84) .. controls (103.73,186.76) and (171.43,136.91) .. (187.72,113.21) ;
\draw [shift={(143.78,156.39)}, rotate = 138.8] [fill={rgb, 255:red, 49; green, 92; blue, 0 }  ,fill opacity=0.52 ][line width=0.08]  [draw opacity=0] (7.2,-1.8) -- (0,0) -- (7.2,1.8) -- cycle    ;
\draw [color={rgb, 255:red, 245; green, 166; blue, 35 }  ,draw opacity=1 ]   (276.88,109.94) .. controls (290.42,127.77) and (343.53,196.54) .. (369.04,203.04) ;
\draw [shift={(319.02,161.35)}, rotate = 228.3] [fill={rgb, 255:red, 245; green, 166; blue, 35 }  ,fill opacity=1 ][line width=0.08]  [draw opacity=0] (7.2,-1.8) -- (0,0) -- (7.2,1.8) -- cycle    ;
\draw [color={rgb, 255:red, 245; green, 166; blue, 35 }  ,draw opacity=1 ]   (187.72,113.21) .. controls (201.26,131.04) and (313.84,197.44) .. (369.04,203.04) ;
\draw [shift={(273.92,168.51)}, rotate = 207.19] [fill={rgb, 255:red, 245; green, 166; blue, 35 }  ,fill opacity=1 ][line width=0.08]  [draw opacity=0] (7.2,-1.8) -- (0,0) -- (7.2,1.8) -- cycle    ;
\draw [color={rgb, 255:red, 49; green, 92; blue, 0 }  ,draw opacity=1 ][line width=0.2]  [dash pattern={on 0.84pt off 2.51pt}]  (368,414.8) .. controls (348.12,406.63) and (273.8,381) .. (192.8,358.8) ;
\draw [shift={(280.66,384.77)}, rotate = 17.61] [fill={rgb, 255:red, 49; green, 92; blue, 0 }  ,fill opacity=1 ][line width=0.08]  [draw opacity=0] (7.2,-1.8) -- (0,0) -- (7.2,1.8) -- cycle    ;
\draw [color={rgb, 255:red, 245; green, 166; blue, 35 }  ,draw opacity=1 ][line width=0.2]  [dash pattern={on 0.84pt off 2.51pt}]  (192.8,358.8) .. controls (221.04,379.04) and (281.8,409) .. (368,414.8) ;
\draw [shift={(277.22,398.41)}, rotate = 197.4] [fill={rgb, 255:red, 245; green, 166; blue, 35 }  ,fill opacity=1 ][line width=0.08]  [draw opacity=0] (7.2,-1.8) -- (0,0) -- (7.2,1.8) -- cycle    ;
\draw [line width=0.2]  [dash pattern={on 0.84pt off 2.51pt}]  (367,237.94) -- (367,254.29) ;
\draw [line width=0.2]  [dash pattern={on 0.84pt off 2.51pt}]  (286.2,251) -- (286.2,293.7) ;
\draw [color={rgb, 255:red, 208; green, 2; blue, 27 }  ,draw opacity=1 ] [dash pattern={on 0.84pt off 2.51pt}]  (397.8,99.8) .. controls (429.19,107.17) and (472,119.29) .. (472.24,251.04) ;
\draw [shift={(462.61,163.8)}, rotate = 253.24] [fill={rgb, 255:red, 208; green, 2; blue, 27 }  ,fill opacity=1 ][line width=0.08]  [draw opacity=0] (7.2,-1.8) -- (0,0) -- (7.2,1.8) -- cycle    ;
\draw [color={rgb, 255:red, 208; green, 2; blue, 27 }  ,draw opacity=0.5 ] [dash pattern={on 0.84pt off 2.51pt}]  (276.88,109.94) .. controls (308.27,117.31) and (451.67,205.31) .. (460.24,249.44) ;
\draw [shift={(383.22,170.21)}, rotate = 215.23] [fill={rgb, 255:red, 208; green, 2; blue, 27 }  ,fill opacity=0.5 ][line width=0.08]  [draw opacity=0] (7.2,-1.8) -- (0,0) -- (7.2,1.8) -- cycle    ;
\draw [color={rgb, 255:red, 208; green, 2; blue, 27 }  ,draw opacity=0.5 ] [dash pattern={on 0.84pt off 2.51pt}]  (187.72,113.21) .. controls (219.11,120.58) and (451.67,205.31) .. (460.24,249.44) ;
\draw [shift={(336.32,169.2)}, rotate = 203.72] [fill={rgb, 255:red, 208; green, 2; blue, 27 }  ,fill opacity=0.5 ][line width=0.08]  [draw opacity=0] (7.2,-1.8) -- (0,0) -- (7.2,1.8) -- cycle    ;
\draw [color={rgb, 255:red, 208; green, 2; blue, 27 }  ,draw opacity=0.5 ] [dash pattern={on 0.84pt off 2.51pt}]  (394.8,227.2) .. controls (429.04,247.04) and (444.24,252.64) .. (456.76,259.61) ;
\draw [shift={(429.39,245.98)}, rotate = 206.96] [fill={rgb, 255:red, 208; green, 2; blue, 27 }  ,fill opacity=0.5 ][line width=0.08]  [draw opacity=0] (7.2,-1.8) -- (0,0) -- (7.2,1.8) -- cycle    ;
\draw [color={rgb, 255:red, 208; green, 2; blue, 27 }  ,draw opacity=0.5 ] [dash pattern={on 0.84pt off 2.51pt}]  (394.64,429.44) .. controls (432.24,419.04) and (477.04,436.64) .. (470.64,279.04) ;
\draw [shift={(466.03,371.2)}, rotate = 102.97] [fill={rgb, 255:red, 208; green, 2; blue, 27 }  ,fill opacity=0.5 ][line width=0.08]  [draw opacity=0] (7.2,-1.8) -- (0,0) -- (7.2,1.8) -- cycle    ;
\draw [line width=0.2]  [dash pattern={on 0.84pt off 2.51pt}]  (109.8,337) -- (163.24,337) ;
\draw  [fill={rgb, 255:red, 231; green, 243; blue, 255 }  ,fill opacity=1 ] (341.8,12) .. controls (341.8,10.9) and (342.7,10) .. (343.8,10) -- (392.8,10) .. controls (393.9,10) and (394.8,10.9) .. (394.8,12) -- (394.8,32.2) .. controls (394.8,33.3) and (393.9,34.2) .. (392.8,34.2) -- (343.8,34.2) .. controls (342.7,34.2) and (341.8,33.3) .. (341.8,32.2) -- cycle ;
\draw  [fill={rgb, 255:red, 255; green, 250; blue, 231 }  ,fill opacity=1 ] (259.8,12) .. controls (259.8,10.9) and (260.7,10) .. (261.8,10) -- (307.8,10) .. controls (308.9,10) and (309.8,10.9) .. (309.8,12) -- (309.8,32.2) .. controls (309.8,33.3) and (308.9,34.2) .. (307.8,34.2) -- (261.8,34.2) .. controls (260.7,34.2) and (259.8,33.3) .. (259.8,32.2) -- cycle ;
\draw  [fill={rgb, 255:red, 231; green, 243; blue, 255 }  ,fill opacity=1 ] (340.8,87.8) .. controls (340.8,86.7) and (341.7,85.8) .. (342.8,85.8) -- (392.8,85.8) .. controls (393.9,85.8) and (394.8,86.7) .. (394.8,87.8) -- (394.8,108) .. controls (394.8,109.1) and (393.9,110) .. (392.8,110) -- (342.8,110) .. controls (341.7,110) and (340.8,109.1) .. (340.8,108) -- cycle ;
\draw  [fill={rgb, 255:red, 255; green, 250; blue, 231 }  ,fill opacity=1 ] (258.55,87.8) .. controls (258.55,86.7) and (259.45,85.8) .. (260.55,85.8) -- (311.8,85.8) .. controls (312.9,85.8) and (313.8,86.7) .. (313.8,87.8) -- (313.8,108) .. controls (313.8,109.1) and (312.9,110) .. (311.8,110) -- (260.55,110) .. controls (259.45,110) and (258.55,109.1) .. (258.55,108) -- cycle ;
\draw  [fill={rgb, 255:red, 255; green, 250; blue, 231 }  ,fill opacity=1 ] (161.8,87.8) .. controls (161.8,86.7) and (162.7,85.8) .. (163.8,85.8) -- (216.28,85.8) .. controls (217.38,85.8) and (218.28,86.7) .. (218.28,87.8) -- (218.28,108) .. controls (218.28,109.1) and (217.38,110) .. (216.28,110) -- (163.8,110) .. controls (162.7,110) and (161.8,109.1) .. (161.8,108) -- cycle ;
\draw  [fill={rgb, 255:red, 231; green, 243; blue, 255 }  ,fill opacity=1 ] (339.8,205) .. controls (339.8,203.9) and (340.7,203) .. (341.8,203) -- (392.8,203) .. controls (393.9,203) and (394.8,203.9) .. (394.8,205) -- (394.8,225.2) .. controls (394.8,226.3) and (393.9,227.2) .. (392.8,227.2) -- (341.8,227.2) .. controls (340.7,227.2) and (339.8,226.3) .. (339.8,225.2) -- cycle ;
\draw  [fill={rgb, 255:red, 255; green, 250; blue, 231 }  ,fill opacity=1 ] (258.55,205) .. controls (258.55,203.9) and (259.45,203) .. (260.55,203) -- (314.8,203) .. controls (315.9,203) and (316.8,203.9) .. (316.8,205) -- (316.8,225.2) .. controls (316.8,226.3) and (315.9,227.2) .. (314.8,227.2) -- (260.55,227.2) .. controls (259.45,227.2) and (258.55,226.3) .. (258.55,225.2) -- cycle ;
\draw  [fill={rgb, 255:red, 255; green, 250; blue, 231 }  ,fill opacity=1 ] (161.8,205) .. controls (161.8,203.9) and (162.7,203) .. (163.8,203) -- (219.8,203) .. controls (220.9,203) and (221.8,203.9) .. (221.8,205) -- (221.8,225.2) .. controls (221.8,226.3) and (220.9,227.2) .. (219.8,227.2) -- (163.8,227.2) .. controls (162.7,227.2) and (161.8,226.3) .. (161.8,225.2) -- cycle ;
\draw  [fill={rgb, 255:red, 255; green, 250; blue, 231 }  ,fill opacity=1 ] (77.89,205) .. controls (77.89,203.9) and (78.78,203) .. (79.89,203) -- (131.8,203) .. controls (132.9,203) and (133.8,203.9) .. (133.8,205) -- (133.8,225.2) .. controls (133.8,226.3) and (132.9,227.2) .. (131.8,227.2) -- (79.89,227.2) .. controls (78.78,227.2) and (77.89,226.3) .. (77.89,225.2) -- cycle ;
\draw  [fill={rgb, 255:red, 231; green, 243; blue, 255 }  ,fill opacity=1 ] (338.8,417.8) .. controls (338.8,416.7) and (339.7,415.8) .. (340.8,415.8) -- (392.8,415.8) .. controls (393.9,415.8) and (394.8,416.7) .. (394.8,417.8) -- (394.8,438) .. controls (394.8,439.1) and (393.9,440) .. (392.8,440) -- (340.8,440) .. controls (339.7,440) and (338.8,439.1) .. (338.8,438) -- cycle ;
\draw  [fill={rgb, 255:red, 255; green, 250; blue, 231 }  ,fill opacity=1 ] (163.24,329.1) .. controls (163.24,327.99) and (164.14,327.1) .. (165.24,327.1) -- (221.8,327.1) .. controls (222.9,327.1) and (223.8,327.99) .. (223.8,329.1) -- (223.8,349.3) .. controls (223.8,350.4) and (222.9,351.3) .. (221.8,351.3) -- (165.24,351.3) .. controls (164.14,351.3) and (163.24,350.4) .. (163.24,349.3) -- cycle ;
\draw  [fill={rgb, 255:red, 255; green, 250; blue, 231 }  ,fill opacity=1 ] (45.33,328.1) .. controls (45.33,326.99) and (46.22,326.1) .. (47.33,326.1) -- (105.8,326.1) .. controls (106.9,326.1) and (107.8,326.99) .. (107.8,328.1) -- (107.8,348.3) .. controls (107.8,349.4) and (106.9,350.3) .. (105.8,350.3) -- (47.33,350.3) .. controls (46.22,350.3) and (45.33,349.4) .. (45.33,348.3) -- cycle ;
\draw  [fill={rgb, 255:red, 255; green, 226; blue, 226 }  ,fill opacity=1 ] (457,255) .. controls (457,253.9) and (457.9,253) .. (459,253) -- (479.8,253) .. controls (480.9,253) and (481.8,253.9) .. (481.8,255) -- (481.8,275.2) .. controls (481.8,276.3) and (480.9,277.2) .. (479.8,277.2) -- (459,277.2) .. controls (457.9,277.2) and (457,276.3) .. (457,275.2) -- cycle ;
\draw  [fill={rgb, 255:red, 231; green, 243; blue, 255 }  ,fill opacity=1 ] (335.8,329.41) .. controls (335.8,328.3) and (336.7,327.41) .. (337.8,327.41) -- (392.8,327.41) .. controls (393.9,327.41) and (394.8,328.3) .. (394.8,329.41) -- (394.8,349.61) .. controls (394.8,350.71) and (393.9,351.61) .. (392.8,351.61) -- (337.8,351.61) .. controls (336.7,351.61) and (335.8,350.71) .. (335.8,349.61) -- cycle ;
\draw [line width=0.2]  [dash pattern={on 0.84pt off 2.51pt}]  (223.8,339) -- (335.8,339) ;
\draw [color={rgb, 255:red, 255; green, 151; blue, 0 }  ,draw opacity=1 ]   (368.8,353.6) -- (368.8,414.4) ;
\draw [shift={(368.8,388.2)}, rotate = 270] [fill={rgb, 255:red, 255; green, 151; blue, 0 }  ,fill opacity=1 ][line width=0.08]  [draw opacity=0] (7.2,-1.8) -- (0,0) -- (7.2,1.8) -- cycle    ;
\draw [color={rgb, 255:red, 255; green, 151; blue, 0 }  ,draw opacity=1 ]   (368,299.29) -- (368,327.8) ;
\draw [shift={(368,317.75)}, rotate = 270] [fill={rgb, 255:red, 255; green, 151; blue, 0 }  ,fill opacity=1 ][line width=0.08]  [draw opacity=0] (7.2,-1.8) -- (0,0) -- (7.2,1.8) -- cycle    ;
\draw [color={rgb, 255:red, 255; green, 151; blue, 0 }  ,draw opacity=1 ] [dash pattern={on 0.84pt off 2.51pt}]  (367,254.29) -- (367,282.8) ;
\draw [shift={(367,272.75)}, rotate = 270] [fill={rgb, 255:red, 255; green, 151; blue, 0 }  ,fill opacity=1 ][line width=0.08]  [draw opacity=0] (7.2,-1.8) -- (0,0) -- (7.2,1.8) -- cycle    ;
\draw [color={rgb, 255:red, 255; green, 151; blue, 0 }  ,draw opacity=1 ] [dash pattern={on 0.84pt off 2.51pt}]  (367,230.43) -- (367,258.94) ;
\draw [shift={(367,248.88)}, rotate = 270] [fill={rgb, 255:red, 255; green, 151; blue, 0 }  ,fill opacity=1 ][line width=0.08]  [draw opacity=0] (7.2,-1.8) -- (0,0) -- (7.2,1.8) -- cycle    ;
\draw [color={rgb, 255:red, 255; green, 151; blue, 0 }  ,draw opacity=1 ]   (367,110.29) -- (367,201.44) ;
\draw [shift={(367,160.07)}, rotate = 270] [fill={rgb, 255:red, 255; green, 151; blue, 0 }  ,fill opacity=1 ][line width=0.08]  [draw opacity=0] (7.2,-1.8) -- (0,0) -- (7.2,1.8) -- cycle    ;
\draw [color={rgb, 255:red, 255; green, 151; blue, 0 }  ,draw opacity=1 ]   (366,34.29) -- (366,82.8) ;
\draw [shift={(366,62.75)}, rotate = 270] [fill={rgb, 255:red, 255; green, 151; blue, 0 }  ,fill opacity=1 ][line width=0.08]  [draw opacity=0] (7.2,-1.8) -- (0,0) -- (7.2,1.8) -- cycle    ;
\draw [color={rgb, 255:red, 208; green, 2; blue, 27 }  ,draw opacity=0.5 ] [dash pattern={on 0.84pt off 2.51pt}]  (393.8,340) .. controls (425.19,347.37) and (409,272.4) .. (457,277.2) ;
\draw [shift={(422.63,300.35)}, rotate = 111.34] [fill={rgb, 255:red, 208; green, 2; blue, 27 }  ,fill opacity=0.5 ][line width=0.08]  [draw opacity=0] (7.2,-1.8) -- (0,0) -- (7.2,1.8) -- cycle    ;
\draw [color={rgb, 255:red, 208; green, 2; blue, 27 }  ,draw opacity=0.5 ] [dash pattern={on 0.84pt off 2.51pt}]  (81.71,326.1) .. controls (93.75,315.74) and (427.4,280.4) .. (457,277.2) ;
\draw [shift={(273.41,297.7)}, rotate = 173.27] [fill={rgb, 255:red, 208; green, 2; blue, 27 }  ,fill opacity=0.5 ][line width=0.08]  [draw opacity=0] (7.2,-1.8) -- (0,0) -- (7.2,1.8) -- cycle    ;
\draw [color={rgb, 255:red, 208; green, 2; blue, 27 }  ,draw opacity=0.5 ] [dash pattern={on 0.84pt off 2.51pt}]  (203.67,327.1) .. controls (215.71,316.74) and (427.4,280.4) .. (457,277.2) ;
\draw [shift={(334.22,297.53)}, rotate = 169.6] [fill={rgb, 255:red, 208; green, 2; blue, 27 }  ,fill opacity=0.5 ][line width=0.08]  [draw opacity=0] (7.2,-1.8) -- (0,0) -- (7.2,1.8) -- cycle    ;
\draw [color={rgb, 255:red, 208; green, 2; blue, 27 }  ,draw opacity=0.5 ] [dash pattern={on 0.84pt off 2.51pt}]  (123.27,227.2) .. controls (128.3,241.48) and (442.36,259.61) .. (456.76,259.61) ;
\draw [shift={(293.79,249.32)}, rotate = 184.41] [fill={rgb, 255:red, 208; green, 2; blue, 27 }  ,fill opacity=0.5 ][line width=0.08]  [draw opacity=0] (7.2,-1.8) -- (0,0) -- (7.2,1.8) -- cycle    ;
\draw [color={rgb, 255:red, 208; green, 2; blue, 27 }  ,draw opacity=0.5 ] [dash pattern={on 0.84pt off 2.51pt}]  (210.63,227.2) .. controls (215.66,241.48) and (442.36,259.61) .. (456.76,259.61) ;
\draw [shift={(337.28,249.51)}, rotate = 186.01] [fill={rgb, 255:red, 208; green, 2; blue, 27 }  ,fill opacity=0.5 ][line width=0.08]  [draw opacity=0] (7.2,-1.8) -- (0,0) -- (7.2,1.8) -- cycle    ;
\draw [color={rgb, 255:red, 208; green, 2; blue, 27 }  ,draw opacity=0.5 ] [dash pattern={on 0.84pt off 2.51pt}]  (303.93,227.2) .. controls (308.96,241.48) and (442.36,259.61) .. (456.76,259.61) ;
\draw [shift={(383.28,249.88)}, rotate = 189.8] [fill={rgb, 255:red, 208; green, 2; blue, 27 }  ,fill opacity=0.5 ][line width=0.08]  [draw opacity=0] (7.2,-1.8) -- (0,0) -- (7.2,1.8) -- cycle    ;
\draw [line width=0.2]  [dash pattern={on 0.84pt off 2.51pt}]  (171.8,252.6) -- (171.8,295.3) ;
\draw [line width=0.2]  [dash pattern={on 0.84pt off 2.51pt}]  (48.6,251) -- (48.6,293.7) ;
\draw [color={rgb, 255:red, 49; green, 92; blue, 0 }  ,draw opacity=1 ][line width=0.2]  [dash pattern={on 0.84pt off 2.51pt}]  (337.8,428) .. controls (97.8,415) and (56.24,390.44) .. (52.8,347.8) ;
\draw [shift={(181.49,414.52)}, rotate = 7.77] [fill={rgb, 255:red, 49; green, 92; blue, 0 }  ,fill opacity=1 ][line width=0.08]  [draw opacity=0] (7.2,-1.8) -- (0,0) -- (7.2,1.8) -- cycle    ;
\draw [color={rgb, 255:red, 245; green, 166; blue, 35 }  ,draw opacity=1 ][line width=0.2]  [dash pattern={on 0.84pt off 2.51pt}]  (52.8,347.8) .. controls (61.8,443) and (109.8,422) .. (337.8,428) ;
\draw [shift={(173.45,425.94)}, rotate = 181.47] [fill={rgb, 255:red, 245; green, 166; blue, 35 }  ,fill opacity=1 ][line width=0.08]  [draw opacity=0] (7.2,-1.8) -- (0,0) -- (7.2,1.8) -- cycle    ;

\draw (324.04,36.55) node [anchor=north west][inner sep=0.75pt]  [font=\tiny]  {$p_{1,0}^{1}$};
\draw (289.6,367.6) node [anchor=north west][inner sep=0.75pt]  [font=\tiny]  {$p_{n,0}^{k}$};
\draw (446.35,108.37) node [anchor=north west][inner sep=0.75pt]  [font=\tiny]  {$b_{1,0}$};
\draw (377.26,126.09) node [anchor=north west][inner sep=0.75pt]  [font=\tiny]  {$q_{1,0}$};
\draw (376.83,45.27) node [anchor=north west][inner sep=0.75pt]  [font=\tiny]  {$q_{0,0}$};
\draw (281.9,53.17) node [anchor=north west][inner sep=0.75pt]  [font=\tiny]  {$q_{0,1}$};
\draw (461.6,404.4) node [anchor=north west][inner sep=0.75pt]  [font=\tiny]  {$b_{n,0}$};
\draw (186,378.6) node [anchor=north west][inner sep=0.75pt]  [font=\tiny]  {$q_{n-1,k}$};
\draw (377,368.4) node [anchor=north west][inner sep=0.75pt]  [font=\tiny]  {$q_{n-1,0}$};
\draw (337.8,330.81) node [anchor=north west][inner sep=0.75pt]  [font=\tiny]  {$S_{n-1,0}$};
\draw (268.8,15.47) node [anchor=north west][inner sep=0.75pt]  [font=\tiny]  {$S_{0,1}$};
\draw (352.7,15.4) node [anchor=north west][inner sep=0.75pt]  [font=\tiny]  {$S_{0,0}$};
\draw (178.08,90.54) node [anchor=north west][inner sep=0.75pt]  [font=\tiny]  {$S_{1,2}$};
\draw (268.53,89.91) node [anchor=north west][inner sep=0.75pt]  [font=\tiny]  {$S_{1,1}$};
\draw (350.39,91.09) node [anchor=north west][inner sep=0.75pt]  [font=\tiny]  {$S_{1,0}$};
\draw (87.19,208.54) node [anchor=north west][inner sep=0.75pt]  [font=\tiny]  {$S_{2,3}$};
\draw (176.38,208.21) node [anchor=north west][inner sep=0.75pt]  [font=\tiny]  {$S_{2,2}$};
\draw (265.56,207.6) node [anchor=north west][inner sep=0.75pt]  [font=\tiny]  {$S_{2,1}$};
\draw (349.56,207.41) node [anchor=north west][inner sep=0.75pt]  [font=\tiny]  {$S_{2,0}$};
\draw (47.33,329.5) node [anchor=north west][inner sep=0.75pt]  [font=\tiny]  {$S_{n-1,n}$};
\draw (166.29,330.5) node [anchor=north west][inner sep=0.75pt]  [font=\tiny]  {$S_{n-1,k}$};
\draw (346.96,421.79) node [anchor=north west][inner sep=0.75pt]  [font=\tiny]  {$S_{n,0}$};
\draw (414.8,322.4) node [anchor=north west][inner sep=0.75pt]  [font=\tiny]  {$b_{n-1,0}$};
\draw (0,13.4) node [anchor=north west][inner sep=0.75pt]  [font=\tiny]  {$\hat{\mathbf{S}}_{0}$};
\draw (-1,89.4) node [anchor=north west][inner sep=0.75pt]  [font=\tiny]  {$\hat{\mathbf{S}}_{1}$};
\draw (-1,207.4) node [anchor=north west][inner sep=0.75pt]  [font=\tiny]  {$\hat{\mathbf{S}}_{2}$};
\draw (-1,328.4) node [anchor=north west][inner sep=0.75pt]  [font=\tiny]  {$\hat{\mathbf{S}}_{n-1}$};
\draw (0,420.4) node [anchor=north west][inner sep=0.75pt]  [font=\tiny]  {$\hat{\mathbf{S}}_{n}$};
\draw (91.6,371.6) node [anchor=north west][inner sep=0.75pt]  [font=\tiny]  {$p_{n,0}^{n}$};
\draw (19,385.6) node [anchor=north west][inner sep=0.75pt]  [font=\tiny]  {$q_{n-1,n}$};
\draw (461.87,256.56) node [anchor=north west][inner sep=0.75pt]  [font=\tiny]  {$F$};
\end{tikzpicture}
 \caption{\small D-SMP model of {\tt RP-VC$_n$}.}
  \label{j2nMarkov}
\end{figure}
The final derivation of {\small$\displaystyle p_{i,j}^{k}$} is presented later upon assuming exponential distribution of the sojourn times and recruitment duration of the vehicles (Sec.~\ref{ERT}), which is a common assumption~\cite{24,29}. Nevertheless, general expressions for {\small$\displaystyle p_{i,j}^{k}$} is obtained in Appendix~\ref{app:generalization} under arbitrary distributions. Let {\small$\mathcal{R}(t)=\Big|\Theta\Big(\mathcal{M}\big(\mathcal{P}(\mathcal{A}),t\big)\Big)\Big|$}, (see (\ref{theta})). Relying on D-SMP presented in Fig.~\ref{j2nMarkov}, we calculate \textit{C-MTTF}{\small$\Big(\mathcal{A}\big|\mathcal{M}\big(\mathcal{P}(\mathcal{A}),t\big),\mathcal{R}(t),t\Big)$} in the next section.

\vspace{-3mm}
\section{Calculating \textit{C-MTTF}{\small$\Big(\mathcal{A}\big|\mathcal{M}\big(\mathcal{P}(\mathcal{A}),t\big),\mathcal{R}(t),t\Big)$}}\label{MTTFCal}
\noindent In this section, we derive the closed-form expression of \textit{C-MTTF}{\small$\Big(\mathcal{A}\big|\mathcal{M}\big(\mathcal{P}(\mathcal{A}),t\big),\mathcal{R}(t),t\Big)$}, when $\mathcal{R}(t)=0$ and $t=0$ (i.e., all the groups have two vehicles), resembling the start of application execution (Sec.~\ref{closedForm}) while {\small$t>0$} is left as future work. Finally, we study a special case of {\tt RP-VC$_n$}, where sojourn times and recruitment duration of the vehicles follow the exponential distribution (Sec.~\ref{ERT}).
\vspace{-2.5mm}
\subsection{Closed-Form of \textit{C-MTTF}{\small$\big(\mathcal{A}|\mathcal{M}\big(\mathcal{P}(\mathcal{A}),t\big),\mathcal{R}(t){=}0,t{=}0\big)$}}\label{closedForm}
We first present the following technical result to calculate {\small$\mathbb{E}\left[\widehat{\mathrm{Q}}_{i,j}\right]$}, utilized to obtain \textit{C-MTTF}{\small$\big(\mathcal{A}|\mathcal{M}\big(\mathcal{P}(\mathcal{A}),t\big),\mathcal{R}(t){=}0,t{=}0\big)$}.
\begin{lemma}\label{expecteTimeToFailure}
Let D-SMP be in state {\small$S_{i,j}$}. Also, let random variable $\widehat{\mathrm{W}}_{i,j}$ denote the sojourn time in state $S_{i,j}$. The expected value of the time until failure from state {\small$S_{i,j}$} is given as follows:
\begin{equation}
\hspace{-4mm}
\begin{aligned}
\mathbb{E}\left[\widehat{\mathrm{Q}}_{i,j}\right] &=\mathbb{E}\left[\widehat{\mathrm{W}}_{i,j}\right] +\mathbb{E}\left[\widehat{\mathrm{Q}}_{i+1,0}\right] q_{i,j} \\ &+\boldsymbol{\sum}_{k=1}^{i}\mathbb{E} \left[\widehat{\mathrm{Q}}_{i-1,k}\right]p_{i,j}^{k}, ~{\small 0\leq i\leq n}, {\small 0\leq j\leq i{+}1, }
\end{aligned}
\hspace{-7mm}
\end{equation}
where {\small$\mathbb{E}\left[\widehat{\mathrm{Q}}_{n{+}1,0}\right]{\triangleq}0$}, $p_{i,j}^{k}$ is given by \eqref{pijk_main}, and $q_{i,j}$ can be calculated by ~\eqref{qij}.
\end{lemma}
\begin{proof}
See Appendix \ref{expecteTimeToFailure_proof}.
\end{proof}
Using Lemma~\ref{expecteTimeToFailure}, we obtain the closed-form of \textit{C-MTTF}{\small$\big(\mathcal{A}|\mathcal{M}\big(\mathcal{P}(\mathcal{A}),t\big),\mathcal{R}(t){=}0,t{=}0\big)$} in the following key theorem.
\begin{theorem}[\textit{C-MTTF}{\small$\big(\mathcal{A}|\mathcal{M}\big(\mathcal{P}(\mathcal{A}),t\big),\mathcal{R}(t){=}0,t{=}0\big)$}]\label{MTTFJ2dn}
Under {\tt RP-VC$_n$}, \textit{C-MTTF}{\small$\big(\mathcal{A}|\mathcal{M}\big(\mathcal{P}(\mathcal{A}),t\big),\mathcal{R}(t){=}0,t{=}0\big)=\mathbb{E}\left[\widehat{\mathrm{Q}}_{0,0}\right]$}, where
\begin{equation}\label{F53}
\begin{aligned}
\mathbb{E}\left[\widehat{\mathrm{Q}}_{0,0}\right]
  &=\mathbb{E}\left[\widehat{\mathrm{W}}_{0,0}\right]+\boldsymbol{\sum} _{i=1}^{n}\beta(i) \mathbb{E}\left[\widehat{\mathrm{W}}_{i,0}\right]\\
  &+\boldsymbol{\sum} _{i=1}^{n} \boldsymbol{\sum} _{j=1}^{i}\beta(i)\mathbb{E}\left[\widehat{\mathrm{W}}_{i-1,j}\right] A(i+1,j).
\end{aligned}
\end{equation}
In (\ref{F53}), $\beta(i)$ is given by
\begin{equation}
  \beta(i)=\boldsymbol{\prod} _{k=1}^{i}\alpha(k),
\end{equation}
where
\begin{equation}\label{alphaEq}
  \alpha(k)= \frac{q_{k-1,0}}{1-\left(\boldsymbol{\sum} _{j=1}^{k} q_{k-1,j} A(k+1,j)\right)}.
\end{equation}
Also, function $A(\cdot,\cdot)$ is a recursive function defined as
\begin{equation}\label{eqq85}
\displaystyle A(i,k) =\alpha(i)\boldsymbol{\sum} _{j=1}^{i}p_{i-1,j}^{k} A(i+1,j) +p_{i-1,0}^{k}.
\end{equation}
where $A(n+1,k)=p_{n,0}^k$,
\end{theorem}
\begin{proof}
See Appendix \ref{MTTFJ2dn_proof}.
\end{proof}
\pgfplotscreateplotcyclelist{customlist}{
             {blue!80!yellow,mark=*},
            {red!60!yellow,mark=square,dashed},
            {green!50!black,mark=star},
            {black!80!blue,mark=diamond,dashed},
             {black!70!yellow,mark=otimes},
            {red!80!blue,mark=triangle,dashed}
}
\begin{figure*}[!t]
    \begin{subfigure}{.33\textwidth}
      \begin{tikzpicture}
          \begin{axis}[
              label style={inner sep=1pt},
tick label style={inner sep=1pt},
          height=0.8in,
          width=2.0in,
          scale only axis,
          xtick={1,2,3,4,5,6},
          xmajorgrids,
          ymin=0.0,
          ylabel=C-MTTF (Hours),
          xlabel=Group Size ($n$),
          ymajorgrids,
           label style={font=\fontsize{5}{5}\selectfont},
          tick label style={font=\fontsize{5}{5}\selectfont},
          grid=both,
          minor grid style={gray!20},
          major grid style={gray!50},
          mark size=1.3pt,
          minor tick num=2,
          legend columns=2
          legend style ={ at={(0.5,1.0)},
            anchor=north west,
            font=\fontsize{2}{3}\selectfont,
            fill=white,align=left},cycle list name=customlist]
        \addplot table [x=size,y=PMTTF10] {results/h1.txt};
        \addlegendentry{\fontsize{4}{5}\selectfont PMTTF};
        \addplot  table [x=size,y=SMTTF10] {results/h1.txt};
        \addlegendentry{\fontsize{4}{5}\selectfont SMTTF ($\lambda_u=6$)};
        \addplot  table [x=size,y=PMTTF15] {results/h1.txt};
        \addlegendentry{\fontsize{4}{5}\selectfont PMTTF};
        \addplot  table [x=size,y=SMTTF15] {results/h1.txt};
        \addlegendentry{\fontsize{4}{5}\selectfont SMTTF ($\lambda_u=4$)};
         \addplot  table [x=size,y=PMTTF20] {results/h1.txt};
        \addlegendentry{\fontsize{4}{5}\selectfont PMTTF};
        \addplot  table [x=size,y=SMTTF20] {results/h1.txt};
        \addlegendentry{\fontsize{4}{5}\selectfont SMTTF ($\lambda_u=3$)};
        \end{axis}
      \end{tikzpicture}
      \caption{\footnotesize $\lambda_z=1$.}
      \label{g1AR}
    \end{subfigure}%
    \begin{subfigure}{.33\textwidth}
      \begin{tikzpicture}
          \begin{axis}[
              label style={inner sep=1pt},
tick label style={inner sep=1pt},
           height=0.8in,
          width=2.0in,
          scale only axis,
          xtick={1,2,3,4,5,6},
          xmajorgrids,
          ymin=0.0,
          ylabel=C-MTTF (Hours),
          xlabel=Group Size ($n$),
          ymajorgrids,
           label style={font=\fontsize{5}{5}\selectfont},
          tick label style={font=\fontsize{5}{5}\selectfont},
          grid=both,
          minor grid style={gray!20},
          major grid style={gray!50},
          mark size=1.3pt,
          minor tick num=2,
          legend columns=2
          legend style ={ at={(0.5,1.0)},
            anchor=north west,
            font=\fontsize{2}{3}\selectfont,
            fill=white,align=left},cycle list name=customlist]
        \addplot table [x=size,y=PMTTF10] {results/h2.txt};
        \addlegendentry{\fontsize{4}{5}\selectfont PMTTF};
        \addplot  table [x=size,y=SMTTF10] {results/h2.txt};
        \addlegendentry{\fontsize{4}{5}\selectfont SMTTF ($\lambda_u=6$)};
        \addplot  table [x=size,y=PMTTF15] {results/h2.txt};
        \addlegendentry{\fontsize{4}{5}\selectfont PMTTF};
        \addplot  table [x=size,y=SMTTF15] {results/h2.txt};
        \addlegendentry{\fontsize{4}{5}\selectfont SMTTF ($\lambda_u=4$)};
         \addplot  table [x=size,y=PMTTF20] {results/h2.txt};
        \addlegendentry{\fontsize{4}{5}\selectfont PMTTF};
        \addplot  table [x=size,y=SMTTF20] {results/h2.txt};
        \addlegendentry{\fontsize{4}{5}\selectfont SMTTF ($\lambda_u=3$)};
        \end{axis}
      \end{tikzpicture}
      \caption{\footnotesize $\lambda_z=\frac{1}{2}$.}
    \label{g1AR}
    \end{subfigure}%
      \begin{subfigure}{.33\textwidth}
      \begin{center}
      \begin{tikzpicture}
          \begin{axis}[
              label style={inner sep=1pt},
tick label style={inner sep=1pt},
           height=0.8in,
          width=2.0in,
          scale only axis,
          xtick={1,2,3,4,5,6},
          xmajorgrids,
          ymin=0.0,
          ylabel=C-MTTF (Hours),
          xlabel=Group Size ($n$),
          ymajorgrids,
           label style={font=\fontsize{5}{5}\selectfont},
          tick label style={font=\fontsize{5}{5}\selectfont},
          grid=both,
          minor grid style={gray!20},
          major grid style={gray!50},
          mark size=1.3pt,
          minor tick num=2,
          legend columns=2
          legend style ={ at={(0.5,1.0)},
            anchor=north west,
            font=\fontsize{2}{3}\selectfont,
            fill=white,align=left},cycle list name=customlist]
        \addplot table [x=size,y=PMTTF10] {results/h3.txt};
        \addlegendentry{\fontsize{4}{5}\selectfont PMTTF};
        \addplot  table [x=size,y=SMTTF10] {results/h3.txt};
        \addlegendentry{\fontsize{4}{5}\selectfont SMTTF ($\lambda_u=6$)};
        \addplot  table [x=size,y=PMTTF15] {results/h3.txt};
        \addlegendentry{\fontsize{4}{5}\selectfont PMTTF};
        \addplot  table [x=size,y=SMTTF15] {results/h3.txt};
        \addlegendentry{\fontsize{4}{5}\selectfont SMTTF ($\lambda_u=4$)};
         \addplot  table [x=size,y=PMTTF20] {results/h3.txt};
        \addlegendentry{\fontsize{4}{5}\selectfont PMTTF};
        \addplot  table [x=size,y=SMTTF20] {results/h3.txt};
        \addlegendentry{\fontsize{4}{5}\selectfont SMTTF ($\lambda_u=3$)};
        \end{axis}
      \end{tikzpicture}
    \caption{\footnotesize $\lambda_z=\frac{1}{3}$.}
    \end{center}
    \end{subfigure}
    \caption{\small Exactness of \textit{C-MTTF}{\footnotesize$\big(\mathcal{A}|\mathcal{M}\big(\mathcal{P}(\mathcal{A}),t\big),\mathcal{R}(t){=}0,t{=}0\big)$} result in Theorem~\ref{mttfEXPTheorem}.}\label{MathCorrectness}
    \vspace{-5mm}
\end{figure*}
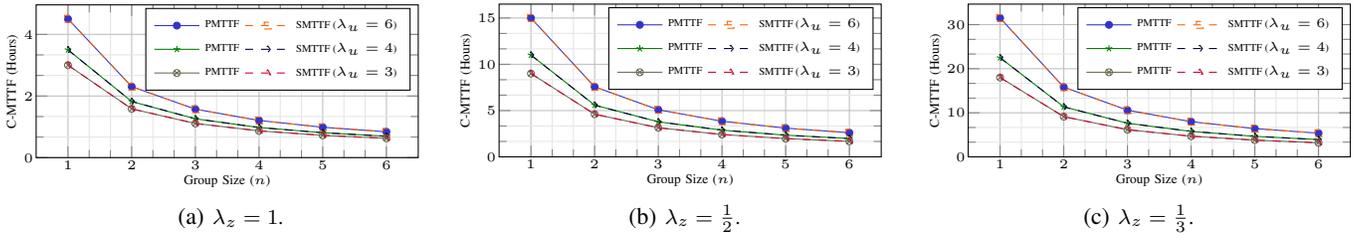
As can be seen from Theorem~\ref{MTTFJ2dn}, calculation of \textit{C-MTTF}{\small$\big(\mathcal{A}|\mathcal{M}\big(\mathcal{P}(\mathcal{A}),t\big),\mathcal{R}(t){=}0,t{=}0\big)$} requires computing $\mathbb{E}[\widehat{\mathrm{W}}_{i,j}]$, which is carried out as follows. The expected time that D-SMP resides in H-state $S_{i,j}\in\hat{\mathbf{S}}_i$ is equivalent to the expected value of {\small$\min\left\{\min_{1\le r \le 2n-i}\left\{S_{i,j}^{\mathsf{Soj}}(r)\right\},\min_{\substack{1\le h \le i}}\left\{S_{i,j}^{\mathsf{Rec}}(h)\right\}\right\}$}, capturing a situation where a vehicle departs the VC or a recruiter completes its recruitment. As a result, we have
\begin{equation}\label{ephi1_main}
\begin{aligned}
\mathbb{E}\left[\widehat{\mathrm{W}}_{i,j}\right] =\boldsymbol{\int}_{0}^{\infty } \textrm{Pr}\Big[\min\Big\{&\min_{1\le r \le 2n-i}\left\{S_{i,j}^{\mathsf{Soj}}(r)\right\},\\
&\min_{\substack{1\le h \le i}}\left\{S_{i,j}^{\mathsf{Rec}}(h)\right\}\Big\}>t\Big] dt.
\end{aligned}
\end{equation}
In next section, we compute {\small $\mathbb{E}\left[\widehat{\mathrm{W}}_{i,j}\right]$} (see \eqref{eqqexp2}) for a version of {\tt RP-VC$_n$}, where sojourn times and recruitment duration of vehicles follow exponential distribution, which is a common assumption \cite{29}. We further generalize our results to the case, where sojourn times and recruitment duration of vehicles follow arbitrary distributions in Appendix \ref{app:generalization}.

\vspace{-2mm}
\subsection{{\tt RP-VC$_n$} with Exponentially Distributed Sojourn Time and Recruitment Duration}\label{ERT}
We present the final theoretical results of this paper for the specific methodology {\tt RP-VC$_n$}${\sim}e$, which refers to {\tt RP-VC$_n$} when the recruitment duration follows exponential distribution with parameter {\small$\lambda_u$}. We obtain the following result, which demonstrates that under {\tt RP-VC$_n$}${\sim}e$, the residual sojourn and recruitment times of the vehicles follow exponential distribution.
\begin{corollary}\label{expj2nCorollary}
Let two random variables {\small$Z'_k$} and {\small$U'_r$}, for {\small$1\le k\le 2n-i$} and {\small$1\le r\le i$}, denote the residual sojourn and recruitment times of the vehicles in H-state {\small$S_{i,j}$}. Under {\tt RP-VC$_n$}${\sim}e$, if D-SMP is in H-state {\small$S_{i,j}$}, random variables {\small$Z'_{k}$} and {\small$U'_{r}$} follow exponential distribution with parameters {\small$\lambda_z$} and {\small$\lambda_u$}, respectively.
\end{corollary}
\begin{proof}
See Appendix \ref{expj2nCorollary_proof}.
\end{proof}
Corollary \ref{expj2nCorollary} states that if recruitment duration of recruiters follow exponential distribution, $\dot{exp}$, $\ddot{exp}$, $\varphi$, $\psi$, and $\gamma$ are all exponential. Considering Corollary \ref{expj2nCorollary}, in the following theorem, we obtain the closed-form expression of \textit{C-MTTF} of application processing under {\tt RP-VC$_n$}${\sim}e$.
\setlength{\belowdisplayskip}{3.8pt}
\begin{theorem}\label{mttfEXPTheorem}
Considering methodology {\tt RP-VC$_n$}${\sim}e$, \textit{C-MTTF}{\small$\big(\mathcal{A}|\mathcal{M}\big(\mathcal{P}(\mathcal{A}),t\big),\mathcal{R}(t)=0,t{=}0\big)=\mathbb{E}\left[\widehat{\mathrm{Q}}_{0,0}\right]$}, where
\begin{equation}\label{MTTFEXP}
\begin{aligned}
&\mathbb{E}\left[\widehat{\mathrm{Q}}_{0,0}\right]
{=} \mathbb{E}\left[\widehat{\mathrm{W}}_{0,0}\right] \left(1{+}\beta'(1)A'(2)\right){+}\beta'(n)\mathbb{E}\left[\widehat{\mathrm{W}}_{n,0}\right]\\
 &+\sum_{i=1}^{n-1}\mathbb{E}\left[\widehat{\mathrm{W}}_{i,0}\right] \Big((i{+}1)\beta'(i{+}1) A'(i{+}2){+}\beta'(i)\Big).
\end{aligned}
\end{equation}
In \eqref{MTTFEXP}, $\beta'(i)$ is given by
\begin{equation}\label{eqq52}
  \beta'(i)=\boldsymbol{\prod} _{k=1}^{i}\alpha'(k),
\end{equation}
where
\begin{equation}\label{eq96}
\begin{aligned}
  \alpha'(k)=\frac{q_{k-1,0}}{1-k\times q_{k-1,0}\times A'(k+1)}.
\end{aligned}
\end{equation}
Also, $A'(\cdot)$ is calculated by
\begin{equation}\label{eq99}
\begin{aligned}
A'(i)=\frac{p_{i-1,0}^{1}}{1-i\times q_{i-1,0}A'(i+1)},
\end{aligned}
\end{equation}
and $A'(n+1)=p_{n,0}^1$.
\end{theorem}
\begin{proof}
See Appendix \ref{mttfEXPTheorem_proof}.
\end{proof}
\noindent To obtain the final result of (\ref{MTTFEXP}), we calculate {\small$\mathbb{E}\left[\widehat{\mathrm{W}}_{i,0}\right]$} below.
\vspace{-2mm}
\begin{corollary}\label{expected_sojourn_time}
Considering methodology {\tt RP-VC$_n$}${\sim}e$, the expected sojourn time in state {\small$S_{i,j}$} is given by
\begin{equation}\label{eqqexp2}
   \mathbb{E}\left[\widehat{\mathrm{W}}_{i,0}\right] =\frac{1}{( 2n-i) \lambda_{z} +i\lambda _{u}}.
\end{equation}
\end{corollary}
\begin{proof}
See Appendix \ref{expected_sojourn_time_proof}.
\end{proof}
We next aim to compute {\small$p_{i,0}^{1}$}, utilized to obtain (\ref{eq99}), (\ref{eqq52}), and (\ref{eq96}). $q_{i-1,0}$ in (\ref{eq99}) and (\ref{eq96}) is computed based on (\ref{qij}).
\begin{corollary}\label{transitionProbCorollary}
Considering methodology {\tt RP-VC$_n$}${\sim}e$, the probability of transition from {\small$S_{i,0}$} to {\small$S_{i-1,0}$} is given by
\begin{equation}\label{eqqexp1}
 p_{i,0}^{1}=\frac{\lambda _{u}}{( 2n-i) \lambda_{z} +i\lambda _{u}}.
\end{equation}
\end{corollary}
\begin{proof}
See Appendix \ref{transitionProbCorollary_proof}.
\end{proof}
Replacing (\ref{eqqexp2}) and (\ref{eqqexp1}) back in Theorem~\ref{mttfEXPTheorem} provides the value of \textit{C-MTTF}{\small$\big(\mathcal{A}|\mathcal{M}\big(\mathcal{P}(\mathcal{A}),t\big),\mathcal{R}(t)=0,t{=}0\big)$} and concludes the mathematical derivations of this paper. For a comprehensive case study of {\tt RP-VC$_n$} refer to Appendix~\ref{app:case_study}.

\vspace{-5mm}
\subsection{Generalizability of $\langle e\rangle$-algebra and DT}
Our mathematical derivations can be utilized in similar cloud-based computing paradigms  for the purpose of reliability analysis and fault-tolerant design. For example, authors in \cite{21} have presented a strategy, referred to by \textit{delay-sensitive and reliable (DSR) placement}, aiming at improving \textit{application placement availability} in cloud-based systems. DSR aims to find an optimal sub-tasks placement $\mathcal{P}$ (i.e., placing virtual machines that process the application sub-tasks on the cloud servers) to minimize the number of utilized servers for processing application,  while satisfying the application availability (i.e., decreasing the probability of failure) and delay constraints. To satisfy availability, DSR places two/three replications of each sub-task on two/three different servers. However, DSR considers the availability of each server and utilizes the inclusion-exclusion principle to compute the availability of the application, which has been proved to have exponential computation complexity. A similar problem can be found in \cite{55,56}.
By exploiting $\langle e\rangle$-algebra and DT, one can model application placement in cloud-based systems as a D-SMP and obtain MTTF of application processing, which can be further used as an optimization metric to design a fault-tolerant system.

 \vspace{-3mm}
\section{Numerical Evaluation}\label{simulation}
\noindent In this section, we carry out an extensive numerical analysis to verify the exactness of our mathematical investigations (Sec.~\ref{EOMA}) and prove the efficiency of {\tt RP-VC$_n$}, as compared to the current art method (Sec.~\ref{ERD}). The default parameters used for modeling a VC are borrowed from \cite{24,29}.
\vspace{-5mm}
\subsection{Exactness of the Mathematical Analysis}\label{EOMA}
We verify the result of Theorem~\ref{mttfEXPTheorem} in a scenario where the sojourn times (i.e., $Z_x$) of the vehicles are i.i.d exponential random variables with parameter $\lambda_z$, the value of which is chosen between $\{1,\frac{1}{2},\frac{1}{3}\}$ (i.e., the average sojourn time of each vehicle is one, two, or three hours). The duration of the recruitment operations are also chosen to be i.i.d exponential random variables with parameter $\lambda_u\in \{6,4,3\}$ (i.e., the average recruitment times are 10, 15, and 20 minutes). The figures are obtained via Monte-Carlo method, where each result is an average of $10^6$ independent runs.\par
Fig.~\ref{MathCorrectness} illustrates how partitioning a large application into $n$ smaller dependent groups (ranging from $n=1$ to $n=6$) can affect \textit{C-MTTF} of {\tt RP-VC$_n$}${\sim}e$. The figure demonstrates a match between the simulation results, shown by simulated MTTF ($SMTTF$) and that of the mathematical model obtained via Theorem~\ref{mttfEXPTheorem}, reflected by predicted MTTF ($PMTTF$). Also, as can be seen from Fig.~\ref{MathCorrectness}, for a fixed value of $\lambda_u$, as the number of groups increases, \textit{C-MTTF} declines because more vehicles are needed to process application $\mathcal{A}$, and thus the chance one of them departing the VC increases.\par
In real situations, due to computing resource deficiency of the vehicles, a single vehicle may not be able to meet the computing requirements of a CI-App. Consequently, if {\tt J}$_2$ could find two vehicles with enough computing resources to start processing a CI-App, when one of the vehicles allocated to the CI-App departs the VC, the other one may need to wait for a long time until a vehicle with enough computing resources arrives at the VC (demonstrated numerically in the Sec. \ref{ERD}). However, {\tt RP-VC$_n$} breaks down a CI-App into smaller groups, through partitioning. Consequently, more vehicles can satisfy the requirements of the CI-App groups. Besides, the time of transferring images/replicas of the groups among vehicles will be decreased. Accordingly, the value of $\lambda_u$ is not fixed in real situations. Fig.~\ref{LU4} indicates that how changing the value of $\lambda_u$ can alter the value of \textit{C-MTTF}. For this experiment, we define a function $\lambda'_u(\alpha)=\lambda_u\times\alpha$, where $\alpha$ is a decline factor. We then calculate \textit{C-MTTF} for $\lambda'_u(\alpha)$ (i.e., recruitment duration of the vehicles follow an exponential distribution with parameter $\lambda'_u(\alpha)$). As can be seen from the figure, as the value of $\alpha$ decreases (i.e., the value of $\lambda'_u(\alpha)$ declines) C-MTTF of {\tt RP-VC$_n$} increases since a smaller $\alpha$ implies a faster vehicle recruitment. Below, we will compare the efficiency of {\tt J}$_2$ and {\tt RP-VC$_n$} in processing CI-App.

\pgfplotscreateplotcyclelist{customlist3}{
             {red!80!yellow,mark=*,dashed},
            {blue!80!yellow,mark=square},
            {red!50!yellow,mark=star},
            {red!50!blue,mark=diamond},
}
\begin{figure}[t!]
\centering
\includegraphics[width=0.7\linewidth,trim=3 3 3 3,clip]{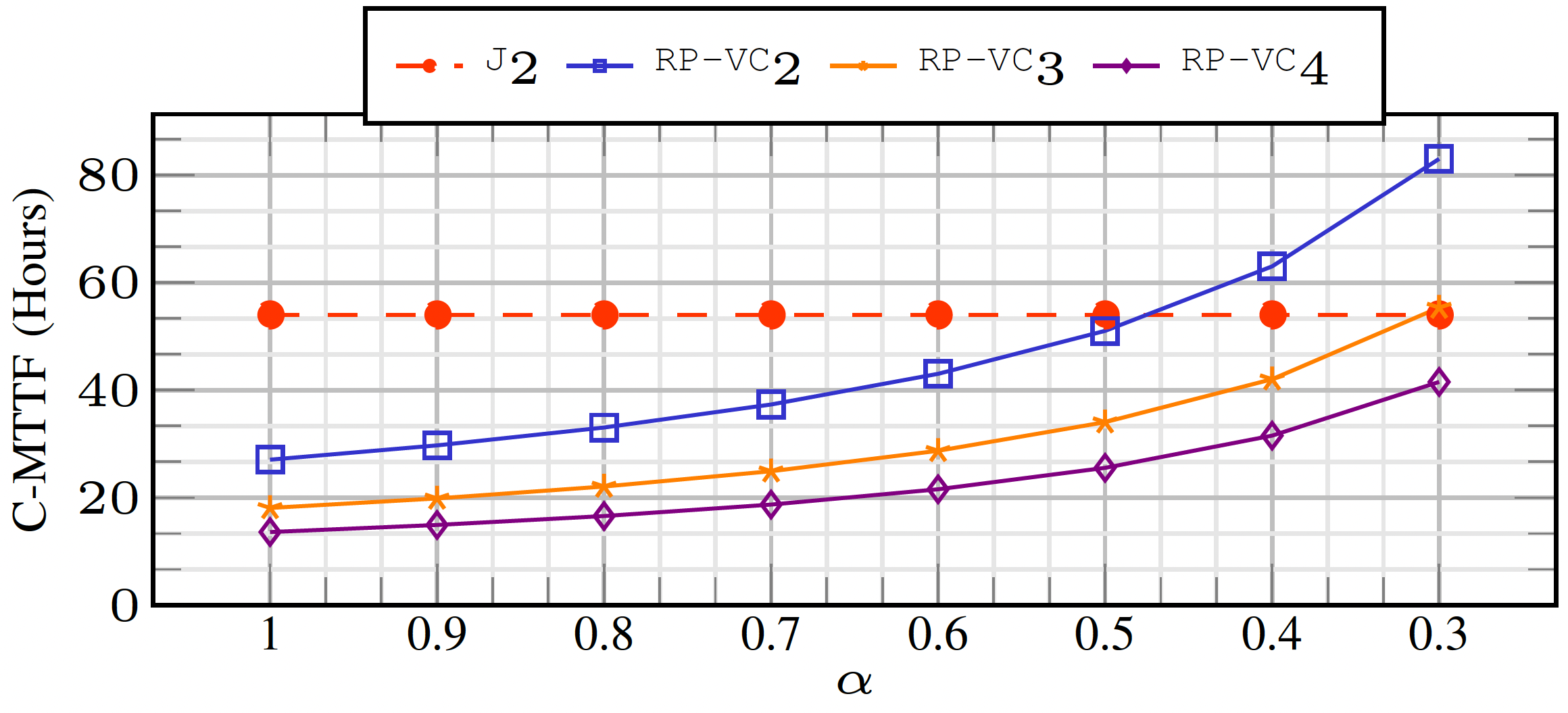}
 \caption{\textit{C-MTTF}{\footnotesize$\big(\mathcal{A}|\mathcal{M}\big(\mathcal{P}(\mathcal{A}),t\big),\mathcal{R}(t){=}0,t{=}0\big)$} for {\footnotesize$\lambda'_u(\alpha){=}\lambda_u{\times}\alpha$}.}\label{LU4}
\end{figure}

\vspace{-2mm}
\subsection{Efficiency of {\tt RP-VC$_n$}}\label{ERD}
In the following, we show how partitioning an application into smaller groups can improve the system performance. We first compare the performance of {\tt RP-VC$_n$} with {\tt J}$_2$ \cite{24}, where both methods consider the execution of applications over two vehicles, while {\tt J} $_2$ does not consider application partitioning. Afterward, {\tt RP-VC$_n$} and {\tt J}$_2$ are compared with their simple versions with no additional replicas referred to by {\tt No-R}. These versions operate the same as the original strategies with the difference that one vehicle is allocated to each group for {\tt RP-VC$_n$} and each application for {\tt J}$_2$.
\subsubsection{Simulation Setup}
To verify the efficiency of {\tt RP-VC$_n$}, we consider the following simulation environment. Vehicles enter the VC based on a Poisson distribution with parameter $\lambda_{ve}$. Users send request for processing CI-App to the system based on a Poisson distribution with parameter $\lambda_{app}$. For simplicity, we use a number, drawn uniformly from interval $[l_a, h_a]$, to refer to the resource requirement (i.e., computing, memory, and network resources) of an application. Similarly, the resource capacity (i.e., computing, memory, and network resources) of each vehicle is represented through a number, drawn uniformly from interval $[l_v, h_v]$. The duration of each application follows exponential distribution with parameter $\lambda_d$. The vehicle's sojourn time follows exponential distribution with parameter $\lambda_z$. Minimum recruitment duration follows exponential distribution with parameter $\lambda_u$.
\textit{Default values of the aforementioned parameters are presented in Table~\ref{simParameters}}.\par
We compare {\tt RP-VC$_n$} with {\tt J}$_2$ and their {\tt No-R} versions in terms of the following quality of service (QoS) metrics:
\begin{equation}\label{AR}
    \textrm{Acceptance Rate (AR)}=\frac{\#\{\textrm{Accepted Applications}\}}{\#\{\textrm{Applications}\}},
\end{equation}
\begin{equation}\label{SR}
    \textrm{Success Rate (SR)}=\frac{\#\{\textrm{Successful Applications}\}}{\#\{\textrm{Applications}\}},
\end{equation}
 \noindent where $\#\{.\}$ reads "number of". The acceptance rate (AR) in (\ref{AR}) captures the number of accepted applications divided by the total number of applications. Similarly, (\ref{SR}) refers to the success rate (SR), which is the number of successfully executed applications divided by the total number of applications.
 \setlength{\tabcolsep}{15pt}
\renewcommand{\arraystretch}{1}
\begin{table}
  \centering
   \caption{\small Default values of simulation parameters}
  \begin{tabular}{|c|c|c|c|}
  \hline
Parameter & Value & Parameter & Value\\
\hline \hline
$\lambda_{ve}$ & 20 & $\lambda_{d}$ & 1 \\
\hline
$\lambda_{app}$ & 2 & $\lambda_{z}$ & 1 \\
\hline
$[l_a,h_a]$ & $[20,30]$ & $\lambda_{u}$ & 6 \\
\hline
$[l_v,h_v]$ & $[20,26]$ & - & - \\
\hline
\end{tabular}
 \label{simParameters}
\end{table}
\pgfplotscreateplotcyclelist{customlist2}{
             {red!80!yellow,mark=*},
            {blue!80!yellow,mark=square},
            {red!50!yellow,mark=star},
            {red!50!blue,mark=diamond},
             {black!70!yellow,mark=otimes},
            {red!80!blue,mark=triangle}
}
\begin{figure}[t!]
\centering
\includegraphics[width=\linewidth,trim=3 3 3 3,clip]{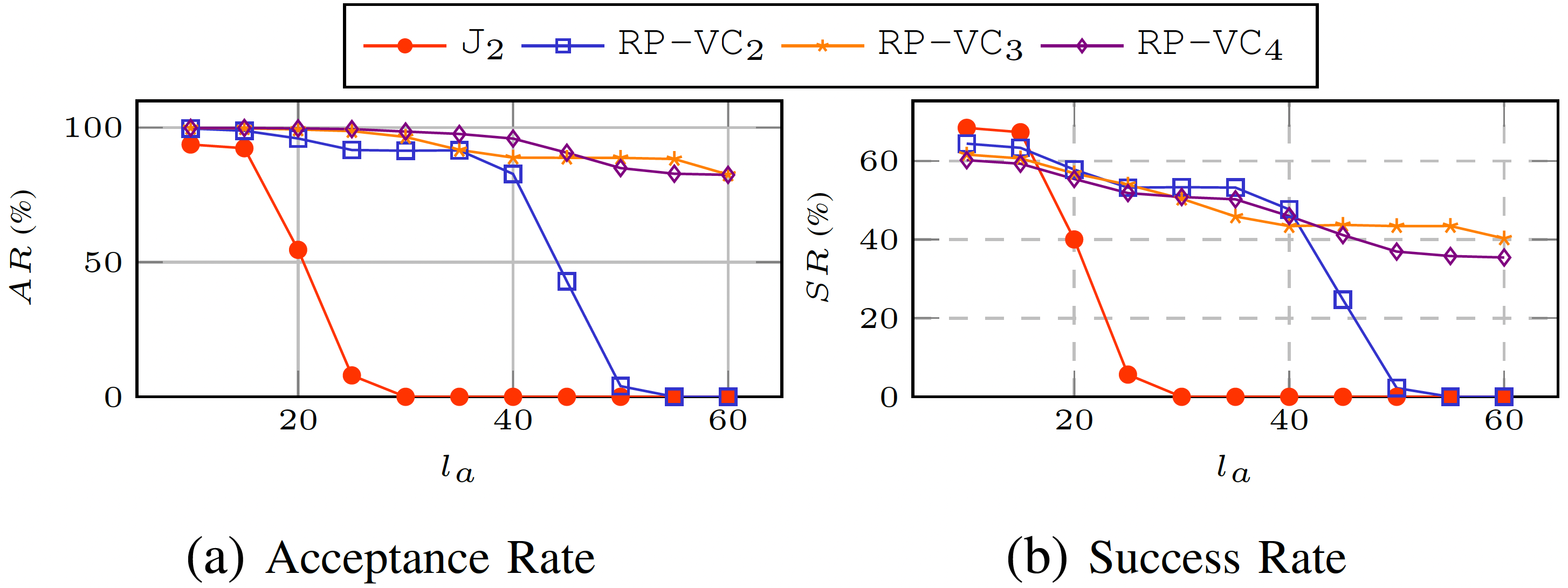}
 \caption{\small Comparison of {\tt RP-VC$_n$} with {\tt J}$_2$ for different application sizes. The $x$-axis shows the lower bound $l_a$ of the application size and the upper bound is calculated by $h_a=l_a+10$.}\label{applicationSize}
 \end{figure}
\subsubsection{Comparison of {\tt RP-VC$_n$} with {\tt J}$_2$ for different application sizes}
Fig.~\ref{applicationSize} illustrates the behavior of {\tt J}$_2$ and {\tt RP-VC$_n$}, for $n{\in}\{2,3,4\}$. It can be observed that for light applications, i.e., $l_a{=}10$ and $l_a{=}15$, {\tt J}$_2$ has a reasonable $AR$ (Fig.~\ref{applicationSize}(a)) with slightly lower value than that of {\tt RP-VC$_n$}, which is $100\%$. However, there is a dramatic decrease in the $AR$ of {\tt J}$_2$ starting from $l_a{=}20$ until $l_a{=}30$, where the $AR$ reaches almost $0\%$. {\tt RP-VC$_n$}, for $n{\in}\{2,3\}$, maintains its high AR for $l_a{=}10$ to $l_a{=}40$ with a moderate decline from $100\%$ to (on average) $85\%$, while {\tt RP-VC$_4$} has an $AR$ close to $100\%$ in this range. From $l_a{>}40$, the $AR$ of {\tt RP-VC$_2$} starts dropping and reaches $0\%$ for $l_a>55$. {\tt RP-VC$_3$} and {\tt RP-VC$_4$} have almost equal $AR$ and (on average) they reach the AR of $85\%$ for $l_a{=}60$.\par
We next study the SR metric in Fig.~\ref{applicationSize}(b). As can be observed from the figure, from $l_a=10$ to $l_a=15$, {\tt J}$_2$ has $80\%$ $SR$, on average. However, $SR$ for {\tt J}$_2$ drops sharply to reach the value of zero from $l_a=15$ to $l_a=30$ and stays zero for $l_a>30$ since acceptance rate is zero in this range. Likewise, for {\tt RP-VC$_2$}, $SR$ experiences a moderate decline from the value of $65\%$ at $l_a=10$ to reach around $50\%$ at $l_a=40$ and drops quickly to zero from $l_a=40$ to $l_a=50$ and remains flat onward. Finally, for {\tt RP-VC$_3$} and {\tt RP-VC$_4$}, $SR$ gradually declines from $65\%$ at $l_a=10$ to reside around $40\%$ at $l_a=60$.
\subsubsection{Comparison of {\tt RP-VC$_n$} with {\tt J}$_2$ for different values of $\lambda_{d}$}
In this section, we study $AR$ (Fig.~\ref{lambdaD}(a)) and $SR$ (Fig.~\ref{lambdaD}(b)) for different application's execution time ranging from two hours to 17 minutes (i.e., $\lambda_d \in\{0.5,1,1.5,2,2.5,3,3.5\}$). The comparison is conducted between {\tt J}$_2$, {\tt RP-VC$_2$}, {\tt RP-VC$_3$}, and {\tt RP-VC$_4$}. Overall, the acceptance rate of all the versions of {\tt RP-VC$_n$} is considerably higher than that of {\tt J}$_2$. This is because {\tt J}$_2$ tries to find two vehicles to fit the entire application. As can be seen from the plots, $AR$ of {\tt J}$_2$ increases from around $25\%$ at $\lambda_d=0.5$ to just below $60\%$ at $\lambda_d=1.5$ and remains flat onward. However, considering {\tt RP-VC$_n$}, the acceptance rate is almost $100\%$ for all the specified values of $\lambda_d$ except for $\lambda_d\in \{0.5,1\}$, where, due to the long execution time of applications for $\lambda_d=0.5$ and $\lambda_d=1$, the value of $AR$ of {\tt RP-VC$_2$} is around $85\%$ and $95\%$, respectively.\par
\begin{figure}[t!]
\centering
\includegraphics[width=\linewidth,trim=3 3 3 3,clip]{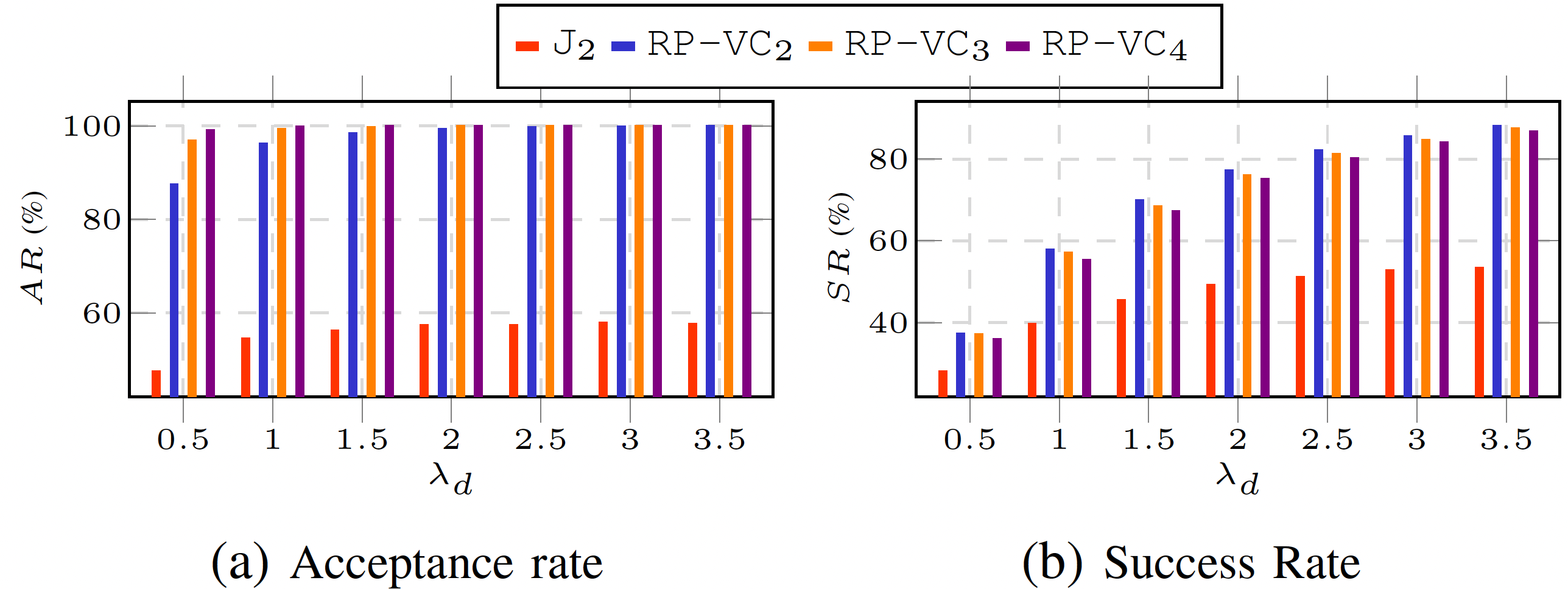}
 \caption{\small Comparison of {\tt RP-VC$_n$} with {\tt J}$_2$\cite{24} for different $\lambda_d$.}
 \label{lambdaD}
 \end{figure}

Considering the SR metric (Fig.~\ref{lambdaD}(b)), it can be seen that the value of SR tends to increase gradually for all strategies. However, upon having very long execution time (e.g., $\lambda_d=0.5$), although the acceptance rate of {\tt RP-VC$_n$} is high (almost $100\%$), the $SR$ is around $40\%$ because the longer the execution of an application takes, the more likely it is that allocated vehicles to the application depart from the VC before finishing its processing. It can be verified from the figure that, for all the values of execution time, $SR$ of {\tt J}$_2$ is lower than all versions of {\tt RP-VC$_n$} (i.e., almost $30\%$ lower than {\tt RP-VC$_n$}).
\subsubsection{Comparison of {\tt RP-VC$_n$} and {\tt J}$_2$ with {\tt NO-R} versions of them}
Fig.~\ref{NOR} presents the comparison of {\tt RP-VC$_n$} and {\tt J}$_2$ with their corresponding {\tt NO-R} versions. According to the histogram of the acceptance rate (Fig.~\ref{NOR}(a)), the {\tt NO-R} versions of strategies {\tt J}$_2$ and {\tt RP-VC$_2$} have slightly higher $AR$ since in {\tt NO-R} versions only one vehicle is allocated to each job (i.e., application for {\tt J}$_2$ and groups for {\tt RP-VC$_2$}) and more vehicles are free. However, for {\tt RP-VC$_3$} and {\tt RP-VC$_4$}, the acceptance rate of the original strategies and their {\tt NO-R} versions are equal to $100\%$. Nevertheless, as can be seen from Fig.~\ref{NOR}(b), all the original strategies have almost $20\%$ greater $SR$ compared to their corresponding {\tt NO-R} versions. This signifies the importance of redundancy-based methodologies in improving the reliability of CI-App processing over VC.

\vspace{-3mm}
\section{Conclusion and Future Work}\label{conclusion}
\noindent We introduced {\tt RP-VC$_n$}, which is a general methodology for reliable processing of
CI-Apps consisting of several dependent sub-tasks in a semi-dynamic VC.
{\tt RP-VC$_n$} is a unified framework for reliable processing of both DAG- and UG-structured CI-Apps through partitioning them into several groups dispersed across vehicles to battle the resource deficiency of a single vehicle
and enhance the \textit{C-MTTF}.
We introduced stochastic event systems (SeS) to specify a class of stochastic systems and demonstrated that a semi-dynamic VC under {\tt RP-VC$_n$} is an SeS. We then modeled {\tt RP-VC} via a non-trivial semi-Markov process ($\beta$-SMP) and characterized the dynamics of $\beta$-SMP through developing $\langle e\rangle$-algebra. Using $\langle e\rangle$-algebra, we relaxed the innate complexities of $\beta$-SMP via proposing \textit{decomposition theorem} and transforming $\beta$-SMP into a decomposed  semi-Markov process (i.e., D-SMP). Relying on D-SMP, we then calculated the \textit{C-MTTF} of our methodology. In addition to multiple future work directions discussed in the paper, we intend to extend the analysis of {\tt RP-VC$_n$} to scenarios with more replications of groups dispersed across vehicles. Also, efficient recruitment of new vehicles (e.g., considering outage probability of links) is an interesting future work.

 \pgfplotscreateplotcyclelist{customlist3}{
             {red!80!yellow,mark=*,dashed},
            {blue!80!yellow,mark=square},
            {red!50!yellow,mark=star},
            {red!50!blue,mark=diamond},
}
\begin{figure}[t!]
\centering
\includegraphics[width=\linewidth,trim=3 3 3 3,clip]{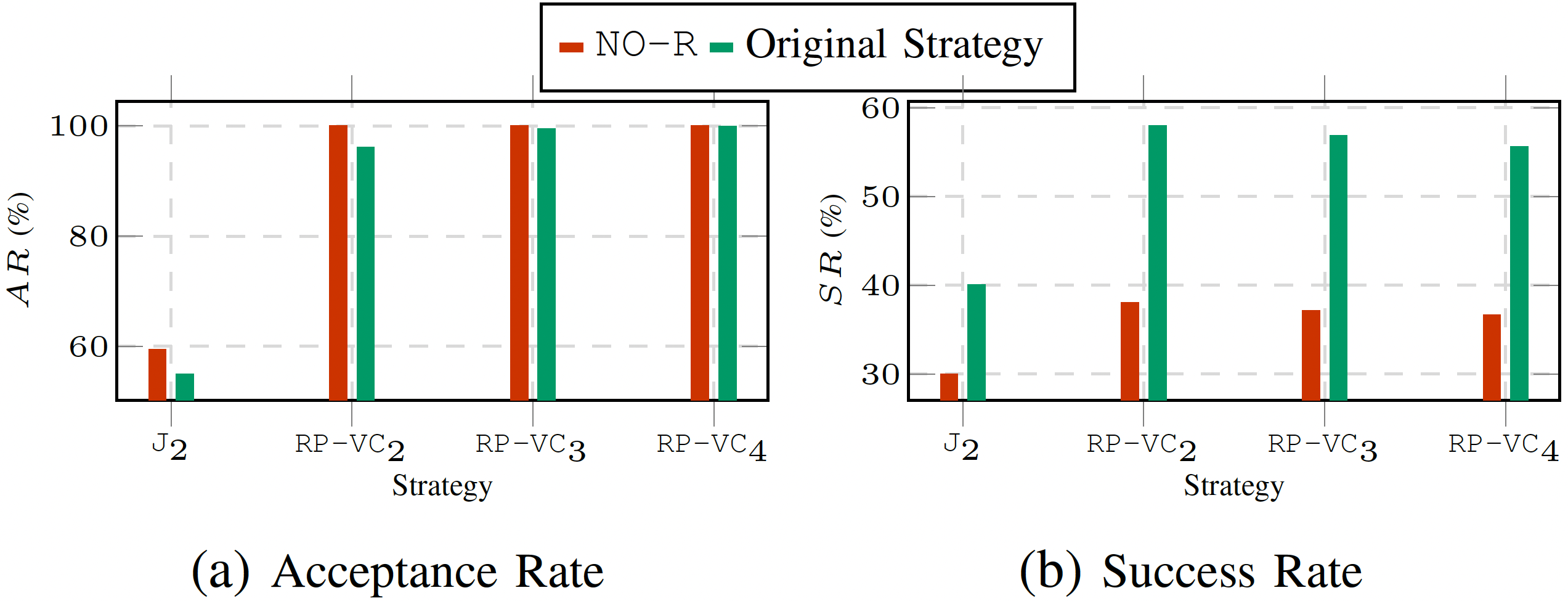}
 \caption{\small Comparison of {\tt RP-VC$_n$} and {\tt J}$_2$ with {\tt NO-R} version of them.}\label{NOR}
 \end{figure}
\vspace{-2mm}
\bibliographystyle{IEEEtran}

\bibliography{j2nRef}

\begin{thebibliography}{10}
\providecommand{\url}[1]{#1}
\csname url@samestyle\endcsname
\providecommand{\newblock}{\relax}
\providecommand{\bibinfo}[2]{#2}
\providecommand{\BIBentrySTDinterwordspacing}{\spaceskip=0pt\relax}
\providecommand{\BIBentryALTinterwordstretchfactor}{4}
\providecommand{\BIBentryALTinterwordspacing}{\spaceskip=\fontdimen2\font plus
\BIBentryALTinterwordstretchfactor\fontdimen3\font minus
  \fontdimen4\font\relax}
\providecommand{\BIBforeignlanguage}[2]{{%
\expandafter\ifx\csname l@#1\endcsname\relax
\typeout{** WARNING: IEEEtran.bst: No hyphenation pattern has been}%
\typeout{** loaded for the language `#1'. Using the pattern for}%
\typeout{** the default language instead.}%
\else
\language=\csname l@#1\endcsname
\fi
#2}}
\providecommand{\BIBdecl}{\relax}
\BIBdecl

\bibitem{42}
Z.~Cheng, M.~Min, M.~Liwang, L.~Huang, and Z.~Gao, ``Multiagent {DDPG}-based
  joint task partitioning and power control in fog computing networks,''
  \emph{IEEE Internet Things J.}, vol.~9, no.~1, pp. 104--116, 2022.

\bibitem{43}
M.~LiWang, S.~Dai, Z.~Gao, X.~Du, M.~Guizani, and H.~Dai, ``A computation
  offloading incentive mechanism with delay and cost constraints under 5{G}
  satellite-ground iov architecture,'' \emph{IEEE Wireless Commu.}, vol.~26,
  no.~4, pp. 124--132, 2019.

\bibitem{45}
X.~Chen and G.~Liu, ``Energy-efficient task offloading and resource allocation
  via deep reinforcement learning for augmented reality in mobile edge
  networks,'' \emph{IEEE Internet Things J.}, vol.~8, no.~13, pp.
  10\,843--10\,856, 2021.

\bibitem{4}
P.~Mach and Z.~Becvar, ``Mobile edge computing: A survey on architecture and
  computation offloading,'' \emph{IEEE Commun. Surveys Tuts.}, vol.~19, no.~3,
  pp. 1628--1656, 2017.

\bibitem{5}
C.~Jiang, X.~Cheng, H.~Gao, X.~Zhou, and J.~Wan, ``Toward computation
  offloading in edge computing: A survey,'' \emph{IEEE Access}, vol.~7, pp.
  131\,543--131\,558, 2019.

\bibitem{8}
M.~Chiang and T.~Zhang, ``Fog and iot: An overview of research opportunities,''
  \emph{IEEE Internet Things J.}, vol.~3, no.~6, pp. 854--864, 2016.

\bibitem{7}
E.~Lee, E.-K. Lee, M.~Gerla, and S.~Y. Oh, ``Vehicular cloud networking:
  architecture and design principles,'' \emph{IEEE Commun. Mag.}, vol.~52,
  no.~2, pp. 148--155, 2014.

\bibitem{12}
M.~Eltoweissy, S.~Olariu, and M.~Younis, ``Towards autonomous vehicular
  clouds,'' in \emph{International Conf. on Ad hoc networks}, 2010, pp. 1--16.

\bibitem{14}
S.~Bitam, A.~Mellouk, and S.~Zeadally, ``{VANET}-cloud: a generic cloud
  computing model for vehicular ad hoc networks,'' \emph{IEEE Wireless Commu.},
  vol.~22, no.~1, pp. 96--102, 2015.

\bibitem{16}
A.~Boukerche and R.~E. {De Grande}, ``Vehicular cloud computing: Architectures,
  applications, and mobility,'' \emph{Comput. Netw.}, vol. 135, pp. 171--189,
  2018.

\bibitem{46}
H.~Liao, X.~Li, D.~Guo, W.~Kang, and J.~Li, ``Dependency-aware application
  assigning and scheduling in edge computing,'' \emph{IEEE Internet Things J.},
  vol.~9, no.~6, pp. 4451--4463, 2022.

\bibitem{17}
S.~Olariu, ``A survey of vehicular cloud research: Trends, applications and
  challenges,'' \emph{IEEE Trans. Intell. Transp. Syst.}, vol.~21, no.~6, pp.
  2648--2663, 2020.

\bibitem{36}
R.~I. Meneguette, A.~Boukerche, and A.~H.~M. Pimenta, ``{AVARAC}: An
  availability-based resource allocation scheme for vehicular cloud,''
  \emph{IEEE Trans. Intell. Transp. Syst.}, vol.~20, no.~10, pp. 3688--3699,
  2019.

\bibitem{24}
P.~Ghazizadeh, R.~Florin, A.~G. Zadeh, and S.~Olariu, ``Reasoning about mean
  time to failure in vehicular clouds,'' \emph{IEEE Trans. Intell. Transp.
  Syst.}, vol.~17, no.~3, pp. 751--761, 2016.

\bibitem{21}
S.~Yang, P.~Wieder, R.~Yahyapour, S.~Trajanovski, and X.~Fu, ``Reliable virtual
  machine placement and routing in clouds,'' \emph{IEEE Trans. Parallel
  Distrib. Syst.}, vol.~28, no.~10, pp. 2965--2978, 2017.

\bibitem{26}
\BIBentryALTinterwordspacing
IBM. (1999) Microservices. [Online]. Available:
  \url{https://www.ibm.com/cloud/learn/microservices}
\BIBentrySTDinterwordspacing

\bibitem{29}
R.~Florin, A.~Ghazizadeh, P.~Ghazizadeh, S.~Olariu, and D.~C. Marinescu,
  ``Enhancing reliability and availability through redundancy in vehicular
  clouds,'' \emph{IEEE Trans. Cloud Comput.}, vol.~9, no.~3, pp. 1061--1074,
  2019.

\bibitem{30}
R.~Florin, A.~Ghazi~Zadeh, P.~Ghazizadeh, and S.~Olariu, ``Towards
  approximating the mean time to failure in vehicular clouds,'' \emph{IEEE
  Trans. Intell. Transp. Syst.}, vol.~19, no.~7, pp. 2045--2054, 2018.

\bibitem{31}
R.~Florin, P.~Ghazizadeh, A.~Ghazi~Zadeh, S.~El-Tawab, and S.~Olariu,
  ``Reasoning about job completion time in vehicular clouds,'' \emph{IEEE
  Trans. Intell. Transp. Syst.}, vol.~18, no.~7, pp. 1762--1771, 2017.

\bibitem{32}
R.~Florin, P.~Ghazizadeh, A.~G. Zadeh, R.~Mukkamala, and S.~Olariu, ``A tight
  estimate of job completion time in vehicular clouds,'' \emph{IEEE Trans.
  Cloud Comput.}, vol.~8, no.~3, pp. 721--734, 2020.

\bibitem{47}
D.~P. Palomar and M.~Chiang, ``A tutorial on decomposition methods for network
  utility maximization,'' \emph{IEEE J. Sel. Areas Commun.}, vol.~24, no.~8,
  pp. 1439--1451, 2006.

\bibitem{48}
M.~Chiang, S.~H. Low, A.~R. Calderbank, and J.~C. Doyle, ``Layering as
  optimization decomposition: A mathematical theory of network architectures,''
  \emph{Proc. IEEE}, vol.~95, no.~1, pp. 255--312, 2007.

\bibitem{28}
S.~Arif, S.~Olariu, J.~Wang, G.~Yan, W.~Yang, and I.~Khalil, ``Datacenter at
  the airport: Reasoning about time-dependent parking lot occupancy,''
  \emph{IEEE Trans. Parallel Distrib. Syst.}, vol.~23, no.~11, pp. 2067--2080,
  2012.

\bibitem{1}
F.~Sun, F.~Hou, N.~Cheng, M.~Wang, H.~Zhou, L.~Gui, and X.~Shen, ``Cooperative
  task scheduling for computation offloading in vehicular cloud,'' \emph{IEEE
  Trans. Veh. Tech.}, vol.~67, no.~11, pp. 11\,049--11\,061, 2018.

\bibitem{33}
R.~I. Meneguette and A.~Boukerche, ``An efficient green-aware architecture for
  virtual machine migration in sustainable vehicular clouds,'' \emph{IEEE
  Trans. Sustain. Comput.}, vol.~5, no.~1, pp. 25--36, 2020.

\bibitem{34}
K.~Zheng, H.~Meng, P.~Chatzimisios, L.~Lei, and X.~Shen, ``An {SMDP}-based
  resource allocation in vehicular cloud computing systems,'' \emph{IEEE Trans.
  Ind. Electron.}, vol.~62, no.~12, pp. 7920--7928, 2015.

\bibitem{35}
C.-C. Lin, D.-J. Deng, and C.-C. Yao, ``Resource allocation in vehicular cloud
  computing systems with heterogeneous vehicles and roadside units,''
  \emph{IEEE Internet Things J.}, vol.~5, no.~5, pp. 3692--3700, 2018.

\bibitem{44}
L.~Zhao, K.~Yang, Z.~Tan, H.~Song, A.~Al-Dubai, A.~Y. Zomaya, and X.~Li,
  ``Vehicular computation offloading for industrial mobile edge computing,''
  \emph{IEEE Trans. Ind. Informat.}, vol.~17, no.~11, pp. 7871--7881, 2021.

\bibitem{39}
T.~K. Refaat, B.~Kantarci, and H.~T. Mouftah, ``Virtual machine migration and
  management for vehicular clouds,'' \emph{Veh. Commun.}, vol.~4, pp. 47--56,
  2016.

\bibitem{40}
P.~Ghazizadeh, ``Resource allocation in vehicular cloud computing,'' Ph.D.
  dissertation, Old Dominion University, July 2014.

\bibitem{49}
C.~Li, S.~Wang, X.~Huang, X.~Li, R.~Yu, and F.~Zhao, ``Parked vehicular
  computing for energy-efficient internet of vehicles: A contract theoretic
  approach,'' \emph{IEEE Internet Things J.}, vol.~6, no.~4, pp. 6079--6088,
  2019.

\bibitem{50}
Y.~Zhang, C.-Y. Wang, and H.-Y. Wei, ``Parking reservation auction for parked
  vehicle assistance in vehicular fog computing,'' \emph{IEEE Trans. Veh.
  Technol.}, vol.~68, no.~4, pp. 3126--3139, 2019.

\bibitem{52}
M.~Feng, M.~Krunz, and W.~Zhang, ``Joint task partitioning and user association
  for latency minimization in mobile edge computing networks,'' \emph{IEEE
  Trans. Veh. Technol.}, vol.~70, no.~8, pp. 8108--8121, 2021.

\bibitem{41}
Z.~Gao, M.~Liwang, S.~Hosseinalipour, H.~Dai, and X.~Wang, ``A truthful auction
  for graph job allocation in vehicular cloud-assisted networks,'' \emph{IEEE
  Trans. Mobile Comput.}, pp. 1--1, 2021.

\bibitem{8847383}
S.~Hosseinalipour, A.~Nayak, and H.~Dai, ``Power-aware allocation of graph jobs
  in geo-distributed cloud networks,'' \emph{IEEE Trans. Parallel Distrib.
  Syst.}, vol.~31, no.~4, pp. 749--765, 2020.

\bibitem{19}
M.~A. Mukwevho and T.~Celik, ``Toward a smart cloud: A review of
  fault-tolerance methods in cloud systems,'' \emph{IEEE Trans. Services
  Comput.}, vol.~14, no.~2, pp. 589--605, 2021.

\bibitem{53}
K.~S. Trivedi, K.~Vaidyanathan, and D.~Selvamuthu, ``Chapter 13 - markov chain
  models and applications,'' in \emph{Modeling and Simulation of Computer
  Networks and Systems}.\hskip 1em plus 0.5em minus 0.4em\relax Boston: Morgan
  Kaufmann, 2015, pp. 393--421.

\bibitem{54}
S.~M. Ross, \emph{Stochastic Processes}.\hskip 1em plus 0.5em minus 0.4em\relax
  New York: John Wiley and Sons, 2-nd edition, 1996.

\bibitem{55}
S.-Y. Kuo, F.-M. Yeh, and H.-Y. Lin, ``Efficient and exact reliability
  evaluation for networks with imperfect vertices,'' \emph{IEEE Trans. Rel.},
  vol.~56, no.~2, pp. 288--300, 2007.

\bibitem{56}
Z.~Zhang, F.~Shao, N.~Zhang, and Y.~Niu, ``Maximizing k-terminal network
  reliability in some sparse graphs,'' \emph{IEEE/ACM Trans. Netw.}, vol.~29,
  no.~1, pp. 190--202, 2021.

\bibitem{10100891}
Z.~Liu, M.~Liwang, S.~Hosseinalipour, H.~Dai, Z.~Gao, and L.~Huang, ``{RFID}:
  Towards low latency and reliable {DAG} task scheduling over dynamic vehicular
  clouds,'' \emph{IEEE Trans. Veh. Technol.}, pp. 1--15, 2023.

\bibitem{993206}
H.~Topcuoglu, S.~Hariri, and M.-Y. Wu, ``Performance-effective and
  low-complexity task scheduling for heterogeneous computing,'' \emph{IEEE
  Trans. Parallel Distrib. Syst.}, vol.~13, no.~3, pp. 260--274, 2002.

\end{thebibliography}
\vspace{-10mm}
\begin{IEEEbiography}
    [{\includegraphics[width=1in,height=1.25in,clip,keepaspectratio]{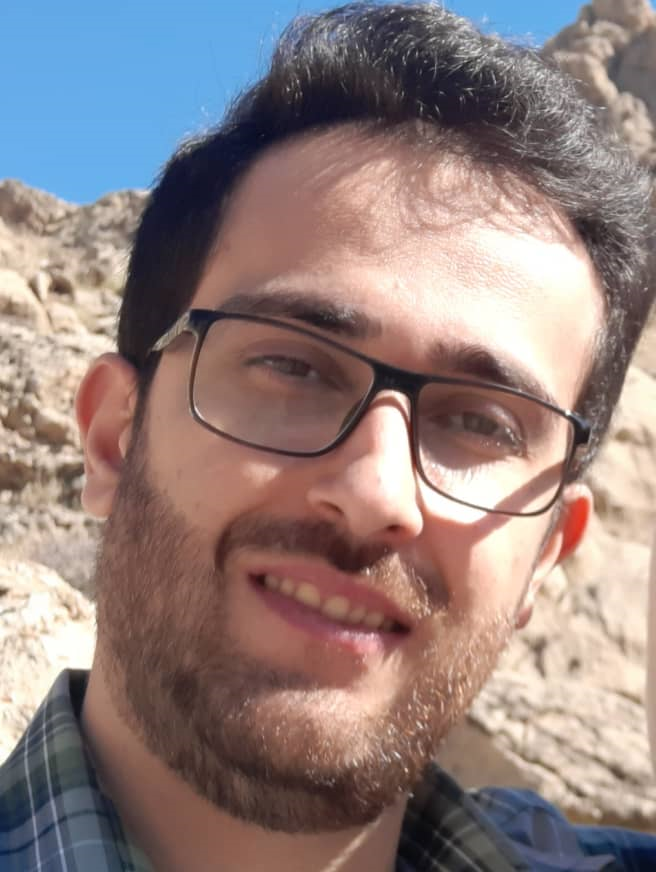}}]{Payam Abdisarabshali}
received the M.Sc. degree in Computer Engineering from Razi University, Iran, with top-rank recognition in 2018. He was a teaching assistant professor (lecturer) with Razi University from 2018 to 2022. He is currently a Ph.D. student at University at Buffalo (SUNY), NY, USA. His research interests include mathematical modeling, mobile computing, distributed machine learning, optimization of next-generation intelligent networks, reliability analysis of computing systems.
\end{IEEEbiography}
\vspace{-10mm}
\begin{IEEEbiography}
    [{\includegraphics[width=1in,height=1.25in,clip,keepaspectratio]{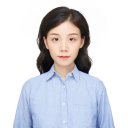}}]{Minghui LiWang} [M'19]
 received the Ph.D. degree from the School of Informatics, Xiamen University, China. She was a Visiting Scholar at North Carolina State University, USA; and a Post Doctoral Fellow with the ECE Department, Western University, Canada, She is currently an assistant professor with the School of Informatics, Xiamen University. Her research interests are wireless communication systems, mobile edge computing, resource optimization management and Internet of Vehicles.
\end{IEEEbiography}
\vspace{-10mm}
\begin{IEEEbiography}
    [{\includegraphics[width=1in,height=1.25in,clip,keepaspectratio]{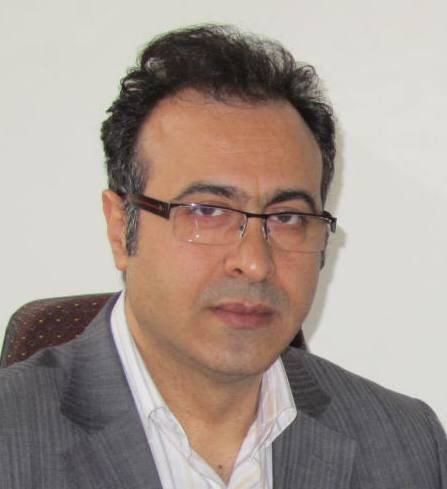}}]{Amir Rajabzadeh}
received the M.S. and Ph.D. degrees in Computer Engineering from Sharif University of Technology, Iran. He is working as an associate professor of Computer Engineering at Razi University, Iran. His main areas of interests are computer architecture, High performance Computing and fault-tolerant systems design. Dr. Rajabzadeh has earned one world, six international, and five national awards in robotic competition, and one national award in mobile computing.
\end{IEEEbiography}
\vspace{-10mm}
\begin{IEEEbiography}
    [{\includegraphics[width=1in,height=1.25in,clip,keepaspectratio]{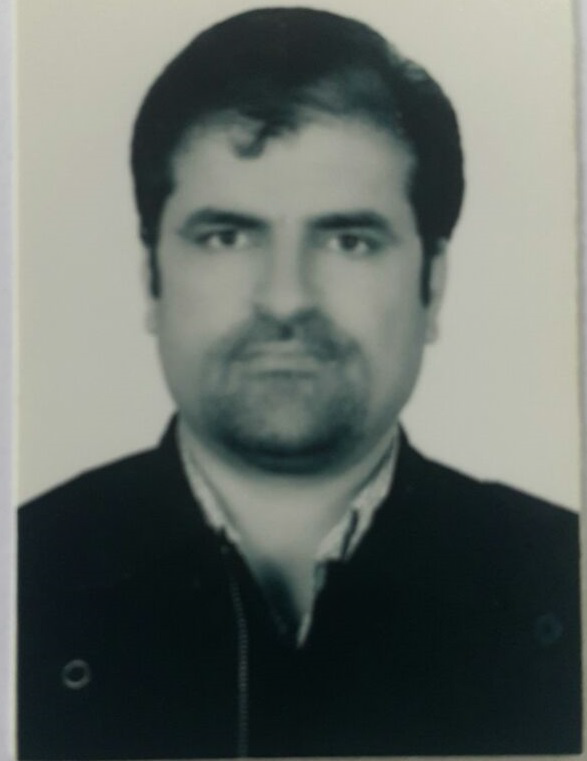}}]{Mahmood Ahmadi}
received the M.S. degree in Computer Architecture from Tehran Polytechnique University, Iran. He received Ph.D. degree from faculty of Electrical Engineering, Mathematics, and Computer Science at Delft University of Technology, Netherlands. He is working as an associate professor of Computer Engineering at Razi University, Iran. His research interests include Computer architecture, network processing, Bloom filters, software-defined networking, and high-performance computing.
\end{IEEEbiography}
\vspace{-10mm}
\begin{IEEEbiography}
    [{\includegraphics[width=1in,height=1.25in,clip,keepaspectratio]{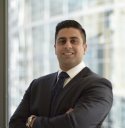}}]{Seyyedali Hosseinalipour}  [M'20]
received the Ph.D. degree in Electrical Engineering from North Carolina State University, NC, USA, in 2020. He was a postdoctoral researcher at Purdue University, IN, USA, from 2020 to 2022. He is currently an assistant professor in the department of Electrical Engineering at University at Buffalo (SUNY), NY, USA. His research interests include analysis of modern wireless networks, machine learning, and fog/edge computing.
\end{IEEEbiography}

 \vfill

\definecolor{Gray}{gray}{0.9}
\setlength{\textfloatsep}{0.05cm}

\appendices

\allowdisplaybreaks
\begingroup
\onecolumn
\setcounter{page}{1}
\section{Proof of Lemma \ref{conditional_closure}}\label{conditional_closure_proof}
\noindent In the following, we present the proof of Lemma \ref{conditional_closure} for \eqref{eq10} and \eqref{eq11}, respectively.\par
\textbf{Proof of \eqref{eq10}.}
Based on Definition \ref{pivot}, $V_1=\hat{\Gamma}(V_1)+\hat{\xi}(V_1)-\hat{\mathbf{\delta}}(V_1)$ and $V_2=\hat{\Gamma}(V_2)+\hat{\xi}(V_2)-\hat{\mathbf{\delta}}(V_2)$. Therefore, $S_1$ is given by
\begin{equation}
\begin{aligned}
  S_1=V_1-V_2&=\hat{\Gamma}(V_1)+\hat{\mathbf{\xi}}(V_1)-\hat{\mathbf{\delta}}(V_1)- \hat{\Gamma}(V_2)-\hat{\mathbf{\xi}}(V_2)+\hat{\mathbf{\delta}}(V_2).
\end{aligned}
\end{equation}
Considering $V_1\dot{\preceq}V_2$, we have $\hat{\mathbf{\delta}}(V_2)=\hat{\mathbf{\delta}}(V_1)$. As a result, we can rewrite $S_1$ as follows:
\begin{equation}
  S_1=V_1-V_2=\hat{\Gamma}(V_1)+\hat{\mathbf{\xi}}(V_1)- \hat{\Gamma}(V_2)-\hat{\mathbf{\xi}}(V_2).
\end{equation}
Furthermore, since $S_1>0$ we have
\begin{equation}\label{eq59}
  \begin{aligned}
  S_1>0 &\Rightarrow \hat{\Gamma}(V_1)+\hat{\mathbf{\xi}}(V_1)- \hat{\Gamma}(V_2)-\hat{\mathbf{\xi}}(V_2)>0\\
  &\Rightarrow \underbrace{\hat{\Gamma}(V_1)+\hat{\mathbf{\xi}}(V_1)-\hat{\mathbf{\xi}}(V_2)}_{(a)}> \hat{\Gamma}(V_2).
  \end{aligned}
\end{equation}
Using Definition \ref{event_dynamic_variable} and taking into account $V_1\dot{\preceq}V_2$, we have $\hat{\Gamma}(V_2)>0$ and $\hat{\mathbf{\xi}}(V_2)>\hat{\mathbf{\xi}}(V_1)$. Accordingly, considering $(a)$ in \eqref{eq59}, $\hat{\Gamma}(V_1)$ is lower bounded by zero as follows:
\begin{equation}
  \begin{aligned}
  \hat{\Gamma}(V_2){>}0, \hat{\mathbf{\xi}}(V_2){>}\hat{\mathbf{\xi}}(V_1) &\Rightarrow
  \hat{\Gamma}(V_1){+}\hat{\mathbf{\xi}}(V_1){-}\hat{\mathbf{\xi}}(V_2){>}0\\
  &\Rightarrow \hat{\Gamma}(V_1){>}\hat{\mathbf{\xi}}(V_2){-}\hat{\mathbf{\xi}}(V_1){>}0.
  \end{aligned}
\end{equation}
Consequently, based on Definition \ref{event_dynamic_variable}, $S_1$ is an EDV, which can be written as follows:
\begin{equation}
    S_1=V_1-V_2=\left(\hat{\Gamma}(V_1)-\left(\hat{\xi}(V_2)-\hat{\xi}(V_1)\right)\right)-\hat{\Gamma}(V_2),
\end{equation}
where $\left(\hat{\Gamma}(V_1)-\left(\hat{\xi}(V_2)-\hat{\xi}(V_1)\right)\right)>\hat{\Gamma}(V_2)$ and $\hat{\Gamma}(V_1)>\hat{\mathbf{\xi}}(V_2)-\hat{\mathbf{\xi}}(V_1)>0$.\par

\textbf{Proof of \eqref{eq11}.} In the following, we conduct similar operations as above to prove \eqref{eq11}. From Definition \ref{pivot}, we know that $S_2$ is given by
\begin{equation}
\begin{aligned}
  S_2=V_2-V_1&=\hat{\Gamma}(V_2)+\hat{\mathbf{\xi}}(V_2)-\hat{\mathbf{\delta}}(V_2)- \hat{\Gamma}(V_1)-\hat{\mathbf{\xi}}(V_1)+\hat{\mathbf{\delta}}(V_1).
\end{aligned}
\end{equation}
Since $V_1\dot{\preceq}V_2$, we have $\hat{\mathbf{\delta}}(V_2)=\hat{\mathbf{\delta}}(V_1)$. Accordingly, we can rewrite $S_2$ as follows:
\begin{equation}
  S_2=V_2-V_1=\hat{\Gamma}(V_2)+\hat{\mathbf{\xi}}(V_2)- \hat{\Gamma}(V_1)-\hat{\mathbf{\xi}}(V_1).
\end{equation}
Further, using $S_2>0$, $\hat{\Gamma}(V_2)$ is lower bounded as
\begin{equation}\label{eq64}
  \begin{aligned}
  S_2>0 &\Rightarrow \hat{\Gamma}(V_2)+\hat{\mathbf{\xi}}(V_2)- \hat{\Gamma}(V_1)-\hat{\mathbf{\xi}}(V_1)>0\\
  &\Rightarrow \hat{\Gamma}(V_2)>\underbrace{\hat{\Gamma}(V_1)-\hat{\mathbf{\xi}}(V_2) +\hat{\mathbf{\xi}}(V_1)}_{(b)}.
  \end{aligned}
\end{equation}
We next aim to bound $(b)$ in \eqref{eq64}, based on which we further bound $\hat{\Gamma}(V_1)$ in terms of $\hat{\mathbf{\xi}}(V_2) -\hat{\mathbf{\xi}}(V_1)$. Based on Definition \ref{event_dynamic_variable} and taking into account that $V_1$ is an EDV, we have $\hat{\Gamma}(V_1)>\hat{\mathbf{\delta}}(V_1)-\hat{\mathbf{\xi}}(V_1)>0$, which leads to the following inequality
\begin{equation}\label{eq65}
  \begin{aligned}
  &\hat{\Gamma}(V_1)>\hat{\mathbf{\delta}}(V_1)-\hat{\mathbf{\xi}}(V_1)\\
  &\Rightarrow  \hat{\Gamma}(V_1)+\hat{\mathbf{\xi}}(V_1)> \hat{\mathbf{\delta}}(V_1)\\
  &\Rightarrow \underbrace{\hat{\Gamma}(V_1)-\hat{\mathbf{\xi}}(V_2)+\hat{\mathbf{\xi}}(V_1)> \hat{\mathbf{\delta}}(V_1)-\hat{\mathbf{\xi}}(V_2)}_{(C)}.\\
  \end{aligned}
\end{equation}
Because $\hat{\mathbf{\delta}}(V_1)=\hat{\mathbf{\delta}}(V_2)$, we have $\hat{\mathbf{\delta}}(V_1)-\hat{\mathbf{\xi}}(V_2)= \hat{\mathbf{\delta}}(V_2)-\hat{\mathbf{\xi}}(V_2)$. Further, $\hat{\mathbf{\delta}}(V_2)-\hat{\mathbf{\xi}}(V_2)>0$ since $V_2$ is an EDV, resulting in $\hat{\mathbf{\delta}}(V_1)-\hat{\mathbf{\xi}}(V_2)>0$. Hence, considering $(c)$ in \eqref{eq65}, $\hat{\Gamma}(V_1)-\hat{\mathbf{\xi}}(V_2)+\hat{\mathbf{\xi}}(V_1)$ is lower bounded by zero as follows:
\begin{equation}\label{eq66}
  \begin{aligned}
   &\hat{\Gamma}(V_1)-\hat{\mathbf{\xi}}(V_2)+\hat{\mathbf{\xi}}(V_1)> \hat{\mathbf{\delta}}(V_1)-\hat{\mathbf{\xi}}(V_2)>0.
  \end{aligned}
\end{equation}
Using $V_1\dot{\preceq}V_2$ and \eqref{eq66}, we have $\hat{\mathbf{\xi}}(V_2)>\hat{\mathbf{\xi}}(V_1)$ and $\hat{\Gamma}(V_1)>\hat{\mathbf{\xi}}(V_2) -\hat{\mathbf{\xi}}(V_1)$, respectively. As a result, we have the following lower bounds for $\hat{\Gamma}(V_1)$
\begin{equation}
  \hat{\Gamma}(V_1)>\hat{\mathbf{\xi}}(V_2) -\hat{\mathbf{\xi}}(V_1)>0.
\end{equation}
Consequently, $S_2$ is an EDV, which can be calculated as follows:
\begin{equation}
    S_2=V_2-V_1=\hat{\Gamma}(V_2)-\left(\hat{\Gamma}(V_1)-\left(\hat{\xi}(V_2)-\hat{\xi}(V_1)\right)\right),
\end{equation}
where $\hat{\Gamma}(V_2)>\left(\hat{\Gamma}(V_1)-\left(\hat{\xi}(V_2)-\hat{\xi}(V_1)\right)\right)$ and $\hat{\Gamma}(V_1)>\left(\hat{\xi}(V_2)-\hat{\xi}(V_1)\right)>0$.

\section{Proof of Lemma \ref{subtraction_inequality}}\label{subtraction_inequality_proof}
\noindent According to Definition \ref{CG}, to prove $(V_1-V') \dot{\preceq} (V_2-V')$, we must prove
\begin{equation}\label{eq69}
\hat{\mathbf{\delta}}(V_1-V')=\hat{\mathbf{\delta}}(V_2-V'),
\end{equation}
and
\begin{equation}\label{eq70}
\hat{\mathbf{\xi}}(V_1-V')<\hat{\mathbf{\xi}}(V_2-V').
\end{equation}
In the following, we prove Lemma \ref{subtraction_inequality} for two cases: (i) if $V'$ is an EDV and (ii) if $V'$ is an SRV. \par
\textbf{(i) $V'$ is an EDV.} Below, we first aim to prove \eqref{eq69}. Using Definition \ref{reducer}, we have
\begin{equation}\label{eq71}
  \hat{\mathbf{\delta}}(V_1-V')=\hat{\mathbf{\delta}}(V_1)+\left(V'+\hat{\mathbf{\delta}}(V')\right),
\end{equation}
and
\begin{equation}\label{eq72}
  \hat{\mathbf{\delta}}(V_2-V')=\hat{\mathbf{\delta}}(V_2)+\left(V'+\hat{\mathbf{\delta}}(V')\right).
\end{equation}
If $\{V_1 \dot{\preceq} V' or V' \dot{\preceq} V_1\}$ and $\{V_2 \dot{\preceq} V' or V' \dot{\preceq} V_2\}$, we have $\hat{\mathbf{\delta}}(V_1){=}\hat{\mathbf{\delta}}(V')$ and $\hat{\mathbf{\delta}}(V_2){=}\hat{\mathbf{\delta}}(V')$. Therefore, we can rewrite \eqref{eq71} and \eqref{eq72} as follows:
\begin{equation}
  \hat{\mathbf{\delta}}(V_1-V')=2\hat{\mathbf{\delta}}(V')+V',
\end{equation}
and
\begin{equation}
  \hat{\mathbf{\delta}}(V_2-V')=2\hat{\mathbf{\delta}}(V')+V'.
\end{equation}
Consequently,
\begin{equation}\label{RC}
\hat{\mathbf{\delta}}(V_1-V')=\hat{\mathbf{\delta}}(V_2-V').
\end{equation}
We next take the following steps to prove \eqref{eq70}. Based on Definition \ref{pivot}, $\hat{\mathbf{\xi}}(V_1-V')$ and $\hat{\mathbf{\xi}}(V_2-V')$ are given as follows:
\begin{equation}\label{eq76}
\begin{aligned}
  \hat{\mathbf{\xi}}(V_1-V')&=V_2-V'- \hat{\Gamma}(V_1-V')+\hat{\mathbf{\delta}}(V_1-V')\\
  &=V_1-V'-\hat{\Gamma}(V_1)+2\hat{\mathbf{\delta}}(V')+V'\\
  &=V_1-\hat{\Gamma}(V_1)+2\hat{\mathbf{\delta}}(V')\\
  &=\hat{\mathbf{\xi}}(V_1)+\hat{\mathbf{\delta}}(V'),
\end{aligned}
\end{equation}
and
\begin{equation}\label{eq77}
\begin{aligned}
  \hat{\mathbf{\xi}}(V_2-V')&=V_2-V'- \hat{\Gamma}(V_2-V')+\hat{\mathbf{\delta}}(V_2-V')\\
  &=V_2-V'-\hat{\Gamma}(V_2)+2\hat{\mathbf{\delta}}(V')+V'\\
  &=V_2-\hat{\Gamma}(V_2)+2\hat{\mathbf{\delta}}(V')\\
  &=\hat{\mathbf{\xi}}(V_2)+\hat{\mathbf{\delta}}(V').
\end{aligned}
\end{equation}
Give that {\small$V_1 \dot{\preceq} V_2$}, we have $\hat{\mathbf{\xi}}(V_1)<\hat{\mathbf{\xi}}(V_2)$, which implies
\begin{equation}\label{IC}
\begin{aligned}
  \hat{\mathbf{\xi}}(V_1)<\hat{\mathbf{\xi}}(V_2) &\Rightarrow \hat{\mathbf{\xi}}(V_1)+\hat{\mathbf{\delta}}(V') <\hat{\mathbf{\xi}}(V_2)+\hat{\mathbf{\delta}}(V')\\
  &\Rightarrow \hat{\mathbf{\xi}}(V_1-V') < \hat{\mathbf{\xi}}(V_2-V').
\end{aligned}
\end{equation}
Finally, taking into consideration (\ref{RC}) and (\ref{IC}), we have
\begin{equation}
  V_1-V'\dot{\preceq} V_2-V',
\end{equation}
which concludes the proof of Lemma \ref{subtraction_inequality} for the case when $V'$ is an EDV. Next, we aim to prove Lemma \ref{subtraction_inequality} when $V'$ is an SRV.\par
\textbf{(ii) $V'$ is an SRV.}
We first present the proof of \eqref{eq69}. According to Definition \ref{pivot}, for an SRV $V'$, we have $\hat{\mathbf{\delta}}(V')=0$. Replacing $\hat{\mathbf{\delta}}(V')=0$ back in \eqref{eq71} and \eqref{eq72} yields
\begin{equation}
  \hat{\mathbf{\delta}}(V_1-V')=\hat{\mathbf{\delta}}(V_1)+V',
\end{equation}
and
\begin{equation}
  \hat{\mathbf{\delta}}(V_2-V')=\hat{\mathbf{\delta}}(V_2)+V'.
\end{equation}
From $V_1\dot{\preceq} V_2$, we have $\hat{\mathbf{\delta}}(V_1)=\hat{\mathbf{\delta}}(V_2)$.
Consequently,
\begin{equation}\label{RC2}
\hat{\mathbf{\delta}}(V_1-V')=\hat{\mathbf{\delta}}(V_2-V').
\end{equation}
Next, we aim to prove \eqref{eq70}. Replacing $\hat{\mathbf{\delta}}(V')=0$ back in \eqref{eq76} and \eqref{eq77} and performing some algebraic manipulations leads to
\begin{equation}
\begin{aligned}
  \hat{\mathbf{\xi}}(V_1-V')=V_1-\hat{\Gamma}(V_1),
\end{aligned}
\end{equation}
and
\begin{equation}
\begin{aligned}
  \hat{\mathbf{\xi}}(V_2-V')=V_2-\hat{\Gamma}(V_2).
\end{aligned}
\end{equation}
From {\small$V_1 \dot{\preceq} V_2$}, we have $\hat{\mathbf{\delta}}(V_1)=\hat{\mathbf{\delta}}(V_2)$ and
\begin{equation}
V_1+\hat{\mathbf{\delta}}(V_1)-\hat{\Gamma}(V_1)<V_2+\hat{\mathbf{\delta}}(V_2)-\hat{\Gamma}(V_2),
\end{equation}
which leads to $V_1-\hat{\Gamma}(V_1)<V_2-\hat{\Gamma}(V_2)$. As a result, we get
\begin{equation}\label{IC2}
\begin{aligned}
  \hat{\mathbf{\xi}}(V_1-V') < \hat{\mathbf{\xi}}(V_2-V').
\end{aligned}
\end{equation}
Consequently, considering \eqref{RC2} and \eqref{IC2}, we have
\begin{equation}
  V_1-V'\dot{\preceq} V_2-V',
\end{equation}
which conclude the proof of Lemma \ref{subtraction_inequality} when $V'$ is an SRV.

\newpage
\section{Proof of Proposition \ref{decomposability_props}}\label{decomposability_props_proof}
\noindent In the following, we present the proof of Proposition \ref{decomposability_props} for (i) $\alpha,\alpha'=\varphi$, (ii) $\alpha,\alpha'=\psi$, (iii) $\alpha,\alpha'=\gamma$, (iv) $\alpha= \psi,\alpha'=\gamma$, and (v) $\alpha= \gamma,\alpha'=\psi$. Let $Z_1$, $Z_2$, $Z_3$, and $Z_4$ be four i.i.d exponential random variables and $U_1$, $U_2$, $U_3$, and $U_4$ be four i.i.d random variables following general distribution $\mathfrak{R}$. Further, consider following five different combinations of e-distributions for EDVs $V_1$ and $V_2$:
\hspace{-5mm}
\begin{enumerate}
    \item $V_1{\sim}\varphi$ and $V_2{\sim}\varphi$:
\begin{equation}\label{pro85}
    \begin{cases}
        V_1{\sim}\varphi{=}U_1{-}Z_2{+}Z_1,~\hat{\Gamma}(V_1){=}U_1,\hat{\xi}(V_1){=}Z_1,\\
        V_2{\sim}\varphi{=}U_2{-}Z_4{+}Z_3,~\hat{\Gamma}(V_2){=}U_2,\hat{\xi}(V_2){=}Z_3.
    \end{cases}
\end{equation}
\item $V_1{\sim}\psi$ and $V_2{\sim}\psi$:
\begin{equation}\label{pro86}
\hspace{-4mm}
    \begin{cases}
        V_1{\sim}\psi{=}U_1{-}(U_3{-}Z_1{+}Z_2),~\hat{\Gamma}(V_1){=}U_1,\hat{\xi}(V_1){=}Z_1,\\
        V_2{\sim}\psi{=}U_2{-}(U_4{-}Z_3{+}Z_4),~\hat{\Gamma}(V_2){=}U_2,\hat{\xi}(V_2){=}Z_3.
    \end{cases}
\hspace{-4mm}
\end{equation}
\item $V_1{\sim}\gamma$ and $V_2{\sim}\gamma$:
\begin{equation}\label{pro87}
\hspace{-4mm}
    \begin{cases}
        V_1{\sim}\gamma{=}(U_1{-}Z_2{+}Z_1){-}U_3,~\hat{\Gamma}(V_1){=}U_1,\hat{\xi}(V_1){=}Z_1,\\
        V_2{\sim}\gamma{=}(U_2{-}Z_4{+}Z_3){-}U_4,~\hat{\Gamma}(V_2){=}U_2,\hat{\xi}(V_2){=}Z_3.
    \end{cases}
 \hspace{-4mm}
\end{equation}
\item $V_1{\sim}\psi$ and $V_2{\sim}\psi$:
\begin{equation}\label{pro88}
 \hspace{-4mm}
    \begin{cases}
        V_1{\sim}\psi{=}U_1{-}(U_3{-}Z_1{+}Z_2),~\hat{\Gamma}(V_1){=}U_1,\hat{\xi}(V_1){=}Z_1,\\
        V_2{\sim}\gamma{=}(U_2{-}Z_4{+}Z_3){-}U_4,~\hat{\Gamma}(V_2){=}U_2,\hat{\xi}(V_2){=}Z_3.
    \end{cases}
 \hspace{-4mm}
\end{equation}
\item $V_1{\sim}\gamma$ and $V_2{\sim}\gamma$:
\begin{equation}\label{pro89}
 \hspace{-4mm}
    \begin{cases}
        V_1{\sim}\gamma{=}(U_1{-}Z_2{+}Z_1){-}U_3,~\hat{\Gamma}(V_1){=}U_1,\hat{\xi}(V_1){=}Z_1,\\
        V_2{\sim}\gamma{=}U_2{-}(U_4{-}Z_3{+}Z_4),~\hat{\Gamma}(V_2){=}U_2,\hat{\xi}(V_2){=}Z_3.
    \end{cases}
 \hspace{-4mm}
\end{equation}
\end{enumerate}
Based on Lemma \ref{conditional_closure}, given that $V_1\dot{\preceq}V_2$, $S_1$ and $S_2$ are given as follows. If $S_1>0$, we have
\begin{equation}\label{pro90}
\begin{aligned}
S_1&=V_1-V_2 =\left(\hat{\Gamma}(V_1)-\left(\hat{\xi}(V_2)- \hat{\xi}(V_1)\right)\right)-\hat{\Gamma}(V_2).
\end{aligned}
\end{equation}
Replacing $\hat{\Gamma}(V_1)=U_1$, $\hat{\Gamma}(V_2)=U_2$, $\hat{\xi}(V_1)=Z_1$, and $\hat{\xi}(V_2)=Z_3$ from \eqref{pro85}, \eqref{pro86}, \eqref{pro87}, \eqref{pro88}, or \eqref{pro89} in \eqref{pro90} yields
\begin{equation}
    S_1=\big(U_1-(Z_3-Z_1)\big)-U_2.
\end{equation}
Based on Lemma \ref{conditional_closure}, $\big(U_1-(Z_3-Z_1)\big)>U_2$, $U_1>(Z_3-Z_1)$, and $Z_3>Z_1$. Consequently, based on Definition \ref{gammadist}, $S_1{\sim}\gamma$.\par
Likewise, if $S_2>0$, we get
\begin{equation}\label{proeq92}
\begin{aligned}
S_2&=V_2-V_1 =\hat{\Gamma}(V_2)-\left(\hat{\Gamma}(V_1)- \left(\hat{\xi}(V_2)-\hat{\xi}(V_1)\right)\right).
\end{aligned}
\end{equation}
By replacing $\hat{\Gamma}(V_1)=U_1$, $\hat{\Gamma}(V_2)=U_2$, $\hat{\xi}(V_1)=Z_1$, and $\hat{\xi}(V_2)=Z_3$ from replacing the result of each of \eqref{pro85}, \eqref{pro86}, \eqref{pro87}, \eqref{pro88}, or \eqref{pro89} in \eqref{proeq92}, $S_2$ is given by
\begin{equation}
    S_2=U_2-(U_1-(Z_3-Z_1)).
\end{equation}
Based on Lemma \ref{conditional_closure}, $U_2>\left(U_1-(Z_3-Z_1)\right)$, $U_1>(Z_3-Z_1)$, and $Z_3>Z_1$. Consequently, based on Definition \ref{psidist}, $S_2{\sim}\psi$.\par

\newpage
\section{Proof of Proposition \ref{absorbency_props}}\label{absorbency_props_proof}
\noindent In the following, the proof of Proposition \ref{absorbency_props} is presented in three parts.
\subsection{Part 1: Proof of $S_1$ and $S_3$}
Let $W{\sim}\varphi=U-(Z_2-Z_1)$ and $Z'_1{\sim}\dot{\mathfrak{X}}=Z_4-Z_3$. Moreover, given that $Z'_1\dot{\preceq}W$, if $S_1>0$, based on Lemma \ref{conditional_closure}, we have
\begin{equation}
 \hspace{-4mm}
\begin{aligned}
S_1=W-Z'_1&=\hat{\Gamma}(W)-\left(\hat{\Gamma}(Z') -\left(\hat{\xi}(W)-\hat{\xi}(Z'_1)\right)\right)\\
&=U-(Z_4-(Z_1-0))\\
&=U-(Z_4-Z_1),
\end{aligned}
 \hspace{-4mm}
\end{equation}
where $U>Z_4-Z_1$ and $Z_4>Z_1$. Consequently, based on definition \ref{phidist}, $S_1{\sim}\varphi$.\par
Likewise, if $S_3>0$, Lemma \ref{conditional_closure} states that
\begin{equation}
 \hspace{-4mm}
\begin{aligned}
S_3=Z'_1-W&=\left(\hat{\Gamma}(Z'_1)-\left(\hat{\xi}(W)- \hat{\xi}(Z'_1)\right)\right) -\hat{\Gamma}(W)\\
&=(Z_4-(Z_1-0))-U\\
&=Z_4-Z_1-U,
\end{aligned}
 \hspace{-4mm}
\end{equation}
where $Z_4>Z_1+U$. Hence, based on Definition \ref{expdist}, $S_3{\sim}\ddot{\mathfrak{X}}$.

\subsection{Part 2: Proof of $S_2$ and $S_4$ when $V{\sim}\gamma$}
Let $V{\sim}\gamma=U_3-(U_1-(Z_2-Z_1))$ and $Z'_2{\sim}\ddot{\mathfrak{X}}=Z_4-Z_3-U_2$. Given that $Z'_2\dot{\preceq}V$, if $S_2>0$, based on Lemma \ref{conditional_closure}, we get
\begin{equation}
\begin{aligned}
S_2=V-Z'_2&= \hat{\Gamma}(V)-\left(\hat{\Gamma}(Z'_2)-\left(\hat{\xi}(V)-\hat{\xi}(Z'_2)\right)\right)\\
&=U_3-(Z_4-(Z_2-0))\\
&=U_3-(Z_4-Z_2),
\end{aligned}
\end{equation}
where $U_3>Z_4-Z_2$ and $Z_4>Z_2$. Consequently, based on Definition \ref{phidist}, $S_2{\sim}\varphi$.\par
Further, if $S_4>0$, according to Lemma \ref{conditional_closure}, the following result is concluded
\begin{equation}
\begin{aligned}
S_4=Z'_2-V&=\left(\hat{\Gamma}(Z'_2)- \left(\hat{\xi}(V)-\hat{\xi}(Z'_2)\right)\right)-\hat{\Gamma}(V)\\
&=(Z_4-(Z_2-0))-U_3\\
&=Z_4-Z_2-U_3,
\end{aligned}
\end{equation}
where $Z_4>U_3+Z_2$. Consequently, based on Definition \ref{expdist}, $S_4{\sim}\ddot{\mathfrak{X}}$.

\subsection{Part 3: Proof of $S_2$ and $S_4$ when $V{\sim}\psi$}
Let $V{\sim}\psi=(U_1-(Z_2-Z_1))-U_3$ and $Z'_2{\sim}\ddot{\mathfrak{X}}=Z_4-Z_3-U_2$. Given that $Z'_2 \dot{\preceq} V$, if $S_2>0$, Lemma \ref{conditional_closure} implies
\begin{equation}
\begin{aligned}
S_2=V-Z'_2&=\hat{\Gamma}(V)-\left(\hat{\Gamma}(Z'_2)-\left(\hat{\xi}(V)-\hat{\xi}(Z'_2)\right)\right)\\
&=U_1-(Z_4-(Z_1-0))\\
&=U_1-(Z_4-Z_1),
\end{aligned}
\end{equation}
where $U_1>Z_4-Z_1$ and $Z_4>Z_1$. Consequently, based on Definition \ref{phidist}, $S_2{\sim}\varphi$.\par
Similarly, if $S_4>0$, Lemma \ref{conditional_closure} yields
\begin{equation}
\begin{aligned}
S_4=Z'_2-V&=\left(\hat{\Gamma}(Z'_2)- \left(\hat{\xi}(V)-\hat{\xi}(Z'_2)\right)\right)-\hat{\Gamma}(V)\\
&=(Z_4-(Z_1-0))-U_1\\
&=Z_4-Z_1-U_1,
\end{aligned}
\end{equation}
where $Z_4>U_1+Z_1$. Consequently, based on Definition \ref{expdist}, $S_4{\sim}\ddot{\mathfrak{X}}$.

\newpage
\section{Proof of Proposition \ref{proposition3}}\label{proposition3_proof}
\noindent We prove Proposition \ref{proposition3} in the following two parts. We first prove \eqref{eqq1} and \eqref{eq21} in Sec. \ref{part1Prop3}. Next, we prove \eqref{eq22} and \eqref{eq23} in Sec. \ref{part2Prop3}.
\subsection{Part 1: Proofs of \eqref{eqq1} and \eqref{eq21}}\label{part1Prop3}
We first aim to prove \eqref{eqq1}. According to Definition \ref{event_dynamic_list}, for $1\le k\le i-1$, we have
\begin{equation}
\begin{aligned}
  \mathcal{H}_1^{(\varphi;\mathfrak{R})}\hspace{-1.4mm}\left\langle n{-}1 \hspace{-0.3mm};\hspace{-0.3mm} 1 \right\rangle(k)\dot{\preceq}\mathcal{H}_1^{(\varphi;\mathfrak{R})}\hspace{-1.4mm}\left\langle n{-}1 \hspace{-0.3mm};\hspace{-0.3mm} 1 \right\rangle(i),
\end{aligned}
\end{equation}
and for $i+1\le k\le n$, we have
\begin{equation}
    \mathcal{H}_1^{(\varphi;\mathfrak{R})}\hspace{-1.4mm}\left\langle n{-}1 \hspace{-0.3mm};\hspace{-0.3mm} 1 \right\rangle(i) \dot{\preceq} \mathcal{H}_1^{(\varphi;\mathfrak{R})}\hspace{-1.4mm}\left\langle n{-}1 \hspace{-0.3mm};\hspace{-0.3mm} 1 \right\rangle(k).
\end{equation}
If $\mathcal{H}_1^{(\varphi;\mathfrak{R})}\hspace{-1.4mm}\left\langle n{-}1 \hspace{-0.3mm};\hspace{-0.3mm} 1 \right\rangle(i)<\mathcal{H}_1^{(\varphi;\mathfrak{R})}\hspace{-1.4mm}\left\langle n{-}1 \hspace{-0.3mm};\hspace{-0.3mm} 1 \right\rangle(k)$, based on Definition \ref{psidist}, $\mathcal{H}_1^{(\varphi;\mathfrak{R})}\hspace{-1.4mm}\left\langle n{-}1 \hspace{-0.3mm};\hspace{-0.3mm} 1 \right\rangle(n)-\mathcal{H}_1^{(\varphi;\mathfrak{R})}\hspace{-1.4mm}\left\langle n{-}1 \hspace{-0.3mm};\hspace{-0.3mm} 1 \right\rangle(i)$ follows $\psi$ distribution for $i\neq n$. Here, $n$ refers to the last element of $\mathcal{H}_1^{(\varphi;\mathfrak{R})}\hspace{-1.4mm}\left\langle n{-}1 \hspace{-0.3mm};\hspace{-0.3mm} 1 \right\rangle$, which follows general distribution $\mathfrak{R}$. Moreover, based on Proposition \ref{decomposability_props}, for $1\le k \le i-1$, we have
\begin{equation}
    \left(\mathcal{H}_1^{(\varphi;\mathfrak{R})}\hspace{-1.4mm}\left\langle n{-}1 \hspace{-0.3mm};\hspace{-0.3mm} 1 \right\rangle(k)-\mathcal{H}_1^{(\varphi;\mathfrak{R})}\hspace{-1.4mm}\left\langle n{-}1 \hspace{-0.3mm};\hspace{-0.3mm} 1 \right\rangle(i)\right){\sim}\gamma,
\end{equation}
and for $i+1\le k \le n$, we have
\begin{equation}
    \left(\mathcal{H}_1^{(\varphi;\mathfrak{R})}\hspace{-1.4mm}\left\langle n{-}1 \hspace{-0.3mm};\hspace{-0.3mm} 1 \right\rangle(k)-\mathcal{H}_1^{(\varphi;\mathfrak{R})}\hspace{-1.4mm}\left\langle n{-}1 \hspace{-0.3mm};\hspace{-0.3mm} 1 \right\rangle(i)\right){\sim}\psi.
\end{equation}
For the case that $i=n$, based on Definition \ref{gammadist}, $\mathcal{H}_1^{(\varphi;\mathfrak{R})}\hspace{-1.4mm}\left\langle n{-}1 \hspace{-0.3mm};\hspace{-0.3mm} 1 \right\rangle(k)-\mathcal{H}_1^{(\varphi;\mathfrak{R})}\hspace{-1.4mm}\left\langle n{-}1 \hspace{-0.3mm};\hspace{-0.3mm} 1 \right\rangle(i)$ follows $\gamma$ distribution for all $1\le k \le n{-}1$. Given that, for $1\le k< x\le n$, $\mathcal{H}_1^{(\varphi;\mathfrak{R})}\hspace{-1.4mm}\left\langle n{-}1 \hspace{-0.3mm};\hspace{-0.3mm} 1 \right\rangle(k)\preceq \mathcal{H}_1^{(\varphi;\mathfrak{R})}\hspace{-1.4mm}\left\langle n{-}1 \hspace{-0.3mm};\hspace{-0.3mm} 1 \right\rangle(x)$, Lemma \ref{subtraction_inequality} yields
\begin{equation}
\begin{aligned}
   \left(\mathcal{H}_1^{(\varphi;\mathfrak{R})}\hspace{-1.4mm}\right.&\left.\left\langle n{-}1 \hspace{-0.3mm};\hspace{-0.3mm} 1 \right\rangle(k)-\mathcal{H}_1^{(\varphi;\mathfrak{R})}\hspace{-1.4mm}\left\langle n{-}1 \hspace{-0.3mm};\hspace{-0.3mm} 1 \right\rangle(i)\right)\dot{\preceq} \left(\mathcal{H}_1^{(\varphi;\mathfrak{R})}\hspace{-1.4mm}\left\langle n{-}1 \hspace{-0.3mm};\hspace{-0.3mm} 1 \right\rangle(x)-\mathcal{H}_1^{(\varphi;\mathfrak{R})}\hspace{-1.4mm}\left\langle n{-}1 \hspace{-0.3mm};\hspace{-0.3mm} 1 \right\rangle(i)\right).
\end{aligned}
\end{equation}
Accordingly, the following two EDLs are concluded. For $1\le k \le i-1$, we have:
\begin{equation}
\begin{aligned}
  \mathcal{L}_1^{(\gamma)}\langle i{-}1 \rangle =\Big[&\mathcal{H}_1^{(\varphi;\mathfrak{R})}\hspace{-1.4mm}\left\langle n{-}1 \hspace{-0.3mm};\hspace{-0.3mm} 1 \right\rangle(1){-}\mathcal{H}_1^{(\varphi;\mathfrak{R})}\hspace{-1.4mm}\left\langle n{-}1 \hspace{-0.3mm};\hspace{-0.3mm} 1 \right\rangle(i),\dots\\
  &,\mathcal{H}_1^{(\varphi;\mathfrak{R})}\hspace{-1.4mm}\left\langle n{-}1 \hspace{-0.3mm};\hspace{-0.3mm} 1 \right\rangle(k){-}\mathcal{H}_1^{(\varphi;\mathfrak{R})}\hspace{-1.4mm}\left\langle n{-}1 \hspace{-0.3mm};\hspace{-0.3mm} 1 \right\rangle(i),\dots\\
  &,\mathcal{H}_1^{(\varphi;\mathfrak{R})}\hspace{-1.4mm}\left\langle n{-}1 \hspace{-0.3mm};\hspace{-0.3mm} 1 \right\rangle(i{-}1){-}\mathcal{H}_1^{(\varphi;\mathfrak{R})}\hspace{-1.4mm}\left\langle n{-}1 \hspace{-0.3mm};\hspace{-0.3mm} 1 \right\rangle(i)\Big],
\end{aligned}
\end{equation}
and for $i+1\le k\le n$, we have
\begin{equation}
\begin{aligned}
  \mathcal{L}_2^{(\psi)}\langle n{-}i \rangle=\Big[&\mathcal{H}_1^{(\varphi;\mathfrak{R})}\hspace{-1.4mm}\left\langle n{-}1 \hspace{-0.3mm};\hspace{-0.3mm} 1 \right\rangle(i+1){-}\mathcal{H}_1^{(\varphi;\mathfrak{R})}\hspace{-1.4mm}\left\langle n{-}1 \hspace{-0.3mm};\hspace{-0.3mm} 1 \right\rangle(i),\dots\\
  &,\mathcal{H}_1^{(\varphi;\mathfrak{R})}\hspace{-1.4mm}\left\langle n{-}1 \hspace{-0.3mm};\hspace{-0.3mm} 1 \right\rangle(k){-}\mathcal{H}_1^{(\varphi;\mathfrak{R})}\hspace{-1.4mm}\left\langle n{-}1 \hspace{-0.3mm};\hspace{-0.3mm} 1 \right\rangle(i),\dots\\
  &,\mathcal{H}_1^{(\varphi;\mathfrak{R})}\hspace{-1.4mm}\left\langle n{-}1 \hspace{-0.3mm};\hspace{-0.3mm} 1 \right\rangle(n){-}\mathcal{H}_1^{(\varphi;\mathfrak{R})}\hspace{-1.4mm}\left\langle n{-}1 \hspace{-0.3mm};\hspace{-0.3mm} 1 \right\rangle(i)\Big],
\end{aligned}
\end{equation}
where $\mathcal{L}_1^{(\gamma)}\langle i{-}1 \rangle \dot{\preceq} \mathcal{L}_2^{(\psi)}\langle n{-}i \rangle$. Let EDL $\mathcal{H}_2^{(\gamma;\psi)}\hspace{-1.4mm}\left\langle i{-}1 \hspace{-0.3mm};\hspace{-0.3mm} n{-}i \right\rangle\hspace{-0.5mm}=\mathcal{L}_1^{(\gamma)}\langle i{-}1 \rangle \odot \mathcal{L}_2^{(\psi)}\langle n{-}i \rangle$. Consequently,
\begin{equation}
\begin{aligned}
   \mathcal{H}_1^{(\varphi;\mathfrak{R})}\hspace{-1.4mm}\left\langle n{-}1 \hspace{-0.3mm};\hspace{-0.3mm} 1 \right\rangle &\circleddash \mathcal{H}_1^{(\varphi;\mathfrak{R})}\hspace{-1.4mm}\left\langle n{-}1 \hspace{-0.3mm};\hspace{-0.3mm} 1 \right\rangle(i)\\
   &=\mathcal{L}_1^{(\gamma)}\langle i{-}1 \rangle \odot \mathcal{L}_2^{(\psi)}\langle n{-}i \rangle\\
   &=\mathcal{H}_2^{(\gamma;\psi)}\hspace{-1.4mm}\left\langle i{-}1 \hspace{-0.3mm};\hspace{-0.3mm} n{-}i \right\rangle\hspace{-0.5mm},
\end{aligned}
\end{equation}
which concludes the proof of \eqref{eqq1}.\par
We take a similar approach conducted above to prove \eqref{eq21}. Since $\mathcal{L}_1^{(\dot{\mathfrak{X}})}\hspace{-1.4mm}\left\langle m\right\rangle \dot{\preceq} \mathcal{H}_1^{(\varphi;\mathfrak{R})}\hspace{-1.4mm}\left\langle n{-}1 \hspace{-0.3mm};\hspace{-0.3mm} 1 \right\rangle$, if $\mathcal{H}_1^{(\varphi;\mathfrak{R})}\hspace{-1.4mm}\left\langle n{-}1 \hspace{-0.3mm};\hspace{-0.3mm} 1 \right\rangle(i)< \mathcal{L}_1^{(\dot{\mathfrak{X}})}\hspace{-1.4mm}\left\langle m\right\rangle \hspace{-0.6mm}(r)$, based on Definition \ref{expdist}, for $i=n$, $\mathcal{L}_1^{(\dot{\mathfrak{X}})}\hspace{-1.4mm}\left\langle m\right\rangle \hspace{-0.6mm}(r)-\mathcal{H}_1^{(\varphi;\mathfrak{R})}\hspace{-1.4mm}\left\langle n{-}1 \hspace{-0.3mm};\hspace{-0.3mm} 1 \right\rangle(i)$ follows $\ddot{\mathfrak{X}}$. In the case that $i\neq n$, based on Proposition \ref{absorbency_props}, we have
\begin{equation}
    \left(\mathcal{L}_1^{(\dot{\mathfrak{X}})}\hspace{-1.4mm}\left\langle m\right\rangle \hspace{-0.6mm}(r){-}\mathcal{H}_1^{(\varphi;\mathfrak{R})}\hspace{-1.4mm}\left\langle n{-}1 \hspace{-0.3mm};\hspace{-0.3mm} 1 \right\rangle(i)\right){\sim}\ddot{\mathfrak{X}}, 1\le r \le m.
\end{equation}
Further, based on Lemma \ref{subtraction_inequality}, for $1\le r< x\le n$, we have
\begin{equation}
\begin{aligned}
   \left(\mathcal{L}_1^{(\dot{\mathfrak{X}})}\hspace{-1.4mm}\left\langle m\right\rangle \hspace{-0.6mm}(r)\right.&\left.-\mathcal{H}_1^{(\varphi;\mathfrak{R})}\hspace{-1.4mm}\left\langle n{-}1 \hspace{-0.3mm};\hspace{-0.3mm} 1 \right\rangle(i)\right) \dot{\preceq} \left(\mathcal{L}_1^{(\dot{\mathfrak{X}})}\hspace{-1.4mm}\left\langle m\right\rangle \hspace{-0.6mm}(x){-}\mathcal{H}_1^{(\varphi;\mathfrak{R})}\hspace{-1.4mm}\left\langle n{-}1 \hspace{-0.3mm};\hspace{-0.3mm} 1 \right\rangle(i)\right).
\end{aligned}
\end{equation}
Consequently, we have
\begin{equation}
\begin{aligned}
  \mathcal{L}_1^{(\dot{\mathfrak{X}})}\hspace{-1.4mm}\left\langle m\right\rangle \hspace{-0.6mm} &\circleddash \mathcal{H}_1^{(\varphi;\mathfrak{R})}\hspace{-1.4mm}\left\langle n{-}1 \hspace{-0.3mm};\hspace{-0.3mm} 1 \right\rangle(i)=\mathcal{L}_2^{(\ddot{\mathfrak{X}})}\hspace{-1.4mm}\left\langle m\right\rangle\hspace{-0.5mm},
\end{aligned}
\end{equation}
where $\mathcal{L}_2^{(\ddot{\mathfrak{X}})}\hspace{-1.4mm}\left\langle m\right\rangle\hspace{-0.5mm}$ is an EDL. Moreover, Lemma \ref{subtraction_inequality} implies $\ddot{\mathrm{Y}}\langle \gamma_{i-1}| \psi_{n-i}\rangle \dot{\succeq} \mathcal{L}_2^{(\ddot{\mathfrak{X}})}\hspace{-1.4mm}\left\langle m\right\rangle\hspace{-0.5mm}$, which concludes the proof.

\subsection{Part 2: Proofs of \eqref{eq22} and \eqref{eq23}}\label{part2Prop3}
In this section, we first present the proof of \eqref{eq22}. If $\mathcal{L}_1^{(\dot{\mathfrak{X}})}\hspace{-1.4mm}\left\langle m\right\rangle \hspace{-0.6mm}(r)<\mathcal{H}_1^{(\varphi;\mathfrak{R})}\hspace{-1.4mm}\left\langle n{-}1 \hspace{-0.3mm};\hspace{-0.3mm} 1 \right\rangle(k)$ for $1\le k \le n-1$, Proposition \ref{absorbency_props} implies that
\begin{equation}
    \left(\mathcal{H}_1^{(\varphi;\mathfrak{R})}\hspace{-1.4mm}\left\langle n{-}1 \hspace{-0.3mm};\hspace{-0.3mm} 1 \right\rangle(k)-\mathcal{L}_1^{(\dot{\mathfrak{X}})}\hspace{-1.4mm}\left\langle m\right\rangle \hspace{-0.6mm}(r)\right){\sim}\varphi.
\end{equation}
Regarding $\mathcal{H}_1^{(\varphi;\mathfrak{R})}\hspace{-1.4mm}\left\langle n{-}1 \hspace{-0.3mm};\hspace{-0.3mm} 1 \right\rangle(n)$, based on Definition \ref{phidist}, $\mathcal{H}_1^{(\varphi;\mathfrak{R})}\hspace{-1.4mm}\left\langle n{-}1 \hspace{-0.3mm};\hspace{-0.3mm} 1 \right\rangle(n)-\mathcal{L}_1^{(\dot{\mathfrak{X}})}\hspace{-1.4mm}\left\langle m\right\rangle \hspace{-0.6mm}(r)$ follows $\varphi$ distribution. Moreover, based on Lemma \ref{subtraction_inequality}, for $1\le k < x \le n$, we have
\begin{equation}
\begin{aligned}
   \left(\mathcal{H}_1^{(\varphi;\mathfrak{R})}\right.&\left.\hspace{-1.4mm}\left\langle n{-}1 \hspace{-0.3mm};\hspace{-0.3mm} 1 \right\rangle(k)-\mathcal{L}_1^{(\dot{\mathfrak{X}})}\hspace{-1.4mm}\left\langle m\right\rangle \hspace{-0.6mm}(r)\right)\dot{\preceq} \left(\mathcal{H}_1^{(\varphi;\mathfrak{R})}\hspace{-1.4mm}\left\langle n{-}1 \hspace{-0.3mm};\hspace{-0.3mm} 1 \right\rangle(x)-\mathcal{L}_1^{(\dot{\mathfrak{X}})}\hspace{-1.4mm}\left\langle m\right\rangle \hspace{-0.6mm}(r)\right).
\end{aligned}
\end{equation}
Accordingly, we have
\begin{equation}
\begin{aligned}
  \mathcal{H}_1^{(\varphi;\mathfrak{R})}&\hspace{-1.4mm}\left\langle n{-}1 \hspace{-0.3mm};\hspace{-0.3mm} 1 \right\rangle \circleddash \mathcal{L}_1^{(\dot{\mathfrak{X}})}\hspace{-1.4mm}\left\langle m\right\rangle \hspace{-0.6mm}(r)= \mathcal{L}_3^{(\varphi)}\hspace{-1.4mm}\left\langle n\right\rangle\hspace{-0.6mm},
\end{aligned}
\end{equation}
where $\mathcal{L}_3^{(\varphi)}\hspace{-1.4mm}\left\langle n\right\rangle\hspace{-0.6mm}$ is an EDL.\par
Next, we aim to prove \eqref{eq23}. It can be shown that for $1\le k \le m$ and $k\ne r$, we have
\begin{equation}
    \left(\mathcal{L}_1^{(\dot{\mathfrak{X}})}\hspace{-1.4mm}\left\langle m\right\rangle \hspace{-0.6mm}(k)-\mathcal{L}_1^{(\dot{\mathfrak{X}})}\hspace{-1.4mm}\left\langle m\right\rangle \hspace{-0.6mm}(r)\right){\sim}\dot{\mathfrak{X}}.
\end{equation}
Further, based on Lemma \ref{subtraction_inequality}, for $1\le k< x\le n$, we have
\begin{equation}
\begin{aligned}
   \left(\mathcal{L}_1^{(\dot{\mathfrak{X}})}\hspace{-1.4mm}\left\langle m\right\rangle \hspace{-0.6mm}(k) \right.&\left.-\mathcal{L}_1^{(\dot{\mathfrak{X}})}\hspace{-1.4mm}\left\langle m\right\rangle \hspace{-0.6mm}(r)\right) \dot{\preceq} \left(\mathcal{L}_1^{(\dot{\mathfrak{X}})}\hspace{-1.4mm}\left\langle m\right\rangle \hspace{-0.6mm}(x)-\mathcal{L}_1^{(\dot{\mathfrak{X}})}\hspace{-1.4mm}\left\langle m\right\rangle \hspace{-0.6mm}(r)\right),
   \end{aligned}
\end{equation}
where $k\ne x \ne r$. Accordingly,
\begin{equation}
\begin{aligned}
  \mathcal{L}_1^{(\dot{\mathfrak{X}})}\hspace{-1.4mm}\left\langle m\right\rangle \hspace{-0.6mm} &- \mathcal{L}_1^{(\dot{\mathfrak{X}})}\hspace{-1.4mm}\left\langle m\right\rangle \hspace{-0.6mm}(r)=\mathcal{L}_4^{(\dot{\mathfrak{X}})}\hspace{-1.4mm}\left\langle m{-}1\right\rangle\hspace{-0.6mm},
\end{aligned}
\end{equation}
where $\mathcal{L}_4^{(\dot{\mathfrak{X}})}\hspace{-1.4mm}\left\langle m{-}1\right\rangle\hspace{-0.6mm}$ is an EDL. Moreover, based on Lemma \ref{subtraction_inequality}, we have
\begin{equation}
 \mathcal{L}_3^{(\varphi)}\hspace{-1.4mm}\left\langle n\right\rangle\hspace{-0.6mm} \dot{\succeq} \mathcal{L}_4^{(\dot{\mathfrak{X}})}\hspace{-1.4mm}\left\langle m{-}1\right\rangle\hspace{-0.6mm},
\end{equation}
which concludes the proof.

\newpage
\section{Proof of Proposition \ref{proposition4}}\label{proposition4_proof}
\noindent In this section, we first prove \eqref{pro4_1} and \eqref{pro4_2} of Proposition \ref{proposition4} in Sec. \ref{part1Prop4}. Next, we prove \eqref{pro4_3} and \eqref{pro4_4} of Proposition \ref{proposition4} in Sec. \ref{part2Prop4}.
\subsection{Part 1: the proofs of \eqref{pro4_1} and \eqref{pro4_2}}\label{part1Prop4}
We first aim to present the proof of \eqref{pro4_1}. Based on Definition \ref{event_dynamic_list}, for $1\le k\le i-1$, we have
\begin{equation}
\begin{aligned}
  \mathcal{H}_1^{(\gamma;\psi)}\hspace{-1.4mm}\left\langle n \hspace{-0.3mm};\hspace{-0.3mm} m \right\rangle(k)\dot{\preceq}\mathcal{H}_1^{(\gamma;\psi)}\hspace{-1.4mm}\left\langle n \hspace{-0.3mm};\hspace{-0.3mm} m \right\rangle(i),
\end{aligned}
\end{equation}
and for $i+1\le k\le n+m$, we have
\begin{equation}
\begin{aligned}
   \mathcal{H}_1^{(\gamma;\psi)}\hspace{-1.4mm}\left\langle n \hspace{-0.3mm};\hspace{-0.3mm} m \right\rangle(i)\dot{\preceq} \mathcal{H}_1^{(\gamma;\psi)}\hspace{-1.4mm}\left\langle n \hspace{-0.3mm};\hspace{-0.3mm} m \right\rangle(k).
\end{aligned}
\end{equation}
If $\mathcal{H}_1^{(\gamma;\psi)}\hspace{-1.4mm}\left\langle n \hspace{-0.3mm};\hspace{-0.3mm} m \right\rangle(i)<\mathcal{H}_1^{(\gamma;\psi)}\hspace{-1.4mm}\left\langle n \hspace{-0.3mm};\hspace{-0.3mm} m \right\rangle(k)$, based on Proposition~\ref{decomposability_props}, for $1\le k \le i-1$, we have
\begin{equation}
\begin{aligned}
    \left(\mathcal{H}_1^{(\gamma;\psi)}\hspace{-1.4mm}\left\langle n \hspace{-0.3mm};\hspace{-0.3mm} m \right\rangle(k)-\mathcal{H}_1^{(\gamma;\psi)}\hspace{-1.4mm}\left\langle n \hspace{-0.3mm};\hspace{-0.3mm} m \right\rangle(i)\right){\sim}\gamma,
\end{aligned}
\end{equation}
and for $i+1\le k \le n+m$, we have
\begin{equation}
\begin{aligned}
    \left(\mathcal{H}_1^{(\gamma;\psi)}\hspace{-1.4mm}\left\langle n \hspace{-0.3mm};\hspace{-0.3mm} m \right\rangle(k)-\mathcal{H}_1^{(\gamma;\psi)}\hspace{-1.4mm}\left\langle n \hspace{-0.3mm};\hspace{-0.3mm} m \right\rangle(i)\right){\sim}\psi.
\end{aligned}
\end{equation}
Moreover, based on Lemma \ref{subtraction_inequality}, for $1\le k< x\le n+m$
\begin{equation}
\begin{aligned}
   \left(\mathcal{H}_1^{(\gamma;\psi)}\hspace{-1.4mm}\left\langle n \hspace{-0.3mm};\hspace{-0.3mm} m \right\rangle(k) -\mathcal{H}_1^{(\gamma;\psi)}\hspace{-1.4mm}\left\langle n \hspace{-0.3mm};\hspace{-0.3mm} m \right\rangle(i)\right)\dot{\preceq} \left(\mathcal{H}_1^{(\gamma;\psi)}\hspace{-1.4mm}\left\langle n \hspace{-0.3mm};\hspace{-0.3mm} m \right\rangle(x)-\mathcal{H}_1^{(\gamma;\psi)}\hspace{-1.4mm}\left\langle n \hspace{-0.3mm};\hspace{-0.3mm} m \right\rangle(i)\right),
\end{aligned}
\end{equation}
where $x\neq k \neq i$. Accordingly, the following two EDLs are concluded. For $1\le k \le i-1$, we have
\begin{equation}
\begin{aligned}
  \mathcal{L}_1^{(\gamma)}\hspace{-1.4mm}\left\langle i{-}1\right\rangle =\Big[&\mathcal{H}_1^{(\gamma;\psi)}\hspace{-1.4mm}\left\langle n \hspace{-0.3mm};\hspace{-0.3mm} m \right\rangle(1)-\mathcal{H}_1^{(\gamma;\psi)}\hspace{-1.4mm}\left\langle n \hspace{-0.3mm};\hspace{-0.3mm} m \right\rangle(i)\dots\\
  &,\mathcal{H}_1^{(\gamma;\psi)}\hspace{-1.4mm}\left\langle n \hspace{-0.3mm};\hspace{-0.3mm} m \right\rangle(k)-\mathcal{H}_1^{(\gamma;\psi)}\hspace{-1.4mm}\left\langle n \hspace{-0.3mm};\hspace{-0.3mm} m \right\rangle(i)\dots\\
  &,\mathcal{H}_1^{(\gamma;\psi)}\hspace{-1.4mm}\left\langle n \hspace{-0.3mm};\hspace{-0.3mm} m \right\rangle(i{-}1)-\mathcal{H}_1^{(\gamma;\psi)}\hspace{-1.4mm}\left\langle n \hspace{-0.3mm};\hspace{-0.3mm} m \right\rangle(i)\Big],
\end{aligned}
\end{equation}
and for $i+1\le k\le n+m$, we have
\begin{equation}
\begin{aligned}
   \mathcal{L}_1^{(\psi)}\hspace{-1.4mm}\left\langle n{+}m{-}i\right\rangle =\Big[&\mathcal{H}_1^{(\gamma;\psi)}\hspace{-1.4mm}\left\langle n \hspace{-0.3mm};\hspace{-0.3mm} m \right\rangle(i{+}1)-\mathcal{H}_1^{(\gamma;\psi)}\hspace{-1.4mm}\left\langle n \hspace{-0.3mm};\hspace{-0.3mm} m \right\rangle(i),\dots\\
  &,\mathcal{H}_1^{(\gamma;\psi)}\hspace{-1.4mm}\left\langle n \hspace{-0.3mm};\hspace{-0.3mm} m \right\rangle(k)-\mathcal{H}_1^{(\gamma;\psi)}\hspace{-1.4mm}\left\langle n \hspace{-0.3mm};\hspace{-0.3mm} m \right\rangle(i),\dots\\
  &,\mathcal{H}_1^{(\gamma;\psi)}\hspace{-1.4mm}\left\langle n \hspace{-0.3mm};\hspace{-0.3mm} m \right\rangle(n{+}m)-\mathcal{H}_1^{(\gamma;\psi)}\hspace{-1.4mm}\left\langle n \hspace{-0.3mm};\hspace{-0.3mm} m \right\rangle(i)\Big],
\end{aligned}
\end{equation}
where, $ \mathcal{L}_1^{(\gamma)}\hspace{-1.4mm}\left\langle i{-}1\right\rangle \dot{\preceq}  \mathcal{L}_1^{(\psi)}\hspace{-1.4mm}\left\langle n{+}m{-}i\right\rangle$. Let $\mathcal{H}_2^{(\gamma;\psi)} \hspace{-1.4mm}\left\langle i{-}1 \hspace{-0.3mm};\hspace{-0.3mm} m{+}n{-}i \right\rangle=\mathcal{L}_1^{(\gamma)}\hspace{-1.4mm}\left\langle i{-}1\right\rangle \odot  \mathcal{L}_1^{(\psi)}\hspace{-1.4mm}\left\langle n{+}m{-}i\right\rangle$. Consequently,
\begin{equation}
\begin{aligned}
   \mathcal{H}_1^{(\gamma;\psi)}\hspace{-1.4mm}\left\langle n \hspace{-0.3mm};\hspace{-0.3mm} m \right\rangle &\circleddash \mathcal{H}_1^{(\gamma;\psi)}\hspace{-1.4mm}\left\langle n \hspace{-0.3mm};\hspace{-0.3mm} m \right\rangle(i)\\
   &=\mathcal{L}_1^{(\gamma)}\hspace{-1.4mm}\left\langle i{-}1\right\rangle \odot  \mathcal{L}_1^{(\psi)}\hspace{-1.4mm}\left\langle n{+}m{-}i\right\rangle\\
   &=\mathcal{H}_2^{(\gamma;\psi)} \hspace{-1.4mm}\left\langle i{-}1 \hspace{-0.3mm};\hspace{-0.3mm} m{+}n{-}i \right\rangle,
\end{aligned}
\end{equation}
which concludes the proof of \eqref{pro4_1}.\par
We next aim to prove \eqref{pro4_2}. If $\mathcal{H}_1^{(\gamma;\psi)}\hspace{-1.4mm}\left\langle n \hspace{-0.3mm};\hspace{-0.3mm} m \right\rangle(i)< \mathcal{L}_1^{(\ddot{\mathfrak{X}})}\hspace{-1.4mm}\left\langle h\right\rangle(r)$, based on Proposition \ref{absorbency_props}, for $1\le r \le h$, we have
\begin{equation}
    \left(\mathcal{L}_1^{(\ddot{\mathfrak{X}})}\hspace{-1.4mm}\left\langle h\right\rangle(r)-\mathcal{H}_1^{(\gamma;\psi)}\hspace{-1.4mm}\left\langle n \hspace{-0.3mm};\hspace{-0.3mm} m \right\rangle(i)\right){\sim}\ddot{\mathfrak{X}}.
\end{equation}
Based on Lemma \ref{subtraction_inequality}, for $1\le r< x\le h$, we have
\begin{equation}
\begin{aligned}
   \left(\mathcal{L}_1^{(\ddot{\mathfrak{X}})}\hspace{-1.4mm}\left\langle h\right\rangle(r)-\mathcal{H}_1^{(\gamma;\psi)}\hspace{-1.4mm}\left\langle n \hspace{-0.3mm};\hspace{-0.3mm} m \right\rangle(i)\right) \dot{\preceq} \left(\mathcal{L}_1^{(\ddot{\mathfrak{X}})}\hspace{-1.4mm}\left\langle h\right\rangle(x)-\mathcal{H}_1^{(\gamma;\psi)}\hspace{-1.4mm}\left\langle n \hspace{-0.3mm};\hspace{-0.3mm} m \right\rangle(i)\right).
\end{aligned}
\end{equation}
Accordingly,
\begin{equation}
\begin{aligned}
  \mathcal{L}_1^{(\ddot{\mathfrak{X}})}\hspace{-1.4mm}\left\langle h\right\rangle \circleddash \mathcal{H}_1^{(\gamma;\psi)}\hspace{-1.4mm}\left\langle n \hspace{-0.3mm};\hspace{-0.3mm} m \right\rangle(i)=\mathcal{L}_2^{(\ddot{\mathfrak{X}})}\hspace{-1.4mm}\left\langle h\right\rangle,
\end{aligned}
\end{equation}
where $\mathcal{L}_2^{(\ddot{\mathfrak{X}})}\hspace{-1.4mm}\left\langle h\right\rangle$ is an EDL. Moreover, Lemma \ref{subtraction_inequality} yields
\begin{equation}
  \mathcal{H}_2^{(\gamma;\psi)} \hspace{-1.4mm}\left\langle i{-}1 \hspace{-0.3mm};\hspace{-0.3mm} m{+}n{-}i \right\rangle \dot{\succeq} \mathcal{L}_2^{(\ddot{\mathfrak{X}})}\hspace{-1.4mm}\left\langle h\right\rangle,
\end{equation}
which conclude the proof of part 1.

\subsection{Part 2: the proofs of \eqref{pro4_3} and \eqref{pro4_4}}\label{part2Prop4}
In this section, we first present the proof of \eqref{pro4_3}. If $\mathcal{H}_1^{(\gamma;\psi)}\hspace{-1.4mm}\left\langle n \hspace{-0.3mm};\hspace{-0.3mm} m \right\rangle(k)<\mathcal{L}_1^{(\ddot{\mathfrak{X}})}\hspace{-1.4mm}\left\langle h\right\rangle(r)$, based on Proposition \ref{absorbency_props}, for $1\le k \le n+m$, we get
\begin{equation}
    \left(\mathcal{H}_1^{(\gamma;\psi)}\hspace{-1.4mm}\left\langle n \hspace{-0.3mm};\hspace{-0.3mm} m \right\rangle(k)-\mathcal{L}_1^{(\ddot{\mathfrak{X}})}\hspace{-1.4mm}\left\langle h\right\rangle(s)\right){\sim}\varphi.
\end{equation}
Moreover, based on Lemma \ref{subtraction_inequality}, for $1\le k< x\le n+m$, the following condition holds
\begin{equation}
 \hspace{-4mm}
\begin{aligned}
   \left(\mathcal{H}_1^{(\gamma;\psi)}\hspace{-1.4mm}\left\langle n \hspace{-0.3mm};\hspace{-0.3mm} m \right\rangle(k) \right.&\left.-\mathcal{L}_1^{(\ddot{\mathfrak{X}})}\hspace{-1.4mm}\left\langle h\right\rangle(r)\right) \dot{\preceq} \left(\mathcal{H}_1^{(\gamma;\psi)}\hspace{-1.4mm}\left\langle n \hspace{-0.3mm};\hspace{-0.3mm} m \right\rangle(x)-\mathcal{L}_1^{(\ddot{\mathfrak{X}})}\hspace{-1.4mm}\left\langle h\right\rangle(r)\right).
\end{aligned}
\hspace{-4mm}
\end{equation}
Accordingly, we get
\begin{equation}
\begin{aligned}
  \mathcal{H}_1^{(\gamma;\psi)}\hspace{-1.4mm}\left\langle n \hspace{-0.3mm};\hspace{-0.3mm} m \right\rangle &\circleddash \mathcal{L}_1^{(\ddot{\mathfrak{X}})}\hspace{-1.4mm}\left\langle h\right\rangle(r)= \mathcal{L}_3^{(\varphi)}\hspace{-1.4mm}\left\langle m+n\right\rangle,
\end{aligned}
\end{equation}
which concludes the proof of \eqref{pro4_3}.\par
We next present the proof of \eqref{pro4_4}. Based on Definition \ref{event_dynamic_list}, for $1\le r,s\le h$ and $r\ne s$, we have
\begin{equation}
  \mathcal{L}_1^{(\ddot{\mathfrak{X}})}\hspace{-1.4mm}\left\langle h\right\rangle(s)\dot{\preceq}\mathcal{L}_1^{(\ddot{\mathfrak{X}})}\hspace{-1.4mm}\left\langle h\right\rangle(r).
\end{equation}
If $\mathcal{L}_1^{(\ddot{\mathfrak{X}})}\hspace{-1.4mm}\left\langle h\right\rangle(r)<\mathcal{L}_1^{(\ddot{\mathfrak{X}})}\hspace{-1.4mm}\left\langle h\right\rangle(s)$, based on Proposition \ref{absorbency_props}, for $1\le s \le h$, we have
\begin{equation}
    \left(\mathcal{L}_1^{(\ddot{\mathfrak{X}})}\hspace{-1.4mm}\left\langle h\right\rangle(s)-\mathcal{L}_1^{(\ddot{\mathfrak{X}})}\hspace{-1.4mm}\left\langle h\right\rangle(r)\right){\sim}\dot{\mathfrak{X}}.
\end{equation}
Furthermore, for $1\le s< x\le h$, Lemma \ref{subtraction_inequality} yields
\begin{equation}
\begin{aligned}
   \left(\mathcal{L}_1^{(\ddot{\mathfrak{X}})}\hspace{-1.4mm}\left\langle h\right\rangle(s)\right.&\left.-\mathcal{L}_1^{(\ddot{\mathfrak{X}})}\hspace{-1.4mm}\left\langle h\right\rangle(r)\right) \dot{\preceq} \left(\mathcal{L}_1^{(\ddot{\mathfrak{X}})}\hspace{-1.4mm}\left\langle h\right\rangle(x)-\mathcal{L}_1^{(\ddot{\mathfrak{X}})}\hspace{-1.4mm}\left\langle h\right\rangle(r)\right).
\end{aligned}
\end{equation}
Accordingly, we have
\begin{equation}
\begin{aligned}
  \mathcal{L}_1^{(\ddot{\mathfrak{X}})}\hspace{-1.4mm}\left\langle h\right\rangle \circleddash \mathcal{L}_1^{(\ddot{\mathfrak{X}})}\hspace{-1.4mm}\left\langle h\right\rangle(r) =\mathcal{L}_4^{(\dot{\mathfrak{X}})}\hspace{-1.4mm}\left\langle h-1\right\rangle.
\end{aligned}
\end{equation}
Moreover, Lemma \ref{subtraction_inequality} implies that
\begin{equation}
  \mathcal{L}_3^{(\varphi)}\hspace{-1.4mm}\left\langle m+n\right\rangle \dot{\succeq} \mathcal{L}_4^{(\dot{\mathfrak{X}})}\hspace{-1.4mm}\left\langle h-1\right\rangle,
\end{equation}
which concludes the proof.

\newpage
\section{Proof of Theorem \ref{decompositionTheorem}}\label{decompositionTheorem_proof}
\noindent We prove the theorem by induction.  The proof is presented in the following steps:
\begin{enumerate}
    \item \textbf{Initial case (base case).} We prove that Decomposition Theorem holds for the initial states $S_0$ and $S_1$ of $\beta$-SMP.
    \item \textbf{Induction step.} We prove if Decomposition Theorem holds for state $S_{n-1}$, it holds for state $S_{n
    }$.
\end{enumerate}

\subsection{Initial case}
Consider $\beta$-SMP presented in Fig. \ref{j2nMarkov1}. Let {\small$\dot{\mathcal{C}}_i$} denote the set of $i$ vehicles processing $i$ different groups, where one of the vehicles of each group has departed the VC (i.e., each group is being processed by only one vehicle).
Therefore, each vehicle {\small$C_\ell\in \dot{\mathcal{C}}_i$} is a recruiter vehicle. Also, let $\ddot{\mathcal{C}}_{2n-2i}$ denote the set of $2n-2i$ vehicles processing $n-i$ different groups, where each group is being processed by two vehicles. Moreover, let $X\in\mathcal{S}$ and $X'\in\mathcal{S}$ denote the current state and the next state of $\beta$-SMP, respectively. Further, let $e_X\in\zeta$ denote the event that occurred at current state $X$. The process starts from state $S_0{=}\dot{\mathcal{C}}_0 \cup \ddot{\mathcal{C}}_{2n}$. Since none of the vehicles has departed the VC yet and there are no recruiters in this state. Decomposition Theorem implies that
\begin{equation}\label{S00Prov0}
 S_0\equiv
    \begin{cases}
     S_{0,0}& \text{initial state}.\\
    \end{cases}
\end{equation}
\begin{equation}\label{S00Prov}
 X'{=}S_0\equiv
    \begin{cases}
     S_{0,1} & \text{if  } X{=}S_{1},e_X{=}e^{\mathsf{R}}(C_{\ell}),~ C_{\ell} {\in} \dot{\mathcal{C}}_{1},
    \end{cases}
\end{equation}
and
\begin{equation}\label{S01Prov}
   X'{=}S_1\equiv
    \begin{cases}
        S_{1,0} & \text{if  } X{=}S_{0},e_X{=}e^{\mathsf{D}}(C_{\ell}),~ \forall C_{\ell} {\in} \ddot{\mathcal{C}}_{2n}. \\
        S_{1,1} & \text{if  } X{=}S_{2},e_X{=}e^{\mathsf{R}}(C_{\ell}),~ O(C_{\ell}){=}1, C_{\ell} {\in} \dot{\mathcal{C}}_{2}\\
        S_{1,2} &\text{if  } X{=}S_{2},e_X{=}e^{\mathsf{R}}(C_{\ell}),~ O(C_{\ell}){=}2, C_{\ell} {\in} \dot{\mathcal{C}}_{2},
    \end{cases}
\end{equation}
where\footnote{$[]$ denotes an empty list.}
\begin{equation}\label{S00Def}
    S_{0,0}=\left\{[Z_1,Z_2,\dots,Z_{2n}], [\,] \right\},
\end{equation}
\begin{equation}\label{S01Def}
    S_{0,1}=\left\{\mathcal{L}_2^{(\ddot{\mathfrak{X}})}\hspace{-1.4mm}\left\langle 2n\right\rangle,[\,]\right\},
\end{equation}
\begin{equation}\label{S10Def}
      S_{1,0}=\left\{\mathcal{L}_1^{(\dot{\mathfrak{X}})}\hspace{-1.4mm}\left\langle {2n{-}1}\right\rangle,\mathcal{H}_1^{(\varphi;\mathfrak{R})}\hspace{-1.4mm}\left\langle {0} \hspace{-0.3mm};\hspace{-0.3mm} 1 \right\rangle\right\},
\end{equation}
\begin{equation}\label{S11Def}
      S_{1,1}=\left\{\mathcal{L}_2^{(\ddot{\mathfrak{X}})}\hspace{-1.4mm}\left\langle 2n{-}1\right\rangle,\mathcal{H}_2^{(\gamma;\psi)}\hspace{-1.4mm}\left\langle 0 \hspace{-0.3mm};\hspace{-0.3mm} 1 \right\rangle\right\},
\end{equation}
and
\begin{equation}\label{S12Def}
      S_{1,2}=\left\{\mathcal{L}_2^{(\ddot{\mathfrak{X}})}\hspace{-1.4mm}\left\langle 2n{-}1\right\rangle,\mathcal{H}_2^{(\gamma;\psi)}\hspace{-1.4mm}\left\langle 1 \hspace{-0.3mm};\hspace{-0.3mm} 0 \right\rangle\right\}.
\end{equation}
Accordingly, we aim to prove that \eqref{S00Prov} and \eqref{S01Prov} hold. To this end, we must prove following cases:
\begin{enumerate}[label={(I-\arabic*)}]
    \item \label{I1} $\beta$-SMP will transit to state $X'\in \{S_1\equiv S_{1,0}\}$, if it is in the initial state $X\in \{S_0\equiv S_{0,0}\}$ and event $e_X=e^{\mathsf{D}}(C_{\ell_k})$ occurs, where $C_{\ell_k} {\in} \ddot{\mathcal{C}}_{2n}$.
    \item \label{I2}  $\beta$-SMP will transit to state $X'\in \{S_2 \equiv S_{2,0}\}$, if it is in state $X\in \{S_1 \equiv S_{1,0}\}$ and event $e_X=e^{\mathsf{D}}(C_{\ell_k})$ occurs, where $C_{\ell_k} {\in} \ddot{\mathcal{C}}_{2n-2}$.
    \item \label{I3} $\beta$-SMP will transit to state $X'\in\{S_1 \equiv S_{1,1}, S_1 \equiv S_{1,2}\}$, if it is in state $X\in \{S_2 \equiv S_{2,0}\}$ and event $e_X=e^{\mathsf{R}}(C_{\ell'_k})$ occurs, where $C_{\ell'_k} {\in} \dot{\mathcal{C}}_{2}$.
    \item \label{I4} $\beta$-SMP will transit to state $X'\in \{S_2 \equiv S_{2,0}\}$, if it is in state $X\in \{S_1 \equiv S_{1,1},S_1 \equiv S_{1,2}\}$ and event $e_X=e^{\mathsf{D}}(C_{\ell_k})$ occurs, where $C_{\ell_k} {\in} \ddot{\mathcal{C}}_{2n-2}$.
    \item \label{I5} $\beta$-SMP will transit to state $X'\in \{S_0 \equiv S_{0,1}\}$, if it is in state $X\in \{S_1\equiv S_{1,0}, S_1 \equiv S_{1,1},S_1 \equiv S_{1,2}\}$ and event $e_X=e^{\mathsf{R}}(C_{\ell'_k})$ occurs, where $C_{\ell'_k} {\in} \dot{\mathcal{C}}_{1}$.
    \item \label{I6} $\beta$-SMP will transit to state $X'\in\{S_1 \equiv S_{1,0}\}$, if it is in state $X\in \{S_0 \equiv S_{0,1}\}$ and event $e_X=e^{\mathsf{D}}(C_{\ell_k})$ occurs, where $C_{\ell_k} {\in} \ddot{\mathcal{C}}_{2n}$.
\end{enumerate}
In the following, we prove the above six cases one by one.\par
\textbf{Proof of case \ref{I1}.}
It is evident that when the process is in the initial state, we have
\begin{equation}\label{eqS01_144}
    S_0\equiv S_{0,0}.
\end{equation}
Let $\mathbb{Z}=\left[Z_{\ell_1},Z_{\ell_2},\dots\right]$ be the list of all random variables denoting the sojourn times of the vehicles. Moreover, let $\mathbb{Z}'=\left[Z'_{\ell_1},Z'_{\ell_2},\dots\right]$ be the set of all random variables denoting the residual sojourn times of the vehicles. Further, let list $\widehat{C}$ denote all of the candidate vehicles which are in the VC when the processing of application $\mathcal{A}$ has been started. As mentioned in Sec.~\ref{SeSVC}, we assume $\widehat{C}$ is large enough such that the number of vehicles is larger than the number of existing task groups in the system. Initially, we have $Z'_{\ell_x}{=}Z_{\ell_x}$. Let {\small$C_{\ell_k}$} and {\small$C_{\ell'_k}$} be two vehicles assigned to a group {\small $G_{k}$}.
In state $S_0$, the departure of a vehicle $C_{\ell_k} {\in} \ddot{\mathcal{C}}_{2n}$ is the only possible event that can occur. Assume that event $e_X=e^{\mathsf{D}}(C_{\ell_k})$ occurs. $\beta$-SMP transits to state $X'{=}S_1{=}\dot{\mathcal{C}}_1 \cup \ddot{\mathcal{C}}_{2n-2}$ and {\small$C_{\ell'_k}$} starts recruiting a new vehicle. Consequently, $Z_{\ell_k}$ (i.e., the sojourn time of {\small$C_{\ell_k}$}) should be subtracted from the residual sojourn times of other vehicles since $Z_{\ell_k}$ units of time have passed from the start of the processing $\mathcal{A}$. Hence, $Z'_{\ell_x}=Z_{\ell_x}-Z_{\ell_k}$, where ${\ell_x} \neq {\ell_k}$, and $Z_{\ell_x}>Z_{\ell_k}$. Therefore, we have
\begin{equation}
\begin{aligned}
 \mathbb{Z}'{=}\left[Z'_{\ell_1}{=}Z_{\ell_1}{-}Z_{\ell_k}, \dots, Z'_{\ell_x}{=}Z_{\ell_x}{-}Z_{\ell_k},\dots\right],
\end{aligned}
\end{equation}
where $|\mathbb{Z}'|=2n-1$ and
\begin{equation}
\mathbb{Z}'(i)\dot{\prec} \mathbb{Z}'(j), 1\le i<j\le 2n-1.
\end{equation}
Accordingly, based on Definition \ref{expdist}, $\mathbb{Z}'= \mathcal{L}_2^{(\dot{\mathfrak{X}})}\hspace{-1.4mm}\left\langle {2n{-}1}\right\rangle$. Similarly, regarding $C_{\ell_y}\in\widehat{C}$, we can show the residual sojourn time of vehicles of $\widehat{C}$ by EDL $\mathcal{L}_{\widehat{C}}^{(\dot{\mathfrak{X}})}\hspace{-1.4mm}\left\langle {\infty}\right\rangle$, where $\infty$ means that the list is large enough, capturing the assumption that the number of vehicles is larger than the number of task groups.
Furthermore, the residual recruitment time of the only recruiter is $U'_{\ell'_k}=U_{\ell'_k}$, where $U_{\ell'_k}$ is the recruitment duration of vehicle $C_{\ell'_k}$ following general distribution $\mathfrak{R}$. Let $\mathcal{H}_1^{(\varphi;\mathfrak{R})}\hspace{-1.4mm}\left\langle {0} \hspace{-0.3mm};\hspace{-0.3mm} 0 \right\rangle$ be an empty H-EDL. Considering Definition \ref{concatentationDef}, we can write
\begin{equation}
    \mathcal{H}_2^{(\varphi;\mathfrak{R})}\hspace{-1.4mm}\left\langle {0} \hspace{-0.1mm};\hspace{-0.3mm} 1 \right\rangle=\mathcal{H}_1^{(\varphi;\mathfrak{R})}\hspace{-1.4mm}\left\langle {0} \hspace{-0.1mm};\hspace{-0.3mm} 0 \right\rangle \odot U'_{\ell'_k}.
\end{equation}
Consequently, if $X{=}S_0$ and $e_X{=}e^{\mathsf{D}}(C_{\ell_k})$, where $C_{\ell_k} {\in} \ddot{\mathcal{C}}_{2n}$, we have
\begin{equation}\label{S10_148}
\begin{aligned}
    X'{=}S_1\equiv S_{1,0}=\Big\{\mathcal{L}_2^{(\dot{\mathfrak{X}})}\hspace{-1.4mm}\left\langle {2n{-}1}\right\rangle,\mathcal{H}_2^{(\varphi;\mathfrak{R})}\hspace{-1.4mm}\left\langle {0} \hspace{-0.1mm};\hspace{-0.3mm} 1 \right\rangle\Big\}.
\end{aligned}
\end{equation}

In the following, we generalize the technique conducted above to conclude the proof of the initial case of induction. Let $S_{i,j}^{\mathsf{Soj}}$ and $S_{i,j}^{\mathsf{Rec}}$ refer to the first and second elements of $S_{i,j}$, respectively. Regarding \eqref{S10_148}, we have $S_{1,0}^{\mathsf{Soj}}=\mathcal{L}_2^{(\ddot{\mathfrak{X}})}\hspace{-1.4mm}\left\langle 2n-1\right\rangle$ and $S_{1,0}^{\mathsf{Rec}}=\mathcal{H}_1^{(\varphi;\mathfrak{R})}\hspace{-1.4mm}\left\langle {0} \hspace{-0.3mm};\hspace{-0.3mm} 1 \right\rangle$. Accordingly, let $C_{\ell_k}$ denote the vehicle whose residual sojourn time is $S_{i,j}^{\mathsf{Soj}}(\ell_k)$. Also, let $C_{\ell'_k}$ denote the recruiter vehicle whose residual recruitment duration is $S_{i,j}^{\mathsf{Rec}}(k)$. Based on Definition \ref{event_dynamic_list},
\begin{equation}
    S_{i,j}^{\mathsf{Rec}}(r)\dot{\preceq} S_{i,j}^{\mathsf{Rec}}(k), r < k.
\end{equation}
As a result, the order of recruiter $C_{\ell'_k}$ is equal to $k$. Mathematically, we have
\begin{equation}\label{eqOrder}
    O(C_{\ell'_k})=\sum_{j=1, j\neq k}^{i} \mathds{1}_{C_{\ell_j} \preceq C_{\ell'_k}}=k.
\end{equation}
Accordingly, we take the following steps to prove the initial case of induction.\par
\textbf{Proof of case \ref{I2}.} Let current state be $X=S_1$. If event $e_X=e^{\mathsf{D}}(C_{\ell_k})$ occurs, $\beta$-SMP will transit to state $X'{=}S_2{=}\dot{\mathcal{C}}_2 \cup \ddot{\mathcal{C}}_{2n-4}$, and {\small$C_{\ell'_k}$} starts recruiting a new vehicle. Therefore, $\mathcal{L}_2^{(\dot{\mathfrak{X}})}\hspace{-1.4mm}\left\langle 2n{-}1\right\rangle(\ell_k)$ (the residual sojourn time of $C_{\ell_k}$) must be subtracted from other random variables (i.e., $\mathcal{L}_2^{(\dot{\mathfrak{X}})}\hspace{-1.4mm}\left\langle 2n{-}1\right\rangle$, $\mathcal{L}_{\widehat{C}}^{(\dot{\mathfrak{X}})}\hspace{-1.4mm}\left\langle {\infty}\right\rangle$, and $\mathcal{H}_2^{(\varphi;\mathfrak{R})}\hspace{-1.4mm}\left\langle {0} \hspace{-0.3mm};\hspace{-0.3mm} 1 \right\rangle$). To this end, using Proposition \ref{proposition3}, we have
\begin{equation}
    \mathcal{L}_3^{(\dot{\mathfrak{X}})}\hspace{-1.4mm}\left\langle 2n{-}2\right\rangle = \mathcal{L}_2^{(\dot{\mathfrak{X}})}\hspace{-1.4mm}\left\langle 2n{-}1\right\rangle \circleddash \mathcal{L}_2^{(\dot{\mathfrak{X}})}\hspace{-1.4mm}\left\langle 2n{-}1\right\rangle(\ell_k),
\end{equation}
and
\begin{equation}
    \mathcal{L}_{\widehat{C}}^{(\dot{\mathfrak{X}})}\hspace{-1.4mm}\left\langle {\infty}\right\rangle = \mathcal{L}_{\widehat{C}}^{(\dot{\mathfrak{X}})}\hspace{-1.4mm}\left\langle {\infty}\right\rangle \circleddash \mathcal{L}_2^{(\dot{\mathfrak{X}})}\hspace{-1.4mm}\left\langle 2n{-}1\right\rangle(\ell_k),
\end{equation}
For simplicity, from here onward, we omit writing the variation of residual sojourn time of vehicles in $\widehat{C}$ (i.e, $\mathcal{L}_{\widehat{C}}^{(\dot{\mathfrak{X}})}\hspace{-1.4mm}\left\langle {\infty}\right\rangle$) since they have the same behavior of recruited vehicles. Further, Proposition~\ref{proposition3} yields
\begin{equation}
    \mathcal{H}_3^{(\varphi;\mathfrak{R})}\hspace{-1.4mm}\left\langle 1 \hspace{-0.3mm};\hspace{-0.3mm} 1 \right\rangle=\left(\mathcal{H}_2^{(\varphi;\mathfrak{R})}\hspace{-1.4mm}\left\langle {0} \hspace{-0.3mm};\hspace{-0.3mm} 1 \right\rangle \circleddash \mathcal{L}_2^{(\dot{\mathfrak{X}})}\hspace{-1.4mm}\left\langle 2n{-}1\right\rangle(\ell_k)\right)\odot U'_{\ell'_k}.
\end{equation}
where, $U'_{\ell'_k}$ is the recruitment time of recruiter $C_{\ell'_k}$. Further, using Proposition \ref{proposition3}, we have
\begin{equation}
   \mathcal{L}_3^{(\dot{\mathfrak{X}})}\hspace{-1.4mm}\left\langle 2n{-}2\right\rangle \dot{\preceq} \mathcal{H}_3^{(\varphi;\mathfrak{R})}\hspace{-1.4mm}\left\langle 1 \hspace{-0.3mm};\hspace{-0.3mm} 1 \right\rangle.
\end{equation}
Consequently, we have the following result:
\begin{equation}\label{eqS20}
\begin{aligned}
    X'{=}S_2\equiv S_{2,0}=\Big\{\mathcal{L}_3^{(\dot{\mathfrak{X}})}\hspace{-1.4mm}\left\langle {2n{-}2}\right\rangle,\mathcal{H}_3^{(\varphi;\mathfrak{R})}\hspace{-1.4mm}\left\langle 1 \hspace{-0.3mm};\hspace{-0.3mm} 1 \right\rangle\Big\}.
\end{aligned}
\end{equation}
Lets assume that current state of $\beta$-SMP is $X{=}S_2\equiv S_{2,0}=\Big\{\mathcal{L}_3^{(\dot{\mathfrak{X}})}\hspace{-1.4mm}\left\langle {2n{-}2}\right\rangle,\mathcal{H}_3^{(\varphi;\mathfrak{R})}\hspace{-1.4mm}\left\langle 1 \hspace{-0.3mm};\hspace{-0.3mm} 1 \right\rangle\Big\}$. Two possible set of events can occur in this state:
\begin{enumerate}
    \item $e_X{=}e^{\mathsf{D}}(C_{\ell_k})$, where $C_{\ell_k} {\in} \ddot{\mathcal{C}}_{2n-4}$.
    \item $e_X{=}e^{\mathsf{R}}(C_{\ell'_k})$, where $C_{\ell'_k} {\in} \dot{\mathcal{C}}_{2}$.
\end{enumerate}
We omit $e_X{=}e^{\mathsf{D}}(C_{\ell_k})$ since it is not a part of the initial case of our induction-based proof. We next investigate the occurrence of events $e_X{=}e^{\mathsf{R}}(C_{\ell'_1})$ and $e_X{=}e^{\mathsf{R}}(C_{\ell'_2})$.\par
\textbf{Proof of case \ref{I3}.} Let $X{=}S_2\equiv S_{2,0}=\Big\{\mathcal{L}_3^{(\dot{\mathfrak{X}})}\hspace{-1.4mm}\left\langle {2n{-}2}\right\rangle,\mathcal{H}_3^{(\varphi;\mathfrak{R})}\hspace{-1.4mm}\left\langle 1 \hspace{-0.3mm};\hspace{-0.3mm} 1 \right\rangle\Big\}$. If event $e_X=e^{\mathsf{R}}(C_{\ell'_s})$ occurs, where $s\in\{1,2\}$, vehicle $C_{\ell'_s}$ successfully recruits a new vehicle $C_{\ell_y}\in \widehat{C}$. Hence, $\beta$-SMP will transit to state $X'=S_1{=}\dot{\mathcal{C}}_1 \cup \ddot{\mathcal{C}}_{2n-2}$. Since $\mathcal{H}_3^{(\varphi;\mathfrak{R})}\hspace{-1.4mm}\left\langle 1 \hspace{-0.3mm};\hspace{-0.3mm} 1 \right\rangle(s)$ units of time have passed, $\mathcal{H}_3^{(\varphi;\mathfrak{R})}\hspace{-1.4mm}\left\langle 1 \hspace{-0.3mm};\hspace{-0.3mm} 1 \right\rangle(s)$ (i.e., the recruitment time of {\small$C_{\ell'_s}$}) should be subtracted from the residual times of other random variables. Specifically, if $X{=}S_2$ and $e_X{=}e^{\mathsf{R}}(C_{\ell'_1})$, where $O(C_{\ell'_1}){=}1$ (see \eqref{eqOrder}), using Proposition~\ref{proposition3}, we have
\begin{equation}
 \hspace{-4mm}
    \mathcal{L}_4^{(\ddot{\mathfrak{X}})}\hspace{-1.4mm}\left\langle {2n{-}1}\right\rangle{=}\left(\hspace{-0.5mm}\mathcal{L}_3^{(\dot{\mathfrak{X}})}\hspace{-1.4mm}\left\langle {2n{-}2}\right\rangle {\circleddash} \mathcal{H}_3^{(\varphi;\mathfrak{R})}\hspace{-1.4mm}\left\langle 1 \hspace{-0.3mm};\hspace{-0.3mm} 1 \right\rangle(1)\hspace{-0.5mm}\right){\odot} Z'_{\ell_y},
 \hspace{-4mm}
\end{equation}
where $Z'_{\ell_y}$ refers to the residual sojourn time of the recruited vehicle $C_{\ell_y}$. Moreover, Proposition \ref{proposition3} yields
\begin{equation}
 \hspace{-4mm}
    \mathcal{H}_4^{(\gamma;\psi)}\hspace{-1.4mm}\left\langle 0 \hspace{-0.3mm};\hspace{-0.3mm} 1 \right\rangle{=}\mathcal{H}_3^{(\varphi;\mathfrak{R})}\hspace{-1.4mm}\left\langle 1 \hspace{-0.3mm};\hspace{-0.3mm} 1 \right\rangle \circleddash \mathcal{H}_3^{(\varphi;\mathfrak{R})}\hspace{-1.4mm}\left\langle 1 \hspace{-0.3mm};\hspace{-0.3mm} 1 \right\rangle(1),
\hspace{-4mm}
\end{equation}
Accordingly, if $X{=}S_2$ and $e_X{=}e^{\mathsf{R}}(C_{\ell'_1})$, where $O(C_{\ell'_1}){=}1$ (see \eqref{eqOrder}), we have
\begin{equation}\label{S10_158}
    X'{=}S_1\equiv S_{1,1}=\left\{\mathcal{L}_4^{(\ddot{\mathfrak{X}})}\hspace{-1.4mm}\left\langle {2n-1}\right\rangle ,\mathcal{H}_4^{(\gamma;\psi)}\hspace{-1.4mm}\left\langle 0 \hspace{-0.1mm};\hspace{-0.3mm} 1 \right\rangle\right\},
\end{equation}
where, using Proposition \ref{proposition3}, we have
\begin{equation}
    \mathcal{L}_4^{(\ddot{\mathfrak{X}})}\hspace{-1.4mm}\left\langle {2n-1}\right\rangle \dot{\preceq}   \mathcal{H}_4^{(\gamma;\psi)}\hspace{-1.4mm}\left\langle 0 \hspace{-0.1mm};\hspace{-0.3mm} 1 \right\rangle.
\end{equation}
Likewise, if $X{=}S_2$ and $e_X{=}e^{\mathsf{R}}(C_{\ell'_2})$, where $O(C_{\ell'_2}){=}2$ (see \eqref{eqOrder}), we have
\begin{equation}
 \hspace{-4mm}
    \mathcal{L}_5^{(\ddot{\mathfrak{X}})}\hspace{-1.4mm}\left\langle {2n{-}1}\right\rangle{=}\left(\hspace{-0.5mm}\mathcal{L}_3^{(\dot{\mathfrak{X}})}\hspace{-1.4mm}\left\langle {2n{-}2}\right\rangle {\circleddash} \mathcal{H}_3^{(\varphi;\mathfrak{R})}\hspace{-1.4mm}\left\langle 1 \hspace{-0.3mm};\hspace{-0.3mm} 1 \right\rangle(2)\hspace{-0.5mm}\right){\odot} Z'_{\ell_y},
 \hspace{-4mm}
\end{equation}
where $Z'_{\ell_y}$ refers to the residual sojourn time of the recruited vehicle $C_{\ell_y}$. Moreover, Proposition \ref{proposition3} yields
\begin{equation}
    \mathcal{H}_5^{(\gamma;\psi)}\hspace{-1.4mm}\left\langle 1 \hspace{-0.3mm};\hspace{-0.3mm} 0 \right\rangle=\mathcal{H}_3^{(\varphi;\mathfrak{R})}\hspace{-1.4mm}\left\langle 1 \hspace{-0.3mm};\hspace{-0.3mm} 1 \right\rangle \circleddash \mathcal{H}_3^{(\varphi;\mathfrak{R})}\hspace{-1.4mm}\left\langle 1 \hspace{-0.3mm};\hspace{-0.3mm} 1 \right\rangle(2).
\end{equation}
Accordingly, we get
\begin{equation}\label{S10_159}
    X'{=}S_1\equiv S_{1,2}=\left\{\mathcal{L}_5^{(\ddot{\mathfrak{X}})}\hspace{-1.4mm}\left\langle {2n-1}\right\rangle ,\mathcal{H}_5^{(\gamma;\psi)}\hspace{-1.4mm}\left\langle 1 \hspace{-0.1mm};\hspace{-0.3mm} 0 \right\rangle\right\}.
\end{equation}
where, using Proposition \ref{proposition3}, we have
\begin{equation}
    \mathcal{L}_5^{(\ddot{\mathfrak{X}})}\hspace{-1.4mm}\left\langle {2n-1}\right\rangle \dot{\preceq} \mathcal{H}_5^{(\gamma;\psi)}\hspace{-1.4mm}\left\langle 1 \hspace{-0.1mm};\hspace{-0.3mm} 0 \right\rangle.
\end{equation}
Next, we aim to prove case \ref{I4}. Assume the current state of $\beta$-SMP is $X=S_1$. According to \eqref{S10_148}, \eqref{S10_158}, and \eqref{S10_159}, $S_1$ is equivalent to three different H-states $S_{1,0}$, $S_{1,1}$, and $S_{1,2}$. Therefore, if event $e_X{=}e^{\mathsf{D}}(C_{\ell_k})$ occurs,  $\beta$-SMP will transit to state $X'{=}S_2{=}\dot{\mathcal{C}}_2 \cup \ddot{\mathcal{C}}_{2n-4}\equiv S_{2,0}$, and {\small$C_{\ell'_k}$} starts recruiting a new vehicle.

So far, we have proved that (i) if $X=S_1\equiv S_{1,0}$, $\beta$-SMP will transit to $X'=S_2\equiv S_{2,0}$ and (ii) if $X=S_2\equiv S_{2,0}$, $\beta$-SMP will transit to $X'=S_1\equiv S_{1,1}$ or $X'=S_1\equiv S_{1,2}$. In the following, utilizing the same technique conducted above, we show that if $X=S_1\equiv S_{i,j}$, where $i=1$ and $1\le j\le 2$, and $e^{\mathsf{D}}(C_{\ell_k})$ occurs, $\beta$-SMP will still transit to $X'=S_2\equiv S_{2,0}$.\par
\textbf{Proof of case \ref{I4}.} Let the current state of $\beta$-SMP be $X=S_1\equiv S_{1,1}=\left\{\mathcal{L}_4^{(\ddot{\mathfrak{X}})}\hspace{-1.4mm}\left\langle {2n-1}\right\rangle ,\mathcal{H}_4^{(\gamma;\psi)}\hspace{-1.4mm}\left\langle 0 \hspace{-0.1mm};\hspace{-0.3mm} 1 \right\rangle\right\}$. If event $e_X{=}e^{\mathsf{D}}(C_{\ell_k})$ occurs, $\mathcal{L}_4^{(\ddot{\mathfrak{X}})}\hspace{-1.4mm}\left\langle {2n-1}\right\rangle(\ell_k)$ (i.e., residual sojourn time of vehicle $C_{\ell_k}$) should be subtracted from $\mathcal{H}_4^{(\gamma;\psi)}\hspace{-1.4mm}\left\langle 0 \hspace{-0.1mm};\hspace{-0.3mm} 1 \right\rangle$ and $\mathcal{L}_4^{(\ddot{\mathfrak{X}})}\hspace{-1.4mm}\left\langle 2n{-}1\right\rangle$. Using Proposition~\ref{proposition4}, we have
\begin{equation}
\hspace{-3mm}
    \mathcal{H}_6^{(\varphi;\mathfrak{R})}\hspace{-1.4mm}\left\langle 1 \hspace{-0.3mm};\hspace{-0.3mm} 1 \right\rangle \hspace{-1mm}= \left(\mathcal{H}_4^{(\gamma;\psi)}\hspace{-1.4mm}\left\langle 0 \hspace{-0.1mm};\hspace{-0.3mm} 1 \right\rangle \circleddash \mathcal{L}_4^{(\ddot{\mathfrak{X}})}\hspace{-1.4mm}\left\langle {2n-1}\right\rangle(\ell_k)\right)\odot U'_{\ell'_k}.
\end{equation}
where $U'_{\ell'_k}$ is the recruitment time of recruiter $C_{\ell'_k}$ since $C_{\ell'_k}$ starts its recruitment operation. Likewise, using Proposition \ref{proposition4}, we get
\begin{equation}
    \mathcal{L}_6^{(\dot{\mathfrak{X}})}\hspace{-1.4mm}\left\langle 2n{-}2\right\rangle = \mathcal{L}_4^{(\ddot{\mathfrak{X}})}\hspace{-1.4mm}\left\langle {2n-1}\right\rangle \circleddash \mathcal{L}_4^{(\ddot{\mathfrak{X}})}\hspace{-1.4mm}\left\langle {2n-1}\right\rangle(\ell_k).
\end{equation}
Consequently, we have the following result
\begin{equation}\label{eqS20}
\begin{aligned}
    X'{=}S_2\equiv S_{2,0}=\Big\{\mathcal{L}_6^{(\dot{\mathfrak{X}})}\hspace{-1.4mm}\left\langle {2n{-}2}\right\rangle,\mathcal{H}_6^{(\varphi;\mathfrak{R})}\hspace{-1.4mm}\left\langle 1 \hspace{-0.3mm};\hspace{-0.3mm} 1 \right\rangle\Big\}.
\end{aligned}
\end{equation}
where, using Proposition \ref{proposition4}, we have
\begin{equation}
    \mathcal{L}_6^{(\dot{\mathfrak{X}})}\hspace{-1.4mm}\left\langle {2n{-}2}\right\rangle \dot{\preceq} \mathcal{H}_6^{(\varphi;\mathfrak{R})}\hspace{-1.4mm}\left\langle 1 \hspace{-0.3mm};\hspace{-0.3mm} 1 \right\rangle.
\end{equation}

Likewise, let the current state of $\beta$-SMP be $X=S_1\equiv S_{1,2}=\left\{\mathcal{L}_5^{(\ddot{\mathfrak{X}})}\hspace{-1.4mm}\left\langle {2n-1}\right\rangle ,\mathcal{H}_5^{(\gamma;\psi)}\hspace{-1.4mm}\left\langle 1 \hspace{-0.1mm};\hspace{-0.3mm} 0 \right\rangle\right\}$. If event $e_X{=}e^{\mathsf{D}}(C_{\ell_k})$ occurs, $\mathcal{L}_5^{(\ddot{\mathfrak{X}})}\hspace{-1.4mm}\left\langle {2n-1}\right\rangle(\ell_k)$ (i.e., residual sojourn time of vehicle $C_{\ell_k}$) should be subtracted from $\mathcal{H}_5^{(\gamma;\psi)}\hspace{-1.4mm}\left\langle 1 \hspace{-0.1mm};\hspace{-0.3mm} 0 \right\rangle$ and $\mathcal{L}_5^{(\ddot{\mathfrak{X}})}\hspace{-1.4mm}\left\langle 2n{-}1\right\rangle$. Using Proposition~\ref{proposition4}, we have
\begin{equation}
\hspace{-3mm}
    \mathcal{H}_7^{(\varphi;\mathfrak{R})}\hspace{-1.4mm}\left\langle 1 \hspace{-0.3mm};\hspace{-0.3mm} 1 \right\rangle \hspace{-1mm}= \left(\mathcal{H}_5^{(\gamma;\psi)}\hspace{-1.4mm}\left\langle 1 \hspace{-0.1mm};\hspace{-0.3mm} 0 \right\rangle \circleddash \mathcal{L}_5^{(\ddot{\mathfrak{X}})}\hspace{-1.4mm}\left\langle {2n-1}\right\rangle(\ell_k)\right)\odot U'_{\ell'_k}.
\end{equation}
where $U'_{\ell'_k}$ is the recruitment time of recruiter $C_{\ell'_k}$ since $C_{\ell'_k}$ starts its recruitment operation. Similarly, using Proposition \ref{proposition4}, we get
\begin{equation}
    \mathcal{L}_7^{(\dot{\mathfrak{X}})}\hspace{-1.4mm}\left\langle 2n{-}2\right\rangle = \mathcal{L}_5^{(\ddot{\mathfrak{X}})}\hspace{-1.4mm}\left\langle {2n-1}\right\rangle \circleddash \mathcal{L}_5^{(\ddot{\mathfrak{X}})}\hspace{-1.4mm}\left\langle {2n-1}\right\rangle(\ell_k).
\end{equation}
Consequently, we have the following result
\begin{equation}\label{eqS20}
\begin{aligned}
    X'{=}S_2\equiv S_{2,0}=\Big\{\mathcal{L}_7^{(\dot{\mathfrak{X}})}\hspace{-1.4mm}\left\langle {2n{-}2}\right\rangle,\mathcal{H}_7^{(\varphi;\mathfrak{R})}\hspace{-1.4mm}\left\langle 1 \hspace{-0.3mm};\hspace{-0.3mm} 1 \right\rangle\Big\}.
\end{aligned}
\end{equation}
where, using Proposition \ref{proposition4}, we have
\begin{equation}
    \mathcal{L}_7^{(\dot{\mathfrak{X}})}\hspace{-1.4mm}\left\langle {2n{-}2}\right\rangle \dot{\preceq} \mathcal{H}_7^{(\varphi;\mathfrak{R})}\hspace{-1.4mm}\left\langle 1 \hspace{-0.3mm};\hspace{-0.3mm} 1 \right\rangle.
\end{equation}

In the following we show that if $\beta$-SMP is in state $X\in \{S_1\equiv S_{1,0}, S_1\equiv S_{1,1}, S_1\equiv S_{1,2}\}$ and event $e_X=e^{\mathsf{R}}(C_{\ell'_1})$ occurs, $\beta$-SMP will transit to state $X'=S_0\equiv S_{0,1}$.\par

\textbf{Proof of case \ref{I5}.}
From \eqref{S10_148}, \eqref{S10_158}, and \eqref{S10_159}, we know that $S_1$ is equivalent to three different H-states $S_{1,0}$, $S_{1,1}$, and $S_{1,2}$. First consider the case that $X{=}S_1 \equiv S_{1,0}$. If event $e_X=e^{\mathsf{R}}(C_{\ell'_1})$ occurs, $\beta$-SMP will transit to state $X'=S_0{=}\dot{\mathcal{C}}_0 \cup \ddot{\mathcal{C}}_{2n}$ and vehicle $C_{\ell'_1}$ successfully recruits a new vehicle $C_{\ell_y} \in \widehat{C}$. Because $\mathcal{H}_2^{(\varphi;\mathfrak{R})}\hspace{-1.4mm}\left\langle {0} \hspace{-0.3mm};\hspace{-0.3mm} 1 \right\rangle(1)$ units of time have passed, $\mathcal{H}_2^{(\varphi;\mathfrak{R})}\hspace{-1.4mm}\left\langle {0} \hspace{-0.3mm};\hspace{-0.3mm} 1 \right\rangle(1)$ should be subtracted from the residual times of other random variables. Therefore, using Proposition~\ref{proposition3}, we have
\begin{equation}
    \mathcal{L}_8^{(\ddot{\mathfrak{X}})}\hspace{-1.4mm}\left\langle {2n}\right\rangle = \left(\mathcal{L}_2^{(\dot{\mathfrak{X}})}\hspace{-1.4mm}\left\langle {2n{-}1}\right\rangle \circleddash \mathcal{H}_2^{(\varphi;\mathfrak{R})}\hspace{-1.4mm}\left\langle {0} \hspace{-0.3mm};\hspace{-0.3mm} 1 \right\rangle(1)\right)\odot Z'_{\ell_y},
\end{equation}
where $Z'_{\ell_y}$ refers to the residual sojourn time of the recruited vehicle $C_{\ell_y}$. Further, $\mathcal{H}_2^{(\varphi;\mathfrak{R})}\hspace{-1.4mm}\left\langle {0} \hspace{-0.3mm};\hspace{-0.3mm} 1 \right\rangle \circleddash \mathcal{H}_2^{(\varphi;\mathfrak{R})}\hspace{-1.4mm}\left\langle {0} \hspace{-0.3mm};\hspace{-0.3mm} 1 \right\rangle(1)$ leads to an empty list. Accordingly, if $X{=}S_1 \equiv S_{1,0}$ and $e_X{=}e^{\mathsf{R}}(C_{\ell'_k})$, where $C_{\ell'_k} {\in} \dot{\mathcal{C}}_{1}$, we have
\begin{equation}\label{eqS01_165}
    X'=S_0\equiv S_{0,1}{=}\left\{\mathcal{L}_8^{(\ddot{\mathfrak{X}})}\hspace{-1.4mm}\left\langle 2n\right\rangle,[\,]\right\}.
\end{equation}
We next show that if $X=S_1\equiv S_{1,1}$ and event $e_X=e^{\mathsf{R}}(C_{\ell'_1})$ occurs, $\beta$-SMP will transit to state $X'=S_0\equiv S_{0,1}$. If event $e_X=e^{\mathsf{R}}(C_{\ell'_1})$ occurs, vehicle $C_{\ell'_1}$ successfully recruits a new vehicle $C_{\ell_y}\in \widehat{C}$ and thus $\beta$-SMP will transit to state $X'=S_0{=}\dot{\mathcal{C}}_0 \cup \ddot{\mathcal{C}}_{2n}$. Because $\mathcal{H}_4^{(\gamma;\psi)}\hspace{-1.4mm}\left\langle 0 \hspace{-0.1mm};\hspace{-0.3mm} 1 \right\rangle(1)$ units of time have passed, $\mathcal{H}_4^{(\gamma;\psi)}\hspace{-1.4mm}\left\langle 0 \hspace{-0.1mm};\hspace{-0.3mm} 1 \right\rangle(1)$ should be subtracted from the residual times of other random variables. Therefore, using Proposition \ref{proposition4}, we have
\begin{equation}
\begin{aligned}
    \mathcal{L}_9^{(\ddot{\mathfrak{X}})}\hspace{-1.4mm}\left\langle {2n}\right\rangle = &\Big(\mathcal{L}_4^{(\ddot{\mathfrak{X}})}\hspace{-1.4mm}\left\langle {2n{-}1}\right\rangle \circleddash \mathcal{H}_4^{(\gamma;\psi)}\hspace{-1.4mm}\left\langle 0 \hspace{-0.1mm};\hspace{-0.3mm} 1 \right\rangle(1)\Big)\odot Z'_{\ell_y},
\end{aligned}
\end{equation}
where $Z'_{\ell_y}$ refers to the residual sojourn time of the recruited vehicle $C_{\ell_y}$. Further, $\mathcal{H}_4^{(\gamma;\psi)}\hspace{-1.4mm}\left\langle 0 \hspace{-0.1mm};\hspace{-0.3mm} 1 \right\rangle \circleddash \mathcal{H}_4^{(\gamma;\psi)}\hspace{-1.4mm}\left\langle 0 \hspace{-0.1mm};\hspace{-0.3mm} 1 \right\rangle(1)$ leads to an empty list. Accordingly, if $X{=}S_1 \equiv S_{1,1}$, and $e_X{=}e^{\mathsf{R}}(C_{\ell'_k})$, where $C_{\ell'_k} {\in} \dot{\mathcal{C}}_{1}$, we have
\begin{equation}\label{eqS01_167}
    X'=S_0\equiv S_{0,1}{=}\left\{\mathcal{L}_9^{(\ddot{\mathfrak{X}})}\hspace{-1.4mm}\left\langle 2n\right\rangle,[\,]\right\}.
\end{equation}
Likewise, in the following, we show that if $X=S_1\equiv S_{1,2}$ and event $e_X=e^{\mathsf{R}}(C_{\ell'_1})$ occurs, $\beta$-SMP will transit to state $X'=S_0\equiv S_{0,1}$. As mentioned above, if event $e_X=e^{\mathsf{R}}(C_{\ell'_1})$ occurs, vehicle $C_{\ell'_1}$ successfully recruits a new vehicle $C_{\ell_y}\in \widehat{C}$ and thus $\beta$-SMP will transit to state $X'=S_0{=}\dot{\mathcal{C}}_0 \cup \ddot{\mathcal{C}}_{2n}$. Because $\mathcal{H}_5^{(\gamma;\psi)}\hspace{-1.4mm}\left\langle 1 \hspace{-0.1mm};\hspace{-0.3mm} 0 \right\rangle(1)$ units of time have passed, $\mathcal{H}_5^{(\gamma;\psi)}\hspace{-1.4mm}\left\langle 1 \hspace{-0.1mm};\hspace{-0.3mm} 0 \right\rangle(1)$ should be subtracted from the residual times of other random variables. Therefore, using Proposition \ref{proposition4}, we have
\begin{equation}
\begin{aligned}
    \mathcal{L}_{10}^{(\ddot{\mathfrak{X}})}\hspace{-1.4mm}\left\langle {2n}\right\rangle = &\Big(\mathcal{L}_5^{(\ddot{\mathfrak{X}})}\hspace{-1.4mm}\left\langle {2n{-}1}\right\rangle \circleddash \mathcal{H}_5^{(\gamma;\psi)}\hspace{-1.4mm}\left\langle 1 \hspace{-0.1mm};\hspace{-0.3mm} 0 \right\rangle(1)\Big)\odot Z'_{\ell_y},
\end{aligned}
\end{equation}
where $Z'_{\ell_y}$ refers to the residual sojourn time of the recruited vehicle $C_{\ell_y}$. Further, $\mathcal{H}_5^{(\gamma;\psi)}\hspace{-1.4mm}\left\langle 1 \hspace{-0.1mm};\hspace{-0.3mm} 0 \right\rangle \circleddash \mathcal{H}_5^{(\gamma;\psi)}\hspace{-1.4mm}\left\langle 1 \hspace{-0.1mm};\hspace{-0.3mm} 0 \right\rangle(1)$ leads to an empty list. Accordingly, if $X{=}S_1 \equiv S_{1,2}$, and $e_X{=}e^{\mathsf{R}}(C_{\ell'_k})$, where $C_{\ell'_k} {\in} \dot{\mathcal{C}}_{1}$, we have
\begin{equation}\label{eqS01_167}
    X'=S_0\equiv S_{0,1}{=}\left\{\mathcal{L}_{10}^{(\ddot{\mathfrak{X}})}\hspace{-1.4mm}\left\langle 2n\right\rangle,[\,]\right\}.
\end{equation}
We next show that if $\beta$-SMP is in state $X=S_0 \equiv S_{0,1}$ and event $e_X=e^{\mathsf{D}}(C_{\ell_k})$ occurs, where $C_{\ell_k} {\in} \ddot{\mathcal{C}}_{2n}$, $\beta$-SMP will transit to state $X'=S_1 \equiv S_{1,0}$.\par
\textbf{Proof of case \ref{I6}.} Let the current state of $\beta$-SMP be $X=S_0 \equiv S_{0,1}{=}\left\{\mathcal{L}_{10}^{(\ddot{\mathfrak{X}})}\hspace{-1.4mm}\left\langle 2n\right\rangle,[\,]\right\}$. If event $e_X=e^{\mathsf{D}}(C_{\ell_k})$ occurs, $\beta$-SMP will transit to state $X'{=}S_1{=}\dot{\mathcal{C}}_1 \cup \ddot{\mathcal{C}}_{2n-2}$, and {\small$C_{\ell'_k}$} starts recruiting a new vehicle. Therefore, $\mathcal{L}_{10}^{(\dot{\mathfrak{X}})}\hspace{-1.4mm}\left\langle 2n\right\rangle(\ell_k)$ (the residual sojourn time of $C_{\ell_k}$) must be subtracted from residual sojourn time of other vehicles. To this end, using Proposition \ref{proposition3}, we have
\begin{equation}
    \mathcal{L}_{11}^{(\dot{\mathfrak{X}})}\hspace{-1.4mm}\left\langle 2n{-}1\right\rangle = \mathcal{L}_{10}^{(\dot{\mathfrak{X}})}\hspace{-1.4mm}\left\langle 2n\right\rangle \circleddash \mathcal{L}_{10}^{(\dot{\mathfrak{X}})}\hspace{-1.4mm}\left\langle 2n\right\rangle(\ell_k),
\end{equation}
Moreover, let $\hat{\mathcal{H}}^{(\varphi;\mathfrak{R})}\hspace{-1.4mm}\left\langle 0 \hspace{-0.3mm};\hspace{-0.3mm} 0 \right\rangle$ denote an empty EDL. The residual recruitment time of recruiter $C_{\ell'_k}$ (i.e., $U'_{\ell'_k}$) can be added to $\hat{\mathcal{H}}^{(\varphi;\mathfrak{R})}\hspace{-1.4mm}\left\langle 0 \hspace{-0.3mm};\hspace{-0.3mm} 0 \right\rangle$. Mathematically, we have
\begin{equation}
    \mathcal{H}_{11}^{(\varphi;\mathfrak{R})}\hspace{-1.4mm}\left\langle 0 \hspace{-0.3mm};\hspace{-0.3mm} 1 \right\rangle=\hat{\mathcal{H}}^{(\varphi;\mathfrak{R})}\hspace{-1.4mm}\left\langle 0\hspace{-0.3mm};\hspace{-0.3mm} 0 \right\rangle \odot U'_{\ell'_k}.
\end{equation}
As a result, we have
\begin{equation}\label{eqS20}
\begin{aligned}
    X'{=}S_1\equiv S_{1,0}=\Big\{\mathcal{L}_{11}^{(\dot{\mathfrak{X}})}\hspace{-1.4mm}\left\langle {2n{-}1}\right\rangle,\mathcal{H}_{11}^{(\varphi;\mathfrak{R})}\hspace{-1.4mm}\left\langle 0 \hspace{-0.3mm};\hspace{-0.3mm} 1 \right\rangle\Big\}.
\end{aligned}
\end{equation}
where, using Proposition \ref{proposition3}, we have
\begin{equation}
    \mathcal{L}_{11}^{(\dot{\mathfrak{X}})}\hspace{-1.4mm}\left\langle {2n{-}1}\right\rangle \dot{\preceq} \mathcal{H}_{11}^{(\varphi;\mathfrak{R})}\hspace{-1.4mm}\left\langle 0 \hspace{-0.3mm};\hspace{-0.3mm} 1 \right\rangle.
\end{equation}

Consequently, considering \eqref{eqS01_144}, \eqref{eqS01_165}, and \eqref{eqS01_167}, we have
\begin{equation}
 X'{=}S_0\equiv
    \begin{cases}
     S_{0,0}& \text{initial state}.\\
     S_{0,1} & \text{if  } X{=}S_{1},e_X{=}e^{\mathsf{R}}(C_{\ell}),~C_{\ell} {\in} \dot{\mathcal{C}}_{1}.
    \end{cases}
\end{equation}
Further, from \eqref{S10_148}, \eqref{S10_158}, and \eqref{S10_159}, we get
\begin{equation}
   X'{=}S_1\equiv
    \begin{cases}
        S_{1,0} & \text{if  } X=S_{0},e_x{=}e^{\mathsf{D}}(C_{\ell}),~\forall C_{\ell} {\in} \ddot{\mathcal{C}}_{2n}. \\
        S_{1,1} & \text{if  } X=S_{2},e_X{=}e^{\mathsf{R}}(C_{\ell}),~O(C_{\ell})=1, C_{\ell} {\in} \dot{\mathcal{C}}_{2}\\
        S_{1,2} &\text{if  } X=S_{2},e_X{=}e^{\mathsf{R}}(C_{\ell}),~ O(C_{\ell})=2, C_{\ell} {\in} \dot{\mathcal{C}}_{2}
    \end{cases}
\end{equation}
which concludes the proof of the initial case of our induction-based proof. We next aim to prove the induction step.
\subsection{Induction Step}
In this section, we prove if Decomposition Theorem holds for state $S_{n-1}$, it holds for state $S_{n}$. Decomposition Theorem states that
\begin{equation}\label{sn_1}
   X'{=}S_{n{-}1}\equiv
    \begin{cases}
        S_{n{-}1,0} & \text{if  } X=S_{n{-}2},e_X{=}e^{\mathsf{D}}(C_{\ell}),~ \forall C_{\ell} {\in} \ddot{\mathcal{C}}_{4}. \\
        S_{n{-}1,1} & \text{if  } X=S_{n},e_X{=}e^{\mathsf{R}}(C_{\ell}),~ O(C_{\ell})=1, \forall C_{\ell} {\in} \dot{\mathcal{C}}_{n}\\
        \dots &\dots\\
        S_{n{-}1,j} &\text{if  } X=S_{n},e_X{=}e^{\mathsf{R}}(C_{\ell}),~ O(C_{\ell})=j, \forall C_{\ell} {\in} \dot{\mathcal{C}}_{n}\\
        \dots & \dots\\
        S_{n{-}1,n} &\text{if  } X=S_{n},e_X{=}e^{\mathsf{R}}(C_{\ell}),~ O(C_{\ell})=n,\forall C_{\ell} {\in} \dot{\mathcal{C}}_{n}
    \end{cases}
\end{equation}
and
  \begin{equation}
 X'{=}S_n\equiv
    \begin{cases}
    S_{n,0} & \text{if  } X=S_{n-1},e_X{=}e^{\mathsf{D}}(C_{\ell}),~ \forall C_{\ell} {\in} \ddot{\mathcal{C}}_{2} \\
    \end{cases}
\end{equation}
To prove the induction step, we must prove the following two cases:
\begin{enumerate}[label={(I-\arabic*)}]
    \item \label{I21} $\beta$-SMP will transit to state $X'\in S_n\equiv S_{n,0}$, when it is in state $X\in \{S_{n-1}\equiv S_{n{-}1,0}, S_{n-1}\equiv S_{n{-}1,1},\dots, S_{n-1}\equiv S_{n{-}1,n}\}$ and event $e_X=e^{\mathsf{D}}(C_{\ell_k})$ occurs, where $C_{\ell_k} {\in} \ddot{\mathcal{C}}_{2}$.
    \item \label{I22} $\beta$-SMP will transit to state $X'=S_{n-1}\equiv S_{n{-}1,k}$, when it is in state $X\in S_n\equiv S_{n,0}$, event $e_X=e^{\mathsf{R}}(C_{\ell'_j})$ occurs, and $ O(C_{\ell'_k})=j$, where $C_{\ell'_j} {\in} \ddot{\mathcal{C}}_{n}$.
\end{enumerate}
In the following, we aim to prove above two cases.\par
\textbf{Proof of case \ref{I21}.} Assume that the current state of $\beta$-SMP is $X=S_{n-1}$. According to \eqref{sn_1}, $S_{n-1}$ is equivalent to $n+1$ different H-states $\{S_{n{-}1,0}, S_{n{-}1,1},\dots, S_{n{-}1,n}\}$. If event $e_X{=}e^{\mathsf{D}}(C_{\ell_k})$ occurs, where $C_{\ell_k} {\in} \ddot{\mathcal{C}}_{2}$, $\beta$-SMP will transit to state $X'{=}S_n{=}\dot{\mathcal{C}}_n \cup \ddot{\mathcal{C}}_{0}$, and {\small$C_{\ell'_k}$} starts recruiting a new vehicle. If $S_{n-1} \equiv S_{n{-}1,0}=\{\mathcal{L}_1^{(\dot{\mathfrak{X}})}\hspace{-1.4mm}\left\langle n{+}1\right\rangle,\mathcal{H}_1^{(\varphi;\mathfrak{R})}\hspace{-1.4mm}\left\langle n{-}2 \hspace{-0.3mm};\hspace{-0.3mm} 1 \right\rangle\}$, $\mathcal{L}_1^{(\dot{\mathfrak{X}})}\hspace{-1.4mm}\left\langle n{+}1\right\rangle(\ell_k)$ should be subtracted from $\mathcal{H}_1^{(\varphi;\mathfrak{R})}\hspace{-1.4mm}\left\langle n{-}2 \hspace{-0.3mm};\hspace{-0.3mm} 1 \right\rangle$ and $ \mathcal{L}_1^{(\dot{\mathfrak{X}})}\hspace{-1.4mm}\left\langle n{+}1\right\rangle$. Using Proposition~\ref{proposition3}, we have
\begin{equation}
    \mathcal{H}_2^{(\varphi;\mathfrak{R})}\hspace{-1.4mm}\left\langle n{-}1 \hspace{-0.3mm};\hspace{-0.3mm} 0 \right\rangle=\left(\mathcal{H}_1^{(\varphi;\mathfrak{R})}\hspace{-1.4mm}\left\langle n{-}2 \hspace{-0.3mm};\hspace{-0.3mm} 1 \right\rangle \circleddash \mathcal{L}_1^{(\dot{\mathfrak{X}})}\hspace{-1.4mm}\left\langle n{+}1\right\rangle(\ell_k)\right) \odot U'_{\ell'_k},
\end{equation}
where $U'_{\ell'_k}$ is the recruitment time of recruiter $C_{\ell'_k}$. Moreover, using Proposition \ref{proposition3}, we get
\begin{equation}
    \mathcal{L}_2^{(\dot{\mathfrak{X}})}\hspace{-1.4mm}\left\langle n\right\rangle = \mathcal{L}_1^{(\dot{\mathfrak{X}})}\hspace{-1.4mm}\left\langle n{+}1\right\rangle \circleddash \mathcal{L}_1^{(\dot{\mathfrak{X}})}\hspace{-1.4mm}\left\langle n{+}1\right\rangle(\ell_k).
\end{equation}
Consequently, we have the following result:
\begin{equation}\label{eqS20}
\begin{aligned}
    X'{=}S_{n}\equiv S_{n,0}=\Big\{\mathcal{L}_2^{(\dot{\mathfrak{X}})}\hspace{-1.4mm}\left\langle {n}\right\rangle,\mathcal{H}_2^{(\varphi;\mathfrak{R})}\hspace{-1.4mm}\left\langle n{-}1 \hspace{-0.3mm};\hspace{-0.3mm} 1 \right\rangle\Big\}.
\end{aligned}
\end{equation}
where, using Proposition \ref{proposition3}, we have
\begin{equation}
    \mathcal{L}_2^{(\dot{\mathfrak{X}})}\hspace{-1.4mm}\left\langle {n}\right\rangle \dot{\preceq} \mathcal{H}_2^{(\varphi;\mathfrak{R})}\hspace{-1.4mm}\left\langle n{-}1 \hspace{-0.3mm};\hspace{-0.3mm} 1 \right\rangle.
\end{equation}

So far, we have proved that if $X=S_{n-1}\equiv S_{n-1,0}$, $\beta$-SMP will transit to $X'=S_n$. In the following, utilizing the same technique conducted above, we show that if $X=S_{n-1}\equiv S_{i,j}$, where $i=n-1$ and $1\le j\le n$, $\beta$-SMP will still transits to $X'=S_n\equiv S_{n,0}$. Assume that event $e_X{=}e^{\mathsf{D}}(C_{\ell_k})$ occurs, where $C_{\ell_k} {\in} \ddot{\mathcal{C}}_{2}$. $\beta$-SMP transits to state $X'{=}S_n{=}\dot{\mathcal{C}}_n \cup \ddot{\mathcal{C}}_{0}$ and $\mathcal{L}_1^{(\dot{\mathfrak{X}})}\hspace{-1.4mm}\left\langle n{+}1\right\rangle(\ell_k)$ should be subtracted from $\mathcal{H}_1^{(\gamma;\psi)}\hspace{-1.4mm}\left\langle j{-}1 \hspace{-0.3mm};\hspace{-0.3mm} i{-}(j{-}1) \right\rangle$ and $\mathcal{L}_1^{(\ddot{\mathfrak{X}})}\hspace{-1.4mm}\left\langle n{+}1\right\rangle$. Using Proposition~\ref{proposition4}, we have
\begin{equation}
    \mathcal{H}_3^{(\varphi;\mathfrak{R})}\hspace{-1.4mm}\left\langle n{-}1 \hspace{-0.3mm};\hspace{-0.3mm} 1 \right\rangle=\left(\mathcal{H}_1^{(\gamma;\psi)}\hspace{-1.4mm}\left\langle j{-}1 \hspace{-0.3mm};\hspace{-0.3mm} i{-}(j{-}1) \right\rangle \circleddash \mathcal{L}_1^{(\dot{\mathfrak{X}})}\hspace{-1.4mm}\left\langle n{+}1\right\rangle(\ell_k)\right)\odot U'_{\ell'_k}.
\end{equation}
where $U'_{\ell'_k}$ is the residual recruitment time of recruiter $C_{\ell'_k}$. Moreover, using Proposition \ref{proposition4}, we get
\begin{equation}
    \mathcal{L}_3^{(\dot{\mathfrak{X}})}\hspace{-1.4mm}\left\langle n\right\rangle = \mathcal{L}_1^{(\ddot{\mathfrak{X}})}\hspace{-1.4mm}\left\langle n{+}1\right\rangle \circleddash \mathcal{L}_1^{(\ddot{\mathfrak{X}})}\hspace{-1.4mm}\left\langle n{+}1\right\rangle(\ell_k).
\end{equation}
Consequently, we have the following result:
\begin{equation}\label{eqS20}
\begin{aligned}
    X'{=}S_n\equiv S_{n,0}=\Big\{\mathcal{L}_3^{(\dot{\mathfrak{X}})}\hspace{-1.4mm}\left\langle {n}\right\rangle,\mathcal{H}_3^{(\varphi;\mathfrak{R})}\hspace{-1.4mm}\left\langle n{-}1 \hspace{-0.3mm};\hspace{-0.3mm} 1 \right\rangle\Big\}.
\end{aligned}
\end{equation}
where, using Proposition \ref{proposition4}, we have
\begin{equation}
    \mathcal{L}_3^{(\dot{\mathfrak{X}})}\hspace{-1.4mm}\left\langle {n}\right\rangle \dot{\preceq} \mathcal{H}_3^{(\varphi;\mathfrak{R})}\hspace{-1.4mm}\left\langle n{-}1 \hspace{-0.3mm};\hspace{-0.3mm} 1 \right\rangle.
\end{equation}
We next present the proof of \ref{I22}.\par
\textbf{Proof of case \ref{I22}.} Assume $\beta$-SMP is in state $X=S_n\equiv S_{n,0}=\left\{\mathcal{L}_3^{(\dot{\mathfrak{X}})}\hspace{-1.4mm}\left\langle {n}\right\rangle,\mathcal{H}_3^{(\varphi;\mathfrak{R})}\hspace{-1.4mm}\left\langle n{-}1 \hspace{-0.3mm};\hspace{-0.3mm} 1 \right\rangle\right\}$. If event $e_X=e^{\mathsf{R}}(C_{\ell'_j})$ occurs, where $j\in\{1,2,\dots n\}$, vehicle $C_{\ell'_j}$ successfully recruits a new vehicle $C_{\ell_y}$ from the vehicles that were in the VC at the start time of processing of application $\mathcal{A}$. Hence, $\beta$-SMP will transit to state $X'=S_{n-1}{=}\dot{\mathcal{C}}_{n-1} \cup \ddot{\mathcal{C}}_{2}$. Consequently, because $\mathcal{H}_3^{(\varphi;\mathfrak{R})}\hspace{-1.4mm}\left\langle n{-}1 \hspace{-0.3mm};\hspace{-0.3mm} 1 \right\rangle(j)$ units of time have passed, $\mathcal{H}_3^{(\varphi;\mathfrak{R})}\hspace{-1.4mm}\left\langle n{-}1 \hspace{-0.3mm};\hspace{-0.3mm} 1 \right\rangle(j)$ should be subtracted from the residual times of other random variables. Therefore, using Proposition \ref{proposition3}, we have
\begin{equation}
\hspace{-4mm}
    \mathcal{L}_4^{(\ddot{\mathfrak{X}})}\hspace{-1.4mm}\left\langle {n{+}1}\right\rangle = \left(\mathcal{L}_3^{(\dot{\mathfrak{X}})}\hspace{-1.4mm}\left\langle {n}\right\rangle \circleddash \mathcal{H}_3^{(\varphi;\mathfrak{R})}\hspace{-1.4mm}\left\langle n{-}1 \hspace{-0.3mm};\hspace{-0.3mm} 1 \right\rangle(j)\right)\odot Z'_{\ell_y},
\hspace{-4mm}
\end{equation}
where $Z'_{\ell_y}$ refers to the residual sojourn time of the recruited vehicle $C_{\ell_y}$. Moreover, Proposition \ref{proposition3} yields
\begin{equation}
\hspace{-4mm}
    \mathcal{H}_4^{(\gamma;\psi)}\hspace{-1.4mm}\left\langle j{-}1 \hspace{-0.3mm};\hspace{-0.3mm} n{-}j \right\rangle{=}\mathcal{H}_3^{(\varphi;\mathfrak{R})}\hspace{-1.4mm}\left\langle n{-}1 \hspace{-0.3mm};\hspace{-0.3mm} 1 \right\rangle \circleddash \mathcal{H}_3^{(\varphi;\mathfrak{R})}\hspace{-1.4mm}\left\langle n{-}1 \hspace{-0.3mm};\hspace{-0.3mm} 1 \right\rangle(j).
\hspace{-4mm}
\end{equation}
Accordingly, if $X{=}S_n$ and $e_X{=}e^{\mathsf{R}}(C_{\ell'_j})$, where $O(C_{\ell'_j}){=}s$ (see \eqref{eqOrder}), we have
\hspace{-5mm}
\begin{equation}
\hspace{-4mm}
    X'{=}S_{n{-}1}{\equiv} S_{n{-}1,j}{=}\left\{\mathcal{L}_4^{(\ddot{\mathfrak{X}})}\hspace{-1.4mm}\left\langle {n{+}1}\right\rangle, \mathcal{H}_4^{(\gamma;\psi)}\hspace{-1.4mm}\left\langle j{-}1 \hspace{-0.3mm};\hspace{-0.3mm} n{-}j \right\rangle\right\},
\hspace{-4mm}
\end{equation}
where, using Proposition \ref{proposition3}, we have
\begin{equation}
    \mathcal{L}_4^{(\ddot{\mathfrak{X}})}\hspace{-1.4mm}\left\langle {n{+}1}\right\rangle \dot{\preceq} \mathcal{H}_4^{(\gamma;\psi)}\hspace{-1.4mm}\left\langle j{-}1 \hspace{-0.3mm};\hspace{-0.3mm} n{-}j \right\rangle,
\end{equation}
which concludes the proof of Decomposition Theorem (Theorem \ref{decompositionTheorem}).

\newpage
\section{Proof of Lemma \ref{expecteTimeToFailure}}\label{expecteTimeToFailure_proof}
\noindent For states $S_{0,0}$ and $S_{0,1}$, we have
\begin{equation}\label{exptected00}
    \mathbb{E}\left[\widehat{\mathrm{Q}}_{0,0}\right] {=}\mathbb{E}\left[\widehat{\mathrm{W}}_{0,0}\right]{+}\mathbb{E}\left[\widehat{\mathrm{Q}}_{1,0}\right],
\end{equation}
and
\begin{equation}\label{exptected01}
    \mathbb{E}\left[\widehat{\mathrm{Q}}_{0,1}\right] {=}\mathbb{E}\left[\widehat{\mathrm{W}}_{0,1}\right]{+}\mathbb{E}\left[\widehat{\mathrm{Q}}_{1,0}\right].
\end{equation}
\eqref{exptected00} and \eqref{exptected01} mean that the expected value of the time until a failure from H-states $S_{0,0}$ and $S_{0,1}$ are equal to the time that D-SMP is in H-states $S_{0,0}$ or $S_{0,1}$ and transits to H-state $S_{1,0}$ and then transits to failure state $F$. \par
If D-SMP is in H-state $S_{i,j}\in\hat{\mathbf{S}}_i$, where $1\le i \le n$, then $\mathbb{E}\left[\widehat{\mathrm{Q}}_{i,j}\right]$ is the sum of the following duration times: (i) the duration time that D-SMP is in H-state $S_{i,j}$ and then transits to the failure state + (ii) the duration time that D-SMP is in H-state $S_{i,j}$ and transits to H-state $S_{i+1,0}$ and then transits to the failure state + (iii) the duration time that D-SMP is in H-state $S_{i,j}$ and transits to states $S_{i-1,j}$ and then goes to the failure state.\par
Accordingly, for each H-state $S_{i,j}\in\hat{\mathbf{S}}_i$, where $0\le i \le n$, we have
\begin{equation}
\begin{aligned}
\hspace{-2mm}
\mathbb{E}\left[\widehat{\mathrm{Q}}_{i,j}\right] & =\mathbb{E}\left[\widehat{\mathrm{W}}_{i,j}\right] b_{i,j} +\Bigg[\mathbb{E}\left[\widehat{\mathrm{W}}_{i,j}\right] +\mathbb{E}\left[\widehat{\mathrm{Q}}_{i+1,0}\right]\Bigg] q_{i,j} +\boldsymbol{\sum} _{k=1}^{i}\Bigg[\mathbb{E}\left[\widehat{\mathrm{W}}_{i,j}\right] +\mathbb{E}\left[\widehat{\mathrm{Q}}_{i-1,k}\right]\Bigg]p_{i,j}^{k}\\
 & =\mathbb{E}\left[\widehat{\mathrm{W}}_{i,j}\right] b_{i,j} +\mathbb{E}\left[\widehat{\mathrm{W}}_{i,j}\right] q_{i,j} +\mathbb{E}\left[\widehat{\mathrm{Q}}_{i+1,0}\right] q_{i,j} +\mathbb{E}\left[\widehat{\mathrm{W}}_{i,j}\right]\sum_{k=1}^{i}p_{i,j}^{k} +\boldsymbol{\sum} _{k=1}^{i}\mathbb{E}\left[\widehat{\mathrm{Q}}_{i-1,k}\right]p_{i,j}^{k}.
 \hspace{-4mm}
\end{aligned}
\end{equation}
Since $\sum_{k=1}^{i}p_{i,j}^{k}+ b_{i,j} +q_{i,j}=1$, then
\begin{equation}
\begin{aligned}
\mathbb{E}\left[\widehat{\mathrm{Q}}_{i,j}\right] =\mathbb{E}\left[\widehat{\mathrm{W}}_{i,j}\right] &+\mathbb{E}\left[\widehat{\mathrm{Q}}_{i+1,0}\right] q_{i,j}+\boldsymbol{\sum} _{k=1}^{i}\mathbb{E}\left[\widehat{\mathrm{Q}}_{i-1,k}\right]p_{i,j}^{k},
\end{aligned}
\end{equation}
which concludes the proof.

\newpage
\section{Proof of Theorem \ref{MTTFJ2dn}} \label{MTTFJ2dn_proof}
\noindent Based on Lemma \ref{expecteTimeToFailure}, $\mathbb{E}\left[\widehat{\mathrm{Q}}_{i,j}\right]$, for each state $S_{i,j}\in\hat{\mathbf{S}}_i$, is calculated as follows. Regarding $\widehat{\mathrm{Q}}_{0,0}$ and $\widehat{\mathrm{Q}}_{0,0}$ we have
\begin{equation}
\begin{aligned}
\mathbb{E}\left[\widehat{\mathrm{Q}}_{0,0}\right] & =\mathbb{E}\left[\widehat{\mathrm{W}}_{0,0}\right] +\mathbb{E}\left[\widehat{\mathrm{Q}}_{1,0}\right],\\
\mathbb{E}\left[\widehat{\mathrm{Q}}_{0,1}\right] & =\mathbb{E}\left[\widehat{\mathrm{W}}_{0,1}\right] +\mathbb{E}\left[\widehat{\mathrm{Q}}_{1,0}\right].\\
\end{aligned}
\end{equation}
In terms of $\widehat{\mathrm{Q}}_{1,0}$, $\widehat{\mathrm{Q}}_{1,1}$ and $\widehat{\mathrm{Q}}_{1,2}$, the results are
\begin{equation}
\begin{aligned}
\mathbb{E}\left[\widehat{\mathrm{Q}}_{1,0}\right] & =\mathbb{E}\left[\widehat{\mathrm{W}}_{1,0}\right] +\mathbb{E}\left[\widehat{\mathrm{Q}}_{2,0}\right] q_{1,0} +\mathbb{E}\left[\widehat{\mathrm{Q}}_{0,1}\right]p_{1,0}^{1},\\
\mathbb{E}\left[\widehat{\mathrm{Q}}_{1,1}\right] & =\mathbb{E}\left[\widehat{\mathrm{W}}_{1,1}\right] +\mathbb{E}\left[\widehat{\mathrm{Q}}_{2,0}\right] q_{1,1} +\mathbb{E}\left[\widehat{\mathrm{Q}}_{0,1}\right]p_{1,1}^{1},\\
\mathbb{E}\left[\widehat{\mathrm{Q}}_{1,2}\right] & =\mathbb{E}\left[\widehat{\mathrm{W}}_{1,2}\right] +\mathbb{E}\left[\widehat{\mathrm{Q}}_{2,0}\right] q_{1,2} +\mathbb{E}\left[\widehat{\mathrm{Q}}_{0,1}\right]p_{1,2}^{1}.
\end{aligned}
\end{equation}
Similarly, for $\widehat{\mathrm{Q}}_{n-2,0}$ to $\widehat{\mathrm{Q}}_{n-2,n-1}$, we have
\begin{equation}\label{qn_2_0}
\begin{aligned}
\mathbb{E}\left[\widehat{\mathrm{Q}}_{n-2,0}\right]  =\mathbb{E}\left[\widehat{\mathrm{W}}_{n-2,0}\right] &+\mathbb{E}\left[\widehat{\mathrm{Q}}_{n-1,0}\right] q_{n-2,0}+\boldsymbol{\sum} _{l=1}^{n-2}\mathbb{E}\left[\widehat{\mathrm{Q}}_{n-3,l}\right]p_{n-2,0}^{l},\\
 & \, \, \, \, \, \, \, \, \, \, \, \, \, \, \, \, \, \, \vdots \\
\mathbb{E}\left[\widehat{\mathrm{Q}}_{n-2,j}\right]  =\mathbb{E}\left[\widehat{\mathrm{W}}_{n-2,j}\right] &+\mathbb{E}\left[\widehat{\mathrm{Q}}_{n-1,0}\right] q_{n-2,j}+\boldsymbol{\sum} _{l=1}^{n-2}\mathbb{E}\left[\widehat{\mathrm{Q}}_{n-3,l}\right]p_{n-2,j}^{l},\\
 & \, \, \, \, \, \, \, \, \, \, \, \, \, \, \, \, \, \, \vdots \\
\mathbb{E}\left[\widehat{\mathrm{Q}}_{n-2,n-1}\right]  =\mathbb{E}\left[\widehat{\mathrm{W}}_{n-2,n-1}\right] &+\mathbb{E}\left[\widehat{\mathrm{Q}}_{n-1,0}\right] q_{n-2,n-1}+\boldsymbol{\sum} _{l=1}^{n-2}\mathbb{E}\left[\widehat{\mathrm{Q}}_{n-3,l}\right]p_{n-2,n-1}^{l}.
 \end{aligned}
\end{equation}

Moreover, for $\widehat{\mathrm{Q}}_{n-1,0}$ to $\widehat{\mathrm{Q}}_{n-1,n}$, we get
\begin{equation}\label{qn_1_0}
\begin{aligned}
\mathbb{E}\left[\widehat{\mathrm{Q}}_{n-1,0}\right]  =\mathbb{E}\left[\widehat{\mathrm{W}}_{n-1,0}\right] &+\mathbb{E}\left[\widehat{\mathrm{Q}}_{n,0}\right] q_{n-1,0}+\boldsymbol{\sum} _{l=1}^{n-1}\mathbb{E}\left[\widehat{\mathrm{Q}}_{n-2,l}\right]p_{n-1,0}^{l},\\
 & \, \, \, \, \, \, \, \, \, \, \, \, \, \, \, \, \, \, \vdots \\
\mathbb{E}\left[\widehat{\mathrm{Q}}_{n-1,j}\right]  =\mathbb{E}\left[\widehat{\mathrm{W}}_{n-1,j}\right] &+\mathbb{E}\left[\widehat{\mathrm{Q}}_{n,0}\right] q_{n-1,j}+\boldsymbol{\sum} _{l=1}^{n-1}\mathbb{E}\left[\widehat{\mathrm{Q}}_{n-2,l}\right]p_{n-1,j}^{l},\\
 & \, \, \, \, \, \, \, \, \, \, \, \, \, \, \, \, \, \, \vdots \\
\mathbb{E}\left[\widehat{\mathrm{Q}}_{n-1,n}\right]  =\mathbb{E}\left[\widehat{\mathrm{W}}_{n-1,n}\right] &+\mathbb{E}\left[\widehat{\mathrm{Q}}_{n,0}\right] q_{n-1,n}+\boldsymbol{\sum} _{l=1}^{n-1}\mathbb{E}\left[\widehat{\mathrm{Q}}_{n-2,l}\right]p_{n-1,n}^{l}.
 \end{aligned}
\end{equation}
Finally, the following relationship is concluded for $\mathbb{E}\left[\widehat{\mathrm{Q}}_{n,0}\right]$.
\begin{equation}\label{qn0}
\begin{aligned}
\mathbb{E}\left[\widehat{\mathrm{Q}}_{n,0}\right]  =\mathbb{E}\left[\widehat{\mathrm{W}}_{n,0}\right] +\mathbb{E}\left[\widehat{\mathrm{Q}}_{n+1,0}\right] q_{n,0}+\boldsymbol{\sum} _{l=1}^{n}\mathbb{E}\left[\widehat{\mathrm{Q}}_{n-1,l}\right]p_{n,0}^{l}.
\end{aligned}
\end{equation}
According to Lemma \ref{expecteTimeToFailure}, $\mathbb{E}\left[\widehat{\mathrm{Q}}_{n+1,0}\right]=0$. Hence, \eqref{qn0} can be rewritten as follows:
\begin{equation}\label{eq78}
\begin{aligned}
\mathbb{E}\left[\widehat{\mathrm{Q}}_{n,0}\right] & =\mathbb{E}\left[\widehat{\mathrm{W}}_{n,0}\right]+\boldsymbol{\sum} _{j=1}^{n}\mathbb{E}\left[\widehat{\mathrm{Q}}_{n-1,j}\right]p_{n,0}^{j}.
\end{aligned}
\end{equation}
We first focus on expanding term $\mathbb{E}\left[\widehat{\mathrm{Q}}_{n-1,j}\right]$. To this end, multiplying $\mathbb{E}\left[\widehat{\mathrm{Q}}_{n-1,j}\right]$ from \eqref{qn_1_0} by $p_{n,0}^j$, for $1\le j\le n$, and taking the summation yields
\begin{equation}\label{eq79}
\hspace{-4mm}
\begin{aligned}
 \sum_{j=1}^{n} p_{n,0}^j&\mathbb{E}\left[\widehat{\mathrm{Q}}_{n-1,j}\right] =\sum_{j=1}^{n}p_{n,0}^j\mathbb{E}\left[\widehat{\mathrm{W}}_{n-1,j}\right]+\sum_{j=1}^{n}p_{n,0}^j\mathbb{E}\left[\widehat{\mathrm{Q}}_{n,0}\right] q_{n-1,j}+\sum_{j=1}^{n}p_{n,0}^j\boldsymbol{\sum} _{l=1}^{n-1}\mathbb{E}\left[\widehat{\mathrm{Q}}_{n-2,l}\right]p_{n-1,j}^{l}\\
 &=\sum_{j=1}^{n}p_{n,0}^j\mathbb{E}\left[\widehat{\mathrm{W}}_{n-1,j}\right]+\sum_{j=1}^{n}p_{n,0}^j\mathbb{E}\left[\widehat{\mathrm{Q}}_{n,0}\right] q_{n-1,j}+\left(\boldsymbol{\sum} _{l=1}^{n-1}\mathbb{E}\left[\widehat{\mathrm{Q}}_{n-2,l}\right]\right) \left(\sum_{j=1}^{n}p_{n,0}^jp_{n-1,j}^{l}\right).
\end{aligned}
\hspace{-4mm}
\end{equation}
Replacing $ \sum_{j=1}^{n} p_{n,0}^j\mathbb{E}\left[\widehat{\mathrm{Q}}_{n-1,j}\right] $ from \eqref{eq79} back in \eqref{eq78} leads to
\begin{equation}\label{eq80}
\hspace{-4mm}
\begin{aligned}
\mathbb{E}\left[\widehat{\mathrm{Q}}_{n,0}\right] & =\mathbb{E}\left[\widehat{\mathrm{W}}_{n,0}\right]+
 \sum_{j=1}^{n}p_{n,0}^j\mathbb{E}\left[\widehat{\mathrm{W}}_{n-1,j}\right] +\sum_{j=1}^{n}p_{n,0}^j\mathbb{E}\left[\widehat{\mathrm{Q}}_{n,0}\right] q_{n-1,j} +\left(\boldsymbol{\sum} _{l=1}^{n-1}\mathbb{E}\left[\widehat{\mathrm{Q}}_{n-2,l}\right]\right) \left(\sum_{j=1}^{n}p_{n,0}^jp_{n-1,j}^{l}\right).
\end{aligned}
\hspace{-4mm}
\end{equation}
Adding $-\sum_{j=1}^{n}p_{n,0}^j\mathbb{E}\left[\widehat{\mathrm{Q}}_{n,0}\right] q_{n-1,j}$ to both sides of \eqref{eq80} and performing some algebraic manipulations result in
\begin{equation}\label{eq81}
\hspace{-4mm}
\begin{aligned}
\mathbb{E}\left[\widehat{\mathrm{Q}}_{n,0}\right]&\left(1-\sum_{j=1}^{n}p_{n,0}^jq_{n-1,j}\right)=\mathbb{E}\left[\widehat{\mathrm{W}}_{n,0}\right]+
 \sum_{j=1}^{n}p_{n,0}^j\mathbb{E}\left[\widehat{\mathrm{W}}_{n-1,j}\right] +\left(\boldsymbol{\sum} _{l=1}^{n-1}\mathbb{E}\left[\widehat{\mathrm{Q}}_{n-2,l}\right]\right) \left(\sum_{j=1}^{n}p_{n,0}^jp_{n-1,j}^{l}\right).\\
\end{aligned}
\hspace{-4mm}
\end{equation}
Dividing both sides of \eqref{eq81} by $\left(1 -\sum_{j=1}^{n}p_{n,0}^jq_{n-1,j}\right)$ yields
\begin{equation}\label{eq82}
\begin{aligned}
\mathbb{E}&\left[\widehat{\mathrm{Q}}_{n,0}\right]  =\frac{1}{1-\left(\boldsymbol{\sum} _{j=1}^{n} q_{n-1,j}p_{n,0}^{j}\right)}\Bigg(\mathbb{E}\left[\widehat{\mathrm{W}}_{n,0}\right] +\boldsymbol{\sum} _{j=1}^{n}\mathbb{E}\left[\widehat{\mathrm{W}}_{n-1,j}\right]p_{n,0}^{j} +\left(\boldsymbol{\sum} _{l=1}^{n-1}\mathbb{E}\left[\widehat{\mathrm{Q}}_{n-2,1}\right]\right) \left(\boldsymbol{\sum} _{j=1}^{n}p_{n-1,j}^{l}p_{n,0}^{j}\right)\Bigg).
\end{aligned}
\end{equation}
We next aim to calculate $\mathbb{E}[\widehat{\mathrm{Q}}_{n-1}^0]$ from \eqref{qn_1_0}. First, we define following auxiliary function:
\begin{equation}\label{alphak}
  \alpha(k)= \frac{q_{k-1,0}}{1-\left(\boldsymbol{\sum} _{j=1}^{k} q_{k-1,j} A(k+1,j)\right)}.
\end{equation}
In \eqref{alphak}, function $A$ is calculated as follows:
\begin{equation}\label{eq85}
\displaystyle A(i,k) =\alpha(i)\boldsymbol{\sum} _{j=1}^{i}p_{i-1,j}^{k} A(i+1,j) +p_{i-1,0}^{k},
\end{equation}
where $A(n+1,k)=p_{n,0}^k$.\par
Multiplying both sides of \eqref{eq82} by $q_{n-1,0}$ leads to
\begin{equation}\label{eq135}
\hspace{-4mm}
\begin{aligned}
\mathbb{E}&\left[\widehat{\mathrm{Q}}_{n,0}\right]q_{n-1,0}  =\frac{q_{n-1,0}}{1-\left(\boldsymbol{\sum} _{j=1}^{n} q_{n-1,j}p_{n,0}^{j}\right)}\left(\mathbb{E}\left[\widehat{\mathrm{W}}_{n,0}\right] +\boldsymbol{\sum} _{j=1}^{n}\mathbb{E}\left[\widehat{\mathrm{W}}_{n-1,j}\right]p_{n,0}^{j} +\left(\boldsymbol{\sum} _{l=1}^{n-1}\mathbb{E}\left[\widehat{\mathrm{Q}}_{n-2,1}\right]\right) \left(\boldsymbol{\sum} _{j=1}^{n}p_{n-1,j}^{l}p_{n,0}^{j}\right)\right).
\end{aligned}
\hspace{-4mm}
\end{equation}
Considering \eqref{alphak}, we can rewrite \eqref{eq135} as follows:
\begin{equation}\label{eq86}
\begin{aligned}
\mathbb{E}&\left[\widehat{\mathrm{Q}}_{n,0}\right]q_{n-1,0}  =\alpha(n)\left(\mathbb{E}\left[\widehat{\mathrm{W}}_{n,0}\right]+\boldsymbol{\sum} _{j=1}^{n}\mathbb{E}\left[\widehat{\mathrm{W}}_{n-1,j}\right]p_{n,0}^{j} +\left(\boldsymbol{\sum} _{l=1}^{n-1}\mathbb{E}\left[\widehat{\mathrm{Q}}_{n-2,1}\right]\right) \left(\boldsymbol{\sum} _{j=1}^{n}p_{n-1,j}^{l}p_{n,0}^{j}\right)\right).
\end{aligned}
\end{equation}
By replacing $\mathbb{E}[\widehat{\mathrm{Q}}_{n,0}]q_{n-1,0}$ from \eqref{eq86} in $\mathbb{E}[\widehat{\mathrm{Q}}_{n-1,0}]$ from \eqref{qn_1_0}, and performing simplification, we get
\begin{equation}\label{eq87}
\begin{aligned}
\mathbb{E}&\left[\widehat{\mathrm{Q}}_{n-1,0}\right]  =\mathbb{E}\left[\widehat{\mathrm{W}}_{n-1,0}\right]+\left(\alpha(n)\left[\mathbb{E}\left[\widehat{\mathrm{W}}_{n,0}\right]
 +\boldsymbol{\sum}_{j=1}^{n}\mathbb{E} \left[\widehat{\mathrm{W}}_{n-1,j}\right]p_{n,0}^{j}\right]\right)\\
 & +\boldsymbol{\sum} _{l=1}^{n-1}\mathbb{E}\left[\widehat{\mathrm{Q}}_{n-2,1}\right]\times\left(\frac{q_{n-1,0}\left(\boldsymbol{\sum} _{j=1}^{n}p_{n-1,j}^{l}p_{n,0}^{j}\right)}{1-\left(\boldsymbol{\sum} _{j=1}^{n} q_{n-1,j}p_{n,0}^{j}\right)} +p_{n-1,0}^{l}\right)\\
 &=\mathbb{E}\left[\widehat{\mathrm{W}}_{n-1,0}\right]+\left(\alpha(n)\left[\mathbb{E}\left[\widehat{\mathrm{W}}_{n,0}\right]
 +\boldsymbol{\sum}_{j=1}^{n}\mathbb{E} \left[\widehat{\mathrm{W}}_{n-1,j}\right]p_{n,0}^{j}\right]\right)+\boldsymbol{\sum} _{j=1}^{n-1}\mathbb{E}\left[\widehat{\mathrm{Q}}_{n-2,1}\right] A(n,j).
\end{aligned}
\end{equation}
Multiplying $\mathbb{E}\left[\widehat{\mathrm{Q}}_{n-2,j}\right]$ from \eqref{qn_2_0} by $A(n,j)$, for $1\le j\le n-1$, and adding them all result in:
\begin{equation}\label{eq88}
\begin{aligned}
 \sum_{j=1}^{n-1}&\mathbb{E}\left[\widehat{\mathrm{Q}}_{n-2,j}\right]A(n,j) =\sum_{j=1}^{n-1}A(n,j)\mathbb{E}\left[\widehat{\mathrm{W}}_{n-2,j}\right]+\sum_{j=1}^{n-1}A(n,j)\mathbb{E}\left[\widehat{\mathrm{Q}}_{n-1,0}\right] q_{n-2,j}+\sum_{j=1}^{n-1}A(n,j)\boldsymbol{\sum} _{l=1}^{n-2}\mathbb{E}\left[\widehat{\mathrm{Q}}_{n-3,l}\right]p_{n-2,j}^{l}\\
 &=\sum_{j=1}^{n-1}A(n,j)\mathbb{E}\left[\widehat{\mathrm{W}}_{n-2,j}\right]+\sum_{j=1}^{n-1}A(n,j)\mathbb{E}\left[\widehat{\mathrm{Q}}_{n-1,0}\right] q_{n-2,j}+\left(\boldsymbol{\sum} _{l=1}^{n-2}\mathbb{E}\left[\widehat{\mathrm{Q}}_{n-3,l}\right]\right) \left(\sum_{j=1}^{n-1}A(n,j)p_{n-2,j}^{l}\right).
\end{aligned}
\end{equation}
Therefore, by replacing $\sum_{j=1}^{n-1}\mathbb{E}\left[\widehat{\mathrm{Q}}_{n-2,j}\right]A(n,j)$ from \eqref{eq88} in \eqref{eq87} and conducting simplifications, the following result is concluded
\begin{equation}
\hspace{-5mm}
\begin{aligned}
\mathbb{E}\left[\widehat{\mathrm{Q}}_{n-1,0}\right]  &=\frac{1}{1-\left(\boldsymbol{\sum} _{j=1}^{n-1} q_{n-2,j} A_{n}^{j}\right)}\vast[\mathbb{E}\left[\widehat{\mathrm{W}}_{n-1,0}\right] +\Bigg(\alpha(n-1)\Bigg[\mathbb{E}\left[\widehat{\mathrm{W}}_{n,0}\right]
 +\boldsymbol{\sum}_{j=1}^{n}\mathbb{E} \left[\widehat{\mathrm{W}}_{n-1,j}\right]p_{n,0}^{j}\Bigg]\Bigg)\\
 & ~~~~~~~~~+\boldsymbol{\sum}_{j=0}^{n-1}\mathbb{E} \left[\widehat{\mathrm{W}}_{n-2,j}\right] A(n,j) +\boldsymbol{\sum} _{l=1}^{n-1}\mathbb{E}\left[\widehat{\mathrm{Q}}_{n-2,1}\right]\left(\boldsymbol{\sum} _{j=1}^{n-1}p_{n-2,j}^{l} A(n,j)\right)\vast].
\end{aligned}
\hspace{-3mm}
\end{equation}
By repeating the above operations up to $\mathbb{E}\left[\widehat{\mathrm{Q}}_{0,0}\right]$, as well as conducting simplifications, the result is as follows:
\begin{equation}
\begin{aligned}
\mathbb{E}\left[\widehat{\mathrm{Q}}_{0,0}\right]  &=\Bigg(\mathbb{E}\left[\widehat{\mathrm{W}}_{0,0}\right] +\alpha(1)\mathbb{E}\left[\widehat{\mathrm{W}}_{1,0}\right]+\alpha(1)\alpha(2)\mathbb{E}\left[\widehat{\mathrm{W}}_{2,0}\right] +\dots+\left(\boldsymbol{\prod} _{k=1}^{n} \alpha(k)\right)\mathbb{E}\left[\widehat{\mathrm{W}}_{n,0}\right]\Bigg)\\
 & +\Bigg(\alpha(1) \boldsymbol{\sum} _{j=1}^{1}\mathbb{E}\left[\widehat{\mathrm{W}}_{0,j}\right] A(2,j) +\alpha(1)\alpha(2)\boldsymbol{\sum} _{j=1}^{2}\mathbb{E}\left[\widehat{\mathrm{W}}_{1,j}\right] A(3,j)+\\
 &~~~~~~~~~~ \cdots+\left(\boldsymbol{\prod} _{k=1}^{n}\alpha(k)\right)\boldsymbol{\sum} _{j=1}^{n}\mathbb{E}\left[\widehat{\mathrm{W}}_{n-1,j}\right] A(n+1,j)\Bigg).
\end{aligned}
\end{equation}
Further, for the sake of simplicity in notations, we define function $\beta(i)$ as follows:
\begin{equation}
  \beta(i)=\boldsymbol{\prod} _{k=1}^{i}\alpha(k).
\end{equation}
Finally, $\mathbb{E}\left[\widehat{\mathrm{Q}}_{0,0}\right]$ is obtained as follows:
\begin{equation}
\hspace{-4mm}
\begin{aligned}
\mathbb{E}\left[\widehat{\mathrm{Q}}_{0,0}\right]  &=\left(\mathbb{E}\left[\widehat{\mathrm{W}}_{0,0}\right] +\sum_{i=1}^{n}\left(\boldsymbol{\prod} _{k=1}^{i} \alpha(k)\right)\mathbb{E}\left[\widehat{\mathrm{W}}_{i,0}\right]\right) +\left(\sum_{i=1}^{n}\left(\boldsymbol{\prod} _{k=1}^{i}\alpha(k)\right)\boldsymbol{\sum} _{j=1}^{i}\mathbb{E}\left[\widehat{\mathrm{W}}_{i-1,j}\right] A(i+1,j)\right)\\
 &=\mathbb{E}\left[\widehat{\mathrm{W}}_{0,0}\right]
 +\sum_{i=1}^{n}\boldsymbol{\sum} _{j=1}^{i}\beta(i)\mathbb{E}\left[\widehat{\mathrm{W}}_{i-1,j}\right] A(i+1,j)+\sum_{i=1}^{n}\beta(i)\mathbb{E}\left[\widehat{\mathrm{W}}_{i,0}\right],
\end{aligned}
\hspace{-4mm}
\end{equation}
which concludes the proof.

\newpage
\section{Proof of corollary \ref{expj2nCorollary}}\label{expj2nCorollary_proof}
\noindent First, consider following technical result.
\begin{lemma}[\hspace{-.1mm}\cite{29}]\label{expLem}
Let independent random variables $Z$ and $U$ follow an exponential probability distribution with parameter $\lambda_z$ and a general distribution $R_U(u)$, respectively. EDV $V = Z-U$ follows the exponential distribution with parameter $\lambda_z$.
\end{lemma}
Accordingly, we next aim to prove Corollary \ref{expj2nCorollary}. Based on Decomposition Theorem (Theorem \ref{decompositionTheorem}), the remaining sojourn times of the vehicles and the remaining recruitment times of the recruiters in state $S_{i,j}$ can follow five different e-distributions $\dot{\mathfrak{X}}$, $\ddot{\mathfrak{X}}$, $\varphi$, $\gamma$, and $\psi$. Without loss of generality, consider the following EDVs:
\begin{equation}
    \begin{cases}
        V_1{\sim}\dot{\mathfrak{X}}{=}Z_1{-}Z_2,~\hat{\Gamma}(V_1){=}Z_1,\\
        V_2{\sim}\ddot{\mathfrak{X}}{=}Z_1{-}Z_2{-}U_1,~\hat{\Gamma}(V_2){=}Z_1,\\
        V_3{\sim}\varphi{=}U_1{-}Z_2{+}Z_1,~\hat{\Gamma}(V_3){=}U_1,\\
        V_4{\sim}\psi{=}U_1{-}(U_2{-}Z_1{+}Z_2),~\hat{\Gamma}(V_4){=}U_1,\\
        V_5{\sim}\gamma{=}(U_1{-}Z_2{+}Z_1){-}U_2,~\hat{\Gamma}(V_5){=}U_1.\\
    \end{cases}
\end{equation}
If $Z_1$ and $Z_2$ are i.i.d exponential random variables with parameter $\lambda_z$, and $U_1$ and $U_2$ are i.i.d exponential random variables with parameter $\lambda_u$, the distribution of the pivot of $V_1$, $V_2$, $V_3$, $V_4$, and $V_5$ are as follows.
\begin{enumerate}
    \item $\hat{\Gamma}(V_1)$ and $\hat{\Gamma}(V_2)$ follow exponential distribution with parameter $\lambda_z$.
    \item $\hat{\Gamma}(V_3)$, $\hat{\Gamma}(V_4)$, and $\hat{\Gamma}(V_5)$ follow exponential distribution with parameter $\lambda_u$.
\end{enumerate}
Using Definition \ref{event_dynamic_variable},
we know that $\hat{\Gamma}(V')>\hat{\delta}(V')-\hat{\xi}(V')$, where $V'\in \{V_1,V_2,V_3,V_4,V_5\}$. Based on Lemma \ref{expLem}, referring to memoryless property, $V_1$ and $V_2$ follow exponential distribution with parameters $\lambda_z$, and $V_3$, $V_4$, and $V_5$ follow exponential distribution with parameter $\lambda_u$, which concludes the proof.

\newpage
\section{Proof of Theorem \ref{mttfEXPTheorem}} \label{mttfEXPTheorem_proof}
\noindent Consider the $MTTF(${\tt RD-VC$_n$}) presented in Theorem \ref{MTTFJ2dn}. We have
\begin{equation}\label{mttfGen}
\begin{aligned}
\mathbb{E}\left[\widehat{\mathrm{Q}}_{0,0}\right]
  &=\mathbb{E}\left[\widehat{\mathrm{W}}_{0,0}\right]+\underbrace{\boldsymbol{\sum} _{i=1}^{n}\beta(i) \mathbb{E}\left[\widehat{\mathrm{W}}_{i,0}\right]}_{a}+\underbrace{\boldsymbol{\sum} _{i=1}^{n} \boldsymbol{\sum} _{j=1}^{i}\beta(i)\mathbb{E}\left[\widehat{\mathrm{W}}_{i-1,j}\right] A(i+1,j)}_{b}.
\end{aligned}
\end{equation}
If we extract the $n^{th}$ term from $(a)$ and the first term from $(b)$, the following result is concluded
\begin{equation}\label{eqq149}
\begin{aligned}
\mathbb{E}\left[\widehat{\mathrm{Q}}_{0,0}\right] & =\mathbb{E}\left[\widehat{\mathrm{W}}_{0,0}\right]
 +\beta(1)\mathbb{E}\left[\widehat{\mathrm{W}}_{0,1}\right] A(2,1)+\sum_{i=2}^{n}\boldsymbol{\sum} _{j=1}^{i}\beta(i)\mathbb{E}\left[\widehat{\mathrm{W}}_{i-1,j}\right] A(i+1,j)+\sum_{i=1}^{n-1}\beta(i)\mathbb{E}\left[\widehat{\mathrm{W}}_{i,0}\right] +\beta(n)\mathbb{E}\left[\widehat{\mathrm{W}}_{n,0}\right]\\
 &=\mathbb{E}\left[\widehat{\mathrm{W}}_{0,0}\right]
 +\beta(1)\mathbb{E}\left[\widehat{\mathrm{W}}_{0,1}\right] A(2,1)+\underbrace{\sum_{i=1}^{n-1}\boldsymbol{\sum} _{j=1}^{i+1}\beta(i+1)\mathbb{E}\left[\widehat{\mathrm{W}}_{i,j}\right] A(i+2,j)}_{c}+\underbrace{\sum_{i=1}^{n-1}\beta(i)\mathbb{E}\left[\widehat{\mathrm{W}}_{i,0}\right] +\beta(n)\mathbb{E}\left[\widehat{\mathrm{W}}_{n,0}\right]}_{d}.
\end{aligned}
\end{equation}
By merging $(c)$ and $(d)$ from \eqref{eqq149}, we get
\begin{equation}\label{eq94}
\begin{aligned}
\mathbb{E}\left[\widehat{\mathrm{Q}}_{0,0}\right] & =\mathbb{E}\left[\widehat{\mathrm{W}}_{0,0}\right]
 +\beta(1)\mathbb{E}\left[\widehat{\mathrm{W}}_{0,1}\right] A(2,1)+\sum_{i=1}^{n-1}\left(\boldsymbol{\sum} _{j=1}^{i+1}\beta(i+1)\mathbb{E}\left[\widehat{\mathrm{W}}_{i,j}\right] A(i+2,j)+\beta(i)\mathbb{E}\left[\widehat{\mathrm{W}}_{i,0}\right]\right) +\beta(n)\mathbb{E}\left[\widehat{\mathrm{W}}_{n,0}\right].
\end{aligned}
\end{equation}
Using Corollary \ref{expj2nCorollary}, it can be inferred that all of the states of $\hat{\mathbf{S}}_i$ are equal. As a result, $\mathbb{E}\left[\widehat{\mathrm{W}}_{i,j}\right]=\mathbb{E}\left[\widehat{\mathrm{W}}_{i,0}\right]$. Consequently, replacing $\mathbb{E}\left[\widehat{\mathrm{W}}_{i,j}\right]$ from  \eqref{eq94}, for $1\le j \le i+1$, with $\mathbb{E}\left[\widehat{\mathrm{W}}_{i,0}\right]$ yields
\begin{equation}
\begin{aligned}
\mathbb{E}\left[\widehat{\mathrm{Q}}_{0,0}\right]&=\mathbb{E}\left[\widehat{\mathrm{W}}_{0,0}\right]
 +\beta(1)\mathbb{E}\left[\widehat{\mathrm{W}}_{0,0}\right] A(2,1)+\sum_{i=1}^{n-1}\left(\boldsymbol{\sum} _{j=1}^{i+1}\beta(i+1)\mathbb{E}\left[\widehat{\mathrm{W}}_{i,0}\right] A(i+2,j)+\beta(i)\mathbb{E}\left[\widehat{\mathrm{W}}_{i,0}\right]\right) +\beta(n)\mathbb{E}\left[\widehat{\mathrm{W}}_{n,0}\right]\\
 &= \mathbb{E}\left[\widehat{\mathrm{W}}_{0,0}\right] \left(1+\beta(1)A(2,1)\right)+\sum_{i=1}^{n-1}\mathbb{E}\left[\widehat{\mathrm{W}}_{i,0}\right] \left(\boldsymbol{\sum} _{j=1}^{i+1}\beta(i+1) A(i+2,j)+\beta(i)\right)+\beta(n)\mathbb{E}\left[\widehat{\mathrm{W}}_{n,0}\right].\\
\end{aligned}
\end{equation}
Moreover, considering Corollary \ref{expj2nCorollary}, $q_{k,j}=q_{k,0}$ and $p_{i,j}^k=p_{i,0}^1$. Therefore, function $\alpha(k)$, $A(i,k)$, and $\beta(i)$ can be rewritten as follows with new names $\alpha'(k)$, $A'(i,k)$, and $\beta'(i)$, respectively. By replacing $q_{k,j}$ and $p_{i,j}^k$ from \eqref{alphaEq} and \eqref{eqq85} with $q_{k,0}$ and $p_{i,0}^1$, respectively, the following results are concluded:
\begin{equation}\label{eqq96}
\begin{aligned}
  \alpha'(k)&= \frac{q_{k-1,0}}{1-\left(\boldsymbol{\sum} _{j=1}^{k} q_{k-1,0} A'(k+1,j)\right)}\\
  &=\frac{q_{k-1,0}}{1-q_{k-1,0}\left(\boldsymbol{\sum} _{j=1}^{k} A'(k+1,j)\right)},
\end{aligned}
\end{equation}
and
\begin{equation}\label{eq97}
\begin{aligned}
\displaystyle A'(i,k)&=\alpha'(i)\boldsymbol{\sum} _{j=1}^{i}p_{i-1,0}^{1} A'(i+1,j) +p_{i-1,0}^{1}\\
&=\alpha'(i) \times p_{i-1,0}^{1}\boldsymbol{\sum} _{j=1}^{i}A'(i+1,j) +p_{i-1,0}^{1}.
\end{aligned}
\end{equation}
Replacing $\alpha'(i)$ from (\ref{eqq96}) in (\ref{eq97}) yields
\begin{equation}
\hspace{-4mm}
\begin{aligned}
\displaystyle A'(i,k)&=\frac{q_{i-1,0} {\times} p_{i-1,0}^{1}\boldsymbol{\sum} _{j=1}^{i}A'(i+1,j)}{1-q_{i-1,0}\left(\boldsymbol{\sum} _{j=1}^{i} A'(i+1,j)\right)} +p_{i-1,0}^{1}\\
&=\frac{p_{i-1,0}^{1}}{1-q_{i-1,0}\left(\boldsymbol{\sum} _{j=1}^{i} A'(i+1,j)\right)}.\\
\end{aligned}
\hspace{-4mm}
\end{equation}
Since $\displaystyle A'(i,k)$ does not depend on $k$ we can remove $k$ from $A'(i,k)$ leading to a simpler notion $A'(i)$. Thus, $A'(i)$ can be further simplified as follows:
\begin{equation}
\begin{aligned}
A'(i)=\frac{p_{i-1,0}^{1}}{1-i\times q_{i-1,0}A'(i+1)},\\
\end{aligned}
\end{equation}
where $A'(n+1)=p_{n,0}^1$. Accordingly, $\beta(i)$ can be rewritten with new name $\beta'(i)$ as follows:
\begin{equation}
  \beta'(i)=\boldsymbol{\prod} _{k=1}^{i}\alpha'(k).
\end{equation}
Moreover, $\alpha'(k)$ can be further simplified as follows:
\begin{equation}
\begin{aligned}
  \alpha'(k)=\frac{q_{k-1,0}}{1-k\times q_{k-1,0}\times A'(k+1)}.
\end{aligned}
\end{equation}
To conclude the proof, $\mathbb{E}\left[\widehat{\mathrm{Q}}_{0,0}\right]$ is obtained as follows:
\begin{equation}
\begin{aligned}
&\mathbb{E}\left[\widehat{\mathrm{Q}}_{0,0}\right] = \mathbb{E}\left[\widehat{\mathrm{W}}_{0,0}\right] \left(1+\beta'(1)A'(2)\right)+\sum_{i=1}^{n-1}\mathbb{E}\left[\widehat{\mathrm{W}}_{i,0}\right] \left((i+1)\beta'(i+1) A'(i+2)+\beta'(i)\right)+\beta'(n)\mathbb{E}\left[\widehat{\mathrm{W}}_{n,0}\right].\\
\end{aligned}
\end{equation}

\newpage
\section{Proof of Corollary \ref{expected_sojourn_time}}\label{expected_sojourn_time_proof}
\noindent Let
\begin{equation}
 M^{(Z)}_{i,0}=\min_{1\le r \le 2n-i}\left\{S_{i,0}^{\mathsf{Soj}}(r)\right\},
\end{equation}
and
\begin{equation}
M^{(U)}_{i,0}=\min_{\substack{1\le h \le i}}\left\{S_{i,0}^{\mathsf{Rec}}(h)\right\}.
\end{equation}
Based on \eqref{ephi1_main} and taking into account that the considered random variables are independent, we have
\begin{equation}
\begin{aligned}
\mathbb{E}\left[\widehat{\mathrm{W}}_{i,0}\right] &=\boldsymbol{\int}_{0}^{\infty } \textrm{Pr}\Bigg[\min\Big\{M^{(Z)}_{i,0},M^{(U)}_{i,0}\Big\}>t\Big] dt\\
&=\boldsymbol{\int} _{0}^{\infty } \textrm{Pr}\left[M^{(Z)}_{i,0}>t\right]\textrm{Pr}\left[M^{(U)}_{i,0}>t\right] dt\\
&=\boldsymbol{\int} _{0}^{\infty } \left(\prod_{k=1}^{2n-i}\textrm{Pr}\left[S_{i,0}^{\mathsf{Soj}}(k)>t\right] \prod_{r=1}^{i}\textrm{Pr}\left[S_{i,0}^{\mathsf{Rec}}(r)>t\right]\right) dt\\
&=\boldsymbol{\int} _{0}^{\infty } e^{-(( 2n-i) \lambda _{c} +i\lambda _{u}) t} dt\\
 & =\frac{1}{( 2n-i) \lambda _{c} +i\lambda _{u}},
\end{aligned}
\end{equation}
which concludes the proof.

\newpage
\section{Proof of Corollary \ref{transitionProbCorollary}}\label{transitionProbCorollary_proof}
\noindent Let $U_k=S_{i,j}^{\mathsf{Rec}}(k)$ and
\begin{equation}
\hspace{-6mm}
V_{i,j,k}^{\mathsf{Min}}=\min\Big\{\hspace{-.5mm}\min_{1\le r \le 2n-i}\left\{S_{i,j}^{\mathsf{Soj}}(r)\right\},\min_{\substack{1\le h \le i,\\h\neq k}}\left\{S_{i,j}^{\mathsf{Rec}}(h)\right\}\hspace{-1mm}\Big\}.  \hspace{-6mm}
\end{equation}
Considering general expression of $p_{i,j}^k$, presented in \ref{pijk_main}, we have
\begin{equation}
\begin{aligned}
p_{i,j}^{k} {=}\textrm{Pr}\Bigg[U_k {<}V_{i,j,k}^{\mathsf{Min}}\Bigg].
\end{aligned}
\end{equation}
Accordingly, $p_{i,0}^1$, for $0\le i \le n$, is calculated as follows:
\begin{equation}\label{pijk_app}
\begin{aligned}
\displaystyle p_{i,0}^1 &=\textrm{Pr}\left[ U_1 <V_{i,0,1}^{\mathsf{Min}} \right]\\
& =\boldsymbol{\int} _{0}^{\infty }\textrm{Pr}\left[ U_{0}< V_{i,0,1}^{\mathsf{Min}} |U_1 =u\right] \lambda _{u} e^{-\lambda _{u}U_1} du\\
 & =\boldsymbol{\int} _{0}^{\infty } \lambda _{u} e^{-(( 2n-i) \lambda _{c} +i\lambda _{u}) u} du\\
 & =\frac{\lambda _{u}}{( 2n-i) \lambda _{c} +i\lambda _{u}},
\end{aligned}
\end{equation}
which concludes the proof.

\newpage
\section{Generalized Transition Probabilities and Expected Sojourn Times}\label{app:generalization}
\noindent In this paper, Decomposition Theorem (Theorem \ref{decompositionTheorem}) decomposes a non-trivial SMP into a decomposed SMP (D-SMP). The actual Mean Time to Failure (MTTF) of task processing is computed for D-SMP in Theorem \ref{MTTFJ2dn}. The above two theorems can be applied to any general distributions. However, in this paper, we assumed that residual sojourn time and recruitment duration of vehicles follow the exponential distribution to obtain the final results of \eqref{pijk_main} and \eqref{ephi1_main}, and then replaced their expressions back in Theorem \ref{MTTFJ2dn} to obtain MTTF under exponential distribution assumption. To generalize the results, in the following, we present Theorem \ref{transition_prob_theorem} and Theorem \ref{expected_sojourn_time_theorem} to obtained \eqref{pijk_main} and \eqref{ephi1_main} for the case where the distribution of residual sojourn time and recruitment duration of vehicles can follow arbitrary distributions.

It is worth mentioning that the final results of the below theorems require finding the joint distribution of the sum of several random variables, which can be derived through analytical approaches such as probability characteristic function and convolution for the exact (if it is possible) and approximation solutions \cite{54}.

\begin{theorem}\label{transition_prob_theorem}
Let $S_{i,j}^{\mathsf{Rec}}(k)$, $S_{i,j}^{\mathsf{Rec}}(h)$, and $S_{i,j}^{\mathsf{Soj}}(r-i)$ refer to the residual \underline{rec}ruitment time of vehicle $k$, residual recruitment time of vehicle $h$, and residual sojourn/\underline{soj}ourn time of vehicle $r$, respectively. When vehicles' sojourn time and recruitment duration follow general distributions $\mathfrak{D}_Z(z)$ and $\mathfrak{R}_U(u)$, respectively, the probability of transition from $S_{i,j}$ to $S_{i-1,k}$, denoted by $p_{i,j}^{k}$, is given by the following three cases:

\textbullet\hspace{.5mm} \textbf{Case 1 (computation of $p_{i,j}^{k}$ when $j=0$, $i\in \{0,\cdots,n\}$, and $k=i$):} In this case, we have
\begin{equation}\label{transition_prob1}
\begin{aligned}
p_{i,0}^{i} &{=}\textrm{Pr}\Bigg[S_{i,0}^{\mathsf{Rec}}(i) {<} \min\Big\{\min_{i+1 \le r \le 2n}\left\{S_{i,0}^{\mathsf{Soj}}(r-i)\right\},\min_{\substack{1 \le h \le i-1}}\left\{S_{i,0}^{\mathsf{Rec}}(h)\right\}\Big\}\Bigg]\\
&=\left(\textrm{Pr}\Bigg[U_{x_{i}}{+}Z_{\ell'_{1}} {<}Z_a\big|Z_{a}>Z_{\ell'_{1}}\Bigg]\right)^{2n{-}i} \left(\textrm{Pr}\Bigg[U_{x_{i}}{+}Z_{\ell'_{1}} {<}U_b{+}Z_b\Bigg]\right)^{i{-}1},
\end{aligned}
\end{equation}
where $Z_{\ell_i}$, $Z_{\ell'_1}$, $Z_a$, $Z_b$, $U_{x_{i}}$, and $U_b$ are independent random variables, from which $Z_{\ell_i}$, $Z_{\ell'_1}$, $Z_a$, and $Z_b$ follow general distribution $\mathfrak{D}_Z(z)$, and $U_{x_{i}}$ and $U_b$ follow general distribution $\mathfrak{R}_U(u)$.

\textbullet\hspace{.5mm} \textbf{Case 2 (computation of $p_{i,j}^{k}$ when $j=0$, $i\in \{0,\cdots,n\}$, and $1 \le k \le i-1$):} In this case, we have
\begin{equation}\label{transition_prob2}
\begin{aligned}
p_{i,0}^{k} &{=}\textrm{Pr}\Bigg[S_{i,0}^{\mathsf{Rec}}(k) {<} \min\Big\{\min_{i+1 \le r \le 2n}\left\{S_{i,0}^{\mathsf{Soj}}(r-i)\right\},\min_{\substack{1 \le h \le i\\h\neq k}}\left\{S_{i,0}^{\mathsf{Rec}}(h)\right\}\Big\}\Bigg]\\
&=\left(\textrm{Pr}\Bigg[U_{x_{k}}{+}Z_{\ell_{k}} {<}Z_a\big|Z_{\ell_a}>Z_{\ell_{k}}\Bigg]\right)^{2n{-}i} \left(\textrm{Pr}\Bigg[U_{x_{k}}{+}Z_{\ell_{k}} {<}U_b{+}Z_b\Bigg]\right)^{i{-}2}\\
&~~~~~~~~~~~~~~~\times\textrm{Pr}\Bigg[U_{x_{k}}{-}Z_{\ell'_1}{+}Z_{\ell_{k}} {<} U_{x_i}\big|Z_{\ell'_1}>Z_{\ell_k},~U_{x_k}>Z_{\ell'_1}-Z_{\ell_k}\Big],
\end{aligned}
\end{equation}
where $Z_{\ell_k}$, $Z_{\ell'_1}$, $Z_a$, $Z_b$, $U_{x_{k}}$, and $U_b$ are independent random variables, from which $Z_{\ell_k}$, $Z_{\ell'_1}$, $Z_a$, and $Z_b$ follow general distribution $\mathfrak{D}_Z(z)$, and $U_{x_{k}}$ and $U_b$ follow general distribution $\mathfrak{R}_U(u)$.

\textbullet\hspace{.5mm} \textbf{Case 3 (computation of $p_{i,j}^{k}$ when $1 \le i \le n-1$, $1 \le j \le i+1$, and $1 \le k \le i$):} In this case, we have

\begin{equation}\label{transition_prob3}
\begin{aligned}
p_{i,j}^{k} &{=}\textrm{Pr}\Bigg[S_{i,j}^{\mathsf{Rec}}(k) {<} \min\Big\{\min_{i+1 \le r \le 2n}\left\{S_{i,j}^{\mathsf{Soj}}(r-i)\right\},\min_{\substack{1\le h\le j-1\\h\neq k}}\left\{S_{i,j}^{\mathsf{Rec}}(h)\right\},\min_{\substack{j \le g \le i \\ g\neq k}}\left\{S_{i,j}^{\mathsf{Rec}}(g)\right\}\Big\}\Bigg]\\
&=\left(\textrm{Pr}\Bigg[U_{x_{k}}{+}Z_{\ell_{k}} {<}Z_{\ell_a}\Bigg]\right)^{2n-i} \left(\textrm{Pr}\Bigg[U_{x_{k}}{+}Z_{\ell_{k}} {<}U_{x_b}+Z_b\Bigg]\right)^{i-1},
\end{aligned}
\end{equation}
where $Z_a$, $Z_b$, and $U_b$ are three independent random variables, from which $Z_a$ and $Z_b$ follow general distributions $\mathfrak{D}_Z(z)$ and $U_b$ follow general distribution $\mathfrak{R}_U(u)$

\end{theorem}

\begin{proof} The probability of transition from $S_{i,j}$ to $S_{i-1,k}$ (i.e., $p_{i,j}^{k}$) is equal to the probability that the recruiter $k$ completes its recruitment operation before any other events (i.e., departure and recruitment events of other vehicles). Mathematically, we get
\begin{equation}\label{pijk}
\begin{aligned}
p_{i,j}^{k} {=}\textrm{Pr}\Bigg[S_{i,j}^{\mathsf{Rec}}(k) {<} \min\Big\{\min_{i+1 \le r \le 2n}\left\{S_{i,j}^{\mathsf{Soj}}(r-i)\right\},\min_{\substack{1 \le h \le i\\ h \neq k}}\left\{S_{i,j}^{\mathsf{Rec}}(h)\right\}\Big\}\Bigg].
\end{aligned}
\end{equation}
Moreover, from Decomposition Theorem (Theorem \ref{decompositionTheorem}), we have (i) $S_{i,0}^{\mathsf{Soj}}{=}\mathcal{L}_1^{(\dot{\mathfrak{X}})}\hspace{-1.4mm}\left\langle {2n{-}i}\right\rangle$ and $S_{i,0}^{\mathsf{Rec}}=\mathcal{H}_1^{(\varphi;\mathfrak{R})}\hspace{-1.4mm}\left\langle {i{-}1} \hspace{-0.3mm};\hspace{-0.3mm} 1 \right\rangle$, and (ii) $S_{i,j}^{\mathsf{Soj}}=\mathcal{L}_2^{(\ddot{\mathfrak{X}})}\hspace{-1.4mm}\left\langle 2n{-}i\right\rangle$ and $S_{i,j}^{\mathsf{Rec}}{=}\mathcal{H}_2^{(\gamma;\psi)}\hspace{-1.4mm}\left\langle j{-}1 \hspace{-0.3mm};\hspace{-0.3mm} i{-}(j{-}1) \right\rangle$. To prove the above-mentioned three cases, we first expand (i) and (ii) and rewrite them in terms of their random variables below.

\textbf{(i):} According to Definition \ref{def19} and Definition \ref{phidist}, we have
\begin{equation}
    \mathcal{H}_1^{(\varphi;\mathfrak{R})}\hspace{-1.4mm}\left\langle {i{-}1} \hspace{-0.3mm};\hspace{-0.3mm} 1 \right\rangle = \left[U_{x_1}-Z_{\ell'_1}+Z_{\ell_1},U_{x_2}-Z_{\ell'_2}+Z_{\ell_2},\cdots,U_{x_{i{-}1}}-Z_{\ell'_{i{-}1}}+Z_{\ell_{i{-}1}}, U_{x_i}\right],
\end{equation}
where for the two last elements of the list, we have $U_{x_{i{-}1}}-Z_{\ell'_{i{-}1}}+Z_{\ell_{i{-}1}} {\dot{\preceq}} U_{x_i}$, and for the rest of the elements $1 \le h \le i-2$, we have $U_{x_h}-Z_{\ell'_h}+Z_{\ell_h} {\dot{\preceq}} U_{x_{h+1}}-Z_{\ell'_{h+1}}+Z_{\ell_{h+1}}$. Likewise, according to Definition \ref{def19} and Definition \ref{expdist}, we have
\begin{equation}
    \mathcal{L}_1^{(\dot{\mathfrak{X}})}\hspace{-1.4mm}\left\langle {2n{-}i}\right\rangle = \left[Z_{\ell_{i{+}1}}-Z_{\ell'_{i{+}1}},Z_{\ell_{i{+}2}}-Z_{\ell'_{i{+}2}},\cdots,Z_{\ell_{2n}}-Z_{\ell'_{2n}}\right],
\end{equation}
where, for all $i+1 \le r \le 2n-1$, it holds that $Z_{\ell_r}-Z_{\ell'_r} {\dot{\preceq}} Z_{\ell_{r{+}1}}-Z_{\ell'_{r{+}1}}$. From Decomposition Theorem, we also have $\mathcal{L}_1^{(\dot{\mathfrak{X}})}\hspace{-1.4mm}\left\langle {2n{-}i}\right\rangle {\dot{\preceq}} \mathcal{H}_1^{(\varphi;\mathfrak{R})}\hspace{-1.4mm}\left\langle {i{-}1} \hspace{-0.3mm};\hspace{-0.3mm} 1 \right\rangle$. Further, from Definition \ref{reducer}, we have $\hat{\mathbf{\delta}}(U_{x_h}-Z_{\ell'_h}+Z_{\ell_h})=Z_{\ell'_h}$ and $\hat{\mathbf{\delta}}(Z_{\ell_r}-Z_{\ell'_r})=Z_{\ell'_r}$. Moreover, according to Definition \ref{CG}, if relation ${\dot{\preceq}}$ holds between two EDVs $V_1$ and $V_2$, we have $\hat{\mathbf{\delta}}(V_1)=\hat{\mathbf{\delta}}(V_2)$. Subsequently, from Definition \ref{reducer} and Definition \ref{CG}, we have
\begin{equation}
\begin{aligned}
\hat{\mathbf{\delta}}(U_{x_h}-Z_{\ell'_h}+Z_{\ell_h})=\hat{\mathbf{\delta}}(U_{x_{1}}-Z_{\ell'_1}+Z_{\ell_{1}}) => Z_{\ell'_h}=Z_{\ell'_1}, 1 \le h \le i-1.
\end{aligned}
\end{equation}
Likewise, from Definition \ref{reducer}, Definition \ref{CG}, and Decomposition Theorem, we get
\begin{equation}
\begin{aligned}
\hat{\mathbf{\delta}}(Z_{\ell_{r}}-Z_{\ell'_r})=\hat{\mathbf{\delta}}(U_{x_{1}}-Z_{\ell'_1}+Z_{\ell_{1}}) => Z_{\ell'_r}=Z_{\ell'_1}, i+1 \le r \le 2n.
\end{aligned}
\end{equation}
Consequently, we can expand  $S_{i,0}^{\mathsf{Soj}}{=}\mathcal{L}_1^{(\dot{\mathfrak{X}})}\hspace{-1.4mm}\left\langle {2n{-}i}\right\rangle$ and $S_{i,0}^{\mathsf{Rec}}=\mathcal{H}_1^{(\varphi;\mathfrak{R})}\hspace{-1.4mm}\left\langle {i{-}1} \hspace{-0.3mm};\hspace{-0.3mm} 1 \right\rangle$ as follows
\begin{equation}\label{Us}
    \mathcal{H}_1^{(\varphi;\mathfrak{R})}\hspace{-1.4mm}\left\langle {i{-}1} \hspace{-0.3mm};\hspace{-0.3mm} 1 \right\rangle = \left[U_{x_1}-Z_{\ell'_1}+Z_{\ell_1},\cdots,U_{x_h}-Z_{\ell'_1}+Z_{\ell_h},\cdots,U_{x_{i{-}1}}-Z_{\ell'_1}+Z_{\ell_{i{-}1}}, U_{x_i}\right]
\end{equation}
and
\begin{equation}\label{Zs}
    \mathcal{L}_1^{(\dot{\mathfrak{X}})}\hspace{-1.4mm}\left\langle {2n{-}i}\right\rangle = \left[Z_{\ell_{i{+}1}}-Z_{\ell'_1},\cdots,Z_{\ell_{r}}-Z_{\ell'_1},\cdots,Z_{\ell_{2n}}-Z_{\ell'_1}\right].
\end{equation}
Moreover, for the random variables used in~\eqref{Us} and \eqref{Zs},  from Definition \ref{expdist} and Definition \ref{phidist}, the following conditions hold (c1) $Z_{\ell'_1}>Z_{\ell_h},~U_{x_h}>Z_{\ell'_1}-Z_{\ell_h},~ 1 \le h \le i-1$ and (c2) $Z_{\ell_{r}}>Z_{\ell'_1},~ i+1 \le r \le 2n$. Since we will use these conditions later, we compact them through two sets as follows:
\begin{equation}\label{}
   \varphi^{\mathsf{C}}_{h}=\{Z_{\ell'_1}>Z_{\ell_h},~U_{x_h}>Z_{\ell'_1}-Z_{\ell_h}\},~ 1 \le h \le i-1
\end{equation}
and
\begin{equation}
   \dot{\mathfrak{X}}^{\mathsf{C}}_{r}=\{Z_{\ell_{r}}>Z_{\ell'_1}\}~ i+1 \le r \le 2n.
\end{equation}

\textbf{(ii):} We initially expand $\mathcal{H}_2^{(\gamma;\psi)}\hspace{-1.4mm}\left\langle j{-}1 \hspace{-0.3mm};\hspace{-0.3mm} i{-}(j{-}1) \right\rangle$ according to Definition \ref{def19}, Definition \ref{psidist}, and Definition \ref{gammadist}  as
\begin{equation}
\begin{aligned}
    \mathcal{H}_2^{(\gamma;\psi)}\hspace{-1.4mm}\left\langle j{-}1 \hspace{-0.3mm};\hspace{-0.3mm} i{-}(j{-}1) \right\rangle = &\Big[(U_{x_1}-Z_{\ell'_1}+Z_{\ell_{1}})-U_{x'_1},\cdots,(U_{x_{j-1}}-Z_{\ell'_{j-1}}+Z_{\ell_{j-1}})-U_{x'_{j-1}}\\
    &,U_{x_j}-(U_{x'_j}-Z_{\ell_j}+Z_{\ell'_j}),\cdots,U_{x_i}-(U_{x'_{i}}-Z_{\ell_i}+Z_{\ell'_{i}})\Big],
\end{aligned}
\end{equation}
where $(U_{x_h}-Z_{\ell'_h}+Z_{\ell_h})-U_{x'_h} {\dot{\preceq}} (U_{x_{h{+}1}}-Z_{\ell'_{h{+}1}}+Z_{\ell_{h{+}1}})-U_{x'_{h{+}1}}$ for all $1 \le h \le j-1$, $(U_{x_{j-1}}-Z_{\ell'_{j-1}}+Z_{\ell_{j-1}})-U_{x'_{j-1}} {\dot{\preceq}} U_{x_j}-(U_{x'_j}-Z_{\ell_j}+Z_{\ell'_j})$, and $U_{x_g}-(U_{x'_g}-Z_{\ell_g}+Z_{\ell'_g}) {\dot{\preceq}} (U_{x_{g{+}1}}-Z_{\ell'_{g{+}1}}+Z_{\ell_{g{+}1}})-U_{x'_{g{+}1}}$ for all $j \le g \le i$. Likewise, for $\mathcal{L}_2^{(\ddot{\mathfrak{X}})}\hspace{-1.4mm}\left\langle 2n{-}i\right\rangle$, according to Definition \ref{def19} and Definition \ref{expdist}, we get
\begin{equation}
    \mathcal{L}_2^{(\ddot{\mathfrak{X}})}\hspace{-1.4mm}\left\langle 2n{-}i\right\rangle = \left[Z_{\ell_{i{+}1}}-U_{x'_{i{+}1}}-Z_{\ell'_{i{+}1}},Z_{\ell_{i{+}2}}-U_{x'_{i{+}2}}-Z_{\ell'_{i{+}2}},\cdots,Z_{\ell_{2n}}-U_{x'_{2n}}-Z_{\ell'_{2n}}\right],
\end{equation}
where $Z_{\ell_r}{-}U_{x'_r}{-}Z_{\ell'_r} {\dot{\preceq}} Z_{\ell_{r{+}1}}{-}U_{x'_{r{+}1}}{-}Z_{\ell'_{r{+}1}}$ for all $i+1 \le r \le 2n$. From Decomposition Theorem, we also have $\mathcal{L}_2^{(\ddot{\mathfrak{X}})}\hspace{-1.4mm}\left\langle 2n{-}i\right\rangle {\dot{\preceq}} \mathcal{H}_2^{(\gamma;\psi)}\hspace{-1.4mm}\left\langle j{-}1 \hspace{-0.3mm};\hspace{-0.3mm} i{-}(j{-}1) \right\rangle$. Further, from Definition \ref{reducer}, we have $\hat{\mathbf{\delta}}((U_{x_h}-Z_{\ell'_h}+Z_{\ell_h})-U_{x'_h})=U_{x'_h}+Z_{\ell'_h}$, $\hat{\mathbf{\delta}}(U_{x_g}-(U_{x'_g}-Z_{\ell_g}+Z_{\ell'_g}))=U_{x'_g}+Z_{\ell'_g}$, and  $\hat{\mathbf{\delta}}(Z_{\ell_r}{-}U_{x'_r}{-}Z_{\ell'_r})=U_{x'_r}+Z_{\ell'_r}$. Moreover, according to Definition \ref{CG}, if relation ${\dot{\preceq}}$ holds between two EDVs $V_1$ and $V_2$, we have $\hat{\mathbf{\delta}}(V_1){=}\hat{\mathbf{\delta}}(V_2)$. Therefore, from Definition \ref{reducer} and Definition \ref{CG}, we have
\begin{equation}
\hat{\mathbf{\delta}}((U_{x_h}-Z_{\ell'_h}+Z_{\ell_h})-U_{x'_h})=\hat{\mathbf{\delta}}((U_{x_1}-Z_{\ell'_1}+Z_{\ell_{1}})-U_{x'_1})=> U_{x'_h}+Z_{\ell'_h}=U_{x'_1}+Z_{\ell'_1}, 1 \le h \le j-1,
\end{equation}
and
\begin{equation}
\hat{\mathbf{\delta}}(U_{x_g}-(U_{x'_g}-Z_{\ell_g}+Z_{\ell'_g}))=\hat{\mathbf{\delta}}((U_{x_1}-Z_{\ell'_1}+Z_{\ell_{1}})-U_{x'_1}) => U_{x'_g}+Z_{\ell'_g}=U_{x'_1}+Z_{\ell'_1}, j \le g \le i.
\end{equation}
Likewise, from Definition \ref{reducer}, Definition \ref{CG}, and Decomposition Theorem, we have
\begin{equation}
\hat{\mathbf{\delta}}(Z_{\ell_r}{-}U_{x'_r}{-}Z_{\ell'_r})=\hat{\mathbf{\delta}}((U_{x_1}-Z_{\ell'_1}+Z_{\ell_{1}})-U_{x'_1})=> U_{x'_r}+Z_{\ell'_r}=U_{x'_1}+Z_{\ell'_1}, i+1 \le r \le 2n.
\end{equation}
Consequently, we can finally expand   $\mathcal{H}_2^{(\gamma;\psi)}\hspace{-1.4mm}\left\langle j{-}1 \hspace{-0.3mm};\hspace{-0.3mm} i{-}(j{-}1) \right\rangle$ as
\begin{equation}\label{Us2}
\begin{aligned}
    \mathcal{H}_2^{(\gamma;\psi)}\hspace{-1.4mm}\left\langle j{-}1 \hspace{-0.3mm};\hspace{-0.3mm} i{-}(j{-}1) \right\rangle = &\Big[(U_{x_1}-Z_{\ell'_1}+Z_{\ell_{1}})-U_{x'_1},\cdots,(U_{x_{j-1}}-Z_{\ell'_{1}}+Z_{\ell_{j-1}})-U_{x'_{1}}\\
    &,U_{x_j}-(U_{x'_1}-Z_{\ell_j}+Z_{\ell'_1}),\cdots,U_{x_i}-(U_{x'_{1}}-Z_{\ell_i}+Z_{\ell'_{1}})\Bigg].
\end{aligned}
\end{equation}
and $\mathcal{L}_2^{(\ddot{\mathfrak{X}})}\hspace{-1.4mm}\left\langle 2n{-}i\right\rangle $ as
\begin{equation}\label{Zs2}
    \mathcal{L}_2^{(\ddot{\mathfrak{X}})}\hspace{-1.4mm}\left\langle 2n{-}i\right\rangle = \left[Z_{\ell_{i{+}1}}-U_{x'_1}-Z_{\ell'_1},\cdots,Z_{\ell_{r}}-U_{x'_1}-Z_{\ell'_1},\cdots,Z_{\ell_{2n}}-U_{x'_1}-Z_{\ell'_1}\right].
\end{equation}
Moreover, for random variables in the lists, considering Definition \ref{psidist} and Definition \ref{gammadist}, the following conditions hold  (c1) $Z_{\ell'_1}-Z_{\ell_{h}},~U_{x_h}>Z_{\ell'_1}-Z_{\ell_{h}},~(U_{x_h}-Z_{\ell'_1}+Z_{\ell_{h}})>U_{x'_1},~1 \le h \le j-1$, (c2) $Z_{\ell_g}>Z_{\ell'_1},~ U_{x'_1}>Z_{\ell_g}-Z_{\ell'_1} ,U_{x_g}>(U_{x'_1}-Z_{\ell_g}+Z_{\ell'_1}),~ j \le g \le i$, and (c3) $Z_{\ell_{r}}>U_{x'_1}+Z_{\ell'_1}~ i+1 \le r \le 2n$.
Since we will use these conditions later, we compact them through three sets as follows:

\begin{equation}
   \gamma^{\mathsf{C}}_{h}=\{Z_{\ell'_1}-Z_{\ell_{h}},~U_{x_h}>Z_{\ell'_1}-Z_{\ell_{h}},~(U_{x_h}-Z_{\ell'_1}+Z_{\ell_{h}})>U_{x'_1}\},~ 1 \le h \le j-1,
\end{equation}
\begin{equation}
   \psi^{\mathsf{C}}_{g}=\{Z_{\ell_g}>Z_{\ell'_1},~ U_{x'_1}>Z_{\ell_g}-Z_{\ell'_1} ,U_{x_g}>(U_{x'_1}-Z_{\ell_g}+Z_{\ell'_1})\},~ j \le g \le i,
\end{equation}
and
\begin{equation}
   \ddot{\mathfrak{X}}^{\mathsf{C}}_{r}=\{Z_{\ell_{r}}>U_{x'_1}+Z_{\ell'_1}\}~ i+1 \le r \le 2n.
\end{equation}
Subsequently, we prove Theorem \ref{transition_prob_theorem} for the following three cases:

\textbullet\hspace{.5mm} \textbf{Case 1 (computation of $p_{i,j}^{k}$ when $j=0$, $i\in \{0,\cdots,n\}$, and $k=i$):}
According to \eqref{pijk}, the transition probability of D-SMP from state $S_{i,0}$ to $S_{i-1,k}$, where $1 \le i \le n$ and $k=i$, is as follows:
\begin{equation}\label{pijk_1}
\begin{aligned}
&p_{i,0}^{i}{=}\textrm{Pr}\Bigg[S_{i,0}^{\mathsf{Rec}}(i) {<} \min\Big\{\min_{i+1 \le r \le 2n}\left\{S_{i,0}^{\mathsf{Soj}}(r-i)\right\},\min_{\substack{1 \le h \le i-1}}\left\{S_{i,0}^{\mathsf{Rec}}(h)\right\}\Big\}\Bigg]\\
&=\textrm{Pr}\Bigg[\mathcal{H}_1^{(\varphi;\mathfrak{R})}\hspace{-1.4mm}\left\langle {i{-}1} \hspace{-0.3mm};\hspace{-0.3mm} 1 \right\rangle(i) {<} \min\Big\{\min_{i+1 \le r \le 2n}\left\{\mathcal{L}_1^{(\dot{\mathfrak{X}})}\hspace{-1.4mm}\left\langle {2n{-}i}\right\rangle(r-i)\right\},\min_{\substack{1 \le h \le i-1}}\left\{\mathcal{H}_1^{(\varphi;\mathfrak{R})}\hspace{-1.4mm}\left\langle {i{-}1} \hspace{-0.3mm};\hspace{-0.3mm} 1 \right\rangle(h)\right\}\Big\}\Bigg].
\end{aligned}
\end{equation}
For all $1\le h \le i-1$ and $i+1 \le r \le 2n$, substituting \eqref{Us} and \eqref{Zs} in \eqref{pijk_1}, and conditioning on $\varphi^{\mathsf{C}}_{h}$ and $\dot{\mathfrak{X}}^{\mathsf{C}}_{r}$, yield
\begin{align}\label{pijk_2}
p_{i,0}^{i} &{=}\textrm{Pr}\Bigg[U_{x_i} {<} \min\Big\{\min\left[Z_{\ell_{i{+}1}}{-}Z_{\ell'_1},\cdots,Z_{\ell_r}{-}Z_{\ell'_1},\cdots,Z_{\ell_{2n}}{-}Z_{\ell'_1}\right]\nonumber\\
&~~~~~~~~~~~~~~~~~~~~~,\min\left[U_{x_1}{-}Z_{\ell'_1}{+}Z_{\ell_{1}},\cdots,U_{x_h}{-}Z_{\ell'_1}{+}Z_{\ell_h},\cdots,U_{i{{-}}1}{-}Z_{\ell'_1}{+}Z_{\ell_{i{-}1}}\right]\Big\}\nonumber\\
&~~~~~~~~~~~~~~~~~~~~\Big|\left\{\dot{\mathfrak{X}}^{\mathsf{C}}_{r}|i+1 \le r \le 2n\right\},\left\{\varphi^{\mathsf{C}}_{h}|1 \le h \le i-1\right\}\Bigg]\nonumber\\
&{=}\textrm{Pr}\Bigg[U_{x_i} {<} \min\Big[Z_{\ell_{i{+}1}}{-}Z_{\ell'_1},\cdots,Z_{\ell_r}{-}Z_{\ell'_1},\cdots,Z_{\ell_{2n}}{-}Z_{\ell'_1}\nonumber\\
&~~~~~~~~~~~~~~~~~~~~~~~~~~~~~,U_{x_1}{-}Z_{\ell'_1}{+}Z_{\ell_{1}},\cdots,U_{x_h}{-}Z_{\ell'_1}{+}Z_{\ell_h},\cdots,U_{x_{i{-}1}}{-}Z_{\ell'_1}{+}Z_{\ell_{i{-}1}}\Bigg]\nonumber\\
&~~~~~~~~~~~~~~~~\Big|\left\{\dot{\mathfrak{X}}^{\mathsf{C}}_{r}|i+1 \le r \le 2n\right\},\left\{\varphi^{\mathsf{C}}_{h}|1 \le h \le i-1\right\}\Bigg]\nonumber\\
&\overset{(e_1)}{=}\textrm{Pr}\Bigg[U_{x_i} {<} \min\Big[Z_{\ell_{i{+}1}}{-}Z_{\ell'_1},\cdots,Z_{\ell_r}{-}Z_{\ell'_1},\cdots,Z_{\ell_{2n}}{-}Z_{\ell'_1}\nonumber\\
&~~~~~~~~~~~~~~~~~~~~~~~~~~~~~,U_{x_1}{-}Z_{\ell'_1}{+}Z_{\ell_{1}},\cdots,U_{x_h}{-}Z_{\ell'_1}{+}Z_{\ell_h},\cdots,U_{x_{i{-}1}}{-}Z_{\ell'_1}{+}Z_{\ell_{i{-}1}}\Bigg]\nonumber\\
&~~~~~~~~~~~~~\Big|\left\{\dot{\mathfrak{X}}^{\mathsf{C}}_{r}|i+1 \le r \le 2n\right\},\left\{\varphi^{\mathsf{C}}_{h}|1 \le h \le i-1\right\}\Bigg] \nonumber\\
&=\textrm{Pr}\Bigg[U_{x_i} {<} \min\Big[Z_{\ell_{i{+}1}},\cdots,Z_{\ell_r},\cdots,Z_{\ell_{2n}}\nonumber\\
&~~~~~~~~~~~~~~~~~~~~~~~~~~~~~,U_{x_1}{+}Z_{\ell_1},\cdots,U_{x_h}{+}Z_h,\cdots,U_{x_{i{-}1}}{+}Z_{\ell_{i{-}1}}\Big]{-}Z_{\ell'_1}\nonumber\\
&~~~~~~~~~~~~~~~\Big|\left\{\dot{\mathfrak{X}}^{\mathsf{C}}_{r}|i+1 \le r \le 2n\right\},\left\{\varphi^{\mathsf{C}}_{h}|1 \le h \le i-1\right\}\Bigg]\nonumber\\
&=\textrm{Pr}\Bigg[U_{x_{i}}{+}Z_{\ell'_{1}} {<} \min\Big[Z_{\ell_{i{+}1}},\cdots,Z_{\ell_r},\cdots,Z_{\ell_{2n}},U_{x_1}{+}Z_{\ell_1},\cdots,U_{x_h}{+}Z_{\ell_h},\cdots,U_{x_{i{-}1}}{+}Z_{\ell_{i{-}1}}\Big]\Bigg] \nonumber\\
&\overset{(e_2)}{=}\prod_{r=i{+}1}^{2n}\textrm{Pr}\Bigg[U_{x_{i}}{+}Z_{\ell'_{1}} {<}Z_{\ell_r}\big|Z_{\ell_r}>Z_{\ell'_{1}}\Bigg] \prod_{\substack{h=1}}^{i{-}1}\textrm{Pr}\Bigg[U_{x_{i}}{+}Z_{\ell'_{1}} {<}U_{x_h}{+}Z_{\ell_h}\Bigg]\nonumber\\
&\overset{(e_3)}{=}\left(\textrm{Pr}\Bigg[U_{x_{i}}{+}Z_{\ell'_{1}} {<}Z_a\big|Z_{a}>Z_{\ell'_{1}}\Bigg]\right)^{2n{-}i} \left(\textrm{Pr}\Bigg[U_{x_{i}}{+}Z_{\ell'_{1}} {<}U_b{+}Z_b\Bigg]\right)^{i{-}1},
\end{align}
where $Z_{\ell_i}$, $Z_{\ell'_1}$, $Z_a$, $Z_b$, $U_{x_{i}}$, and $U_b$ are independent random variables, from which $Z_{\ell_i}$, $Z_{\ell'_1}$, $Z_a$, and $Z_b$ follow general distribution $\mathfrak{D}_Z(z)$, and $U_{x_{i}}$ and $U_b$ follow general distribution $\mathfrak{R}_U(u)$. Further $(e_1)$ and $(e_2)$ are the results of the independency of the random variables assumed in Sec. IV-B. Moreover, $(e_3)$ is due the fact that $Z_{\ell_r}$ are i.i.d random variables for all $i+1 \le r \le 2n $ and $U_{x_h}$ are i.i.d random variables for all $1 \le h \le i$.

\textbullet\hspace{.5mm} \textbf{Case 2 (computation of $p_{i,j}^{k}$ when $j=0$, $i\in \{0,\cdots,n\}$, and $1 \le k \le i-1$):} According to \eqref{pijk}, the transition probability of D-SMP from state $S_{i,0}$ to $S_{i-1,k}$, where $1 \le i \le n$ and $1 \le k \le i-1$, is as follows:
\begin{equation}\label{pijk_145}
\begin{aligned}
&p_{i,0}^{k}{=}\textrm{Pr}\Bigg[S_{i,0}^{\mathsf{Rec}}(k) {<} \min\Big\{\min_{i+1 \le r \le 2n}\left\{S_{i,0}^{\mathsf{Soj}}(r-i)\right\},\min_{\substack{1 \le h \le i\\h\neq k}}\left\{S_{i,0}^{\mathsf{Rec}}(h)\right\}\Big\}\Bigg]\\
&=\textrm{Pr}\Bigg[\mathcal{H}_1^{(\varphi;\mathfrak{R})}\hspace{-1.4mm}\left\langle {i{-}1} \hspace{-0.3mm};\hspace{-0.3mm} 1 \right\rangle(k) {<} \min\Big\{\min_{i+1 \le r \le 2n}\left\{\mathcal{L}_1^{(\dot{\mathfrak{X}})}\hspace{-1.4mm}\left\langle {2n{-}i}\right\rangle(r-i)\right\},\min_{\substack{1 \le h \le i\\h\neq k}}\left\{\mathcal{H}_1^{(\varphi;\mathfrak{R})}\hspace{-1.4mm}\left\langle {i{-}1} \hspace{-0.3mm};\hspace{-0.3mm} 1 \right\rangle(h)\right\}\Big\}\Bigg].
\end{aligned}
\end{equation}
For all $1 \le h \le i-1,1 \le k \le i-1$, where $h\neq k$, and for all $i+1 \le r \le 2n$, replacing \eqref{Us} and \eqref{Zs} in \eqref{pijk_145}, and conditioning on $\varphi^{\mathsf{C}}_{h}$, $\varphi^{\mathsf{C}}_{k}$, and $\dot{\mathfrak{X}}^{\mathsf{C}}_{r}$, results in

\begin{align}\label{pijk_3}
p_{i,0}^{k} &{=}\textrm{Pr}\Bigg[U_{x_{k}}{-}Z_{\ell'_1}{+}Z_{\ell_{k}} {<} \min\Big\{\min\left[Z_{\ell_{i{+}1}}{-}Z_{\ell'_1},\cdots,Z_{\ell_r}{-}Z_{\ell'_1},\cdots,Z_{\ell_{2n}}{-}Z_{\ell'_1}\right]\nonumber\\
&~~~~~~~~~~~~~~~~~~~~~,\min\left[U_{x_1}{-}Z_{\ell'_1}{+}Z_{\ell_{1}},\cdots,U_{x_h}{-}Z_{\ell'_1}{+}Z_{\ell_h},\cdots,U_{i{{-}}1}{-}Z_{\ell'_1}{+}Z_{\ell_{i{-}1}}, U_{x_i}\right]\Big\}\nonumber\\
&~~~~~~~~~~~~~~~~~~~~\Big|\left\{\varphi^{\mathsf{C}}_{h}|1 \le h \le i-1,  h\neq k\right\}, \left\{\varphi^{\mathsf{C}}_{k}\right\},\left\{\dot{\mathfrak{X}}^{\mathsf{C}}_{r}|i+1 \le r \le 2n\right\}\Bigg]\nonumber\\
&{=}\textrm{Pr}\Bigg[U_{x_{k}}{-}Z_{\ell'_1}{+}Z_{\ell_{k}} {<} \min\Big[Z_{\ell_{i{+}1}}{-}Z_{\ell'_1},\cdots,Z_{\ell_r}{-}Z_{\ell'_1},\cdots,Z_{\ell_{2n}}{-}Z_{\ell'_1}\nonumber\\
&~~~~~~~~~~~~~~~~~~~~~~~~~~~~~,U_{x_1}{-}Z_{\ell'_1}{+}Z_{\ell_{1}},\cdots,U_{x_h}{-}Z_{\ell'_1}{+}Z_{\ell_h},\cdots,U_{x_{i{-}1}}{-}Z_{\ell'_1}{+}Z_{\ell_{i{-}1}}, U_{x_i}\Bigg]\nonumber\\
&~~~~~~~~~~~~~~~~\Big|\left\{\varphi^{\mathsf{C}}_{h}|1 \le h \le i-1,  h\neq k\right\}, \left\{\varphi^{\mathsf{C}}_{k}\right\},\left\{\dot{\mathfrak{X}}^{\mathsf{C}}_{r}|i+1 \le r \le 2n\right\}\Bigg]\nonumber\\
&\overset{(e_1)}{=}\textrm{Pr}\Bigg[U_{x_{k}}{-}Z_{\ell'_1}{+}Z_{\ell_{k}} {<} \min\Big[Z_{\ell_{i{+}1}}{-}Z_{\ell'_1},\cdots,Z_{\ell_r}{-}Z_{\ell'_1},\cdots,Z_{\ell_{2n}}{-}Z_{\ell'_1}\nonumber\\
&~~~~~~~~~~~~~~~~~~~~~~~~~~~~~,U_{x_1}{-}Z_{\ell'_1}{+}Z_{\ell_{1}},\cdots,U_{x_h}{-}Z_{\ell'_1}{+}Z_{\ell_h},\cdots,U_{x_{i{-}1}}{-}Z_{\ell'_1}{+}Z_{\ell_{i{-}1}}\Bigg]\nonumber\\
&~~~~~~~~~~~~~\Big|\left\{\varphi^{\mathsf{C}}_{h}|1 \le h \le i-1,  h\neq k\right\}, \left\{\varphi^{\mathsf{C}}_{k}\right\},\left\{\dot{\mathfrak{X}}^{\mathsf{C}}_{r}|i+1 \le r \le 2n\right\}\Bigg]\nonumber\\
&~~~~~~~~~~~~~~~~~\times\textrm{Pr}\Bigg[U_{x_{k}}{-}Z_{\ell'_1}{+}Z_{\ell_{k}} {<} U_{x_i}\big|Z_{\ell'_1}>Z_{\ell_k},~U_{x_k}>Z_{\ell'_1}-Z_{\ell_k}\Bigg] \nonumber\\
&=\textrm{Pr}\Bigg[U_{x_{k}}{-}Z_{\ell'_1}{+}Z_{\ell_{k}} {<} \min\Big[Z_{\ell_{i{+}1}},\cdots,Z_{\ell_r},\cdots,Z_{\ell_{2n}}\nonumber\\
&~~~~~~~~~~~~~~~~~~~~~~~~~~~~~,U_{x_1}{+}Z_{\ell_1},\cdots,U_{x_h}{+}Z_h,\cdots,U_{x_{i{-}1}}{+}Z_{\ell_{i{-}1}}\Big]{-}Z_{\ell'_1}\nonumber\\
&~~~~~~~~~~~~~~~\Big|\left\{\varphi^{\mathsf{C}}_{h}|1 \le h \le i-1,  h\neq k\right\}, \left\{\varphi^{\mathsf{C}}_{k}\right\},\left\{\dot{\mathfrak{X}}^{\mathsf{C}}_{r}|i+1 \le r \le 2n\right\}\Bigg]\nonumber\\
&~~~~~~~~~~~~~~~~~~~~~~~~~~~~~\times\textrm{Pr}\Bigg[U_{x_{k}}{-}Z_{\ell'_1}{+}Z_{\ell_{k}} {<} U_{x_i}\big|Z_{\ell'_1}>Z_{\ell_k},~U_{x_k}>Z_{\ell'_1}-Z_{\ell_k}\Bigg] \nonumber\\
&=\textrm{Pr}\Bigg[U_{x_{k}}{+}Z_{\ell_{k}} {<} \min\Big[Z_{\ell_{i{+}1}},\cdots,Z_{\ell_r},\cdots,Z_{\ell_{2n}},U_{x_1}{+}Z_{\ell_1},\cdots,U_{x_h}{+}Z_{\ell_h},\cdots,U_{x_{i{-}1}}{+}Z_{\ell_{i{-}1}}\Big]\Bigg] \nonumber\\
&~~~~~~~~~~~~~~~~~~~~~~~~~~~~~\times\textrm{Pr}\Bigg[U_{x_{k}}{-}Z_{\ell'_1}{+}Z_{\ell_{k}} {<} U_{x_i}\big|Z_{\ell'_1}>Z_{\ell_k},~U_{x_k}>Z_{\ell'_1}-Z_{\ell_k}\Bigg] \nonumber\\
&\overset{(e_2)}{=}\prod_{r=i{+}1}^{2n}\textrm{Pr}\Bigg[U_{x_{k}}{+}Z_{\ell_{k}} {<}Z_{\ell_r}\big|Z_{\ell_r}>Z_{\ell_{k}}\Bigg] \prod_{\substack{h=1,\\h\neq k}}^{i{-}1}\textrm{Pr}\Bigg[U_{x_{k}}{+}Z_{\ell_{k}} {<}U_{x_h}{+}Z_{\ell_h}\Bigg]\nonumber\\
&~~~~~~~~~~~~~~~~~~~~~~~~~\times\textrm{Pr}\Bigg[U_{x_{k}}{-}Z_{\ell'_1}{+}Z_{\ell_{k}} {<} U_{x_i}\big|Z_{\ell'_1}>Z_{\ell_k},~U_{x_k}>Z_{\ell'_1}-Z_{\ell_k}\Bigg] \nonumber\\
&\overset{(e_3)}{=}\left(\textrm{Pr}\Bigg[U_{x_{k}}{+}Z_{\ell_{k}} {<}Z_a\big|Z_{a}>Z_{\ell_{k}}\Bigg]\right)^{2n{-}i} \left(\textrm{Pr}\Bigg[U_{x_{k}}{+}Z_{\ell_{k}} {<}U_b{+}Z_b\Bigg]\right)^{i{-}2}\nonumber\\
&~~~~~~~~~~~~~~~~~~~~~~~~~\times\textrm{Pr}\Bigg[U_{x_{k}}{-}Z_{\ell'_1}{+}Z_{\ell_{k}} {<} U_{x_i}\big|Z_{\ell'_1}>Z_{\ell_k},~U_{x_k}>Z_{\ell'_1}-Z_{\ell_k}\Big],
\end{align}
where $Z_{\ell_k}$, $Z_{\ell'_k}$, $Z_a$, $Z_b$, $U_{x_{k}}$, and $U_b$ are independent random variables, from which $Z_{\ell_k}$, $Z_{\ell'_k}$, $Z_a$, and $Z_b$ follow general distribution $\mathfrak{D}_Z(z)$, and $U_{x_{k}}$ and $U_b$ follow general distribution $\mathfrak{R}_U(u)$. Further $(e_1)$ and $(e_2)$ are the results of the independence assumption of the random variables. Moreover, $(e_3)$ is due the fact that $Z_{\ell_r}$ are i.i.d random variables  for all $i+1 \le r \le 2n $, and $U_{x_h}$ are i.i.d random variables for all $1 \le h \le i$.

\textbullet\hspace{.5mm} \textbf{Case 3 (computation of $p_{i,j}^{k}$ when $1 \le i \le n-1$, $1 \le j \le i+1$, and $1 \le k \le i$):} According to \eqref{pijk}, the transition probability of D-SMP from state $S_{i,j}$ to $S_{i-1,k}$, where $1 \le i \le n-1$, $1 \le j \le i+1$, and $1 \le k \le i$, is as follows:
\begin{align}\label{pijk_21}
&p_{i,j}^{k} {=}\textrm{Pr}\Bigg[S_{i,j}^{\mathsf{Rec}}(k) {<} \min\Big\{\min_{i+1 \le r \le 2n}\left\{S_{i,j}^{\mathsf{Soj}}(r-i)\right\},\min_{\substack{1 \le h \le j-1\\ h\neq k}}\left\{S_{i,j}^{\mathsf{Rec}}(h)\right\}, \min_{\substack{j \le g \le i \\ g\neq k}}\left\{S_{i,j}^{\mathsf{Rec}}(g)\right\}\Big\}\Bigg]\nonumber\\
&=\textrm{Pr}\Bigg[\mathcal{H}_2^{(\gamma;\psi)}\hspace{-1.4mm}\left\langle j{-}1 \hspace{-0.3mm};\hspace{-0.3mm} i{-}(j{-}1) \right\rangle(k) {<} \min\Big\{\min_{i+1 \le r \le 2n}\left\{\mathcal{L}_2^{(\ddot{\mathfrak{X}})}\hspace{-1.4mm}\left\langle 2n{-}i\right\rangle(r)\right\},\min_{\substack{1\le h\le j-1\\h\neq k}}\left\{\mathcal{H}_2^{(\gamma;\psi)}\hspace{-1.4mm}\left\langle j{-}1 \hspace{-0.3mm};\hspace{-0.3mm} i{-}(j{-}1) \right\rangle(h)\right\}\nonumber\\
&~~~~~~~~~~~~~~~~~~~~~~~~~~~~~~~~~~~~~~~~~~~~~~~~~~~~,\min_{\substack{j\le g \le i\\g\neq k}}\left\{\mathcal{H}_2^{(\gamma;\psi)}\hspace{-1.4mm}\left\langle j{-}1 \hspace{-0.3mm};\hspace{-0.3mm} i{-}(j{-}1) \right\rangle(g)\right\}\Big\}\Bigg].
\end{align}

We break down the proof of case 3 into the following two sub-cases:\par
\textbullet\hspace{.5mm} \textbf{Sub-case [3-a] (when $1 \le k \le j-1$):} Substituting \eqref{Us2} and \eqref{Zs2} in \eqref{pijk_21}, and conditioning on $\ddot{\mathfrak{X}}^{\mathsf{C}}_{r}$, $\gamma^{\mathsf{C}}_{h}$, and $\psi^{\mathsf{C}}_{g}$, for all $i+1 \le r \le 2n$, $1 \le h \le j-1$, and $j \le g \le i$, results in
\begin{align}\label{pijk_22}
p_{i,j}^{k} &{=}\textrm{Pr}\Bigg[(U_{x_{k}}{-}Z_{\ell'_1}{+}Z_{\ell_{k}}){-}U_{x'_1} {<} \min\Big\{\nonumber\\
&~~~~~~~\min\left[Z_{\ell_{i{+}1}}{-}U_{x'_1}{-}Z_{\ell'_1},\cdots,Z_{\ell_{r}}{-}U_{x'_1}{-}Z_{\ell'_1},\cdots,Z_{\ell_{2n}}{-}U_{x'_1}{-}Z_{\ell'_1}\right]\nonumber\\
&~~~~~~~,\min\Big[(U_{x_1}{-}Z_{\ell'_{1}}{+}Z_{\ell_1}){-}U_{x'_{1}},\cdots,(U_{x_h}{-}Z_{\ell'_{1}}{+}Z_{\ell_h}){-}U_{x'_{1}},\cdots,(U_{x_{j-1}}-Z_{\ell'_{1}}+Z_{\ell_{j-1}})-U_{x'_{1}}\nonumber\\
&~~~~~~~,U_{x_{j}}{-}(U_{x'_1}{-}Z_{\ell_{j}}+Z_{\ell'_1}),\cdots,U_{x_{g}}{-}(U_{x'_1}{-}Z_{\ell_{g}}+Z_{\ell'_1}),\cdots,U_{x_{i}}{-}(U_{x'_{1}}{-}Z_{\ell_{i}}+Z_{\ell'_{1}})\Big]\Big\}\nonumber\\
&~~~~~~~~~~~~~~~~~~~~~~~~\Big|\left\{\ddot{\mathfrak{X}}^{\mathsf{C}}_{r}|i+1 \le r \le 2n\right\},\left\{\gamma^{\mathsf{C}}_{h}|1 \le h \le j-1, k\neq h\right\}, \left\{\gamma^{\mathsf{C}}_{k}\right\} ,\left\{\psi^{\mathsf{C}}_{g}|j \le g \le i\right\}\Bigg]\nonumber\\
&{=}\textrm{Pr}\Bigg[U_{x_{k}}{-}Z_{\ell'_1}+Z_{\ell_{k}}{-}U_{x'_1} {<} \min\Big[Z_{\ell_{i{+}1}}{-}U_{x'_1}{-}Z_{\ell'_1},\cdots,Z_{\ell_{r}}{-}U_{x'_1}{-}Z_{\ell'_1},\cdots,Z_{\ell_{2n}}{-}U_{x'_1}{-}Z_{\ell'_1}\nonumber\\
&~~~~~~~~,(U_{x_1}{-}Z_{\ell'_1}+Z_{\ell_1}){-}U_{x'_1},\cdots,(U_{x_h}{-}Z_{\ell'_{1}}{+}Z_{\ell_h}){-}U_{x'_{1}},\cdots,(U_{x_{j-1}}{-}Z_{\ell'_{1}}+Z_{\ell_{j-1}}){-}U_{x'_{1}}\nonumber\\
&~~~~~~~~,U_{x_{j}}{-}(U_{x'_1}{-}Z_{\ell_{j}}+Z_{\ell'_1}),\cdots,U_{x_{j}}{-}(U_{x'_1}{-}Z_{\ell_{j}}+Z_{\ell'_1}),\cdots,U_{x_{i}}{-}(U_{x'_{1}}{-}Z_{\ell_{i}}+Z_{\ell'_{1}})\Bigg]\nonumber\\
&~~~~~~~~~~~~~~~~~~~~~~~~\Big|\left\{\ddot{\mathfrak{X}}^{\mathsf{C}}_{r}|i+1 \le r \le 2n\right\},\left\{\gamma^{\mathsf{C}}_{h}|1 \le h \le j-1, k\neq h\right\}, \left\{\gamma^{\mathsf{C}}_{k}\right\} ,\left\{\psi^{\mathsf{C}}_{g}|j \le g \le i\right\}\Bigg]\nonumber\\
&{=}\textrm{Pr}\Bigg[U_{x_{k}}{-}Z_{\ell'_1}{+}Z_{\ell_{k}}{-}U_{x'_1} {<} \min\Big[Z_{\ell_{i{+}1}},\cdots,Z_{\ell_{r}},\cdots,Z_{\ell_{2n}}\nonumber\\
&~~~~~~~~~~~~~~~~~~~~~~~~,U_{x_1}{+}Z_{\ell_1},\cdots,U_{x_h}{+}Z_{\ell_h},\cdots,U_{j{-}1}{+}Z_{j{-}1}\nonumber\\
&~~~~~~~~~~~~~~~~~~~~~~~~,U_{x_{j}}{+}Z_{\ell_{j}},\cdots,U_{x_{g}}{+}Z_{\ell_{g}},\cdots,U_{x_{i}}{+}Z_{\ell_{i}}\Big]{-}U_{x'_1}{-}Z_{\ell'_1}\nonumber\\
&~~~~~~~~~~~~~~~~~~~~~~~~\Big|\left\{\ddot{\mathfrak{X}}^{\mathsf{C}}_{r}|i+1 \le r \le 2n\right\},\left\{\gamma^{\mathsf{C}}_{h}|1 \le h \le j-1, k\neq h\right\}, \left\{\gamma^{\mathsf{C}}_{k}\right\} ,\left\{\psi^{\mathsf{C}}_{g}|j \le g \le i\right\}\Bigg]\nonumber\\
&{=}\textrm{Pr}\Bigg[U_{x_{k}}{+}Z_{\ell_{k}} {<} \min\Big[Z_{\ell_{i{+}1}},\cdots,Z_{\ell_{r}},\cdots,Z_{\ell_{2n}}\nonumber\\
&~~~~~~~~~~~~~~~~~~~~~~~~,U_{x_1}{+}Z_{\ell_1},\cdots,U_{x_h}{+}Z_{\ell_h},\cdots,U_{j{-}1}{+}Z_{j{-}1}\nonumber\\
&~~~~~~~~~~~~~~~~~~~~~~~~,U_{x_{j}}{+}Z_{\ell_{j}},\cdots,U_{x_{g}}{+}Z_{\ell_{g}},\cdots,U_{x_{i}}{+}Z_{\ell_{i}}\Big]\Bigg]\nonumber\\
&\overset{(e_1)}{=}\prod_{r=i{+}1}^{2n}\textrm{Pr}\Bigg[U_{x_{k}}{+}Z_{\ell_{k}} {<}Z_{\ell_r}\Bigg] \prod_{\substack{h=1,\\h\neq k}}^{i}\textrm{Pr}\Bigg[U_{x_{k}}{+}Z_{\ell_{k}} {<}U_{x_h}+Z_h\Bigg]\nonumber\\
&=\left(\textrm{Pr}\Bigg[U_{x_{k}}{+}Z_{\ell_{k}} {<}Z_{\ell_a}\Bigg]\right)^{2n-i} \left(\textrm{Pr}\Bigg[U_{x_{k}}{+}Z_{\ell_{k}} {<}U_{x_b}+Z_b\Bigg]\right)^{i-1},
\end{align}
where $Z_a$, $Z_b$, and $U_b$ are three independent random variables, from which $Z_a$ and $Z_b$ follow general distributions $\mathfrak{D}_Z(z)$ and $U_b$ follow general distribution $\mathfrak{R}_U(u)$. Further $(e_1)$ is the result of the independence assumption of the random variables.

\textbullet\hspace{.5mm} \textbf{Sub-case [3-b] (when $j \le k \le i$):} Substituting \eqref{Us2} and \eqref{Zs2} in \eqref{pijk_21}, and conditioning on $\ddot{\mathfrak{X}}^{\mathsf{C}}_{r}$, $\gamma^{\mathsf{C}}_{h}$, and $\psi^{\mathsf{C}}_{g}$, for all $i+1 \le r \le 2n$, $1 \le h \le j-1$, and $j \le g \le i$, results in
\begin{align}\label{pijk_23}
p_{i,j}^{k} &{=}\textrm{Pr}\Bigg[U_{x_{k}}{-}(U_{x'_1}{-}Z_{\ell_{k}}{+}Z_{\ell'_1}) {<} \min\Big\{\nonumber\\
&~~~~~~~\min\left[Z_{\ell_{i{+}1}}{-}U_{x'_1}{-}Z_{\ell'_1},\cdots,Z_{\ell_{r}}{-}U_{x'_1}{-}Z_{\ell'_1},\cdots,Z_{\ell_{2n}}{-}U_{x'_1}{-}Z_{\ell'_1}\right]\nonumber\\
&~~~~~~~,\min\Big[(U_{x_1}{-}Z_{\ell'_{1}}{+}Z_{\ell_1}){-}U_{x'_{1}},\cdots,(U_{x_h}{-}Z_{\ell'_{1}}{+}Z_{\ell_h}){-}U_{x'_{1}},\cdots,(U_{x_{j-1}}-Z_{\ell'_{1}}+Z_{\ell_{j-1}})-U_{x'_{1}}\nonumber\\
&~~~~~~~,U_{x_{j}}{-}(U_{x'_1}{-}Z_{\ell_{j}}+Z_{\ell'_1}),\cdots,U_{x_{g}}{-}(U_{x'_1}{-}Z_{\ell_{g}}+Z_{\ell'_1}),\cdots,U_{x_{i}}{-}(U_{x'_{1}}{-}Z_{\ell_{i}}+Z_{\ell'_{1}})\Big]\Big\}\nonumber\\
&~~~~~~~~~~~~~~~~~~~~~~~~\Big|\left\{\ddot{\mathfrak{X}}^{\mathsf{C}}_{r}|i+1 \le r \le 2n\right\},\left\{\gamma^{\mathsf{C}}_{h}|1 \le h \le j-1\right\}, \left\{\psi^{\mathsf{C}}_{k}\right\} ,\left\{\psi^{\mathsf{C}}_{g}|j \le g \le i, k\neq g\right\}\Bigg]\nonumber\\
&{=}\textrm{Pr}\Bigg[U_{x_{k}}{-}Z_{\ell'_1}+Z_{\ell_{k}}{-}U_{x'_1} {<} \min\Big[Z_{\ell_{i{+}1}}{-}U_{x'_1}{-}Z_{\ell'_1},\cdots,Z_{\ell_{r}}{-}U_{x'_1}{-}Z_{\ell'_1},\cdots,Z_{\ell_{2n}}{-}U_{x'_1}{-}Z_{\ell'_1}\nonumber\\
&~~~~~~~~,(U_{x_1}{-}Z_{\ell'_1}+Z_{\ell_1}){-}U_{x'_1},\cdots,(U_{x_h}{-}Z_{\ell'_{1}}{+}Z_{\ell_h}){-}U_{x'_{1}},\cdots,(U_{x_{j-1}}{-}Z_{\ell'_{1}}+Z_{\ell_{j-1}}){-}U_{x'_{1}}\nonumber\\
&~~~~~~~~,U_{x_{j}}{-}(U_{x'_1}{-}Z_{\ell_{j}}+Z_{\ell'_1}),\cdots,U_{x_{j}}{-}(U_{x'_1}{-}Z_{\ell_{j}}+Z_{\ell'_1}),\cdots,U_{x_{i}}{-}(U_{x'_{1}}{-}Z_{\ell_{i}}+Z_{\ell'_{1}})\Bigg]\nonumber\\
&~~~~~~~~~~~~~~~~~~~~~~~~\Big|\left\{\ddot{\mathfrak{X}}^{\mathsf{C}}_{r}|i+1 \le r \le 2n\right\},\left\{\gamma^{\mathsf{C}}_{h}|1 \le h \le j-1\right\}, \left\{\psi^{\mathsf{C}}_{k}\right\} ,\left\{\psi^{\mathsf{C}}_{g}|j \le g \le i, k\neq g\right\}\Bigg]\nonumber\\
&{=}\textrm{Pr}\Bigg[U_{x_{k}}{-}Z_{\ell'_1}{+}Z_{\ell_{k}}{-}U_{x'_1} {<} \min\Big[Z_{\ell_{i{+}1}},\cdots,Z_{\ell_{r}},\cdots,Z_{\ell_{2n}}\nonumber\\
&~~~~~~~~~~~~~~~~~~~~~~~~,U_{x_1}{+}Z_{\ell_1},\cdots,U_{x_h}{+}Z_{\ell_h},\cdots,U_{j{-}1}{+}Z_{j{-}1}\nonumber\\
&~~~~~~~~~~~~~~~~~~~~~~~~,U_{x_{j}}{+}Z_{\ell_{j}},\cdots,U_{x_{g}}{+}Z_{\ell_{g}},\cdots,U_{x_{i}}{+}Z_{\ell_{i}}\Big]{-}U_{x'_1}{-}Z_{\ell'_1}\nonumber\\
&~~~~~~~~~~~~~~~~~~~~~~~~\Big|\left\{\ddot{\mathfrak{X}}^{\mathsf{C}}_{r}|i+1 \le r \le 2n\right\},\left\{\gamma^{\mathsf{C}}_{h}|1 \le h \le j-1\right\}, \left\{\psi^{\mathsf{C}}_{k}\right\} ,\left\{\psi^{\mathsf{C}}_{g}|j \le g \le i, k\neq g\right\}\Bigg]\nonumber\\
&{=}\textrm{Pr}\Bigg[U_{x_{k}}{+}Z_{\ell_{k}} {<} \min\Big[Z_{\ell_{i{+}1}},\cdots,Z_{\ell_{r}},\cdots,Z_{\ell_{2n}}\nonumber\\
&~~~~~~~~~~~~~~~~~~~~~~~~,U_{x_1}{+}Z_{\ell_1},\cdots,U_{x_h}{+}Z_{\ell_h},\cdots,U_{j{-}1}{+}Z_{j{-}1}\nonumber\\
&~~~~~~~~~~~~~~~~~~~~~~~~,U_{x_{j}}{+}Z_{\ell_{j}},\cdots,U_{x_{g}}{+}Z_{\ell_{g}},\cdots,U_{x_{i}}{+}Z_{\ell_{i}}\Big]\Bigg]\nonumber\\
&\overset{(e_1)}{=}\prod_{r=i{+}1}^{2n}\textrm{Pr}\Bigg[U_{x_{k}}{+}Z_{\ell_{k}} {<}Z_{\ell_r}\Bigg] \prod_{\substack{h=1,\\h\neq k}}^{i}\textrm{Pr}\Bigg[U_{x_{k}}{+}Z_{\ell_{k}} {<}U_{x_h}+Z_h\Bigg]\nonumber\\
&=\left(\textrm{Pr}\Bigg[U_{x_{k}}{+}Z_{\ell_{k}} {<}Z_{\ell_a}\Bigg]\right)^{2n-i} \left(\textrm{Pr}\Bigg[U_{x_{k}}{+}Z_{\ell_{k}} {<}U_{x_b}+Z_b\Bigg]\right)^{i-1},
\end{align}
where $Z_a$, $Z_b$, and $U_b$ are three independent random variables, from which $Z_a$ and $Z_b$ follow general distributions $\mathfrak{D}_Z(z)$ and $U_b$ follow general distribution $\mathfrak{R}_U(u)$. Further $(e_1)$ is the result of the independence assumption of the random variables.

Inspecting the results in \eqref{pijk_2} and \eqref{pijk_3} verifies the results in \eqref{transition_prob1} and \eqref{transition_prob2}. Similarly, the results of \eqref{pijk_22} and \eqref{pijk_23} verify the result in \eqref{transition_prob3}. This concludes the proof of the theorem.
\end{proof}

\begin{theorem}\label{expected_sojourn_time_theorem}
Let $S_{i,j}^{\mathsf{Rec}}(h)$ and $S_{i,j}^{\mathsf{Soj}}(r-i)$ refer to the residual \underline{rec}ruitment time of vehicle $h$ and the residual sojourn/\underline{soj}ourn time of vehicle $r$, respectively. When vehicles' sojourn time and recruitment duration follow general distributions $\mathfrak{D}_Z(z)$ and $\mathfrak{R}_U(u)$, respectively, the expected sojourn time of D-SMP in state $S_{i,j}$ can be represented through the following two cases:

\textbullet\hspace{.5mm} \textbf{Case 1 (computation of $p_{i,j}^{k}$ when $j=0$ and $i\in \{0,\cdots,n\}$):}
\begin{align}\label{expected_value1}
\mathbb{E}\left[\widehat{\mathrm{W}}_{i,0}\right] &=\boldsymbol{\int}_{0}^{\infty } \textrm{Pr}\Bigg[\min\Big\{\min_{i+1 \le r \le 2n}\left\{S_{i,j}^{\mathsf{Soj}}(r-i)\right\},\min_{\substack{1 \le h \le i}}\left\{S_{i,j}^{\mathsf{Rec}}(h)\right\}\Big\}>t\Big] dt\\
&=\boldsymbol{\int}_{0}^{\infty } \Bigg(\Big(\textrm{Pr} \left[Z_{a}-Z_{b}>t\big|Z_{a}>Z_{b}\right]\Big)^{2n-i} \nonumber\\
&~~~~~~~~~~~~\times\Big(\textrm{Pr}\left[U_{c}+Z_{c}-Z_{b}>t\big| Z_{b}>Z_{c}, U_{c}>Z_{b}-Z_{c}\right]\Big)^{i-1} \textrm{Pr}\left[ U_{d}>t\right]\Bigg) dt,
\end{align}
where $Z_{a}$, $Z_{b}$, $Z_{c}$, $U_{c}$, and $U_{d}$ are independent random variables, from which $Z_{a}$, $Z_{b}$, and $Z_{c}$ follow general distributions $\mathfrak{D}_Z(z)$, and $U_{c}$ and $U_{d}$ follow general distribution $\mathfrak{R}_U(u)$.

\textbullet\hspace{.5mm} \textbf{Case 2 (computation of $p_{i,j}^{k}$ when $1 \le i \le n-1$ and $1 \le j \le i+1$):}
\begin{align}\label{expected_value2}
\mathbb{E}&\left[\widehat{\mathrm{W}}_{i,j}\right] =\boldsymbol{\int}_{0}^{\infty } \textrm{Pr}\Bigg[\min\Big\{\min_{i+1 \le r \le 2n}\left\{S_{i,j}^{\mathsf{Soj}}(r-i)\right\},\min_{\substack{1\le h\le j-1}}\left\{S_{i,j}^{\mathsf{Rec}}(h)\right\}, \min_{\substack{j\le g \le i}}\left\{S_{i,j}^{\mathsf{Rec}}(g)\right\}\Big\}>t\Big] dt\nonumber\\
&=\boldsymbol{\int}_{0}^{\infty } \Bigg(\Big(\textrm{Pr} \left[Z_{a}-U_{b}-Z_{b}>t\Big|Z_{a}>U_{b}+Z_{b}\right]\Big)^{2n-i}\nonumber\\
&~~~~~~~~~~~~~\times \Big(\textrm{Pr}\left[(U_{c}{-}Z_{b}{+}Z_{c}){-}U_{b}>t\Big|Z_{b}>Z_{c},U_{c}>Z_{b}{-}Z_{c},(U_{c}{-}Z_{b}{+}Z_{c})>U_{b}\right]\Big)^{j{-}1}\nonumber\\
&~~~~~~~~~~~~~\times \Big(\textrm{Pr}\left[U_{d}{-}(U_{b}{-}Z_{d}{+}Z_{b})>t\Big| Z_{d}>Z_{b}, U_{b}>Z_{d}{-}Z_{b}, U_{d}>(U_{b}{-}Z_{d}{+}Z_{b})\right]\Big)^{i{-}(j{-}1)}\Bigg) dt,
\end{align}
where $Z_{a}$, $Z_{b}$, $Z_{c}$, $Z_{d}$, $U_{b}$, $U_{c}$, and $U_{d}$ are independent random variables, from which $Z_{a}$, $Z_{b}$, $Z_{c}$, and $Z_{d}$ follow general distributions $\mathfrak{D}_Z(z)$, and $U_{b}$, $U_{c}$, and $U_{d}$ follow general distribution $\mathfrak{R}_U(u)$.
\end{theorem}

\begin{proof} The expected time that D-SMP resides in state $S_{i,j}\in\hat{\mathbf{S}}_i$ is equivalent to the expected value of {\small$\min\Big\{\min_{i+1 \le r \le 2n}\left\{S_{i,j}^{\mathsf{Soj}}(r-i)\right\},\min_{\substack{h \in \{1,\cdots,i\}}}\left\{S_{i,j}^{\mathsf{Rec}}(h)\right\}\Big\}$}, capturing a situation where a vehicle departs the VC or a recruiter completes its recruitment. Mathematically, we have
\begin{equation}\label{ephi1}
\begin{aligned}
\mathbb{E}\left[\widehat{\mathrm{W}}_{i,j}\right] =\boldsymbol{\int}_{0}^{\infty } \textrm{Pr}\Bigg[\min\Big\{\min_{i+1 \le r \le 2n}\left\{S_{i,j}^{\mathsf{Soj}}(r-i)\right\},\min_{\substack{h \in \{1,\cdots,i\}}}\left\{S_{i,j}^{\mathsf{Rec}}(h)\right\}\Big\}>t\Bigg] dt.
\end{aligned}
\end{equation}
Subsequently, we prove case 1 and case 2 as follows:

\textbullet\hspace{.5mm} \textbf{Case 1 (computation of $p_{i,j}^{k}$ when $j=0$ and $i\in \{0,\cdots,n\}$):}
According to \eqref{ephi1}, the expected sojourn time of D-SMP in state $S_{i,0}$, where {\small$i \in\{0,\cdots,n\}$}, is as follows:
\begin{equation}\label{ephi1_22}
\begin{aligned}
\mathbb{E}\left[\widehat{\mathrm{W}}_{i,0}\right] =\boldsymbol{\int}_{0}^{\infty } \textrm{Pr}\Bigg[\min\Big\{\min_{i+1 \le r \le 2n}\left\{\mathcal{L}_1^{(\dot{\mathfrak{X}})}\hspace{-1.4mm}\left\langle {2n{-}i}\right\rangle(r)\right\},\min_{\substack{1 \le h \le i}}\left\{\mathcal{H}_1^{(\varphi;\mathfrak{R})}\hspace{-1.4mm}\left\langle {i{-}1} \hspace{-0.3mm};\hspace{-0.3mm} 1 \right\rangle(h)\right\}\Big\}>t\Bigg] dt.
\end{aligned}
\end{equation}
Substituting \eqref{Us} and \eqref{Zs} in \eqref{ephi1_22}, and conditioning on $\dot{\mathfrak{X}}^{\mathsf{C}}_{r}$ and $\varphi^{\mathsf{C}}_{h}$, for all $i+1 \le r \le 2n$ and $1\le h \le i-1$, yield
\begin{align}\label{ephi1_2}
\mathbb{E}\left[\widehat{\mathrm{W}}_{i,0}\right] &=\boldsymbol{\int}_{0}^{\infty } \textrm{Pr}\Bigg[\min\Big\{\min\left[Z_{\ell_{i{+}1}}-Z_{\ell'_1},\cdots,Z_{\ell_{r}}-Z_{\ell'_1},\cdots,Z_{\ell_{2n}}-Z_{\ell'_1}\right]\nonumber\\
&~~~~~~~~~~~~~~~~~~~~~~,\min\left[U_{x_1}-Z_{\ell'_1}+Z_{\ell_1},\cdots,U_{x_h}-Z_{\ell'_1}+Z_{\ell_h},\cdots,U_{x_{i{-}1}}-Z_{\ell'_1}+Z_{\ell_{i{-}1}}, U_{x_i}\right]\Big\}>t\nonumber\\
&~~~~~~~~~~~~~~~~~~~~~~~~\Big|\left\{\dot{\mathfrak{X}}^{\mathsf{C}}_{r}|i+1 \le r \le 2n\right\},\left\{\varphi^{\mathsf{C}}_{h}|1 \le h \le i-1\right\}\Big] dt\nonumber\\
&=\boldsymbol{\int}_{0}^{\infty } \textrm{Pr}\Bigg[\min\Big[Z_{\ell_{i{+}1}}-Z_{\ell'_1},\cdots,Z_{\ell_{r}}-Z_{\ell'_1},\cdots,Z_{\ell_{2n}}-Z_{\ell'_1}\nonumber\\
&~~~~~~~~~~~~~~~~~~~~~~,U_{x_1}-Z_{\ell'_1}+Z_{\ell_1},\cdots,U_{x_h}-Z_{\ell'_1}+Z_{\ell_h},\cdots,U_{x_{i{-}1}}-Z_{\ell'_1}+Z_{\ell_{i{-}1}}, U_{x_i}\Big]>t\nonumber\\
&~~~~~~~~~~~~~~~~~~~~~~~~\Big|\left\{\dot{\mathfrak{X}}^{\mathsf{C}}_{r}|i+1 \le r \le 2n\right\},\left\{\varphi^{\mathsf{C}}_{h}|1 \le h \le i-1\right\}\Big] dt\nonumber\\
&=\boldsymbol{\int}_{0}^{\infty } \textrm{Pr}\Bigg[\min\Big[Z_{\ell_{i{+}1}}-Z_{\ell'_1},\cdots,Z_{\ell_{r}}-Z_{\ell'_1},\cdots,Z_{\ell_{2n}}-Z_{\ell'_1}\nonumber\\
&~~~~~~~~~~~~~~~~~~~~~~,U_{x_1}-Z_{\ell'_1}+Z_{\ell_1},\cdots,U_{x_h}-Z_{\ell'_1}+Z_{\ell_h},\cdots,U_{x_{i{-}1}}-Z_{\ell'_1}+Z_{\ell_{i{-}1}}\Big]>t\nonumber\\
&~~~~~~~~~~~~~~~~~~~~~~~~\Big|\left\{\dot{\mathfrak{X}}^{\mathsf{C}}_{r}|i+1 \le r \le 2n\right\},\left\{\varphi^{\mathsf{C}}_{h}|1 \le h \le i-1\right\}\Big] \textrm{Pr}\left[ U_{x_i}>t\right] dt\nonumber\\
&=\boldsymbol{\int}_{0}^{\infty } \textrm{Pr}\Bigg[\min\Big[Z_{\ell_{i{+}1}},\cdots,Z_{\ell_{2n}},U_{x_1}+Z_{\ell_1},U_{x_2}+Z_{\ell_2},\cdots,U_{x_{i{-}1}}+Z_{\ell_{i{-}1}}\Big]-Z_{\ell'_1}>t\nonumber\\
&~~~~~~~~~~~~~~~~~~~~~~~~\Big|\left\{\dot{\mathfrak{X}}^{\mathsf{C}}_{r}|i+1 \le r \le 2n\right\},\left\{\varphi^{\mathsf{C}}_{h}|1 \le h \le i-1\right\}\Big] \textrm{Pr}\left[ U_{x_i}>t\right] dt\nonumber\\
&=\boldsymbol{\int}_{0}^{\infty } \left(\prod_{r=i{+}1}^{2n} \textrm{Pr} \left[Z_{\ell_r}-Z_{\ell'_1}>t\big| \dot{\mathfrak{X}}^{\mathsf{C}}_{r}\right] \times \prod_{\substack{h=1}}^{i{-}1} \textrm{Pr}\left[U_{x_h}+Z_{\ell_h}-Z_{\ell'_{1}}>t\big|\varphi^{\mathsf{C}}_{h}\right]\times \textrm{Pr}\left[U_{x_i}>t\right]\right) dt\nonumber\\
&=\boldsymbol{\int}_{0}^{\infty } \Bigg(\Big(\textrm{Pr} \left[Z_{a}{-}Z_{b}{>}t\big|Z_{a}{>}Z_{b}\right]\Big)^{2n-i}   \Big(\textrm{Pr}\left[U_{c}{+}Z_{c}{-}Z_{b}{>}t\big| Z_{b}{>}Z_{c},U_{c}{>}Z_{b}{-}Z_{c}\right]\Big)^{i-1} \textrm{Pr}\left[ U_{d}>t\right]\Bigg) dt,
\end{align}
where $Z_{a}$, $Z_{b}$, $Z_{c}$, $U_{c}$, and $U_{d}$ are independent random variables, from which $Z_{a}$, $Z_{b}$, and $Z_{c}$ follow general distributions $\mathfrak{D}_Z(z)$, and $U_{c}$ and $U_{d}$ follow general distribution $\mathfrak{R}_U(u)$.

\textbullet\hspace{.5mm} \textbf{Case 2 (computation of $p_{i,j}^{k}$ when $1 \le i \le n-1$ and $1 \le j \le i+1$):}
According to \eqref{ephi1}, the expected sojourn time of D-SMP in state $S_{i,j}$, where $1 \le i \le n-1$ and $1 \le j \le i+1$, is as follows:
\begin{align}\label{ephi2_1}
\mathbb{E}\left[\widehat{\mathrm{W}}_{i,j}\right] &=\boldsymbol{\int}_{0}^{\infty } \textrm{Pr}\Bigg[\min\Bigg\{\min_{i+1 \le r \le 2n}\left\{\mathcal{L}_2^{(\ddot{\mathfrak{X}})}\hspace{-1.4mm}\left\langle 2n{-}i\right\rangle(r-i)\right\},\min_{\substack{1\le h\le j-1}}\left\{\mathcal{H}_2^{(\gamma;\psi)}\hspace{-1.4mm}\left\langle j{-}1 \hspace{-0.3mm};\hspace{-0.3mm} i{-}(j{-}1) \right\rangle(h)\right\}\nonumber\\
&~~~~~~~~~~~~~~~~~~~~~~~~~~~~~~~~~~~~~~~~,\min_{\substack{j\le g \le i}}\left\{\mathcal{H}_2^{(\gamma;\psi)}\hspace{-1.4mm}\left\langle j{-}1 \hspace{-0.3mm};\hspace{-0.3mm} i{-}(j{-}1) \right\rangle(g)\right\}\Bigg\}>t\Bigg] dt.
\end{align}
Substituting \eqref{Us2} and \eqref{Zs2} in \eqref{ephi2_1}, and conditioning on $\ddot{\mathfrak{X}}^{\mathsf{C}}_{r}$, $\gamma^{\mathsf{C}}_{h}$, and $\psi^{\mathsf{C}}_{g}$, for all $i+1 \le r \le 2n$, $1 \le h \le j-1$, and $j \le g \le i$, result in
\begin{align}\label{ephi2_2}
\mathbb{E}\left[\widehat{\mathrm{W}}_{i,j}\right] &=\boldsymbol{\int}_{0}^{\infty } \textrm{Pr}\Bigg[\min\Big[Z_{\ell_{i{+}1}}-U_{x'_1}-Z_{\ell'_1},\cdots,Z_{\ell_{r}}-U_{x'_1}-Z_{\ell'_1},\cdots,Z_{\ell_{2n}}-U_{x'_1}-Z_{\ell'_1}\nonumber\\
&~~~~~~~~~~~~,(U_{x_1}-Z_{\ell'_1}+Z_{\ell_{1}})-U_{x'_1},\cdots,(U_{x_h}-Z_{\ell'_1}+Z_{\ell_h})-U_{x'_1},\cdots,(U_{x_{j-1}}-Z_{\ell'_{1}}+Z_{\ell_{j-1}})-U_{x'_{1}}\nonumber\\
&~~~~~~~~~~~~,U_{x_j}-(U_{x'_1}-Z_{\ell_j}+Z_{\ell'_1}),\cdots,U_{x_g}-(U_{x'_1}-Z_{\ell_g}+Z_{\ell'_1}),\cdots,U_{x_i}-(U_{x'_{1}}-Z_{\ell_i}+Z_{\ell'_{1}})\Big]>t\nonumber\\
&~~~~~~~~~~~~~~~~~~~~~~~~\Big|\left\{\ddot{\mathfrak{X}}^{\mathsf{C}}_{r}|i+1 \le r \le 2n\right\},\left\{\gamma^{\mathsf{C}}_{h}|1 \le h \le j-1\right\},\left\{\psi^{\mathsf{C}}_{g}|j \le g \le i\right\}\Big] dt\nonumber\\
&=\boldsymbol{\int}_{0}^{\infty } \textrm{Pr}\Bigg[\min\Big[Z_{\ell_{i{+}1}},\cdots,Z_{\ell_{r}},\cdots,Z_{\ell_{2n}},U_{x_1}+Z_{\ell_{1}},\cdots,U_{x_h}+Z_{\ell_{h}},\cdots,U_{x_{j-1}}+Z_{\ell_{j-1}}\nonumber\\
&~~~~~~~~~~~~~~~~~~~~~~~~~~~~~,U_{x_j}+Z_{\ell_j},\cdots,U_{x_g}+Z_{\ell_g},\cdots,U_{x_i}+Z_{\ell_i}\Big]-U_{x'_1}-Z_{\ell'_1}>t\nonumber\\
&~~~~~~~~~~~~~~~~~~~~~~~~\Big|\left\{\ddot{\mathfrak{X}}^{\mathsf{C}}_{r}|i+1 \le r \le 2n\right\},\left\{\gamma^{\mathsf{C}}_{h}|1 \le h \le j-1\right\},\left\{\psi^{\mathsf{C}}_{g}|j \le g \le i\right\}\Big] dt\nonumber\\
&=\boldsymbol{\int}_{0}^{\infty } \Bigg(\prod_{r=i{+}1}^{2n} \textrm{Pr} \left[Z_{\ell_r}-U_{x'_1}-Z_{\ell'_1}>t\big|\ddot{\mathfrak{X}}^{\mathsf{C}}_{r}\right] \prod_{\substack{h=1}}^{j{-}1} \textrm{Pr}\left[(U_{x_h}-Z_{\ell'_1}+Z_{\ell_{h}})-U_{x'_1}>t\big|\gamma^{\mathsf{C}}_{h}\right]\nonumber\\
&~~~~~~~~~~~~~~~~~~~~~~~~~~~~~~~~~~~~~~~~~~~~~~~~~~~~~~~\times \prod_{\substack{g=j}}^{i} \textrm{Pr}\left[U_{x_g}-(U_{x'_1}-Z_{\ell_g}+Z_{\ell'_1})>t\big|\psi^{\mathsf{C}}_{g}\right]\Bigg) dt\nonumber\\
&=\boldsymbol{\int}_{0}^{\infty } \Bigg(\Big(\textrm{Pr} \left[Z_{a}-U_{b}-Z_{b}>t\Big|Z_{a}>U_{b}+Z_{b}\right]\Big)^{2n-i}\nonumber\\
&~~~~~~~~~~~~~\times \Big(\textrm{Pr}\left[(U_{c}{-}Z_{b}{+}Z_{c}){-}U_{b}>t\Big|Z_{b}>Z_{c},U_{c}>Z_{b}{-}Z_{c},(U_{c}{-}Z_{b}{+}Z_{c})>U_{b}\right]\Big)^{j{-}1}\nonumber\\
&~~~~~~~~~~~~~\times \Big(\textrm{Pr}\left[U_{d}{-}(U_{b}{-}Z_{d}{+}Z_{b})>t\Big| Z_{d}>Z_{b}, U_{b}>Z_{d}{-}Z_{b}, U_{d}>(U_{b}{-}Z_{d}{+}Z_{b})\right]\Big)^{i{-}(j{-}1)}\Bigg) dt,
\end{align}
where $Z_{a}$, $Z_{b}$, $Z_{c}$, $Z_{d}$, $U_{b}$, $U_{c}$, and $U_{d}$ are independent random variables, from which $Z_{a}$, $Z_{b}$, $Z_{c}$, and $Z_{d}$ follow general distributions $\mathfrak{D}_Z(z)$, and $U_{b}$, $U_{c}$, and $U_{d}$ follow general distribution $\mathfrak{R}_U(u)$.

Inspecting the results in \eqref{ephi1_2} and \eqref{ephi2_2} verifies the results in \eqref{expected_value1} and \eqref{expected_value2}. This concludes the proof of the theorem.
\end{proof}

\newpage

\section{Case Study of {\tt RP-VC$_n$}}\label{app:case_study}
In this appendix, we aim to clarify the following critical points of {\tt RP-VC$_n$} framework during the execution of a DAG application.
\begin{enumerate}[label={(P-\arabic*)}]
    \item \label{p1} The dynamics of the system, i.e., the transition between different states of decomposed semi-Markov process (D-SMP), discussed in Sec \ref{DSMPModel}, due to departure and recruitment events of vehicles.
    \item \label{p2} The time-dependent residual sojourn time and recruitment duration of vehicles.
\end{enumerate}

Real-world examples of DAG tasks are face recognition task, antivirus task, chess game task, and modified molecular dynamic code \cite{10100891,993206}. As an example, we plot the DAG structure of modified molecular dynamic code in Fig. \ref{fig:dag_real}.

To describe the two above-mentioned points of  {\tt RP-VC$_n$} through a case study, for the sake of simplicity of explanations, we consider a simple application $\mathcal{A}\hspace{-1mm}=\hspace{-1mm}\big(\mathcal{V},\mathcal{E}\big)$ with five dependent sub-tasks, where  {\small$\mathcal{V}\hspace{-1mm}=\hspace{-1mm}\{T_1,T_2,\dotsc,T_5\}$} and {\small$\mathcal{E}\hspace{-1mm}=\hspace{-1mm}\{(1,2),(1,3),(2,4),(3,5)\}$}. The DAG representation of application $\mathcal{A}$ is depicted in Fig. \ref{fig:dag1}.



Let {\small$\mathcal{A}$} be grouped by partition {\small$\mathcal{P}(\mathcal{A})\hspace{-1mm}=\hspace{-1mm}\big(\mathcal{G},\mathcal{E}\big)$} into three groups {\small$\mathcal{G}\hspace{-1mm}=\hspace{-1mm}\{G_1,G_2,G_3\}$}, where {\small$G_1\hspace{-1mm}=\hspace{-1mm}\{T_1,T_2\}$}, {\small$G_2\hspace{-1mm}=\hspace{-1mm}\{T_3,T_5\}$}, and {\small$G_3\hspace{-1mm}=\hspace{-1mm}\{T_4\}$}. Also, let $\mathcal{C}=\{C_1, C_2,\cdots\}$ refer to the set of all vehicles in VC at the start of the processing of application $\mathcal{A}$. Moreover, let $\mathcal{Z}{=}\{Z_1,Z_2,\cdots\}$ and $\mathcal{U}{=}\{U_1,U_2,\cdots\}$ be two sets of i.i.d random variables following two different general distributions $\mathfrak{D}_Z(z)$ and $\mathfrak{R}_U(u)$, respectively, where $Z_\ell\in\mathcal{Z}$ captures sojourn time of vehicle $C_\ell\in\mathcal{C}$, and $U_x \in \mathcal{U}$ refers to the recruitment duration of $C_x\in\mathcal{C}$. Furthermore, we refer to the \textit{residual} sojourn time of $C_\ell\in\mathcal{C}$ by $Z'_\ell$, and the \textit{residual} recruitment duration of vehicle $C_x\in\mathcal{C}$ by $U'_x$.
\begin{figure}[h]
\centering
\includegraphics[width=0.5\linewidth,trim=1 1 1 1,clip]{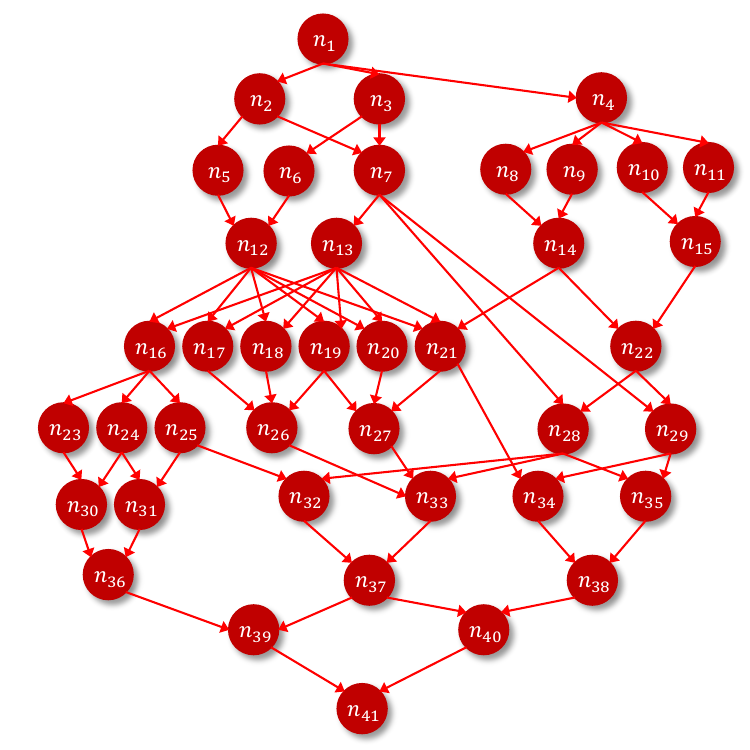}
\caption{The DAG of the molecular dynamics code \cite{10100891,993206}.}
\label{fig:dag_real}
\vspace{-4mm}
\end{figure}
\begin{figure}[h]
\centering
\includegraphics[width=0.3\linewidth,trim=1 1 1 1,clip]{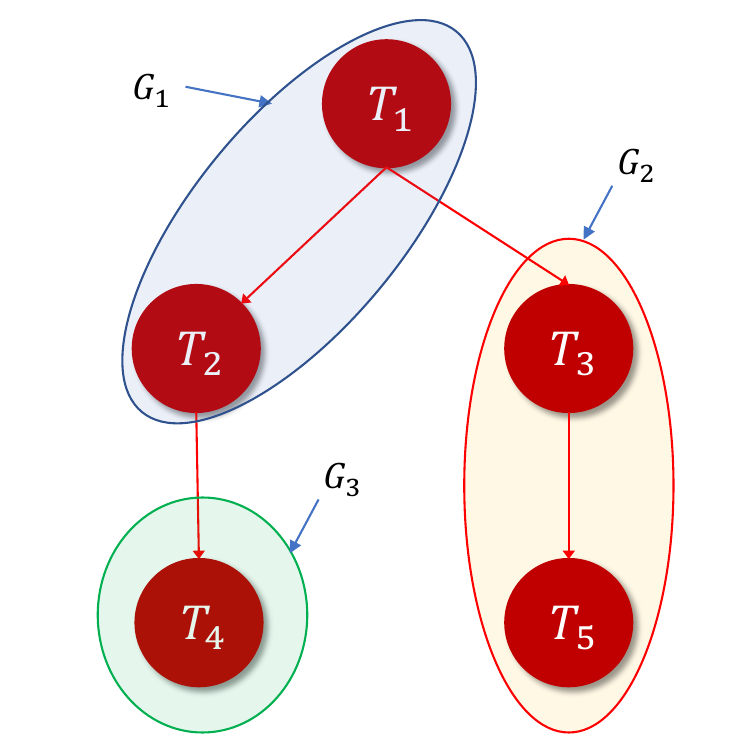}
\caption{\hspace{-0.8mm} The DAG representation of application $\mathcal{A}\hspace{-1mm}=\hspace{-1mm}\big(\mathcal{V},\mathcal{E}\big)$.}
\label{fig:dag1}
\vspace{-4mm}
\end{figure}
Fig. \ref{fig:exam} illustrates the dynamics of {\tt RP-VC$_n$} for $n=3$ (i.e., having three groups of sub-tasks) during the execution of application $\mathcal{A}$. The departure and recruitment events of vehicles, time instants, and the order of triggering events (depicted by dotted directed lines), are shown on the left side of the figure. For a better clarification, the figure depicts the active vehicles for each group (i.e., the vehicles that are processing a sub-task) at different time instants (under the curly bracket titled ``Deployment"). Furthermore, the variation of residual sojourn times and recruitment duration of vehicles over time are shown in the middle of the figure under the title of ``Residual times of random variables".

The figure also shows the relationship between $\beta$-inhomogeneous SMP ($\beta$-SMP), defined in Sec. \ref{BSMPR2n} and shown in Fig. \ref{j2nMarkov1} in the main text, and decomposed SMP (D-SMP), illustrated in Sec. \ref{DSMPModel}. Specifically, it is shown that each state of $\beta$-SMP (i.e., the right most column) is decomposed into several sub-states in D-SMP (i.e., the middle columns). For instance, state $S_1$ of $\beta$-SMP is decomposed into three sub-states $S_{1,0}$, $S_{1,1}$, and $S_{1,2}$. In addition, the number of vehicles of sets {\small$\dot{\mathcal{C}}_i$} and $\ddot{\mathcal{C}}_{2n-2i}$, where $n=3$ and $i\in\{0,1,2,3\}$, are also depicted.

Consider the column under the curly bracket ``Vehicles" in the figure. Recall from the main text that {\small$\dot{\mathcal{C}}_i$} denotes the set of $i$ vehicles processing $i$ different groups, where one of the vehicles of each group has departed the VC (i.e., each group is being processed by only one vehicle). Therefore, each vehicle {\small$C_\ell\in \dot{\mathcal{C}}_i$} is a recruiter vehicle. Also,  $\ddot{\mathcal{C}}_{2n-2i}$ denotes the set of $2n-2i$ vehicles processing $n-i$ different groups, where each group is being processed by two vehicles.

Consider the column under the curly bracket ``Deployment" in the figure. Initially, {\tt RP-VC$_3$} deploys $G_1$ on vehicles $\{C_0,C_1\}$, $G_2$ on vehicles $\{C_2, C_3\}$, and $G_3$ on vehicles $\{C_4,C_5\}$. Due to the departure and recruitment events of vehicles, D-SMP visits each state multiple times. In the following, we provide an in-depth study of the dynamic of {\tt RP-VC$_3$} for five representative time instants.

\textbullet\hspace{.5mm} \textbf{System's Initial State.} As can be seen from the figure, until time instant $t_1$, the process is in state $S_{0,0}$ of D-SMP (which is a sub-state of state $S_0$ of $\beta$-SMP), referring to the initial state of the process. In state $S_{0,0}$, residual sojourn times of vehicles are equivalent to their initial values (i.e., $Z'_\ell=Z_\ell,~ \forall C_\ell\in \ddot{\mathcal{C}}_{6}$). Note that in this state, there are no recruiter vehicles (i.e., $\dot{\mathcal{C}}_0=\emptyset$) since all of the groups of application $\mathcal{A}$ are being processed by two vehicles.

\textbullet\hspace{.5mm} \textbf{Time instant $\bm{t_1}$.} Assume that, at time $t_1$, event $e^{\mathsf{D}}(C_0)$ occurs and vehicle $C_0$ of group $G_1$ departs the VC. Consequently, D-SMP transits to state $S_{1,0}$ (which is a sub-state of state $S_1$ of $\beta$-SMP) and visits it for the first time. In this state, $G_2$ and $G_3$ are being processed by two vehicles, while $G_1$ is being processed only by one vehicle. Moreover, in state $S_{1,0}$, vehicle $C_1$, which is the only remaining  vehicle of group $G_1$, starts recruiting a new vehicle, the duration of which is $U_1$. Further, $Z_0$ should be subtracted from the sojourn times of other vehicles (i.e., $Z'_\ell=Z_\ell-Z_0$, where $Z_\ell>Z_0$) because $Z_{0}$ time units have passed since the start of processing of $\mathcal{A}$.

\textbullet\hspace{.5mm} \textbf{Time instant $\bm{t_2}$.} Assume that, at time $t_2$, $C_1$ completes its recruitment operation after $U_1$ units of time and recruits $C_6$ for group $G_1$. Therefore, D-SMP returns to state $S_{0,1}$ (which is also a sub-state of state $S_0$ of $\beta$-SMP). Moreover, in state $S_{0,1}$, $U_1$ is subtracted from residual times of random variables since $U_{1}$ units of time have passed since $t_1$, which results in $Z'_\ell=Z_\ell-U_1-Z_0$.

\textbullet\hspace{.5mm} \textbf{Time instant $\bm{t_3}$.} Assume that, at time $t_3$, event $e^\mathsf{D}(C_1)$ occurs and vehicle $C_1$ of group $G_1$ departs the VC. Therefore, D-SMP transits to $S_{1,0}$ again and visits it for the second time. In this state, $C_6$, which is the only vehicle processing group $G_1$, starts recruiting a new vehicle. Recall that the duration of recruitment operation of $C_6$ is captured by random variable $U_6$. In this visit of $S_{1,0}$, $Z_1-U_1-Z_0$ is subtracted from the \textit{residual} sojourn time of other vehicles since $Z_1-U_1-Z_0$ units of time have passed from $t_2$ leading to $Z'_\ell=Z_\ell-Z_1$.

\textbullet\hspace{.5mm} \textbf{Time instant $\bm{t_4}$.} Assume that, at time $t_4$, vehicle $C_2$ associated with group $G_2$ departs the VC and D-SMP transits to state $S_{2,0}$ (which is a sub-state of state $S_2$ of $\beta$-SMP) and visits this state for the first time. In this situation, the remaining vehicle of $G_2$, i.e., $C_3$, starts recruiting a new vehicle. Note that, in state $S_{2,0}$, two recruiters, i.e., $C_6$ and $C_3$, are simultaneously recruiting two different vehicles for two different groups $G_1$ and $G_2$, respectively. In this state, to calculate residual times, $Z_2-Z_1$ should be subtracted from residual sojourn time and recruitment duration of other vehicles since $Z_2-Z_1$ units of time have passed since $t_3$. Consequently, residual sojourn times of vehicles are $Z'_\ell=Z_\ell-Z_1-(Z_2-Z_1)=Z_\ell-Z_2,~ \forall C_\ell\in \ddot{\mathcal{C}}_{2}\cup \dot{\mathcal{C}}_{2}$, where $Z_\ell>Z_2$. Similarly, residual recruitment duration of $U'_6=U_6-Z_2+Z_1$ and $U'_3=U_3$, where $Z_2>Z_1$, and $U_6>(Z_2-Z_1)$.

\textbullet\hspace{.5mm} \textbf{Time instant $\bm{t_5}$.} Assume that, at time $t_5$, $C_3$ completes its recruitment operation by recruiting $C_7$ and thus D-SMP transits to state $S_{1,1}$ (which is also a sub-state of state $S_1$ of $\beta$-SMP). In this state, the residual times of random variables are $Z'_\ell=Z_\ell-U_3-Z_2$ and $U'_6=(U_6-Z_2+Z_1)-U_3$, where $Z_\ell>U_3+Z_2$, $U_6>Z_2-Z_1$, and $(U_6-Z_2+Z_1)>U_3$.

It can be seen that how the decomposed SMP (D-SMP) can capture the precise state of the system during the processing of a computation-intensive application. The transition between different states and residual times of random variables form a non-trivial process, which is characterized through Decomposition Theorem (Theorem \ref{decompositionTheorem}) in the main text. Furthermore, it can be construed that the first time that all vehicles that are processing the same group leave the VC a failure will happen, which is characterized by MTTF of the system obtained in Theorem \ref{MTTFJ2dn}.

\begin{figure*}
\centering
\includegraphics[width=\linewidth,trim=1 1 1 1,clip]{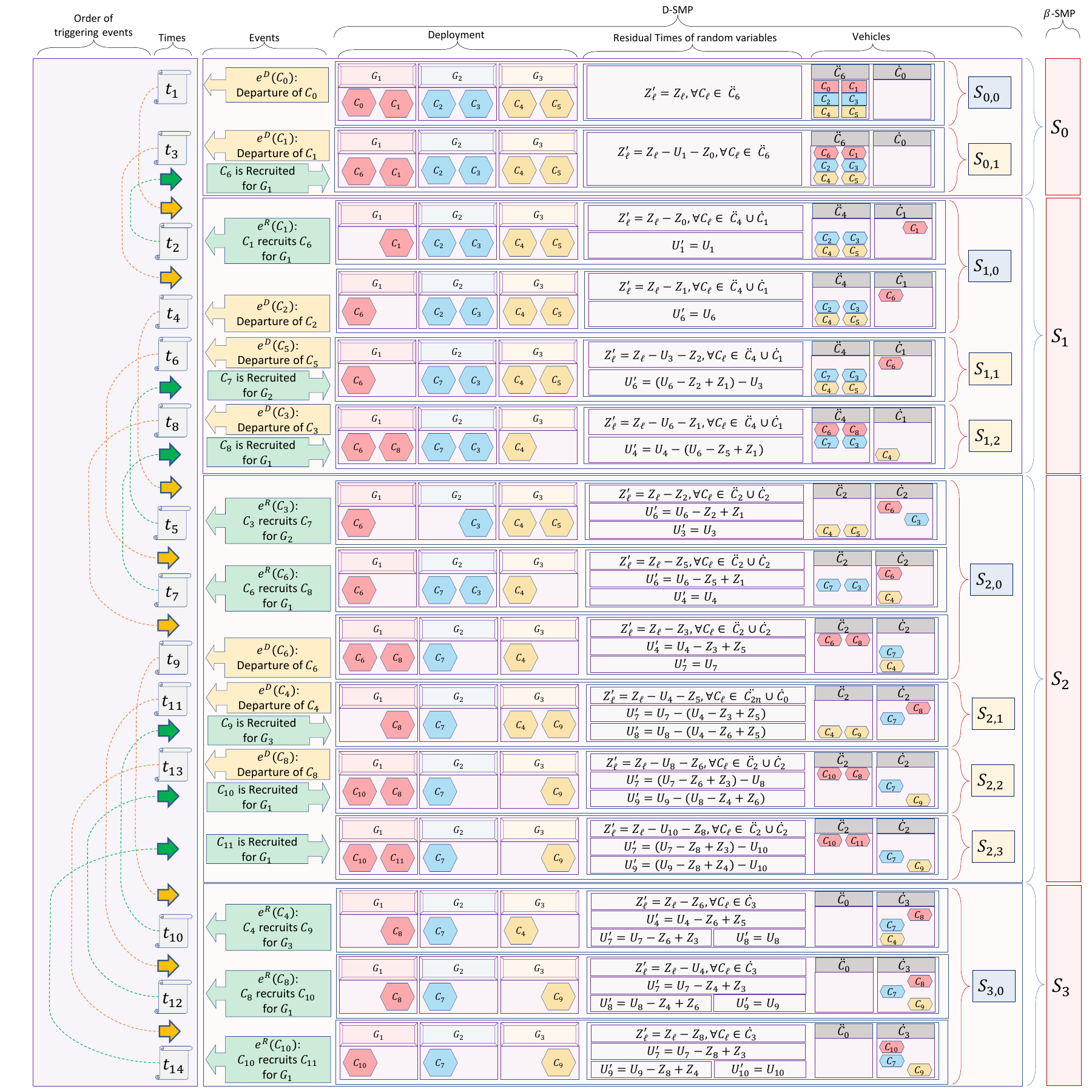}
\caption{A schematic of the dynamics of process of {\tt RP-VC$_n$}.}
\label{fig:exam}
\vspace{-4mm}
\end{figure*}

\end{document}